\documentclass[journal,onecolumn]{IEEEtran}

\usepackage{graphicx}
\usepackage{amssymb}
\usepackage{epsfig}
\usepackage{amsmath}
\usepackage{amsthm}
\usepackage{amsfonts}
\usepackage{setspace}
\usepackage{url}
\usepackage{ifthen}
\usepackage[dvips]{color}
\usepackage[active]{srcltx}
\usepackage{epstopdf}

\newtheorem{theorem}{Theorem}
\newtheorem{cor}{Corollary}
\newtheorem{lem}{Lemma}

\newcommand{\underb}{\underbrace}

\newcommand{\udl}{\underline}
\newcommand{\lf}{\left}
\newcommand{\ri}{\right}
\newcommand{\lv}{\lVert}
\newcommand{\rv}{\rVert}
\newcommand{\ol}{\overline}
\newcommand{\mrm}{\mathrm}
\newcommand{\cube}{\mathrm{cube}}
\newcommand{\eff}{\mathrm{eff}}
\newcommand{\rc}{\mathrm{rc}}
\newcommand{\tr}{\mathrm{tr}}
\newcommand{\ex}{\mathrm{ex}}
\newcommand{\beq}{\begin{equation}}
\newcommand{\eeq}{\end{equation}}
\newcommand{\SNR}{\mathsf{SNR}}
\newcommand{\CN}{\mathcal{CN}}
\newcommand{\RD}{\mathbf{D}}
\newcommand{\DT}{D\cdot T}
\newcommand{\bal}[2]{{\setlength\arraycolsep{2pt}\begin{eqnarray}\label{#1} #2 \end{eqnarray}}}
\newcommand{\baln}[2]{{\setlength\arraycolsep{2pt}\begin{eqnarray*}\label{#1} #2 \end{eqnarray*}}}
\newcommand{\bals}[2]{{\setlength\arraycolsep{1pt}\begin{eqnarray}\label{#1} #2 \end{eqnarray}}}

\newcommand{\beqn}[2]{\begin{equation*}\label{#1} {#2} \end{equation*}}
\newcommand{\singlecolumntype}{1}
\newcommand{\nn}{\nonumber}
\newcommand{\con}[2]{\ifthenelse{\equal{\singlecolumntype}{1}}{#1}{#2}}

\providecommand{\ICs}{IC's}

\begin{document}

%
\title{On the Diversity-Multiplexing Tradeoff of Unconstrained Multiple-Access Channels}

\author{Yair Yona, and Meir Feder,~\IEEEmembership{Fellow,~IEEE}
\thanks{The material in
this paper was presented in part at the Information Theory and
Applications Workshop (ITA), 2014.}
\thanks{This research was supported by THE ISRAEL SCIENCE FOUNDATION
grant No. 634/09.}
\thanks{Y. Yona was with the Department of Electrical Engineering --
Systems, Tel-Aviv University, Ramat-Aviv 69978, Israel. He is now with the department of Electrical Engineering, University of California Los Angeles, Los Angeles, CA 90095, USA (email: yairyo99@ucla.edu).

M. Feder is with the Department of Electrical Engineering --
Systems, Tel-Aviv University, Ramat-Aviv 69978, Israel (e-mail:
meir@eng.tau.ac.il).}}

\maketitle

\begin{abstract}
In this work the optimal diversity-multiplexing tradeoff (DMT) is
investigated for the multiple-input multiple-output fading
multiple-access channel with no power constraints (infinite
constellations). For $K$ users ($K>1$), $M$ transmit antennas for
each user, and $N$ receive antennas, infinite constellations in
general and lattices in particular are shown to attain the optimal
DMT of finite constellations for $N\ge (K+1) M-1$, i.e., user limited
regime. On the other hand for $N<\lf(K+1\ri)M-1$ it is shown that
infinite constellations can not attain the optimal DMT. This is in
contrast to the point-to-point case in which infinite constellations
are DMT optimal for any $M$ and $N$. In general, this work shows
that when the network is heavily loaded, i.e.,
$K>\max\lf(1,\frac{N-M+1}{M}\ri)$, taking into account the shaping
region in the decoding process plays a crucial role in pursuing the
optimal DMT. By investigating the cases in which infinite
constellations are optimal and suboptimal, this work also gives a
geometrical interpretation to the DMT of infinite constellations in
multiple-access channels.
\end{abstract}

\IEEEpeerreviewmaketitle

\section{Introduction}
Employing multiple  antennas in a point-to-point wireless channel
increases the number of degrees of freedom available for
transmission. This is illustrated for the ergodic case in
\cite{TelatarCapacityMIMO},\cite{FoschiniCapacityMIMO}, where $M$
transmit and $N$ receive antennas increase the capacity by a factor
of $\min\lf(M,N\ri)$. The number of degrees of freedom utilized by
the transmission scheme is referred to as \emph{multiplexing gain}.
Another advantage of employing multiple antennas is the potential
increase in the transmitted signal reliability. The fact that
multiple antennas increase the number of independent links between
antenna pairs, enables the error probability to decrease, i.e., add
diversity. If for high signal to noise ratio ($\SNR$) the error
probability is proportional to $\SNR^{-d}$, then we state that the
\emph{diversity order} is $d$.

For the point-to-point setting, Zheng and Tse \cite{TseDivMult2003}
characterized the optimal diversity-multiplexing tradeoff (DMT) of
the quasi-static Rayleigh flat-fading channel, i.e., for each
multiplexing gain they found the best attainable diversity order.
The optimal DMT is a piecewise linear function connecting the points
$\lf(M-l\ri)\lf(N-l\ri)$, $l=0,\dots,\min\lf( M,N\ri)$. The
transmission scheme in \cite{TseDivMult2003} uses random codes.
Subsequent works presented more structured schemes that attain the
optimal DMT. El Gamal et al. \cite{ElGamalLAST2004} showed by using
probabilistic methods that lattice space-time (LAST) codes attain
the optimal DMT by using minimum-mean square error (MMSE) estimation
followed by lattice decoding. Later, explicit coding schemes based
on lattices and cyclic-division algebra \cite{EliaExplicitSTC},
\cite{BelfiorePerfectCodes} were shown to attain the optimal DMT by
using maximum-likelihood (ML) decoding, and also by using MMSE
estimation followed by lattice decoding
\cite{EliaJaldenDMTOptLRLinearLaticeDecoding}. A subtle but very
important point is that these coding schemes take into consideration
the finiteness of the codebook in the decoder. A question that
remained open was whether lattices can achieve the optimal DMT by
using \emph{regular} lattice decoding, i.e., decoder that takes into
account the infinite lattice without considering the shaping region
or the power constraint. In order to answer this question, the work
in \cite{YonaFederICOptimalDMT} presented an analysis of the
performance of infinite constellations (\ICs) in multiple-input
multiple-output (MIMO) fading channels. A new tradeoff was presented
between the \ICs\ average number of dimensions per channel use, i.e.,
the IC dimensionality divided by the number of channel uses, and the
best attainable DMT. By choosing the right average number of
dimensions per channel use, it was shown
\cite{YonaFederICOptimalDMT} that IC's in general and more
specifically lattices using regular lattice decoding, attain the
optimal DMT of finite constellations.

For the multiple-access channel, where a number of users transmit to
a single receiver, the number of users in the network affects the
multiplexing gain and the diversity order. For instance, for a
network with $K$ users transmitting at the same rate, the number of
available degrees of freedom for each user is
$\min\lf(M,\frac{N}{K}\ri)$. Tse, Viswanath and Zheng
\cite{ZhengTseMACDMT2004} characterized the optimal DMT of a network
with $K$ users, where each user has $M$ transmit antennas and the
receiver has $N$ antennas. For the symmetric case, in which the
users transmit at the same multiplexing gain $r$, i.e.,
$r_{1}=\dots=r_{K}=r$, the optimal DMT takes the following elegant
form \cite{ZhengTseMACDMT2004}:
\begin{itemize}
\item For $r\in\lf[0,\min\lf(\frac{N}{K+1},M\ri)\ri]$ the optimal symmetric
DMT equals to the optimal DMT of a point-to-point channel with $M$
transmit and $N$ receive antennas
$d^{\ast,\lf(FC\ri)}_{M,N}\lf(r\ri)$.
\item For $r\in\lf[\min\lf(\frac{N}{K+1},M\ri),\min\lf(M,\frac{N}{K}\ri)\ri]$ the optimal
symmetric DMT equals to the optimal DMT of a point-to-point channel
with all $K$ users pulled together $d^{\ast,\lf(FC\ri)}_{K\cdot
M,N}\lf(Kr\ri)$.
\end{itemize}
Similar to the development in the point-to-point case, random codes
were used in \cite{ZhengTseMACDMT2004}. Later Nam and El Gamal
\cite{NamElGamalOptimalDMT2007} showed that a random ensemble of
LAST codes attains the optimal DMT of the multiple-access channel
using MMSE estimation followed by lattice decoding over the lattice
induced by the $K$ users. An explicit coding scheme based on
lattices and cyclic division algebra that attains the optimal DMT
using ML decoding was presented in \cite{LUDMTOptimalCodesMAC2011}.

In this paper we study the optimal DMT of lattices using regular
lattice decoding, i.e., decoding without taking into consideration
the power constraint, for the MIMO Rayleigh fading multiple-access
channel. The result is rather surprising; unlike the point-to-point
case in which the tradeoff between dimensions and diversity enables
to attain the optimal DMT, we show that for the multiple-access
channel the optimal DMT is attained only for $N\ge \lf(K+1\ri)M-1$,
i.e., user limited regime. On the other hand when the network is
heavily loaded we show that IC's or lattices using regular lattice
decoding, can not attain the optimal DMT.

In the first part of this paper an upper bound on the optimal
\emph{symmetric} DMT IC's can achieve is derived. The upper bound is
attained by finding for each multiplexing gain $r$, the average
number of dimensions per channel use for each user, that maximizes
the diversity order. In the case $N<\lf(K+1\ri)M-1$ it is shown that
the optimal DMT of IC's does not coincide with the optimal DMT of
finite constellations. Moreover, for $N<\lf(K-1\ri)M+1$ it is shown
that the optimal DMT of IC's in the symmetric case is inferior
compared to the optimal DMT of finite constellations, for any value
of $r$ except for the edges $r=0$, $\frac{N}{K}$. On the other hand
for $N\ge (K+1)M-1$, by choosing the correct average number of
dimensions per channel use for each user, it is shown that the upper
bound on the optimal DMT of IC's coincides with the optimal DMT of
finite constellations
$d^{\ast,\lf(FC\ri)}_{M,N}\lf(\max\lf(r_{1},\dots,r_{K}\ri)\ri)$.

In the second part of this paper, a transmission scheme that attains
the optimal DMT for $N\ge (K+1)M-1$ is presented. Each user in this
scheme transmits according to the DMT optimal scheme for the
point-to-point channel, presented in \cite{YonaFederICOptimalDMT}.
By analyzing the receiver joint ML decoding performance, it is shown
that this transmission scheme attains the optimal DMT of finite
constellations. We wish to emphasize that the proposed transmission
scheme is more involved than simply using orthogonalization between
users, which in general is shown to be suboptimal for IC's. The
proposed transmission scheme requires $N+M-1$ channel uses to attain
the optimal DMT,  which is smaller than $N+KM-1$, the number of
channel uses required in \cite{ZhengTseMACDMT2004} (the dependence
in the number of users lies in the fact that $N\ge \lf(K+1\ri)M-1$).
Finally, the algebraic analysis of the transmission scheme
geometrically explains why for $N\ge (K+1)M-1$ the optimal DMT
equals to the optimal DMT of the point-to-point channel of each
user, i.e., why the optimal DMT equals
$d^{\ast,\lf(FC\ri)}_{M,N}\lf(\max\lf(r_{1},\dots,r_{K}\ri)\ri)$.

As a basic illustrative example for the results we consider the
following two cases. For the first case assume a network with two
users ($K=2$), where each user has a single transmit antenna
($M=1$), and a receiver with a single receive antenna ($N=1$). In
this case the optimal DMT of finite constellations in the symmetric
case \cite{ZhengTseMACDMT2004} equals $1-r$ for
$r\in\lf[0,\frac{1}{3}\ri]$, and $2-4r$ for
$r\in\lf[\frac{1}{3},\frac{1}{2}\ri]$. For IC's it is shown in this
setting that the optimal DMT for the symmetric case equals $1-2r$
for $r\in\lf[0,\frac{1}{2}\ri]$, which is strictly inferior except
for $r=0$, $\frac{1}{2}$. In the second case, by merely adding
another receive antenna, i.e., $M=1$, $N=K=2$, the optimal DMT of
IC's coincides with finite constellations optimal DMT
$d^{\ast,\lf(FC\ri)}_{1,2}\lf(\max\lf(r_{1},r_{2}\ri)\ri)$.

It is important to note that for $N<\lf(K+1\ri)M-1$ this paper shows
the sub-optimality of IC's compared to the optimal DMT of finite
constellations. However, in this case an explicit analytical
expression for the upper bound on the optimal DMT of IC's is given
only for the symmetric case, whereas for the general case the upper
bound is presented in the form of optimization problem. Indeed, for
$N<\lf(K+1\ri)M-1$ it still remains an open problem to find an
explicit expression for the general upper bound (the non-symmetric
case) on the optimal DMT of IC's, together with a transmission
scheme that achieves it. On the other hand, when $N\ge
\lf(K+1\ri)M-1$ this paper provides both analytical upper bound to
the optimal DMT of IC's, and also a transmission scheme that attains
it.

The outline of the paper is as follows. In section
\ref{sec:BasicDefinitions} basic definitions for the fading
multiple-access channel and IC's are given. Section
\ref{sec:LowerBoundErrorProb} presents an upper bound on the optimal
DMT of IC's, and shows the sub-optimality of IC's for
$N<\lf(K+1\ri)M-1$. Transmission scheme that attains the optimal DMT
of finite constellations for $N\ge \lf(K+1\ri)M-1$ is presented in
section \ref{sec:AttainingtheOptimalDMT}. Finally, in section
\ref{sec:Discussionsec} we discuss the results in this paper and
present for the multiple-access channel a geometrical interpretation
to the DMT of IC's.

\section{Basic Definitions}\label{sec:BasicDefinitions}
\subsection{Channel Model}
We consider a $K$-user multiple access channel for which each user
has $M$ transmit antennas, and the receiver has $N$ antennas. We
assume perfect knowledge of all channels at the receiver, and no
channel knowledge at the transmitters.  We also assume quasi static
flat-fading channel for each user. The channel model is as follows:
\begin{equation} \label{eq:Channel Fading}
\underline{y}_{t}=\sum_{i=1}^{K}H^{\lf(i\ri)}\cdot\udl{x}_{t}^{\lf(i\ri)}+\rho^{-\frac{1}{2}}\underline{n}_{t}\qquad
t=1,\dots, T
\end{equation}
where $\udl{x}_{t}^{\lf(i\ri)}$, $t=1,\dots,T$ is user $i$
transmitted signal, $\underline{n}_{t}\sim
\CN(\underline{0},\frac{2}{2\pi e}I_{N})$ is the additive noise for
which $\CN$ denotes complex-normal, $I_{N}$ is the $N$-dimensional
unit matrix, and $\udl{y}_{t}\in\mathbb{C}^{N}$. $H^{\lf(i\ri)}$ is
the fading matrix of user $i$. It consists of $N$ rows and $M$
columns, where $h^{\lf(i\ri)}_{l,j}\sim \CN(0,1)$, $1\le l \le N$,
$1\le j \le M$, are the entries of $H^{\lf(i\ri)}$. The scalar
$\rho^{-\frac{1}{2}}$ multiplies each element of
$\underline{n}_{t}$, where $\rho$ can be interpreted as the average
$\SNR$ of each user at the receive antennas for power constrained
constellations that satisfy $\frac{1}{T} \sum_{t=1}^{T}
E\{\lVert\udl{x}^{\lf(i\ri)}_{t}\rVert^{2}\}\le \frac{2}{2\pi e}$.

Next we wish to define an equivalent channel to \eqref{eq:Channel
Fading}. Let us define the extended transmission vector
\begin{equation}\label{eq:ConcatenatedofICsintheTransmitter}
\udl{x}=\lf(\udl{x}_{1}^{\lf(1\ri)\dagger},\dots,\udl{x}_{1}^{\lf(K\ri)\dagger},\dots,\udl{x}_{T}^{\lf(1\ri)\dagger},\dots,\udl{x}_{T}^{\lf(K\ri)\dagger}\ri)^{\dagger}
\end{equation}
i.e., first concatenate the users in each channel use, and then
concatenate the vectors between channel uses. Now we define
$H=\lf(H^{\lf(1\ri)},\dots,H^{\lf(K\ri)}\ri)$ which is an $N\times
KM$ matrix. By defining $H_{ex}$ as an $NT\times KMT$ block diagonal
matrix for which each block on the diagonal equals $H$,
$\underline{n}_{\mathrm{ex}}=\rho^{-\frac{1}{2}}\cdot
\lf(\underline{n}_{1}^{\dagger},
\dots,\underline{n}_{T}^{\dagger}\ri)^{\dagger}\in\mathbb{C}^{NT}$
and $\udl{y}_{\mathrm{ex}}\in\mathbb{C}^{NT}$, we can rewrite the
channel model in \eqref{eq:Channel Fading}
\begin{equation}\label{eq:ExtendedChannelModel}
\underline{y}_{\mathrm{ex}}=H_{\mrm{ex}}\cdot\underline{x}+\underline{n}_{\mathrm{ex}}.
\end{equation}

Let $L=\min\lf(N,KM\ri)$, and let $\sqrt{\lambda}_{i}$, $1\le i\le
L$ be the real valued, non-negative singular values of $H$. We
assume $\sqrt{\lambda}_{L}\ge \dots\ge\sqrt{\lambda}_{1}>0$. For
large values of $\rho$, we state that $f(\rho)\dot{\ge}g(\rho)$ when
$\lim_{\rho\to\infty}\frac{\ln \lf(f(\rho)\ri) }{\ln(\rho)}\ge
\frac{\ln \lf(g(\rho)\ri) }{\ln(\rho)}$, and also define
$\dot{\le}$, $\dot{=}$ in a similar manner by substituting $\ge$
with $\le$, $=$ respectively.

\subsection{Infinite Constellations}
Infinite constellation (IC) is a countable set
$S=\{s_{1},s_{2},\dots\}$ in $\mathbb{C}^{n}$. Let
$\cube_{l}(a)\subset\mathbb{C}^{n}$ be a (probably rotated)
$l$-complex dimensional cube ($l\le n$) with edge of length $a$
centered around zero. We define an IC $S_{l}$ to be $l$-complex
dimensional if there exists rotated $l$-complex dimensional cube
$\cube_{l}(a)$ such that $S_{l}\subset\lim_{a\to\infty}\cube_{l}(a)$
and $l$ is minimal. $M(S_{l},a)=|S_{l}\bigcap \cube_{l}(a)|$ is the
number of points of the IC $S_{l}$ inside $\cube_{l}(a)$. In
\cite{PoltirevJournal}, the $n$-complex dimensional IC density was
defined as
$$\gamma_{\mathrm{G}}=\limsup_{a\to\infty}\frac{M(S_{n},a)}{a^{2n}}$$
and the volume to noise ratio (VNR) for the additive white Gaussian
noise (AWGN) channel was given as
$$\mu_{\mathrm{G}}=\frac{\gamma_{\mathrm{G}}^{-\frac{1}{n}}}{2\pi
e\sigma^{2}}$$ where $\sigma^{2}$ is the noise variance of each
component.

We now turn to the IC definitions at the transmitters. We define the
average number of dimensions per channel use as the IC dimension
divided by the number of channel uses. Let us consider user $i$,
where $1\le i\le K$. We denote the average number of dimensions per
channel use by $D_{i}$. Let us consider a $D_{i}T$-complex
dimensional sequence of IC's - $S^{\lf(i\ri)}_{D_{i}T}(\rho)$, where
$D_{i}\le M$, $T$ is the number of channel uses, and
$\sum_{i=1}^{K}D_{i}\le L$. First we define
$\gamma_{tr}^{\lf(i\ri)}=\rho^{r_{i}T}$ as the density of
$S_{KT}^{\lf(i\ri)}(\rho)$ at transmitter $i$. Similarly to the
definitions in \cite{YonaFederICOptimalDMT} the multiplexing gain of
user's $i$ IC is defined as
\con{
\begin{equation}
r_{i}=\lim_{\rho\to\infty}\frac{1}{T}\log_{\rho}(\gamma_{\mathrm{tr}}^{\lf(i\ri)}+1)=
\lim_{\rho\to\infty}\frac{1}{T}\log_{\rho}(\rho^{r_{i}T}+1),\quad
0\le r_{i}\le D_{i}.
\end{equation}}
{\bal{}{
r_{i}=\lim_{\rho\to\infty}\frac{1}{T}\log_{\rho}(\gamma_{\mathrm{tr}}^{\lf(i\ri)}+1)&=&
\nn\\
\lim_{\rho\to\infty}\frac{1}{T}\log_{\rho}(\rho^{r_{i}T}&+&1),\quad
0\le r_{i}\le D_{i}.}}
The VNR at the transmitter of user $i$ is
\ifthenelse{\equal{\singlecolumntype}{1}}
{\begin{equation}
\mu_{\mathrm{tr}}^{^{\lf(i\ri)}}=\frac{{\gamma_{\mathrm{tr}}^{\lf(i\ri)}}^{-\frac{1}{D_{i}T}}}{2\pi
e\sigma^{2}}=\rho^{1-\frac{r_{i}}{D_{i}}}
\end{equation}}{\begin{equation}
\mu_{\mathrm{tr}}^{^{\lf(i\ri)}}=\frac{{\gamma_{\mathrm{tr}}^{\lf(i\ri)}}^{-\frac{1}{D_{i}T}}}{2\pi
e\sigma^{2}}=\rho^{1-\frac{r_{i}}{D_{i}}}
\end{equation}}
where $\sigma^{2}=\frac{\rho^{-1}}{2\pi e}$ is each component's
additive noise variance. Now let us concatenate the users IC's in
accordance with \eqref{eq:ConcatenatedofICsintheTransmitter}. We
denote $D=\sum_{i=1}^{K}D_{i}$. The concatenation yields an
equivalent $D T$-complex dimensional IC, $S_{\DT}\lf(\rho\ri)$, that
has multiplexing gain $\sum_{i=1}^{K}r_{i}$, density
$\gamma_{tr}=\rho^{\lf(\sum_{i=1}^{K}r_{i}\ri)T}$ and VNR
$\mu_{tr}=\rho^{1-\frac{\sum_{i=1}^{K}r_{i}}{D}}$. In this case we
get in \eqref{eq:ExtendedChannelModel} that the transmitted signal
$\udl{x}\in S_{D T}\lf(\rho\ri)\subset\mathbb{C}^{KMT}$.

At the receiver we first define the set $H_{ex}\cdot cube_{\DT}(a)$
as the multiplication of each point in $cube_{\DT}(a)$ with the
matrix $H_{ex}$. In a similar manner, the IC induced by the channel
at the receiver is $S_{\DT}^{'}=H_{ex}\cdot S_{\DT}$. The set
$H_{ex}\cdot cube_{\DT}(a)$ is almost surely $\DT$-complex
dimensional (where $D\le L$). In this case
\con{\beqn{}{M(S_{\DT},a)=|S_{\DT}\bigcap
\cube_{\DT}(a)|=|S_{\DT}^{'}\bigcap (H_{ex}\cdot\cube_{\DT}(a))|.}}{\baln{}{M(S_{\DT},a)=|S_{\DT}&\bigcap&
\cube_{\DT}(a)|
\nn\\
&=&|S_{\DT}^{'}\bigcap (H_{ex}\cdot\cube_{\DT}(a))|.}}
We define the receiver density as
$$\gamma_{\mathrm{rc}}=\limsup_{a\to\infty}\frac{M(S_{\DT},a)}{\mathbf{Vol}(H_{ex\cdot}\cube_{\DT}(a))}$$
i.e., the upper limit on the ratio of the number of IC points in
$H_{ex}\cdot\cube_{\DT}(a)$, and the volume of
$H_{ex}\cdot\cube_{\DT}(a)$. Note that for $N\ge KM$ and $D=KM$ we
get
$\gamma_{\mathrm{rc}}=\rho^{\sum_{i=1}^{K}r_{i}T}\cdot\prod_{i=1}^{KM}\lambda_{i}^{-T}$
and
$\mu_{\mathrm{rc}}=\rho^{1-\frac{\sum_{i=1}^{K}r_{i}}{KM}}\cdot\prod_{i=1}^{KM}\lambda_{i}^{\frac{1}{KM}}$.
The joint decoder average decoding error probability, over the
points of the effective IC $S_{\DT}(\rho)$, for a certain channel
realization $H$, is defined as
\begin{equation}\label{eq:AverageDecodingErrorProbability}
\overline{Pe}(H,\rho)=\limsup_{a\to\infty}\frac{\sum_{\underline{x}^{'}\in
S_{\DT}^{'}\bigcap
(H_{ex}\cdot\cube_{\DT}(a))}Pe(\underline{x}^{'},H,\rho)}{M(S_{\DT},a)}
\end{equation}
where $Pe(\underline{x}^{'},H,\rho)$ is the error probability
associated with $\udl{x}^{'}$. The average decoding error
probability of $S_{\DT}(\rho)$ over all channel realizations is
$\overline{Pe}(\rho)=E_{H}\{\overline{Pe}(H,\rho)\}$. The
\emph{diversity order} is defined as
\begin{equation}\label{eq:DiversityOrder}
d=-\lim_{\rho\to\infty}\log_{\rho}(\overline{Pe}(\rho)).
\end{equation}

In practice finite constellations are transmitted even when
performing regular lattice decoding at the receiver. Based on the
results in \cite{LoeligerAveragingBounds} it was shown in
\cite{YonaFederICOptimalDMT} that finite constellation with
multiplexing gain $r$ can be carved from a lattice with multiplexing
gain $r$, while maintaining the same performance when regular
lattice decoder is employed at the receiver. In our case it also
applies to each of the users, i.e., carving finite constellations
with multiplexing gains tuple $\lf(r_{1},\dots,r_{K}\ri)$ that
satisfy the power constraint, from lattices with multiplexing gains
tuple $\lf(r_{1},\dots,r_{K}\ri)$. At the receiver the performance
is maintained by performing regular lattice decoding on the
effective lattice.

\subsection{Additional Notations}
We further denote by $d^{\ast,\lf(FC\ri)}_{M,N}\lf(r\ri)$ the
optimal DMT of finite constellations, and by
$d^{\ast,D}_{M,N}\lf(r\ri)$ the upper bound on the optimal DMT of
any IC with average number of dimensions per channel use $D$, both
in a point to point channel with $M$ transmit and $N$ receive
antennas. For the multiple access channel with $K$ users, $M$
transmit antennas for each user, and $N$ receive antennas, we denote
by $d^{\ast,\lf(FC\ri)}_{K,M,N}\lf(r\ri)$ the optimal DMT of finite
constellations in the symmetric case, and by
$d^{\ast,\lf(IC\ri)}_{K,M,N}\lf(r\ri)$,
$d^{\ast,\lf(IC\ri)}_{K,M,N}\lf(r_{1},\dots,r_{K}\ri)$ the upper
bounds on the optimal DMT of the unconstrained multiple-access
channel for the symmetric case, and for multiplexing gains tuple
$\lf(r_{1},\dots,r_{K}\ri)$ respectively.

We denote $r_{max}=\max \lf(r_{1},\dots,r_{K}\ri)$, i.e., the maximal
multiplexing gain in the multiplexing gains tuple. In addition for
any $A\subseteq \lf\{1,\dots,K\ri\}$ we define $R_{A}=\sum_{a\in
A}r_{a}$ and $D_{A}=\sum_{a\in A}D_{a}$.

\section{Upper Bound on the Best Diversity-Multiplexing Tradeoff}\label{sec:LowerBoundErrorProb}
In this section we show that for $N<(K+1)M-1$ the DMT of the
unconstrained multiple-access channel is suboptimal compared to the
optimal DMT of finite constellations. On the other hand for $N\ge
(K+1)M-1$, we derive an upper bound on the optimal DMT that
coincides with the optimal DMT of finite constellations.

In subsection \ref{subsec:LowerBoundErrorProb_A} we lower bound the
error probability of any IC for the multiple-access channel, by
using lower bounds on the error probability of any IC in the
point-to-point channel. We use these lower bounds to formulate an
upper bound on the optimal DMT of IC's for the multiple-access
channel, in the form of an optimization problem. In subsection
\ref{subsec:LowerBoundErrorProb_B} we solve this optimization
problem for the symmetric case. We compare the optimal DMT of IC's
to the optimal DMT of finite constellations, and find the cases for
which IC's are suboptimal in subsection
\ref{subsec:LowerBoundErrorProb_C}. Finally in subsection
\ref{subsec:TheReasonsForSuboptimality} we give a convexity argument
that shows for the symmetric case that whenever the optimal DMT is
not a convex function IC's are suboptimal

\subsection{Upper Bound on the Diversity-Multiplexing-Tradeoff}\label{subsec:LowerBoundErrorProb_A}
We lower bound the error probability of the unconstrained
multiple-access channel in Lemma
\ref{lem:lowerboundMACErrorProbabaility}. Based on this lower bound
we present in Theorem \ref{Th:MDTFormulation} an upper bound on the
optimal DMT of IC's.

Assume user $i$ transmits over $D_{i}T$-complex dimensional IC, with
average number of dimensions per channel use $D_{i}$ and $T$ channel
uses. The following lemma lower bounds the average decoding error
probability of the $K$-users
$\overline{Pe}^{\lf(D_{1},\dots,D_{K},T\ri)}\lf(\rho,r_{1},\dots,r_{K}\ri)$,
where $\lf(D_{1},\dots,D_{K}\ri)$ is the tuple of average number of
dimensions per channel use, $T$ is the number of channel uses and
$\lf(r_{1},\dots,r_{K}\ri)$ is the tuple of multiplexing gains.
\begin{lem}\label{lem:lowerboundMACErrorProbabaility}
\con{
\beqn{}{
\overline{Pe}^{\lf(D_{1},\dots,D_{K},T\ri)}\lf(\rho,r_{1},\dots,r_{K}\ri)\ge \max_{ A\subseteq\lf\{1,\dots,K\ri\}}\lf(Pe^{\lf(D_{A},T\ri)}\lf(\rho,R_{A}\ri)\ri)}}
{
\baln{}{
\overline{Pe}^{\lf(D_{1},\dots,D_{K},T\ri)}&\lf(\rho,r_{1},\dots,r_{K}\ri)&\ge
\nn\\
&\max_{ A\subseteq\lf\{1,\dots,K\ri\}}&\lf(Pe^{\lf(D_{A},T\ri)}\lf(\rho,R_{A}\ri)\ri)}}
where $Pe^{\lf(D_{A},T\ri)}\lf(\rho,R_{A}\ri)$ is the lower bound
derived in \cite{YonaFederICOptimalDMT} for the error probability of
any IC with $T$ channel uses, $D_{A}=\sum_{a\in A}D_{a}$ average
number of dimensions per channel use, and multiplexing gain
$R_{A}=\sum_{a\in A}r_{a}$, in a point-to-point channel with
$|A|\cdot M$ transmit and $N$ receive antennas.
\end{lem}
\begin{proof}
 By considering the extended channel model
\eqref{eq:ExtendedChannelModel}, we get that the $K$ distributed
transmitters transmit an effective
$\lf(\sum_{i=1}^{K}D_{i}\ri)T$-complex dimensional IC, over $T$
channel uses, with multiplexing gain $\sum_{i=1}^{K}r_{i}$. The
error probability of this IC is lower bounded by the lower bound for
the error probability of any IC with average number of dimensions
per channel use $\sum_{i=1}^{K}D_{i}$, $T$ channel uses, and
multiplexing gain $\sum_{i=1}^{K}r_{i}$, in a point-to-point channel
with $KM$ transmit and $N$ receive antennas. Such a lower bound on
the error probability was derived in \cite{YonaFederICOptimalDMT}
for each channel realization (\cite{YonaFederICOptimalDMT} Theorem
1), and then for the average over all channel realizations when
$\rho$ is large (\cite{YonaFederICOptimalDMT} Theorem 2). Now
consider the set $A\subset\lf\{1,\dots,K\ri\}$. In case a genie
tells the receiver the transmitted messages of users
$\lf\{1,\dots,K\ri\}\setminus A$, the optimal receiver attains an
error probability that lower bounds the $K$-user optimal receiver
error probability. Without loss of optimality, the optimal receiver
can subtract them from the received signal, and get a new
$|A|$-users unconstrained multiple-access channel with average
number of dimensions per channel use $\lf\{D_{a}\ri\}_{a\in A}$, $T$
channel uses, and multiplexing gain $\sum_{a\in A}r_{a}$. In a
similar manner, the error probability of this $|A|$-users channel is
lower bounded by the lower bound on the error probability of any IC
with $\sum_{a\in A}D_{a}$ average number of dimensions per channel
use, $T$ channel uses, and multiplexing gain $\sum_{a\in A}r_{a}$,
derived in \cite{YonaFederICOptimalDMT}. Hence, the maximal lower
bound on the error probability for $A\subseteq \lf\{1,\dots,
K\ri\}$, also sets a lower bound for the error probability. This
concludes the proof.
\end{proof}

Next we wish to formulate an upper bound on the DMT of IC's in the
$K$-user unconstrained multiple-access channel. We derive this bound
based on the lower bound on the error probability presented in Lemma
\ref{lem:lowerboundMACErrorProbabaility}, and on an upper bound on
the DMT of IC's for the point-to-point channel, presented in
\cite{YonaFederICOptimalDMT}. Let us begin by presenting the upper
bound on the DMT for the point-to-point channel.

\begin{theorem}[\cite{YonaFederICOptimalDMT} Theorem
2]\label{prop:P2POptimalDMTUpperBoundDMT} For any sequence of IC's
$S_{\DT}\lf(\rho\ri)$ with $D$ average number of dimensions per
channel use, in a point-to-point channel with $M$ transmit and $N$
receive antennas, the DMT $d^{\DT}_{M,N}\lf(r\ri)$ is upper bounded
by
$$d^{\DT}_{M,N}\lf(r\ri)\le d^{\ast,D}_{M,N}\lf(r\ri)=\frac{M\cdot N}{D}\lf(D-r\ri)$$
for $0\le D\le \frac{M\cdot N}{N+M-1}$, and
$$d^{\DT}_{M,N}\lf(r\ri)\le d^{\ast,D}_{M,N}\lf(r\ri)=\frac{\lf(M-l\ri)\lf(N-l\ri)}{D-l}\cdot\lf(D-r\ri)$$
for $\frac{M\cdot N -\lf(l-1\ri)l}{N+M-1-2\lf(l-1\ri)}\le
D\le\frac{M\cdot N -l\lf(l+1\ri)}{N+M-1-2l}$, and
$l=1,\dots,\min\lf(M,N\ri)-1$. In all cases $0\le r\le D$.
\end{theorem}

Based on Lemma \ref{lem:lowerboundMACErrorProbabaility} and Theorem
\ref{prop:P2POptimalDMTUpperBoundDMT} we formulate the following
upper bound on the optimal DMT of the multiple-access channel.
\begin{theorem}\label{Th:MDTFormulation}
The optimal DMT of any sequence of IC's with multiplexing gains
tuple $\lf(r_{1},\dots,r_{K}\ri)$ is upper bounded by
\con{
\begin{equation*}
d^{\ast,\lf(IC\ri)}_{K,M,N} \lf(r_{1},\dots,r_{K}\ri)= \max_{ \lf
(D_{1},\dots,D_{K}\ri)\in \RD} \min_{ A\subseteq
\lf\{1,\dots,K\ri\}} \lf(d^{\ast,D_{A}}_{|A|\cdot
M,N}\lf(R_{A}\ri)\ri)
\end{equation*}
}
{\baln{}{d^{\ast,\lf(IC\ri)}_{K,M,N} &\lf(r_{1},\dots,r_{K}\ri)&=
\nn\\
&\max_{ \lf
(D_{1},\dots,D_{K}\ri)\in \RD}& \min_{ A\subseteq
\lf\{1,\dots,K\ri\}} \lf(d^{\ast,D_{A}}_{|A|\cdot
M,N}\lf(R_{A}\ri)\ri)
}
}
where $\RD= \lf\{D_{1},\dots,D_{K}\mid 0\le D_{i}\le
M,\sum_{i=1}^{K}D_{i}\le L\ri\}$.
\end{theorem}
\begin{proof}
Following Lemma \ref{lem:lowerboundMACErrorProbabaility} we get a
lower bound for the error probability of any sequence of effective
IC's $S_{\sum_{i=1}^{K}D_{i}T}\lf(\rho\ri)$, transmitted by the $K$
users. This lower bound can be translated to an upper bound on the
diversity order. In addition, this lower bound on the error
probability depends on lower bounds on the error probabilities for
the point-to-point channel. Hence, we can use the upper bound on the
DMT in the point-to-point channel, presented in Theorem
\ref{prop:P2POptimalDMTUpperBoundDMT}, to get the following upper
bound on the DMT of a tuple of average number of dimensions per
channel use $\lf (D_{1},\dots,D_{K}\ri)$
\begin{equation*}
\min_{ A\subseteq \lf\{1,\dots,K\ri\}} \lf(d^{\ast,D_{A}}_{|A|\cdot
M,N}\lf(R_{A}\ri)\ri).
\end{equation*}
Maximizing over $\lf (D_{1},\dots,D_{K}\ri)\in \RD$ yields the upper
bound on the optimal DMT.
\end{proof}

\subsection{Characterizing the Optimal Symmetric DMT}\label{subsec:LowerBoundErrorProb_B}
We wish to characterize an upper bound on the optimal DMT of IC's in
the symmetric case, i.e., $r_{1}=\dots=r_{K}=r$. Later we will use
this upper bound in order to show the sub-optimality of the
unconstrained multiple-access channel in the case
$N<\lf(K+1\ri)M-1$. In addition, we will show that the upper bound
coincides with the optimal DMT of finite constellations in the case
$N\ge\lf(K+1\ri)M-1$.

Lemmas \ref{lem:TheCasewhereSingleUserIstheWorst},
\ref{lem:TheCasewhereSingleUserIsNotNecessarilytheWorst},
\ref{lem:TheCasewhereDimensionSmallerThanFirstDMTLine},
\ref{lem:TheCasewhereSingleUserandKUserDMTUpperBoundCoincide}
present the relations between $d^{\ast,i\cdot D}_{i\cdot
M,N}\lf(i\cdot r\ri)$, $i=1,\dots, K$ for different values of $N$.
We use these lemmas in order to upper bound the optimal DMT in the
symmetric case in Theorem \ref{th:ICOptimalSymmetricDMTUpperNound}.

Based on Theorem \ref{Th:MDTFormulation} we can state that the
optimal DMT for the symmetric case for $K$ users is upper bounded by
\begin{equation}\label{eq:OptimalUpperBoundMACDMTSymmetric}
d^{\ast,\lf(IC\ri)}_{K,M,N}
\lf(r\ri)=\max_{\lf(D_{1},\dots,D_{K}\ri)\in \RD} \min_{ A\subseteq
\lf\{1,\dots,K\ri\}} \lf(d^{\ast,D_{A}}_{|A|\cdot M,N}\lf(|A|\cdot
r\ri)\ri)
\end{equation}
where $0\le r\le\frac{L}{K}$, i.e., we wish solve the aforementioned optimization problem for each
$0\le r\le\frac{L}{K}$. In order to solve this optimization problem
we first solve a simpler optimization problem for the case
$D_{1}=\dots=D_{K}=D$, i.e., each user transmits over $D$ average
number of dimensions per channel use. In this case the upper bound
in \eqref{eq:OptimalUpperBoundMACDMTSymmetric} takes a simpler form
\begin{equation}\label{eq:OptimalUpperBoundMACDMTSymmetricEqualDim}
\max_{D}\min_{1\le i\le K} \lf(d^{\ast,i\cdot
D}_{i\cdot M,N}\lf(i\cdot r\ri)\ri)
\end{equation}
where $0\le D\le \frac{L}{K}$. After solving this optimization
problem, we will show that choosing $D_{1}=\dots=D_{K}=D$ also
yields the optimal solution for
\eqref{eq:OptimalUpperBoundMACDMTSymmetric}.

In order to solve the optimization problem in
\eqref{eq:OptimalUpperBoundMACDMTSymmetricEqualDim}, we first need
to present some properties on the relations between $d^{\ast,i\cdot D}_{i\cdot
M,N}\lf(i\cdot r\ri)$, $1\le i\le K$. We begin by
presenting a property on the behavior of
$d^{\ast,D}_{M,N}\lf(\cdot\ri)$ as a function of $D$.
\begin{cor}[\cite{YonaFederICOptimalDMT} Corollary 1]\label{prop:ICP2pDMTAnchorPoints}
For $0\le D\le \frac{M\cdot N}{N+M-1}$ we have the following
equality
$$d^{\ast,D}_{M,N}\lf(0\ri)=MN,$$
whereas for $\frac{M\cdot N -\lf(l-1\ri)l}{N+M+1-2 \lf(l-1\ri)}\le
D\le\frac{M\cdot N -l\lf(l+1\ri)}{N+M-1-2l}$, and
$l=1,\dots,\min\lf(M,N\ri)-1$ we get
$$d^{\ast,D}_{M,N}\lf(l\ri)=(M-l)\cdot (N-l).$$
\end{cor}
A simple interpretation of Corollary \ref{prop:ICP2pDMTAnchorPoints}
is that for $0\le D\le \frac{M\cdot N}{N+M-1}$ the straight lines
$d^{\ast,D}_{M,N}\lf(\cdot\ri)$ that represent the upper bounds on
the DMT, all have the same ``anchor'' point at multiplexing gain
$r=0$, i.e., they all have diversity order $MN$ at $r=0$, and each
line equals to zero at $r=D$. On the other hand, for $\frac{M\cdot N
-\lf(l-1\ri)\lf(l\ri)}{N+M+1-2l}\le D\le\frac{M\cdot N
-\lf(l\ri)\lf(l+1\ri)}{N+M-1-2l}$, and
$l=1,\dots,\min\lf(M,N\ri)-1$, the straight lines equal to
$\lf(M-l\ri)\lf(N-l\ri)$ for multiplexing gain $r=l$, and again each
line equals to zero for $r=D$. Figure \ref{fig:PropAnchorPoint}
illustrates this property for $M=N=2$.
\begin{figure}[h]
\centering
\con{\epsfig{figure=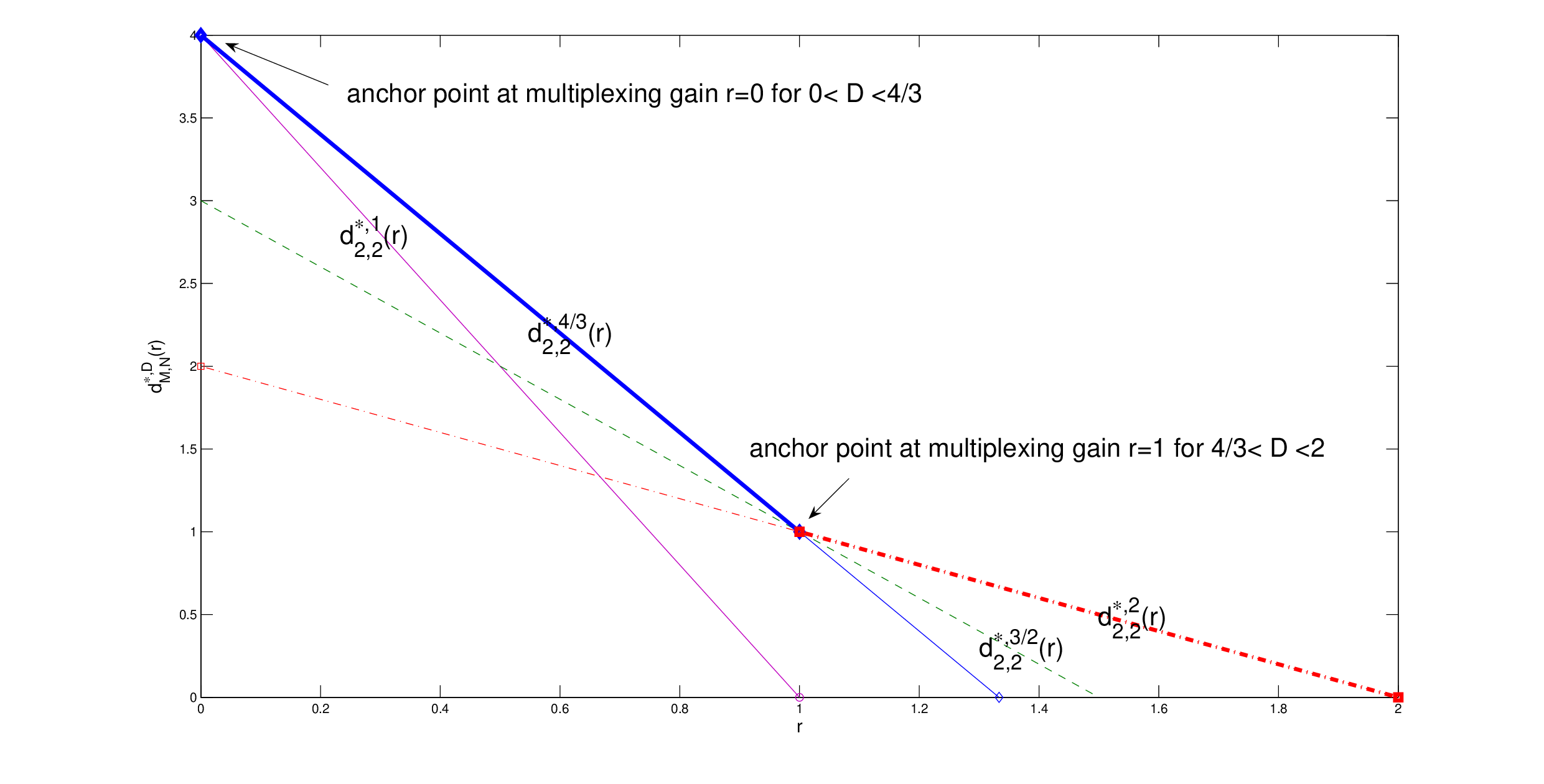,height=5.5cm}\caption{Upper
bound on the DMT for any IC of $D$ average number of
dimensions per channel use, in a point to point channel with $M=N=2$. Note that
$d^{\ast,1}_{2,2}\lf(r\ri)$ and
$d^{\ast,\frac{4}{3}}_{2,2}\lf(r\ri)$ are straight lines that equal to
$MN=4$ at multiplexing gain $r=0$, whereas
$d^{\ast,\frac{3}{2}}_{2,2}\lf(r\ri)$ and
$d^{\ast,2}_{2,2}\lf(r\ri)$ are straight lines that equal to
$\lf(M-1\ri)\lf(N-1\ri)=1$ at multiplexing gain $r=1$, in
accordance with Corollary \ref{prop:ICP2pDMTAnchorPoints}. In bold
is the optimal DMT of finite constellations.}
\label{fig:PropAnchorPoint}
}{
\epsfig{figure=1.eps,height=4.5cm}\caption{Upper
bound on the DMT for any IC of $D$ average number of
dimensions per channel use, in a point to point channel with $M=N=2$. Note that
$d^{\ast,1}_{2,2}\lf(r\ri)$ and
$d^{\ast,\frac{4}{3}}_{2,2}\lf(r\ri)$ are straight lines that equal to
$MN=4$ at multiplexing gain $r=0$, whereas
$d^{\ast,\frac{3}{2}}_{2,2}\lf(r\ri)$ and
$d^{\ast,2}_{2,2}\lf(r\ri)$ are straight lines that equal to
$\lf(M-1\ri)\lf(N-1\ri)=1$ at multiplexing gain $r=1$, in
accordance with Corollary \ref{prop:ICP2pDMTAnchorPoints}. In bold
is the optimal DMT of finite constellations.}
\label{fig:PropAnchorPoint}}
\end{figure}
The next corollary presents the relation between
$d^{\ast,D}_{M,N}\lf(l\ri)$ and
$d^{\ast,\lf(FC\ri)}_{M,N}\lf(r\ri)$.
\begin{cor}\label{prop:RelationBedDandFCd}
For any $0\le D\le \min \lf(M,N\ri)$ we have the following
inequality
\begin{equation*}
d^{\ast,D}_{M,N}\lf(r\ri)\le d^{\ast,\lf(FC\ri)}_{M,N}\lf(r\ri)
\end{equation*}
for $0\le r\le D$. Furthermore, when $l\le r\le l+1$ and
$l=0,\dots,\min \lf(M,N\ri)-1$
\begin{equation*}
d^{\ast,\lf(FC\ri)}_{M,N}\lf(r\ri) = NM-l\cdot
\lf(l+1\ri)-\lf(N+M-1-2\cdot l\ri)r.
\end{equation*}
\end{cor}
\begin{proof}
The proof follows from \cite[Corollary 2]{YonaFederICOptimalDMT}
stating that for any $l=0,\dots,\min \lf(M,N\ri)-1$ and $l\le r\le
l+1$
\begin{equation*}
\max_{D}d^{\ast,D}_{M,N}\lf(r\ri)\le
d^{\ast,D_{l}}_{M,N}\lf(r\ri)=d^{\ast,\lf(FC\ri)}_{M,N}\lf(r\ri)
\end{equation*}
where $D^{\ast}_{l}=\frac{N\cdot M-l\cdot \lf(l+1\ri)}{N+M-1-2l}$.
Therefore, for any $0\le D\le \lf(M,N\ri)-1$ we get
\begin{equation*}
d^{\ast,D}_{M,N}\lf(r\ri)\le d^{\ast,\lf(FC\ri)}_{M,N}\lf(r\ri).
\end{equation*}
for $0\le r\le D$.

The explicit expression for $d^{\ast,\lf(FC\ri)}_{M,N}\lf(r\ri)$ is
obtained by the straight lines that connect the points
$\lf(l,\lf(N-l\ri)\cdot \lf(M-l\ri)\ri)$ and
$\lf(l+1,\lf(N-l-1\ri)\cdot \lf(M-l-1\ri)\ri)$, for $l=0,\dots,\min
\lf(M,N\ri)-1$.
\end{proof}

Another property relates to the optimal DMT of finite constellations
for the multiple-access channel in the symmetric case.
\begin{theorem}[\cite{ZhengTseMACDMT2004} Theorem
3]\label{prop:FCOptDMT} The optimal DMT of finite constellations in
the symmetric case equals
\con{
\begin{equation*}
d^{\ast,\lf(FC\ri)}_{K,M,N}\lf(r\ri)=\lf\{\begin{array}{lr}
d^{\ast,\lf(FC\ri)}_{M,N}\lf(r\ri) &0\le r\le\min \lf(\frac{N}{K+1},M\ri)\\
d^{\ast,\lf(FC\ri)}_{KM,N}\lf(K\cdot r\ri) &\min
\lf(\frac{N}{K+1},M\ri)\le r\le\min \lf(\frac{N}{K},M\ri)
\end{array} \ri.
\end{equation*}}
{\baln{}{
&d^{\ast,\lf(FC\ri)}_{K,M,N}\lf(r\ri)=&
\nn\\
&\lf\{\begin{array}{lr}
d^{\ast,\lf(FC\ri)}_{M,N}\lf(r\ri) &0\le r\le\min \lf(\frac{N}{K+1},M\ri)\\
d^{\ast,\lf(FC\ri)}_{KM,N}\lf(K\cdot r\ri) &\min
\lf(\frac{N}{K+1},M\ri)\le r\le\min \lf(\frac{N}{K},M\ri)
\end{array} \ri.&
}}
\end{theorem}

In order to solve the optimization problem in
\eqref{eq:OptimalUpperBoundMACDMTSymmetricEqualDim} we present
several lemmas related to the inequalities between $d^{\ast,i\cdot
D}_{i\cdot M,N}\lf(i\cdot r\ri)$ for $1\le i\le K$. The proofs of
these lemmas rely mainly on Corollary
\ref{prop:ICP2pDMTAnchorPoints}, Corollary
\ref{prop:RelationBedDandFCd} and Theorem \ref{prop:FCOptDMT}.
\begin{lem}\label{lem:TheCasewhereSingleUserIstheWorst}
For $N\ge \lf(K+1\ri)M-1$ we get
$$d^{\ast,D}_{M,N}\lf(r\ri)\le d^{\ast,i\cdot D}_{i\cdot
M,N}\lf(i\cdot r\ri)\quad 2\le i\le K$$ for any $0\le r\le D$ and
$0\le D\le M$.
\end{lem}
\begin{proof}
The proof is in appendix
\ref{append:TheCasewhereSingleUserIstheWorst}.
\end{proof}
An example for Lemma \ref{lem:TheCasewhereSingleUserIstheWorst} for
$M=K=2$ and $N=4$ is illustrated in Figure
\ref{fig:TheCasewhereSingleUserIstheWorst}.

\begin{figure}[h]
\centering
\con{\epsfig{figure=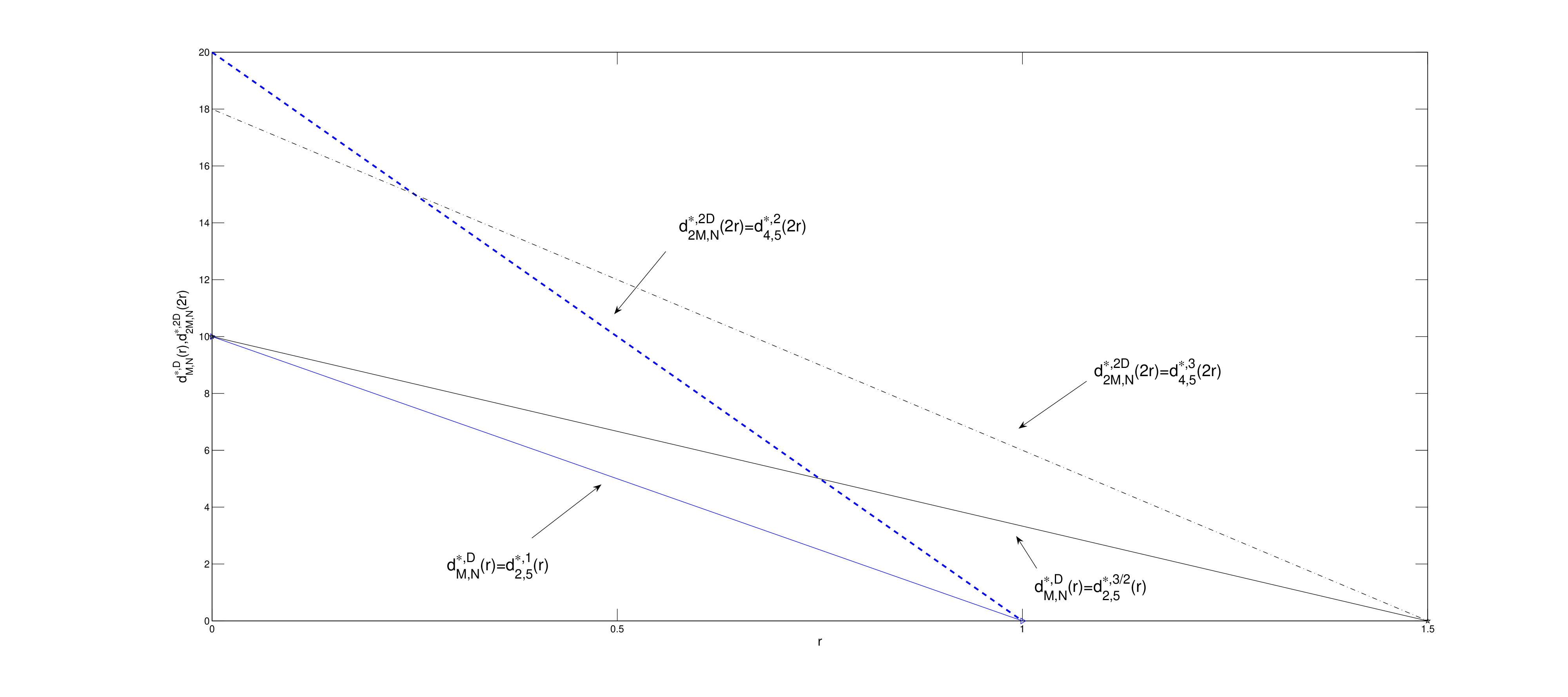,height=5.5cm}\caption{Illustration
of Lemma \ref{lem:TheCasewhereSingleUserIstheWorst} for the case
$M=K=2$ and $N=5$. We compare the straight lines
$d^{\ast,D}_{M,N}\lf(r\ri)$ and $d^{\ast,2D}_{2M,N}\lf(2r\ri)$ for
$D=1$ and $D=\frac{3}{2}$. It can be seen that for this setting
$d^{\ast,2D}_{2M,N}\lf(2r\ri)>d^{\ast,D}_{M,N}\lf(r\ri)$.}\label{fig:TheCasewhereSingleUserIstheWorst}}{\epsfig{figure=2.eps,height=4cm}\caption{Illustration
of Lemma \ref{lem:TheCasewhereSingleUserIstheWorst} for the case
$M=K=2$ and $N=5$. We compare the straight lines
$d^{\ast,D}_{M,N}\lf(r\ri)$ and $d^{\ast,2D}_{2M,N}\lf(2r\ri)$ for
$D=1$ and $D=\frac{3}{2}$. It can be seen that for this setting
$d^{\ast,2D}_{2M,N}\lf(2r\ri)>d^{\ast,D}_{M,N}\lf(r\ri)$.}\label{fig:TheCasewhereSingleUserIstheWorst}}
\end{figure}

\begin{lem}\label{lem:TheCasewhereSingleUserIsNotNecessarilytheWorst}
For $N <  \lf(K+1\ri)M-1$ we get
$$d^{\ast,D}_{M,N}\lf(r\ri)\le d^{\ast,i\cdot D}_{i\cdot
M,N}\lf(i\cdot r\ri)\quad 2\le i\le K-1$$ for any $0\le D\le
\frac{L}{K}$ and $0\le r\le D$.
\end{lem}

\begin{proof}
The proof is in appendix
\ref{append:TheCasewhereSingleUserIsNotNecessarilytheWorst}
\end{proof}

From Lemmas \ref{lem:TheCasewhereSingleUserIstheWorst},
\ref{lem:TheCasewhereSingleUserIsNotNecessarilytheWorst} we can see
that the optimization problem in
\eqref{eq:OptimalUpperBoundMACDMTSymmetricEqualDim} involves only
$d^{\ast,D}_{M,N}\lf(r\ri)$ and $d^{\ast,K\cdot D}_{K\cdot
M,N}\lf(K\cdot r\ri)$. We now prove two more properties that will
enable us to find the optimal DMT of IC's in the symmetric case.

\begin{lem}\label{lem:TheCasewhereDimensionSmallerThanFirstDMTLine}
For $N< \lf(K-1\ri)M +1$ we get
$$\max_{0\le D\le\frac{L}{K}}\min_{1\le i\le K}d^{\ast,i\cdot D}_{i\cdot M,N} \lf(i\cdot r\ri)=d^{\ast,\frac{N}{K}}_{M,N}\lf(r\ri)= M\cdot N-M\cdot K\cdot r$$
where $0\le r\le \frac{N}{K}$.
\end{lem}
\begin{proof}
The proof is in appendix
\ref{append:TheCasewhereDimensionSmallerThanFirstDMTLine}
\end{proof}

From Lemma \ref{lem:TheCasewhereDimensionSmallerThanFirstDMTLine} we
can see that for the multiple-access channel, when $N<
\lf(M-1\ri)K+1$ the optimal DMT of IC's is smaller than finite
constellations optimal DMT for any value of $r$ except for $r=0$ and
$r=\frac{N}{k}$. Figure
\ref{fig:TheCasewhereDimensionSmallerThanFirstDMTLine} illustrates
Lemma \ref{lem:TheCasewhereDimensionSmallerThanFirstDMTLine} for the
case $M=N=K=2$. Now let us show the cases for which
$d^{\ast,D}_{M,N} \lf(r\ri)$ and $d^{\ast,K\cdot D}_{K\cdot M,N}
\lf(K\cdot r\ri)$ coincide.

\begin{figure}[h]
\centering
\con{\epsfig{figure=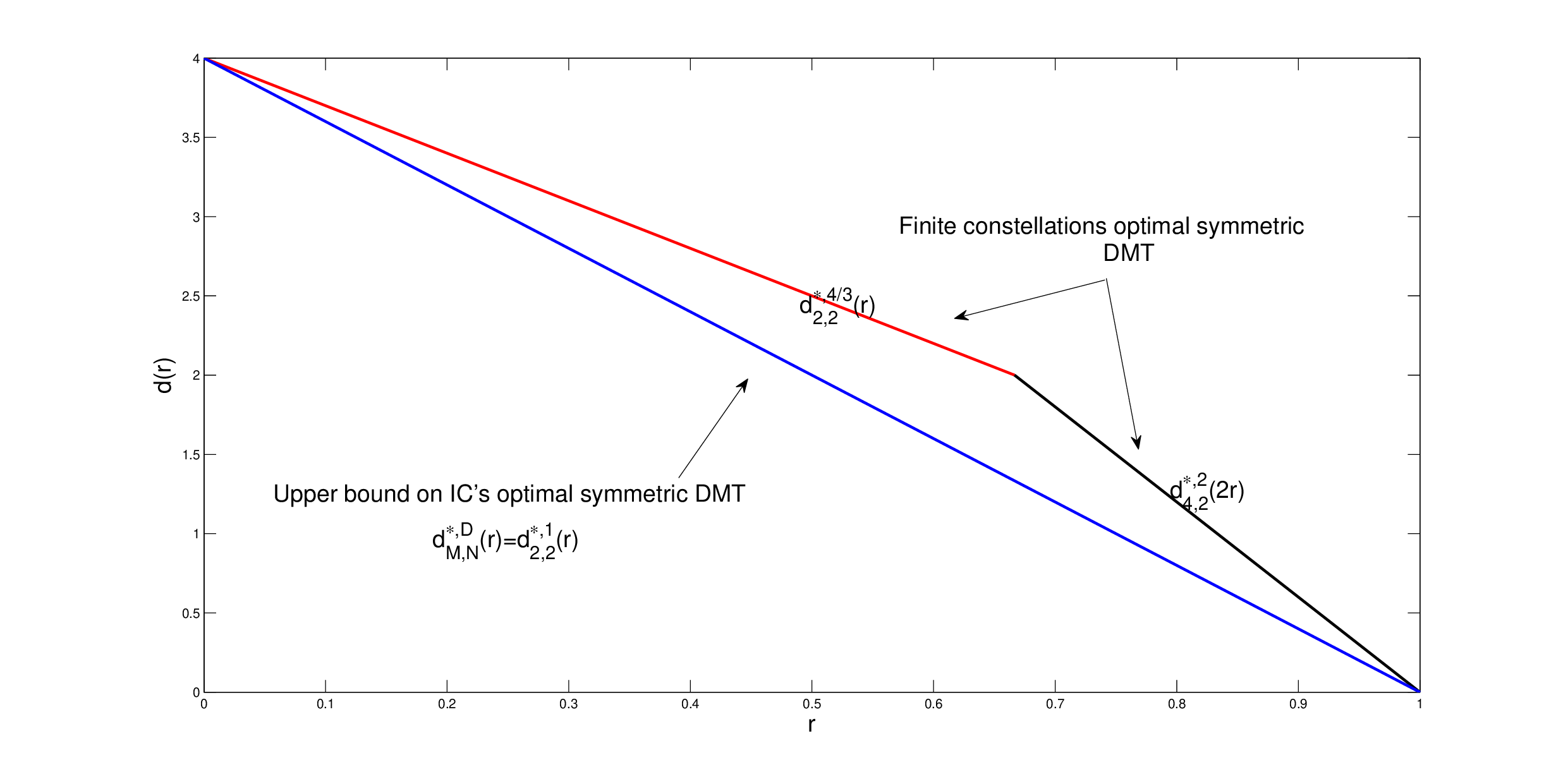,height=5.5cm}\caption{Illustration of Lemma \ref{lem:TheCasewhereDimensionSmallerThanFirstDMTLine} for the case $M=N=K=2$.
 In this case the optimal DMT is smaller than the optimal DMT of finite constellations, for any value of $r$ except for $r=0$, 1.}\label{fig:TheCasewhereDimensionSmallerThanFirstDMTLine}}{\epsfig{figure=3.eps,height=4.8cm}\caption{Illustration of Lemma \ref{lem:TheCasewhereDimensionSmallerThanFirstDMTLine} for the case $M=N=K=2$.
 In this case the optimal DMT is smaller than the optimal DMT of finite constellations, for any value of $r$ except for $r=0$, 1.}\label{fig:TheCasewhereDimensionSmallerThanFirstDMTLine}}
\end{figure}

The following lemma serves as another building block in upper
bounding the optimal DMT in the symmetric case when
$N=\lf(K-1\ri)M+1+l$, $l=0,\dots, 2M-3$. It finds the average number
of dimensions per channel use that leads to the equality
$d^{\ast,D}_{M,N} \lf(r\ri)=d^{\ast,K\cdot D}_{K\cdot M,N}
\lf(K\cdot r\ri)$ for any value of $r$, and also shows for which
values of $r$ these straight lines are equal to the optimal DMT of
finite constellations in a point-to-point channel.
\begin{lem}\label{lem:TheCasewhereSingleUserandKUserDMTUpperBoundCoincide}
For $N= \lf(K-1\ri)M+1+l< \lf(K+1\ri)M-1$, where $l=0,\dots, 2M-3$,
we get for average number of dimensions per channel use per user
$D_{l}=\frac{MN-\lfloor\frac{l}{2}\rfloor\cdot
\lf(\lfloor\frac{l}{2}\rfloor+1\ri)-2\cdot\lf(\lfloor\frac{l}{2}\rfloor+1\ri)\cdot\lf(\frac{l}{2}-\lfloor\frac{l}{2}\rfloor\ri)}{N+M-1-l}$
that
\con{
\begin{equation*}
d^{\ast,D_{l}}_{M,N} \lf(r\ri) =d^{\ast,K\cdot D_{l}}_{K\cdot
M,N}\lf(K\cdot
r\ri)=d^{\ast}\lf(r\ri)=MN-\lfloor\frac{l}{2}\rfloor\cdot
\lf(\lfloor\frac{l}{2}\rfloor+1\ri)-2\cdot\lf(\lfloor\frac{l}{2}\rfloor+1\ri)\cdot\lf(\frac{l}{2}-\lfloor\frac{l}{2}\rfloor\ri)-
\lf(N+M-1-l\ri)r
\end{equation*}}{\baln{}{
&d^{\ast,D_{l}}_{M,N}\lf(r\ri) =d^{\ast,K\cdot D_{l}}_{K\cdot
M,N}\lf(K\cdot
r\ri)=d^{\ast}\lf(r\ri)=MN-&
\nn\\
&\lfloor\frac{l}{2}\rfloor\cdot
\lf(\lfloor\frac{l}{2}\rfloor+1\ri)-2\cdot\lf(\lfloor\frac{l}{2}\rfloor+1\ri)\cdot\lf(\frac{l}{2}-\lfloor\frac{l}{2}\rfloor\ri)&
\nn\\
&-
\lf(N+M-1-l\ri)r&
}}
where $0\le r\le D_{l}$. In addition
\begin{equation*}
d^{\ast,\lf(FC\ri)}_{M,N}\lf(\lfloor\frac{l}{2}
\rfloor+1\ri)=d^{\ast}\lf(\lfloor\frac{l}{2} \rfloor+1\ri)
\end{equation*}
and also
\con{
\begin{equation*}
d^{\ast,\lf(FC\ri)}_{KM,N}\lf( \lf(K-1\ri)M+ \lfloor\frac{l+1}{2}
\rfloor\ri) = d^{\ast}\lf( \frac{\lf(K-1\ri)M+ \lfloor\frac{l+1}{2}
\rfloor}{K}\ri)
\end{equation*}}{\baln{}{&d^{\ast,\lf(FC\ri)}_{KM,N}\lf( \lf(K-1\ri)M+ \lfloor\frac{l+1}{2}
\rfloor\ri) =&
\nn\\
&d^{\ast}\lf( \frac{\lf(K-1\ri)M+ \lfloor\frac{l+1}{2} \rfloor}{K}\ri)&}}
\end{lem}
\begin{proof}
The proof is in appendix
\ref{append:TheCasewhereSingleUserandKUserDMTUpperBoundCoincide}.
\end{proof}
An example that illustrates Lemma
\ref{lem:TheCasewhereSingleUserandKUserDMTUpperBoundCoincide} for
$M=K=2$ and $N=4$ is given in Figure
\ref{fig:TheCasewhereSingleUserandKUserDMTUpperBoundCoincide}.

\begin{figure}[h]
\centering
\con{\epsfig{figure=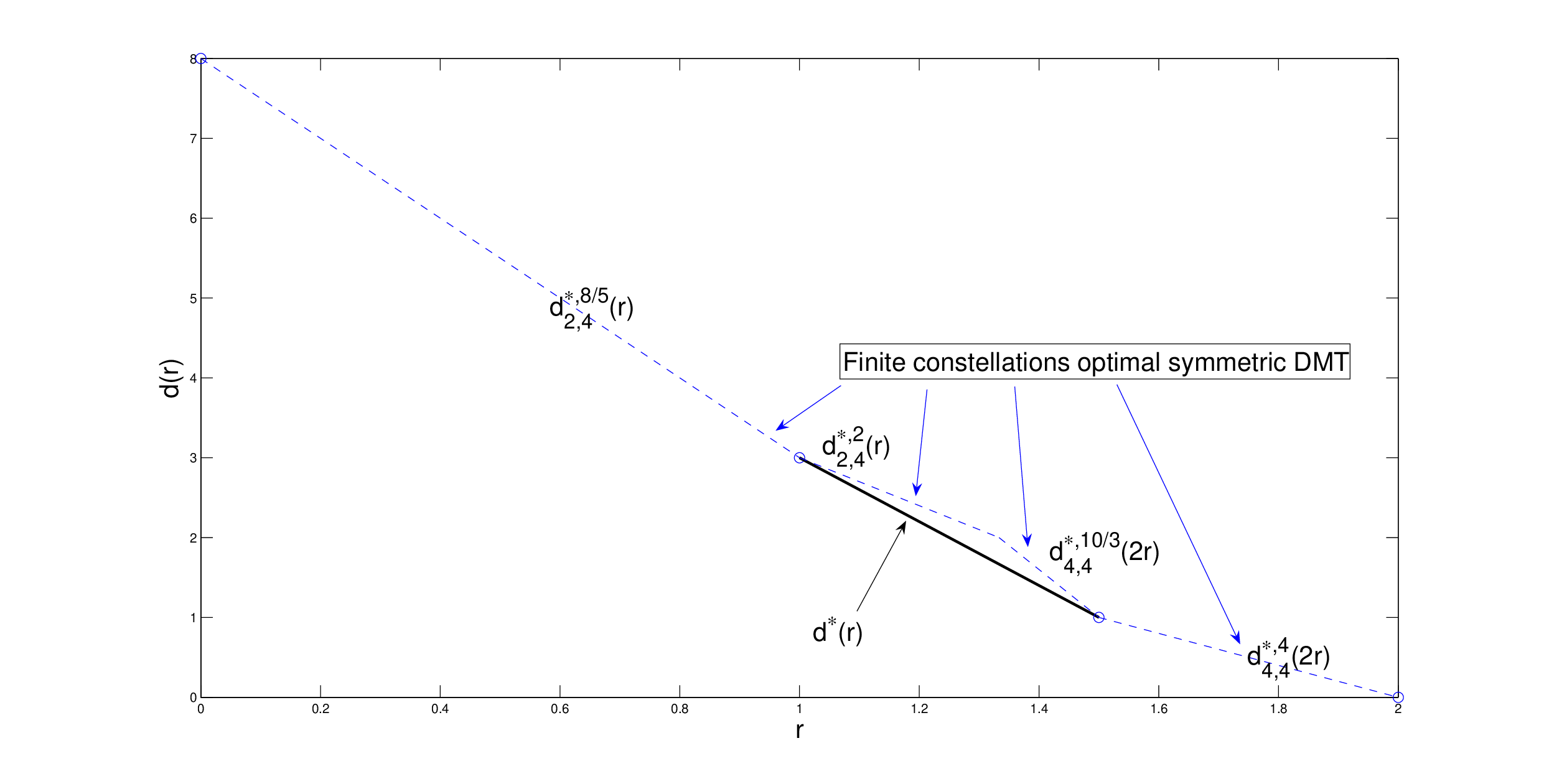,height=7cm}\caption{$d^{\ast}\lf(r\ri)$
for $M=K=2$ and $N=4$, i.e., $l=1$. Note that
$d^{\ast}\lf(1\ri)=d^{\ast,8/5}_{2,4}\lf(1\ri)=d^{\ast,\lf(FC\ri)}_{2,4}\lf(1\ri)=d^{\ast,2}_{2,4}\lf(1\ri)$
and $d^{\ast}\lf(\frac{3}{2}\ri)=
d^{\ast,4}_{4,4}\lf(3\ri)=d^{\ast,\lf(FC\ri)}_{4,4}\lf(3\ri)=d^{\ast,\frac{10}{3}}_{4,4}\lf(3\ri)$.}\label{fig:TheCasewhereSingleUserandKUserDMTUpperBoundCoincide}}{\epsfig{figure=4.eps,height=4.8cm}\caption{$d^{\ast}\lf(r\ri)$
for $M=K=2$ and $N=4$, i.e., $l=1$. Note that
$d^{\ast}\lf(1\ri)=d^{\ast,8/5}_{2,4}\lf(1\ri)=d^{\ast,\lf(FC\ri)}_{2,4}\lf(1\ri)=d^{\ast,2}_{2,4}\lf(1\ri)$
and $d^{\ast}\lf(\frac{3}{2}\ri)=
d^{\ast,4}_{4,4}\lf(3\ri)=d^{\ast,\lf(FC\ri)}_{4,4}\lf(3\ri)=d^{\ast,\frac{10}{3}}_{4,4}\lf(3\ri)$.}\label{fig:TheCasewhereSingleUserandKUserDMTUpperBoundCoincide}}
\end{figure}

We are now are ready to characterize the upper bound on the optimal
DMT of IC's in the symmetric case. Recall that for
$N=\lf(K-1\ri)M+1+l< \lf(K+1\ri)M-1$, $l=0,\dots, 2M-3$
\con{
\begin{equation*}
d^{\ast}\lf(r\ri) = MN-\lfloor\frac{l}{2}\rfloor\cdot
\lf(\lfloor\frac{l}{2}\rfloor+1\ri)-2\cdot\lf(\lfloor\frac{l}{2}\rfloor+1\ri)\cdot\lf(\frac{l}{2}-\lfloor\frac{l}{2}\rfloor\ri)-
\lf(N+M-1-l\ri)r.
\end{equation*}}
{\baln{}{&d^{\ast}\lf(r\ri) = MN-\lfloor\frac{l}{2}\rfloor\cdot
\lf(\lfloor\frac{l}{2}\rfloor+1\ri)&
\nn\\
&-2\cdot\lf(\lfloor\frac{l}{2}\rfloor+1\ri)\cdot\lf(\frac{l}{2}-\lfloor\frac{l}{2}\rfloor\ri)-
\lf(N+M-1-l\ri)r.&}}
\begin{theorem}\label{th:ICOptimalSymmetricDMTUpperNound}
The optimal DMT of any sequence of IC's in the symmetric case is upper bounded by:\\
For $N\ge \lf(K+1\ri)M-1$
\begin{equation*}
d^{\ast,\lf(IC\ri)}_{K,M,N}\lf(r\ri) =
d^{\ast,\lf(FC\ri)}_{M,N}\lf(r\ri).
\end{equation*}
For $N<\lf(K-1\ri)M+1$
\begin{equation*}
d^{\ast,\lf(IC\ri)}_{K,M,N}\lf(r\ri)= M\cdot N-K\cdot M\cdot r.
\end{equation*}
For  $N=\lf(K-1\ri)M+1+l< \lf(K+1\ri)M-1$, where $l=0,\dots,2M-3$
\con{
\begin{equation*}
d^{\ast,\lf(IC\ri)}_{K,M,N}\lf(r\ri)=\left\{\begin{array}{ll}
d^{\ast,\lf(FC\ri)}_{M,N}\lf(r\ri) & 0\le r\le \lfloor\frac{l}{2}\rfloor+1\\
d^{\ast}\lf(r\ri) & \lfloor\frac{l}{2}\rfloor+1\le r\le \frac{\lf(K-1\ri)M+\lfloor\frac{l+1}{2}\rfloor}{K}\\
d^{\ast,\lf(FC\ri)}_{KM,N}\lf(Kr\ri) &
\frac{\lf(K-1\ri)M+\lfloor\frac{l+1}{2}\rfloor}{K}\le r\le
\frac{L}{K}
\end{array}
\right.
\end{equation*}}
{\baln{}{&d^{\ast,\lf(IC\ri)}_{K,M,N}\lf(r\ri)=&
\nn\\
&\left\{\begin{array}{ll}
d^{\ast,\lf(FC\ri)}_{M,N}\lf(r\ri) & 0\le r\le \lfloor\frac{l}{2}\rfloor+1\\
d^{\ast}\lf(r\ri) & \lfloor\frac{l}{2}\rfloor+1\le r\le \frac{\lf(K-1\ri)M+\lfloor\frac{l+1}{2}\rfloor}{K}\\
d^{\ast,\lf(FC\ri)}_{KM,N}\lf(Kr\ri) &
\frac{\lf(K-1\ri)M+\lfloor\frac{l+1}{2}\rfloor}{K}\le r\le
\frac{L}{K}
\end{array}
\right.&}}
\end{theorem}
\begin{proof}
The proof is in appendix
\ref{append:ICOptimalSymmetricDMTUpperNound}.
\end{proof}
Figure \ref{fig:TheCasewhereSingleUserandKUserDMTUpperBoundCoincide}
also presents $d^{\ast,\lf(IC\ri)}_{K,M,N}\lf(r\ri)$ for $M=K=2$ and
$N=4$ (which leads to $l=1$).

\subsection{Comparison to Finite Constellations}\label{subsec:LowerBoundErrorProb_C}
In this subsection we compare the optimal DMT of finite
constellations to the upper bound on the optimal DMT of IC's (in
general, not only for the symmetric case). This comparison enables
us to show that for $N\ge \lf(K+1\ri)M-1$ the upper bound on the
optimal DMT of IC's coincides with the optimal DMT of finite
constellations. On the other hand for $N<\lf(K+1\ri)M-1$ we show
that the upper bound on the optimal DMT of IC's is inferior compared
to the optimal DMT of finite constellations. This leads to the
conclusion that in the case $N<\lf(K+1\ri)M-1$, the best DMT any
sequence of IC's can attain is suboptimal compared to the optimal
DMT of finite constellations.

In Lemma \ref{lem:ConLemCompareOptimalSymmetricDMTICandFC} we
compare the upper bound on the optimal DMT of IC's in the symmetric
case, to the optimal DMT of finite constellations. Then we use this
result to prove in Theorem
\ref{th:ConvThOptimalityandSuboptimalityofICinMACChannel} that the
optimal DMT of IC's is suboptimal when $N<\lf(K+1\ri)M-1$.

We begin by showing when the upper bound on the optimal DMT of IC's
in the symmetric case, $d^{\ast,\lf(IC\ri)}_{K,M,N}\lf(r\ri)$, is
suboptimal compared to the optimal DMT of finite constellations.
\begin{lem}\label{lem:ConLemCompareOptimalSymmetricDMTICandFC}
For either $N\ge \lf(K+1\ri)M-1$ or $K=2$, $M=s+1$, $N=3\cdot s$,
where $s\ge 1$ and $s\in\mathbb{Z}$ we get
\begin{equation*}
d^{\ast,\lf(IC\ri)}_{K,M,N}\lf(r\ri)=d^{\ast,\lf(FC\ri)}_{K,M,N}\lf(r\ri).
\end{equation*}
For $N<\lf(K-1\ri)M+1$
\begin{equation*}
d^{\ast,\lf(IC\ri)}_{K,M,N}\lf(r\ri)<d^{\ast,\lf(FC\ri)}_{K,M,N}\lf(r\ri)\quad
0<r<\frac{N}{K}.
\end{equation*}
For $N=\lf(K-1\ri)M+1+l< \lf(K+1\ri)M-1$ and $l=0,\dots,2M-3$
\begin{equation*}
d^{\ast,\lf(IC\ri)}_{K,M,N}\lf(r\ri)<d^{\ast,\lf(FC\ri)}_{K,M,N}\lf(r\ri)
\end{equation*}
where $\lfloor\frac{l}{2}\rfloor+1<r<\frac{\lf(K-1\ri)M+\lfloor\frac{l+1}{2}\rfloor}{K}$.
\end{lem}
\begin{proof}
The full proof is in appendix
\ref{Append:ConLemCompareOptimalSymmetricDMTICandFC}. In a nutshell
the proof is based on the properties of $d^{\ast,D}_{M,N}\lf(r\ri)$
derived in Corollary \ref{prop:ICP2pDMTAnchorPoints} as well as
Corollary \ref{prop:RelationBedDandFCd}, and also on the results in
Theorem \ref{th:ICOptimalSymmetricDMTUpperNound}. It is important to
note that for $K=2$, $M=s+1$ and $N=3\cdot s$ we get that
$d^{\ast,\lf(IC\ri)}_{K,M,N}\lf(r\ri)=d^{\ast,\lf(FC\ri)}_{K,M,N}\lf(r\ri)$
because in this case
$\lfloor\frac{l}{2}\rfloor+1=\frac{\lf(K-1\ri)M+\lfloor\frac{l+1}{2}\rfloor}{K}$.
\end{proof}
The sub-optimality of $d^{\ast,\lf(IC\ri)}_{K,M,N}\lf(r\ri)$ for
$N<\lf(K-1\ri)M+1$ is illustrated in Figure
\ref{fig:TheCasewhereDimensionSmallerThanFirstDMTLine}, whereas the
sub-optimality for $N=\lf(K-1\ri)M+1+l$ and $l=0,\dots,2m-3$ is
illustrated in Figure
\ref{fig:TheCasewhereSingleUserandKUserDMTUpperBoundCoincide}.

We now present the cases for which the upper bound on the optimal
DMT of the unconstrained multiple-access channel coincides with the
optimal DMT of finite constellations, and the cases where the
optimal DMT of the unconstrained multiple-access channel is
suboptimal compared to the optimal DMT of finite constellations.

\begin{theorem}\label{th:ConvThOptimalityandSuboptimalityofICinMACChannel}
For $N\ge \lf(K+1\ri)M-1$ the optimal DMT of the unconstrained
multiple-access channel is upper bounded by
$d^{\ast,\lf(FC\ri)}_{M,N}\lf(\max \lf(r_{1},\dots,r_{K}\ri) \ri)$
the optimal DMT of finite constellations. In the case
$N<\lf(K+1\ri)M-1$, the best DMT that can be attained for the
unconstrained multiple-access channel is inferior compared to the
optimal DMT of finite constellations.
\end{theorem}
\begin{proof}
The full proof is in appendix
\ref{append:ConvThOptimalityandSuboptimalityofICinMACChannel_1}. The
proof outline is as follows. Recall that in Theorem
\ref{Th:MDTFormulation} we have shown that the optimal DMT of IC's
is upper bounded by
\con{
\begin{equation*}
d^{\ast,\lf(IC\ri)}_{K,M,N} \lf(r_{1},\dots,r_{K}\ri)= \max_{ \lf
(D_{1},\dots,D_{K}\ri)\in \RD} \min_{ A\subseteq
\lf\{1,\dots,K\ri\}} \lf(d^{\ast,D_{A}}_{|A|\cdot M,N}\lf(R_{A}\ri)\ri).
\end{equation*}}
{\baln{}{&d^{\ast,\lf(IC\ri)}_{K,M,N} \lf(r_{1},\dots,r_{K}\ri)=&
\nn\\
&\max_{ \lf
(D_{1},\dots,D_{K}\ri)\in \RD} \min_{ A\subseteq
\lf\{1,\dots,K\ri\}} \lf(d^{\ast,D_{A}}_{|A|\cdot M,N}\lf(R_{A}\ri)\ri).&}}
For $N\ge \lf(K+1\ri)M-1$ we show that this term is upper and lower
bounded by
$d^{\ast,\lf(FC\ri)}_{M,N}\lf(\max\lf(r_{1},\dots,r_{K}\ri)\ri)$,
which is the optimal DMT of finite constellations in this case.

In the case $N< \lf(K+1\ri)M-1$ we show that the optimal DMT is not
attained by finding a set of multiplexing gain tuples
$\lf(r_{1},\dots,r_{K}\ri)\in B$ for which
$d^{\ast,\lf(IC\ri)}_{K,M,N}\lf(r_{1},\dots,r_{K}\ri)<d^{\ast,\lf(FC\ri)}_{K,M,N}\lf(r_{1},\dots,r_{K}\ri)$.
Based on Lemma \ref{lem:ConLemCompareOptimalSymmetricDMTICandFC} we
get for $r_{1}=\dots=r_{K}=r$ that there exists a set of
multiplexing gains for which
$d^{\ast,\lf(IC\ri)}_{K,M,N}\lf(r\ri)<d^{\ast,\lf(FC\ri)}_{K,M,N}\lf(r\ri)$,
except for $K=2$, $M=s+1$ and $N=3\cdot s$, where $s\ge 1$ is an
integer. For this case showing that
$d^{\ast,\lf(IC\ri)}_{2,s+1,3\cdot
s}\lf(r_{1},r_{2}\ri)<d^{\ast,\lf(FC\ri)}_{2,s+1,3\cdot
s}\lf(r_{1},r_{2}\ri)$ is more involved and requires considering the
case $r_{1}\neq r_{2}$ (see appendix
\ref{append:ConvThOptimalityandSuboptimalityofICinMACChannel_1} for
the full proof). An illustrative example for the method of proof for
this case is presented in Figures
\ref{fig:ConvThOptimalityandSuboptimalityofICinMACChannel},
\ref{fig:ConvThOptimalityandSuboptimalityofICinMACChannel_b}.
\end{proof}

\begin{figure}[h]
\centering
\con{\epsfig{figure=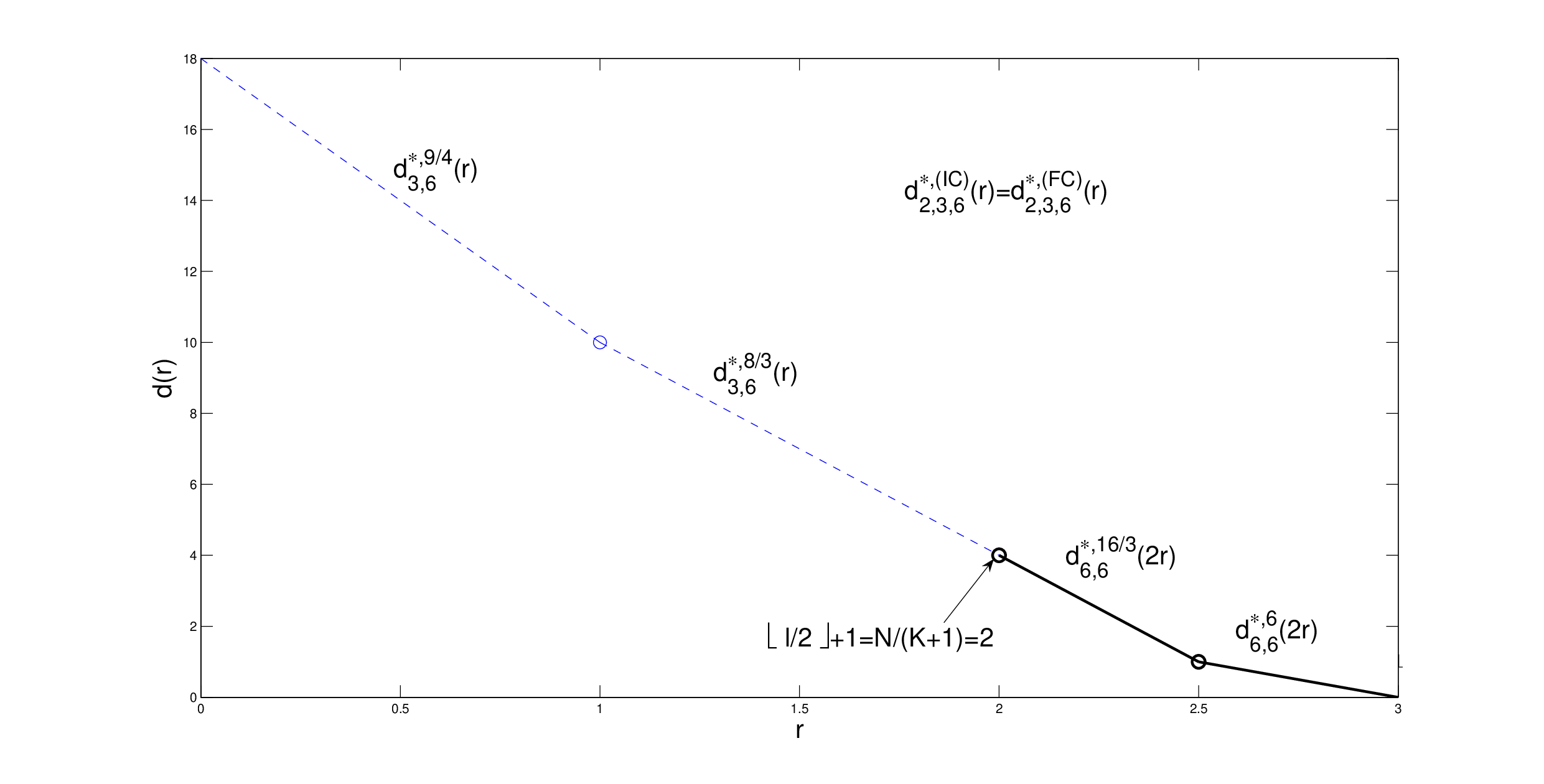,height=5.5cm}\caption{The
upper bound on the optimal DMT of IC's in the symmetric case for
$K=2$, $M=3$, $N=6$. Note that for this case we get
$\lfloor\frac{l}{2}\rfloor+1=\frac{N}{K+1}=\frac{\lf(K-1\ri)
M+1+\lfloor\frac{l+1}{2}\rfloor}{K}$. In addition this upper bound
coincides with the optimal DMT of finite constellations in the
symmetric case. Finally, for this case we get
$d^{\ast,\frac{8}{3}}_{3,6}\lf(r\ri)=d^{\ast,\frac{16}{3}}_{6,6}\lf(2r\ri)$
.}\label{fig:ConvThOptimalityandSuboptimalityofICinMACChannel}
}{\epsfig{figure=5.eps,height=4.8cm}\caption{The
upper bound on the optimal DMT of IC's in the symmetric case for
$K=2$, $M=3$, $N=6$. Note that for this case we get
$\lfloor\frac{l}{2}\rfloor+1=\frac{N}{K+1}=\frac{\lf(K-1\ri)
M+1+\lfloor\frac{l+1}{2}\rfloor}{K}$. In addition this upper bound
coincides with the optimal DMT of finite constellations in the
symmetric case. Finally, for this case we get
$d^{\ast,\frac{8}{3}}_{3,6}\lf(r\ri)=d^{\ast,\frac{16}{3}}_{6,6}\lf(2r\ri)$
.}\label{fig:ConvThOptimalityandSuboptimalityofICinMACChannel}}
\end{figure}

\begin{figure}[h]
\centering
\con{\epsfig{figure=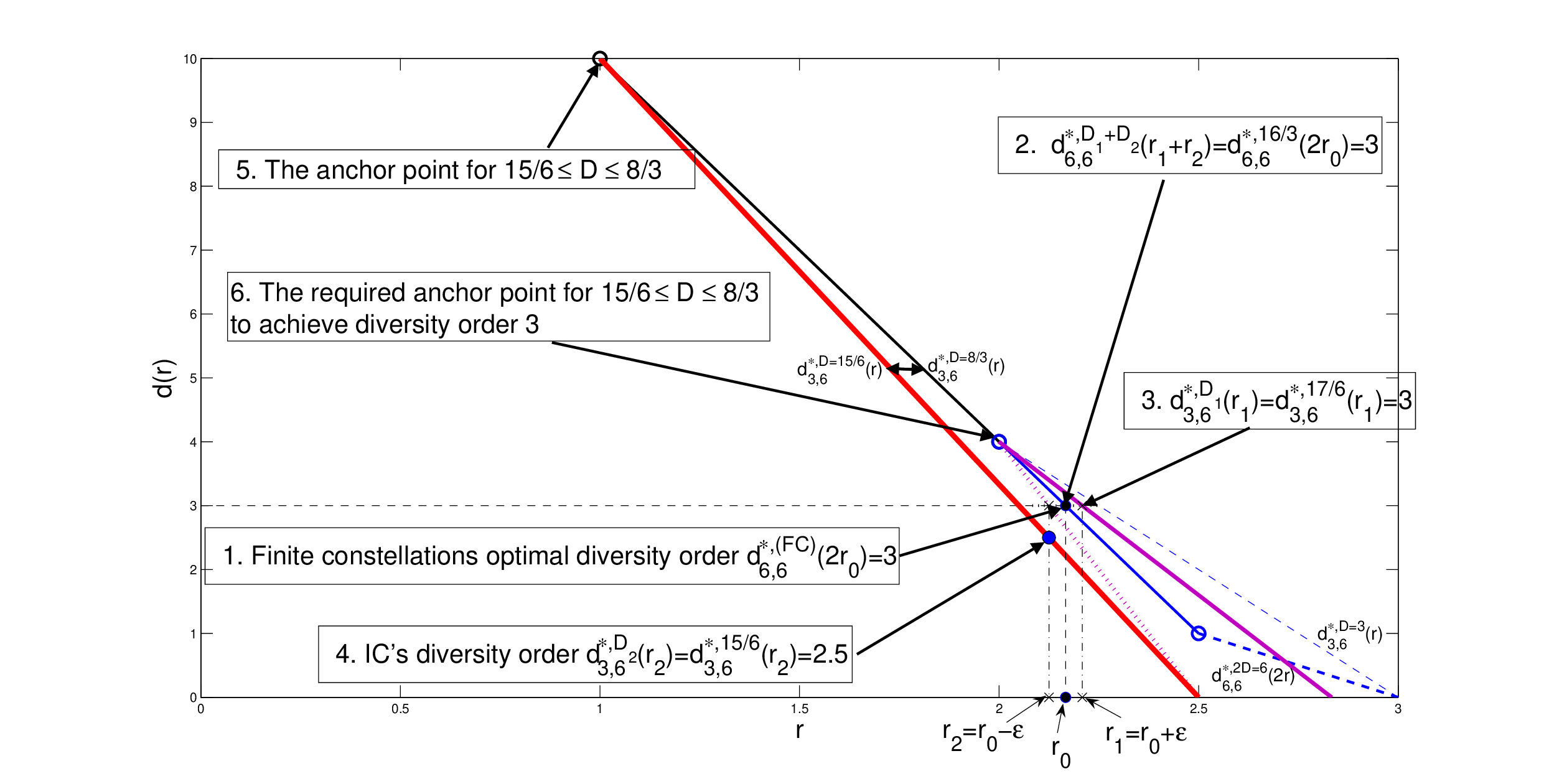,height=9cm}\caption{Illustration
of the sub-optimality of the unconstrained multiple-access channel
for $M=3$, $N=6$ and $K=2$. In this example we take
$r_{1}=r_{0}+\epsilon=\frac{13}{6}+\frac{1}{24}$ and
$r_{2}=r_{0}-\epsilon=\frac{13}{6}-\frac{1}{24}$, where
$r_{0}=\frac{13}{6}$. In this case the optimal diversity order of
finite constellations equals $\min
\lf(d^{\ast,\lf(FC\ri)}_{3,6}\lf(r_{1}\ri),d^{\ast,\lf(FC\ri)}_{3,6}\lf(r_{2}\ri),d^{\ast,\lf(FC\ri)}_{6,6}\lf(r_{1}+r_{2}\ri)\ri)$.
From the figure it can be seen that the minimum is obtained for
$d^{\ast,\lf(FC\ri)}_{6,6}\lf(r_{1}+r_{2}\ri)=d^{\ast,\lf(FC\ri)}_{6,6}\lf(2r_{0}\ri)=3$.
On the other hand IC's diversity order equals $\min
\lf(d^{\ast,D_{1}}_{3,6}\lf(r_{1}\ri),d^{\ast,D_{2}}_{3,6}\lf(r_{2}\ri),d^{\ast,D_{1}+D_{2}}_{6,6}\lf(2
r_{0}\ri)\ri)$. In this example we choose
$D_{1}=\frac{8}{3}+\frac{1}{6}$, $D_{2}=\frac{8}{3}-\frac{1}{6}$. In
this case we get $d^{\ast,D_{1}+D_{2}}_{6,6}\lf(2 r_{0}\ri)=
d^{\ast,\frac{16}{3}}_{6,6}\lf(2 r_{0}\ri)=3$,
$d^{\ast,D_{1}}_{3,6}\lf(r_{1}\ri)=d^{\ast,\frac{17}{6}}_{3,6}\lf(r_{1}\ri)=3$
and
$d^{\ast,D_{2}}_{3,6}\lf(r_{2}\ri)=d^{\ast,\frac{15}{6}}_{3,6}\lf(r_{2}\ri)=\frac{5}{2}<3$.
Hence, in this case the diversity order of IC's is smaller than the
optimal diversity order of finite constellations. It results from
the fact that for $0<D\le\frac{8}{3}$ the straight lines
$d^{\ast,D}_{3,6}\lf(r\ri)$ rotate around anchor points with
multiplexing gain smaller than 2, whereas they should rotate around
anchor point with multiplexing gain
2.}\label{fig:ConvThOptimalityandSuboptimalityofICinMACChannel_b}
}{\epsfig{figure=6.eps,height=4.8cm}\caption{Illustration
of the sub-optimality of the unconstrained multiple-access channel
for $M=3$, $N=6$ and $K=2$. In this example we take
$r_{1}=r_{0}+\epsilon=\frac{13}{6}+\frac{1}{24}$ and
$r_{2}=r_{0}-\epsilon=\frac{13}{6}-\frac{1}{24}$, where
$r_{0}=\frac{13}{6}$. In this case the optimal diversity order of
finite constellations equals $\min
\lf(d^{\ast,\lf(FC\ri)}_{3,6}\lf(r_{1}\ri),d^{\ast,\lf(FC\ri)}_{3,6}\lf(r_{2}\ri),d^{\ast,\lf(FC\ri)}_{6,6}\lf(r_{1}+r_{2}\ri)\ri)$.
From the figure it can be seen that the minimum is obtained for
$d^{\ast,\lf(FC\ri)}_{6,6}\lf(r_{1}+r_{2}\ri)=d^{\ast,\lf(FC\ri)}_{6,6}\lf(2r_{0}\ri)=3$.
On the other hand IC's diversity order equals $\min
\lf(d^{\ast,D_{1}}_{3,6}\lf(r_{1}\ri),d^{\ast,D_{2}}_{3,6}\lf(r_{2}\ri),d^{\ast,D_{1}+D_{2}}_{6,6}\lf(2
r_{0}\ri)\ri)$. In this example we choose
$D_{1}=\frac{8}{3}+\frac{1}{6}$, $D_{2}=\frac{8}{3}-\frac{1}{6}$. In
this case we get $d^{\ast,D_{1}+D_{2}}_{6,6}\lf(2 r_{0}\ri)=
d^{\ast,\frac{16}{3}}_{6,6}\lf(2 r_{0}\ri)=3$,
$d^{\ast,D_{1}}_{3,6}\lf(r_{1}\ri)=d^{\ast,\frac{17}{6}}_{3,6}\lf(r_{1}\ri)=3$
and
$d^{\ast,D_{2}}_{3,6}\lf(r_{2}\ri)=d^{\ast,\frac{15}{6}}_{3,6}\lf(r_{2}\ri)=\frac{5}{2}<3$.
Hence, in this case the diversity order of IC's is smaller than the
optimal diversity order of finite constellations. It results from
the fact that for $0<D\le\frac{8}{3}$ the straight lines
$d^{\ast,D}_{3,6}\lf(r\ri)$ rotate around anchor points with
multiplexing gain smaller than 2, whereas they should rotate around
anchor point with multiplexing gain
2.}\label{fig:ConvThOptimalityandSuboptimalityofICinMACChannel_b}}
\end{figure}

\subsection{Discussion: Convexity Vs. Non-Convexity of the Optimal DMT}\label{subsec:TheReasonsForSuboptimality} It is interesting to note that the
upper bound on the optimal DMT of IC's in the symmetric case is a
convex function, whereas the optimal DMT of finite constellations is
not necessarily so. The convexity of the optimal DMT of IC's can be
shown rather easily by the following arguments. It is based on the
fact that a function that equals to the maximum between straight
lines is a convex function. For $N\ge \lf(K+1\ri)M-1$ the optimal
DMT of IC's in the symmetric case is simply upper bounded by
$d^{\ast,\lf(FC\ri)}_{M,N}\lf(r\ri)$ which is a maximization between
straight lines, and therefore is a convex function. In the case
$N<\lf(K-1\ri)M+1$ the upper bound on the optimal DMT of IC's in the
symmetric case is a straight line. Finally, for $N=\lf(K-1\ri)M+1+l<
\lf(K+1\ri)M-1$, where $l=0,\dots, 2M-3$, the upper bound on the
optimal symmetric DMT of IC's equals to the maximization between the
first $\lfloor\frac{l}{2}\rfloor+1$ straight lines constituting
$d^{\ast,\lf(FC\ri)}_{M,N}\lf(r\ri)$, $d^{\ast}\lf(r\ri)$, and the
last $M-\lfloor\frac{l+1}{2}\rfloor$ straight lines constituting
$d^{\ast,\lf(FC\ri)}_{K\cdot M,N}\lf(K\cdot r\ri)$. This
maximization also yields a convex function.

On the other hand the optimal DMT of finite constellations in the
symmetric case is not necessarily a convex function. See Figure
\ref{fig:TheCasewhereSingleUserandKUserDMTUpperBoundCoincide} for
illustration. In fact the optimal DMT is not a convex function
whenever $N< \lf(K-1\ri)M+1$, or $N=\lf(K-1\ri)M+1+l<
\lf(K+1\ri)M-1$ and $\lfloor\frac{l}{2}\rfloor+1 \neq
\frac{\lf(K-1\ri)M+\lfloor\frac{l+1}{2}\rfloor}{K}$ where
$l=0,\dots, 2M-3$. It results from the following arguments. For
$N<\lf(K-1\ri)M+1$ we get $\frac{MN}{N+M-1}>\frac{N}{K}$, and so
$d^{\ast,\frac{MN}{N+M-1}}_{M,N}\lf(\frac{N}{K}\ri)>0$. In addition
$d^{\ast,\lf(FC\ri)}_{K,M,N}\lf(r\ri)=d^{\ast,\frac{MN}{N+M-1}}_{M,N}\lf(r\ri)$
for $0\le r\le\min \lf(1,\frac{N}{K+1}\ri)$. Based on these facts
and on the facts that $d^{\ast,\lf(FC\ri)}_{K,M,N}\lf(r\ri)$ is a
piecewise linear function and
$d^{\ast,\lf(FC\ri)}_{K,M,N}\lf(\frac{N}{K}\ri)=0$, we get that
$d^{\ast,\lf(FC\ri)}_{K,M,N}\lf(r\ri)$ is not a convex function. For
$N=\lf(K-1\ri)M+1+l< \lf(K+1\ri)M-1$ and $l=0,\dots, 2M-3$, we
know that
\begin{equation*}
d^{\ast,\lf(IC\ri)}_{K,M,N}\lf(r\ri)=d^{\ast}\lf(r\ri)<d^{\ast,\lf(FC\ri)}_{K,M,N}\lf(r\ri)
\end{equation*}
for $\lfloor\frac{l}{2}\rfloor+1<r<\frac{\lf(K-1\ri)M+\lfloor\frac{l+1}{2}\rfloor}{K}$.
Since $d^{\ast}\lf(r\ri)$ is a straight line it necessarily means
that $d^{\ast,\lf(FC\ri)}_{K,M,N}\lf(r\ri)$ is not a convex function
whenever
$\lfloor\frac{l}{2}\rfloor+1\neq\frac{\lf(K-1\ri)M+\lfloor\frac{l+1}{2}\rfloor}{K}$.
For the case $\lfloor\frac{l}{2}\rfloor+1 =
\frac{\lf(K-1\ri)M+\lfloor\frac{l+1}{2}\rfloor}{K}$ we get
$d^{\ast,\lf(FC\ri)}_{K,M,N}\lf(r\ri)=d^{\ast,\lf(IC\ri)}_{K,M,N}\lf(r\ri)$,
and so in this case the optimal DMT of finite constellations in the
symmetric case is also a convex function. Finally, for $N\ge
\lf(K+1\ri)M-1$ the optimal DMT in the symmetric case equals
$d^{\ast,\lf(FC\ri)}_{M,N}$ and as aforementioned it is a convex
function. Therefore, we can state that \emph{whenever the optimal
DMT of finite constellations in the symmetric case is not a convex
function, IC's are suboptimal}.

Finally, a question that may arise is whether it is possible to find
an extension of orthogonal designs
\cite{CalderbankSpaceTimeOrthogonal} to the multiple-access channel,
i.e., a transmission scheme that enables to separate the space-time
code from the symbols required for transmission. The most notable
example of such a transmission scheme is the Alamouti scheme
\cite{AlamoutiScheme} for the case of two transmit antennas and a
single receive antenna. For example, in this case transmitting the
information itself over the space-time code enables to obtain the
optimal DMT $d^{\ast,\lf(FC\ri)}_{2,1}\lf(r\ri)$ regardless of the
constellation size. For the multiple-access channel, if we examine
the optimal DMT of finite constellations for the symmetric case, for
$M=2$, $K=2$ and $N=1$ we get
\begin{equation*}
d^{\ast,\lf(FC\ri)}_{2,2,1}\lf(r\ri)=\left\{\begin{array}{ll}
d^{\ast,\lf(FC\ri)}_{2,1}\lf(r\ri) & 0\le r\le\frac{l}{3}\\
d^{\ast,\lf(FC\ri)}_{4,1}\lf(2r\ri) & \frac{1}{3}\le r\le
\frac{1}{2}
\end{array}
\right.
\end{equation*}
which imply that in the range $0\le r\le\frac{1}{3}$ each user can
obtain the same performance as the Alamouti scheme. However, our
results show that for this setting we get $N=1<\lf(K-1\ri)M+1=3$.
Therefore, the optimal DMT of IC's for the symmetric case is upper bounded by
\begin{equation*}
d^{\ast,\lf(IC\ri)}_{2,2,1}\lf(r\ri)=d^{\ast,\lf(FC\ri)}_{2,1}\lf(2r\ri)
\end{equation*}
which is strictly smaller than $d^{\ast,\lf(FC\ri)}_{2,1}\lf(r\ri)$
except for $r=0$, as illustrated in Figure
\ref{fig:TheReasonsForSuboptimality}. This leads us to the
conclusion that for the multiple-access channel, the signals
required for transmission affect the performance and can not be
separated from the space-time code. This is due to the fact that
when the constellation size is infinite, the performance is
sub-optimal. Hence, in this sense there is no extension of
orthogonal designs to the multiple-access channel.

\begin{figure}[h]
\centering
\con{\epsfig{figure=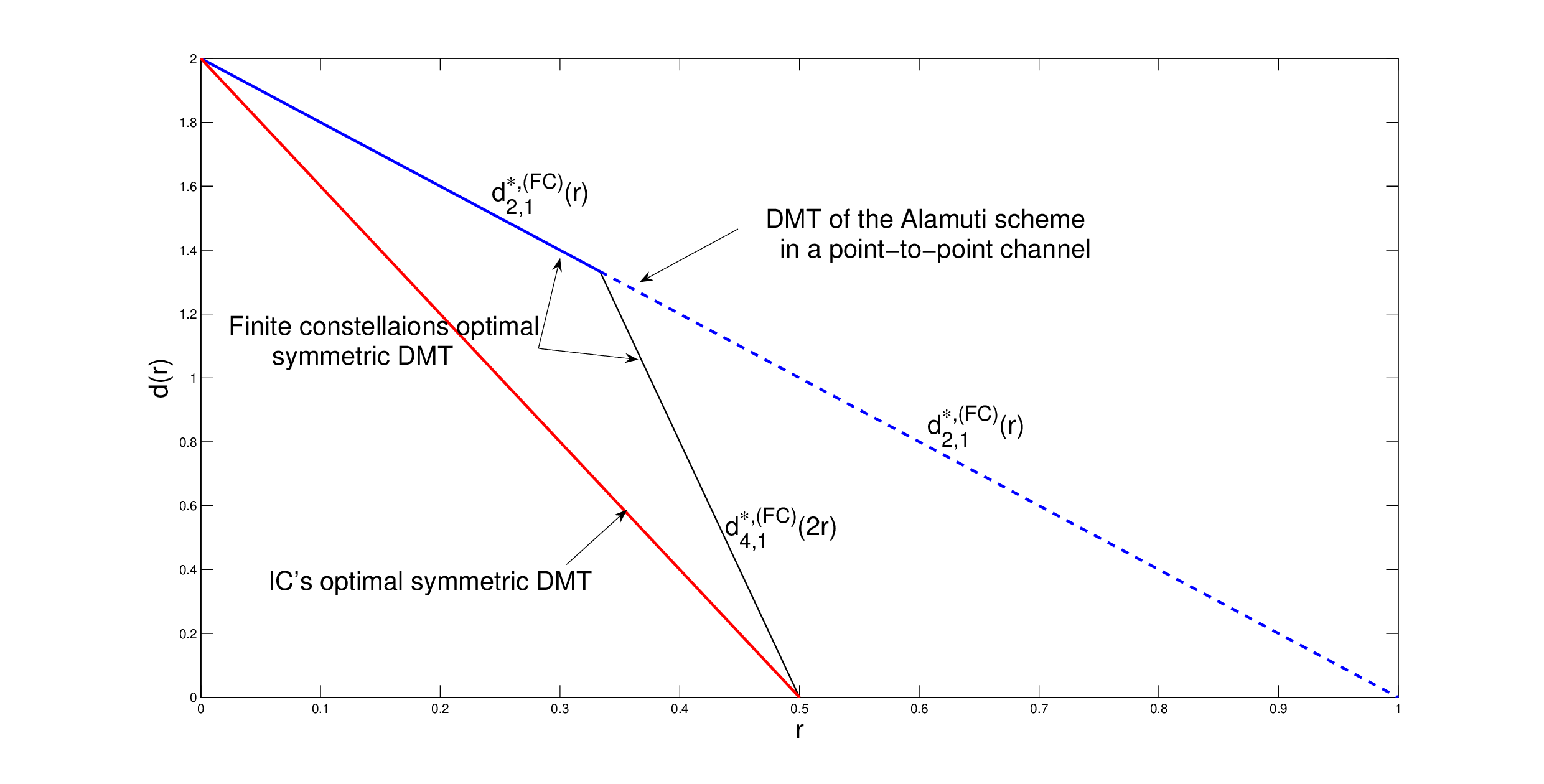,height=5.5cm}\caption{Comparison
between the optimal DMT of finite constellations in the symmetric
case and the upper bound on the optimal DMT of IC's, for $M=K=2$ and
$N=1$. Note that in the range $0\le r\le \frac{1}{3}$ finite
constellations achieve the Alamouti performance, whereas IC's do not. This illustrates that in the multiple-access channel the
constellation and the space-time code can not be
separated.}\label{fig:TheReasonsForSuboptimality}}{\epsfig{figure=7.eps,height=4.8cm}\caption{Comparison
between the optimal DMT of finite constellations in the symmetric
case and the upper bound on the optimal DMT of IC's, for $M=K=2$ and
$N=1$. Note that in the range $0\le r\le \frac{1}{3}$ finite
constellations achieve the Alamouti performance, whereas IC's do not. This illustrates that in the multiple-access channel the
constellation and the space-time code can not be
separated.}\label{fig:TheReasonsForSuboptimality}}
\end{figure}

\section{Attaining the Optimal DMT for $N\ge \lf(K+1\ri)M-1$}\label{sec:AttainingtheOptimalDMT}
In this section we show that the upper bound on the DMT of the
unconstrained multiple-access channel, derived in section
\ref{sec:LowerBoundErrorProb}, is achievable for $N\ge
\lf(K+1\ri)M-1$ by a sequence of IC's in general and lattices in
particular. Essentially, we show for $N\ge \lf(K+1\ri)M-1$ that IC's
attain DMT that equals to
$d^{\ast,\lf(FC\ri)}_{K,M,N}\lf(r_{1},\dots,r_{K}\ri)=d^{\ast,\lf(FC\ri)}_{M,N}\lf(\max
\lf(r_{1},\dots,r_{K}\ri)\ri)$.

We begin by showing in subsection
\ref{subsec:OrthogonalizationSubOptimality} that simple orthogonal
transmission approaches such as time-division multiple-access (TDMA)
or code-division multiple-access (CDMA) will result in sub-optimal
performance for $N\ge \lf(K+1\ri)M-1$. Then, we introduce in
subsection \ref{subsec:TheTransmissionScheme} the transmission
scheme for each user, followed by presentation of the effective
channel induced by the transmission scheme in subsection
\ref{subsec:TheEffectiveChannel}. We derive in subsection
\ref{subsec:DirectsubsecUpperBoundErrorProb} for each channel
realization an upper bound for the error probability of the ML
decoder of an ensemble of $K$ IC's. Finally, in subsection
\ref{subsec:DirectAchieveOptimalDMT} we average this upper bound
over the channel realizations, and show that the optimal DMT is
attained for $N\ge \lf(K+1\ri)M-1$ .

\subsection{Orthogonal Transmission is Sub-optimal}\label{subsec:OrthogonalizationSubOptimality} In
this subsection we show the sub-optimality of transmission methods
that create at the receiver orthogonalization between different
independent streams, for any channel realization. The advantage of
these transmission schemes is their simplicity. By assigning the
IC's or lattices correctly in the space, they enable to consider
each stream independently and reduce the decoding problem to the
point-to-point scenario. Such an approach is very natural when
considering IC's in general and lattices in particular, as it
involves assigning the streams with dimensions or subspaces that
remain orthogonal at the receiver for each channel realization. The
IC related to a certain stream lies within the assigned subspace. We
show for $N\ge \lf(K+1\ri)M-1$ that such transmission method is
sub-optimal as it requires each user to give up too many dimensions
to create the orthogonalization.

At the receiver, orthogonal transmission scheme enables each
independent stream to lie within a subspace orthogonal to the other
streams, for each channel realization. In order for a transmission
scheme to fulfil this property, the streams must be assigned with
orthogonal subspaces already at the transmitter, i.e., must be
assigned with orthogonal subspaces in $\mathbb{C}^{MT}$ assuming
there are $T$ channel uses. Hence, orthogonal transmission schemes
require the partition of at most $M$ number of dimensions per
channel use between \emph{all} users. On the other hand, $N\ge
\lf(K+1\ri)M-1$ leads to $N \ge K\cdot M$, and so potentially the
$K$ users could transmit together up to $KM$ dimensions per channel
use, but not orthogonally. The optimal DMT for the symmetric case
for $N\ge \lf(K+1\ri)M-1$ is $d^{\ast,\lf(FC\ri)}_{M,N}\lf(r\ri)$.
From Corollary \ref{prop:ICP2pDMTAnchorPoints} and Theorem
\ref{th:ICOptimalSymmetricDMTUpperNound} we know that in the range
$M-1\le r\le M$ the optimal DMT is obtained only when each user
transmits over $M$ average number of dimensions per channel use,
i.e., the $K$ users must transmit together $KM$ dimensions per
channel use. Hence, orthogonal transmission is not provided with
enough dimensions per channel use to obtain the last line of the
optimal DMT. This leads to its sub-optimality.

As a first example we consider an orthogonal transmission scheme
that takes the natural partition to $K$ streams induced by the
multiple-access channel. In order to obtain orthogonalization for
this case, at each channel use a different user transmits, while the
others wait for their turn to transmit. This transmission method is
coined TDMA. Let us consider the symmetric case for which each user
transmits at multiplexing gain $r$. For this case, for $T$ channel
uses and $K$ users, each user transmits over $\frac{T}{K}$ channel
uses. Therefore, each user can achieve the point-to-point
performance of a channel with $M$ transmit and $N$ receive antennas,
using $\frac{T}{K}$ channel uses. However, in order for each user to
transmit at multiplexing gain $r$ per channel use, he must transmit
at multiplexing gain $Kr$ over those $\frac{T}{K}$ channel uses,
which leads to DMT performance of
$d^{\ast,\lf(FC\ri)}_{M,N}\lf(Kr\ri)$. This shows the sub-optimality
of TDMA.

Another transmission approach is assigning an independent stream for
each transmit antenna. This is equivalent to considering a
multiple-access channel with $KM$ users, each with a single transmit
antenna. Let us consider for example a multiple-access channel with
$M=1$, $K$ users and $N\ge K$. In this case the optimal DMT for the
symmetric case equals $d^{\ast,\lf(FC\ri)}_{1,N}\lf(r\ri)$. On the
other hand for CDMA each user is assigned with an orthogonal
subspace in $\mathbb{C}^{T}$, assuming there are $T$ channel uses.
In this way each stream can obtain the performance of a
point-to-point channel with a single transmit antenna and $N$
receive antennas. However, for the orthogonalization to hold each
user is assigned with $\frac{T}{K}$ dimensional subspace, which must
be orthogonal to the other users subspaces. Hence, in order for each
user to obtain multiplexing gain $r$ per channel use, he must
transmit at multiplexing gain $Kr$ over the $\frac{T}{K}$
dimensional subspace. This leads to suboptimal DMT performance of
$d^{\ast,\lf(FC\ri)}_{1,N}\lf(Kr\ri)$.


\subsection{The Transmission Scheme}\label{subsec:TheTransmissionScheme}
From subsection \ref{subsec:OrthogonalizationSubOptimality} we get
that an optimal transmission scheme must allow different users to
lie in overlapping subspaces at the receiver, i.e., at the receiver
the users can not reside in orthogonal subspaces. Essentially, for
the proposed transmission scheme each user transmits as if the
channel was a point-to-point channel with $M$ transmit and $N$
receive antennas. Hence, each user transmission matrix is identical
to the transmission matrix presented in
\cite{YonaFederICOptimalDMT}.

We denote the transmission matrix of user $i$ by
$G_{l}^{\lf(i\ri)}$, where $l=0,\dots,M-1$ and $i=1,\dots,K$.
$G_{l}^{\lf(i\ri)}$ has $M$ rows that represent the transmission
antennas, and $T_{l}=N+M-1-2\cdot l$ columns that represent the
number of channel uses. $G_{l}^{\lf(i\ri)}$ transmits over
$D_{l}=\frac{NM-l \lf(l+1\ri)}{N+M-1-2l}$ average number of
dimensions per channel use in the following manner.

Consider a channel with $M$ transmit and $N$ receive antennas.
\begin{enumerate}
\item For
$D_{M-1}=\frac{M(N-M+1)}{N-M+1}=M$: the matrix $G_{M-1}^{\lf(i\ri)}$
has $N-M+1$ columns (channel uses). In the first column transmit
symbols $x_{1},\dots, x_{M}$ on the $M$ antennas, and in the $N-M+1$
column transmit symbols $x_{M(N-M)+1},\dots, x_{M(N-M+1)}$ on the
$M$ antennas.
\item For $D_{l}$, $l=0,\dots,L-2$: the matrix $G_{l}^{\lf(i\ri)}$ has $M+N-1-2\cdot l$ columns.
We add to $G_{l+1}^{\lf(i\ri)}$, the transmission scheme for
$D_{l+1}$, two columns in order to get $G_{l}^{\lf(i\ri)}$. In the
first added column transmit $l+1$ symbols on antennas $1,\dots,l+1$.
In the second added column transmit different $l+1$ symbols on
antennas $M-l,\dots,M$.
\end{enumerate}

According to the definition of the transmission scheme we can see
that the different users transmit the same average number of
dimensions per channel use. Let us denote the transmission scheme of
the first $k$ users by
\begin{equation}
G_{l}^{\lf(1,\dots, k\ri)}=\lf(G_{l}^{\lf(1\ri)\dagger},\dots,
G_{l}^{\lf(k\ri)\dagger}\ri)^{\dagger}\quad k=1,\dots,K.
\end{equation}
$G_{l}^{\lf(1,\dots, k\ri)}$ is a $k \cdot M\times T_{l}$ matrix.
Note that $G_{l}^{\lf(1,\dots, k\ri)}$ transmits over $k\cdot
D_{l}\cdot T_{l}$ dimensions. Later in this section we show that
$G_{l}^{\lf(1,\dots, K\ri)}$ attains the optimal DMT in the range
$l\le r_{max}\le l+1$.

\emph{Example}: $M=2$, $N=5$ and $K=2$. In this case the
transmission scheme for $D_{0}=\frac{10}{6}$, $D_{1}=\frac{8}{4}$
($G_{0}^{\lf(1.2\ri)}$, $G_{1}^{\lf(1.2\ri)}$ respectively) is as
follows:
\con{
\begin{equation}\label{eq:DirectSubsecTransmissiomScheme1}
G_{l}^{\lf(1,2\ri)}=\lf(\begin{array}{c}
G_{l}^{\lf(1\ri)}\\
G_{l}^{\lf(2\ri)}
\end{array}\ri)=\underb{\underb{\left(\begin{array}{ccccc}
x_{1} & x_{3} & x_{5} & x_{7} & | \\
x_{2}   &x_{4} & x_{6} & x_{8} & |\\
-- & -- & -- & -- & | \\
x_{9} & x_{11} & x_{13} & x_{15} & | \\
x_{10}   &   x_{12} & x_{14} & x_{16} & |
\end{array}\right.}_{D_{1}=\frac{8}{4},G_{1}^{\lf(1,2\ri)}}
\left.\begin{array}{cc}
x_{17} & 0 \\
0 & x_{18}\\
-- & -- \\
x_{19} & 0  \\
0 & x_{20}
\end{array}\right)}_{D_{0}=\frac{10}{6},G_{0}^{\lf(1,2\ri)}}.\end{equation}}
{\bal{eq:DirectSubsecTransmissiomScheme1}{&G_{l}^{\lf(1,2\ri)}=\lf(\begin{array}{c}
G_{l}^{\lf(1\ri)}\\
G_{l}^{\lf(2\ri)}
\end{array}\ri)=&
\nn\\
&\underb{\underb{\left(\begin{array}{ccccc}
x_{1} & x_{3} & x_{5} & x_{7} & | \\
x_{2}   &x_{4} & x_{6} & x_{8} & |\\
-- & -- & -- & -- & | \\
x_{9} & x_{11} & x_{13} & x_{15} & | \\
x_{10}   &   x_{12} & x_{14} & x_{16} & |
\end{array}\right.}_{D_{1}=\frac{8}{4},G_{1}^{\lf(1,2\ri)}}
\left.\begin{array}{cc}
x_{17} & 0 \\
0 & x_{18}\\
-- & -- \\
x_{19} & 0  \\
0 & x_{20}
\end{array}\right)}_{D_{0}=\frac{10}{6},G_{0}^{\lf(1,2\ri)}}.&}}

\subsection{The Effective Channel}\label{subsec:TheEffectiveChannel}
Next we define the effective channel matrix induced by the
transmission scheme of the first $k$ users
$G_{l}^{\lf(1,\dots,k\ri)}$, where $k=1,\dots,K$. Let us denote the
first $k$ users transmission at time instance $t$ by
\begin{equation*}
\udl{x}_{t}=\lf(\udl{x}_{t}^{\lf(1\ri)\dagger},\dots,\udl{x}_{t}^{\lf(k\ri)\dagger}\ri)^{\dagger}\quad
t=1,\dots, T_{l}.
\end{equation*}
In accordance with the channel model from \eqref{eq:Channel Fading}
we get
\begin{equation*}
\udl{y}_{t}=H^{\lf(1,\dots,k\ri)}\cdot \udl{x}_{t}\quad
t=1,\dots,T_{l}.
\end{equation*}
where
$H^{\lf(1,\dots,k\ri)}=\lf(H^{\lf(1\ri)},\dots,H^{\lf(k\ri)}\ri)$,
is an $N\times k\cdot M$ matrix. The multiplication
$H^{\lf(1,\dots,k\ri)}\cdot G_{l}^{\lf(1,\dots,k\ri)}$ yields a
matrix with $N$ rows and $T_{l}$ columns, for which each column
equals to $H^{\lf(1,\dots,k\ri)}\cdot\udl{x}_{t}$, $t=1\dots T_{l}$.
Each user is transmitting $D_{l}\cdot T_{l}$-complex dimensional IC
with $D_{l}\cdot T_{l}$-complex symbols, i.e., $G_{l}^{\lf(i\ri)}$
has exactly $D_{l}\cdot T_{l}$ non-zero values representing the
$D_{l}\cdot T_{l}$ complex-dimensional IC within
$\mathbb{C}^{MT_{l}}$. Together, the first $k$ users transmit an
effective $k\cdot D_{l}\cdot T_{l}$-dimensional complex IC within
$\mathbb{C}^{k\cdot M T_{l}}$. For each column of
$G_{l}^{\lf(1,\dots,k\ri)}$, denoted by $\udl{g}_{m}^{\lf(k\ri)}$,
$m=1\dots ,T_{l}$, we define the effective channel that
$\udl{g}_{m}^{\lf(k\ri)}$ sees as $\widehat{H}_{m}$. It consists of
the columns of $H^{\lf(1,\dots,k\ri)}$ that correspond to the
non-zero entries of $\udl{g}_{m}^{\lf(k\ri)}$, i.e.,
$H^{\lf(1,\dots,k\ri)}\cdot\udl{g}_{m}^{\lf(k\ri)}=\widehat{H}_{m}\cdot\udl{\widehat{g}}_{m}^{\lf(k\ri)}$,
where $\udl{\widehat{g}}_{m}^{\lf(k\ri)}$ equals to the non-zero
entries of $\udl{g}_{m}^{\lf(k\ri)}$. As an example assume without
loss of generality that only the first $l_{m}$ entries of
$\udl{g}_{m}^{\lf(k\ri)}$ are not zero. In this case
$\widehat{H}_{m}$ is an $N\times l_{m}$ matrix that equals to the
first $l_{m}$ columns of $H^{\lf(1,\dots,k\ri)}$. In accordance with
\eqref{eq:ExtendedChannelModel}, $H_{\eff}^{(l),k}$ is an
$NT_{l}\times k D_{l}\cdot T_{l}$ block diagonal matrix consisting
of $T_{l}$ blocks. Since each block in $H_{\eff}^{(l),k}$
corresponds to the multiplication of $H^{\lf(1,\dots,k\ri)}$ with
different column in $G_{l}^{\lf(1,\dots,k\ri)}$, the blocks of
$H_{\eff}^{(l),k}$ equal $\widehat{H}_{m}$, $m=1,\dots,T_{l}$. Note
that in the effective matrix $NT_{l}\ge k\cdot D_{l}\cdot T_{l}$.

Next we elaborate on the structure of the blocks of
$H_{\eff}^{(l),k}$. For this reason we denote the \emph{m'th} column
of $H^{\lf(1,\dots,k\ri)}$ by $\udl{h}_{m}$, $m=1,\dots, k\cdot M$.
The transmission scheme has $N+M-1-2\cdot l$ columns. The entries of
the first $N-M+1$ columns of $G_{l}^{\lf(1,\dots,k\ri)}$,
$\udl{g}_{1}^{\lf(k\ri)},\dots,\udl{g}_{N-M+1}^{\lf(k\ri)}$ are all
different from zero.  Hence, the first $N-M+1$ blocks of
$H_{\eff}^{(l),k}$ are
\begin{equation}\label{eq:DirectSubsecTheEffectiveChannel1}
\widehat{H}_{m}=H^{\lf(1,\dots,k\ri)}\qquad m=1,\cdots, N-M+1.
\end{equation}
After the first $N-M+1$ columns we have $M-1-l$ pairs of columns. For each pair we have
\con{\begin{align}\label{eq:DirectSubsecTheEffectiveChannel2}
\widehat{H}_{N-M+2v}=\widehat{H}_{N-M+2 \lf(v-1\ri)}&\setminus
\lf\{\udl{h}_{M-\lf(v-1\ri)},\udl{h}_{2M-\lf(v-1\ri)},\dots,\udl{h}_{kM-\lf(v-1\ri)}\ri\}\nonumber\\
&=\{\udl{h}_{1},\dots,\udl{h}_{M-v},\udl{h}_{M+1},\dots,\udl{h}_{2M-v},\dots,\udl{h}_{\lf(k-1\ri)M+1},\dots,\udl{h}_{k\cdot
M-v}\}
\end{align}}{\bal{eq:DirectSubsecTheEffectiveChannel2}{&\widehat{H}_{N-M+2v}=&
\nn\\
&\widehat{H}_{N-M+2 \lf(v-1\ri)}\setminus
\lf\{\udl{h}_{M-\lf(v-1\ri)},\udl{h}_{2M-\lf(v-1\ri)},\dots,\udl{h}_{kM-\lf(v-1\ri)}\ri\}&
\nn\\
&=\{\udl{h}_{1},\dots,\udl{h}_{M-v},\udl{h}_{M+1},\dots,\udl{h}_{2M-v},\dots&
\nn\\
&\dots,\udl{h}_{\lf(k-1\ri)M+1},\dots,\udl{h}_{k\cdot
M-v}\}&}}
and
\con{\begin{align}\label{eq:DirectSubsecTheEffectiveChannel3}
\widehat{H}_{N-M+2v+1}=\widehat{H}_{N-M+2 \lf(v-1\ri)+1}&\setminus
\lf\{\udl{h}_{v},\udl{h}_{v+M},\dots,\udl{h}_{v+kM}\ri\}\nonumber\\
&=\{\udl{h}_{v+1},\dots,\udl{h}_{M},\udl{h}_{M+v+1},\dots,\udl{h}_{2M},\dots,\udl{h}_{\lf(k-1\ri)M+v+1},\dots,\udl{h}_{k\cdot
M}\}
\end{align}}{\bal{eq:DirectSubsecTheEffectiveChannel3}{&\widehat{H}_{N-M+2v+1}=&
\nn\\
&=\widehat{H}_{N-M+2 \lf(v-1\ri)+1}\setminus
\lf\{\udl{h}_{v},\udl{h}_{v+M},\dots,\udl{h}_{v+kM}\ri\}&
\nn\\
&\{\udl{h}_{v+1},\dots,\udl{h}_{M},\udl{h}_{M+v+1},\dots,\udl{h}_{2M},\dots&
\nn\\
&\dots,\udl{h}_{\lf(k-1\ri)M+v+1},\dots,\udl{h}_{k\cdot
M}\}&}}
where $v=1,\dots, M-1-l$.

\emph{Example}: consider $M=2$, $N=5$ and $K=2$ as presented in
\eqref{eq:DirectSubsecTransmissiomScheme1}. In this case $l=0,1$ and
we have $D_{0}=\frac{10}{6}$ and $D_{1}=\frac{8}{4}=2$ respectively.
In addition
$H^{\lf(1,2\ri)}=\lf(H^{\lf(1\ri)},H^{\lf(2\ri)}\ri)=\lf(\udl{h}_{1},\udl{h}_{2},\udl{h}_{3},\udl{h}_{4}\ri)$.
We begin with $k=1$. In this case we get a point-to-point channel
with $2$ transmit and 5 receive antennas
$H^{\lf(1\ri)}=\lf(\udl{h}_{1},\udl{h}_{2}\ri)$, which leads to the
following effective channels
\begin{enumerate}
\item $D_{1}=2$: $H_{\eff}^{(l=1),k=1}$ is generated from the
multiplication of the $5\times 2$ matrix $H^{\lf(1\ri)}$ with the
four columns of the transmission matrix $G_{1}^{\lf(1\ri)}$. In this
case $H_{\eff}^{(1),1}$ is a $20\times 8$ block diagonal matrix,
consisting of four blocks, where each block equals to
$H^{\lf(1\ri)}$.
\item $D_{0}=\frac{10}{6}$: $H_{\eff}^{(l=0),k=1}$ is a $30\times 10$ block
diagonal matrix consisting of six blocks. The first four blocks are
equal to $H^{\lf(1\ri)}$. The additional two blocks (induced by
columns 5-6 of $G_{0}^{\lf(1\ri)}$) are vectors. We get that
$\widehat{H}_{5}=\udl{h}_{1}$ and $\widehat{H}_{6}=\udl{h}_{2}$.
\end{enumerate}
For $k=2$ the effective channel induced by $G_{l}^{\lf(1,2\ri)}$ is
as follows.
\begin{enumerate}
\item $D_{1}=2$: In this case the effective channel $H_{eff}^{\lf(l=1\ri),k=2}$ is a $20\times 16$ matrix consisting of four blocks, where each block equals $H^{\lf(1,2\ri)}=\lf(H^{\lf(1\ri)},H^{\lf(2\ri)}\ri)$.
\item $D_{0}=\frac{10}{6}$: In this case the effective channel $H_{eff}^{\lf(l=0\ri),k=2}$ is a $30\times
20$ matrix consisting of six blocks. The first four blocks equal to
$H^{\lf(1,2\ri)}$, whereas the other two blocks are
$\widehat{H}_{5}=\lf(\udl{h}_{1},\udl{h}_{3}\ri)$ and
$\widehat{H}_{6}=\lf(\udl{h}_{2},\udl{h}_{4}\ri)$.
\end{enumerate}
We present $H_{\eff}^{(0),2}$ of our example in equation
\eqref{eq:DirectSubsecTheEffectiveChannel4}.

\begin{figure*}
\begin{equation}\label{eq:DirectSubsecTheEffectiveChannel4}
H_{\eff}^{(l=0),k=2}=\left(\begin{array}{cccccc}
H^{\lf(1,2\ri)} &\boldsymbol{0}  &\boldsymbol{0}  &\boldsymbol{0} &\boldsymbol{0} &\boldsymbol{0}\\
\boldsymbol{0} &H^{\lf(1,2\ri)} &\boldsymbol{0} &\boldsymbol{0} &\boldsymbol{0} &\boldsymbol{0}\\
\boldsymbol{0} &\boldsymbol{0} &H^{\lf(1,2\ri)} &\boldsymbol{0} &\boldsymbol{0} &\boldsymbol{0}\\
\boldsymbol{0} &\boldsymbol{0} &\boldsymbol{0} &H^{\lf(1,2\ri)} &\boldsymbol{0} &\boldsymbol{0}\\
\boldsymbol{0} &\boldsymbol{0} &\boldsymbol{0} &\boldsymbol{0} &\lf(\udl{h}_{1},\udl{h}_{3}\ri) &\boldsymbol{0}\\
\boldsymbol{0} &\boldsymbol{0} &\boldsymbol{0} &\boldsymbol{0}
&\boldsymbol{0} &\lf(\udl{h}_{2},\udl{h}_{4}\ri)
\end{array}\right)
\end{equation}
\end{figure*}

Now let us consider the rows of $G_{l}^{\lf(1,\dots,k\ri)}$. Each
row of the transmission matrix is related to the column of
$H^{\lf(1,\dots,k\ri)}$ that multiplies it, i.e., row $j$ in
$G_{l}^{\lf(1,\dots,k\ri)}$ corresponds to column $\udl{h}_{j}$. In
case there is a non zero entry of row $j$ in column $m$ of
$G_{l}^{\lf(1,\dots,k\ri)}$, it means that $\udl{h}_{j}$ occurs in
$\widehat{H}_{m}$. In the next lemma we examine the number of
occurrences of a certain column of $H^{\lf(1,\dots,k\ri)}$ in the
blocks of $H_{eff}^{\lf(l\ri),k}$.

\begin{lem}\label{lem:DirectSubsecTheEffectiveChannel1}
For any $k=1,\dots,K$ consider column $\udl{h}_{a\cdot M+b}$ in
$H^{\lf(1,\dots,k\ri)}$, where $a=0,\dots,k-1$ and $b=1,\dots,M$. In
this case $\udl{h}_{a\cdot M+b}$ occurs only in the first
$m=N-M+1+\min \lf(M-l-1,M-b\ri)+ \min \lf(M-l-1,b-1\ri)$ blocks of
$H_{eff}^{\lf(l\ri),k}$.
\end{lem}
\begin{proof}
Straight forward from the definition of the blocks of
$H_{eff}^{\lf(l\ri),k}$ in
\eqref{eq:DirectSubsecTheEffectiveChannel1},
\eqref{eq:DirectSubsecTheEffectiveChannel2} and
\eqref{eq:DirectSubsecTheEffectiveChannel3}.
\end{proof}

\subsection{Upper Bound on the Error Probability}\label{subsec:DirectsubsecUpperBoundErrorProb} In this
subsection we derive for each channel realization an upper bound for
the error probability of the joint ML decoder of $K$ ensembles of
IC's transmitted on the unconstrained multiple-access channel,
assuming each IC is $D_{l}\cdot T_{l}$-complex dimensional.

In accordance with the definitions in
\ref{subsec:TheEffectiveChannel} we denote the effective channel of
any set of users pulled together by $H_{eff}^{\lf(l\ri),\lf(s\ri)}$,
where $s\subseteq \lf\{1,\dots,K\ri\}$\footnote{Note that in
\ref{subsec:TheEffectiveChannel} we considered the case of the first
$k$ users for $k=1,\dots, K$. The extension to any $s\subseteq
\lf\{1,\dots,K\ri\}$ is straight forward.}. We define
$|H_{\eff}^{(l),\lf(s\ri) \dagger}\cdot
H_{\eff}^{(l),\lf(s\ri)}|=\rho^{-\sum_{i=1}^{|s|\cdot D_{l}\cdot
T_{l}}\eta_{i}^{\lf(s\ri)}}$ , where
$\rho^{-\frac{\eta_{i}^{\lf(s\ri)}}{2}}$ is the \emph{i'th} singular
value of $H_{\eff}^{(l),\lf(s\ri)}$, $1\le i\le |s|\cdot D_{l}\cdot
T_{l}$. We also define
$\udl{\eta}^{\lf(s\ri)}=(\eta_{1}^{\lf(s\ri)},\dots,\eta_{|s|\cdot
D_{l}\cdot T_{l}}^{\lf(s\ri)})^{T}$. Note that in our setting
$NT_{l}\ge K\cdot D_{l}\cdot T_{l}$.

\begin{theorem}\label{Th:DirectUpperBoundErrorProb}
Consider $K$ ensembles of $D_{l}\cdot T_{l}$-complex dimensional
IC's transmitted on the unconstrained multiple-access channel with
effective channel $H_{eff}^{\lf(l\ri),K}$ and densities
$\gamma_{tr}^{\lf(i\ri)}=\rho^{T_{l}r_{i}}$, $i=1,\dots,K$. The
average decoding error probability of the joint ML decoder is upper
bounded by
\con{\begin{align*}
\ol{Pe}(H_{\eff}^{(l),K},\rho)\le \sum_{s\subseteq
\lf\{1,\dots,K\ri\}}\ol{Pe}(\udl{\eta}^{\lf(s\ri)},\rho)=
\sum_{s\subseteq \lf\{1,\dots,K\ri\}} D(|s|\cdot D_{l}\cdot
T_{l})\rho^{-\cdot T_{l}(|s|D_{l}-\sum_{i\in
s}r_{i})+\sum_{i=1}^{|s|\cdot
D_{l}\cdot T_{l}}\eta_{i}^{\lf(s\ri)}}\nonumber\\
 =\sum_{s\subseteq \lf\{1,\dots,K\ri\}}D(|s|\cdot
D_{l}\cdot T_{l})\rho^{-T_{l}(|s|D_{l}-\sum_{i\in
s}r_{i})}\cdot|H_{\eff}^{(l),\lf(s\ri)\dagger}\cdot
H_{\eff}^{(l),\lf(s\ri)}|^{-1}
\end{align*}}{\baln{}{&\ol{Pe}(H_{\eff}^{(l),K},\rho)\le \sum_{s\subseteq
\lf\{1,\dots,K\ri\}}\ol{Pe}(\udl{\eta}^{\lf(s\ri)},\rho)=&
\nn\\
&\sum_{s\subseteq \lf\{1,\dots,K\ri\}} D(|s|\cdot D_{l}\cdot
T_{l})\rho^{-T_{l}(|s|D_{l}-\sum_{i\in
s}r_{i})+\sum_{i=1}^{|s|\cdot
D_{l}\cdot T_{l}}\eta_{i}^{\lf(s\ri)}}&
\nn\\
&=\sum_{s\subseteq \lf\{1,\dots,K\ri\}}D(|s|\cdot
D_{l}\cdot T_{l})\rho^{-T_{l}(|s|D_{l}-\sum_{i\in
s}r_{i})}\times&
\nn\\
&|H_{\eff}^{(l),\lf(s\ri)\dagger}\cdot
H_{\eff}^{(l),\lf(s\ri)}|^{-1}&
}}
where $D(|s|\cdot D_{l}\cdot T_{l})$ is a constant independent of
$\rho$, and $\eta_{i}^{\lf(s\ri)}\ge 0$ for any $s\subseteq
\lf\{1,\dots,K\ri\}$ and any $1\le i\le |s|\cdot D_{l}\cdot T_{l}$.
\end{theorem}
\begin{proof}
The proof is based on dividing the error event into events of error
for different sets of users (disjoint events). Then we show that the
upper bound on the error probability for the point-to-point channel
derived in \cite{YonaFederICOptimalDMT} can be used to upper bound
the probability for each of these events. The full proof is in
appendix \ref{Append:DirectUpperBoundErrorProb}.
\end{proof}

We wish to emphasize that the constraint of $\eta_{i}^{\lf(s\ri)}\ge
0$, for $i=1,\dots, |s|\cdot D_{l}\cdot T_{l}$ and for any
$s\subseteq \lf\{1,\dots,K\ri\}$ results from the fact that the
\emph{same} ensemble is upper bounded for \emph{any} channel
realization. In cases where it is possible to fit an ensemble to
each channel realization, i.e., in the case where the transmitter
knows the channel, the upper bound applies also without this
restriction.

\subsection{Achieving the Optimal
DMT}\label{subsec:DirectAchieveOptimalDMT} In this subsection we
show that the transmission scheme proposed in
\ref{subsec:TheTransmissionScheme} attains the optimal DMT for $N\ge
\lf(K+1\ri)M-1$, $d^{\ast,\lf(FC\ri)}_{M,N} \lf(\max
\lf(r_{1},\dots, r_{K}\ri)\ri)$. We base the proof on the upper
bound for the error probability derived in Theorem
\ref{Th:DirectUpperBoundErrorProb}. This upper bound consists of the
sum of several terms, one for each $s\subseteq \lf\{1,\dots,K\ri\}$.
Each term depends on the determinant corresponding to its effective
channel $|H_{eff}^{\lf(l\ri),\lf(s\ri)\dagger}\cdot
H_{eff}^{\lf(l\ri),\lf(s\ri)}|^{-1}$. For each term (for each $s$)
we upper bound this determinant in Lemma
\ref{lem:DirectLowerBoundDeterminantHeff} (different bounds than the
bounds used in \cite{YonaFederICOptimalDMT}) to get a new upper
bound on the error probability. The upper bound is based on the fact
that a determinant equals to the multiplication of the orthogonal
elements of its columns (when the number of rows is larger than the
number of columns). We average the upper bound over the channel
realizations and show it attains the optimal DMT in Theorem
\ref{Th:DirectLowerBoundDiversityOrder}, and also prove that the
results apply to lattices when regular lattice decoder is employed
at the receiver, in Theorem
\ref{Th:LowerBoundDiversityOrderLattices}.

Each term in the upper bound in Theorem
\ref{Th:DirectUpperBoundErrorProb} can be viewed as the error
probability of a point-to-point channel with $|s|\cdot M$ transmit
antennas and $N$ receive antennas, while transmitting an $|s|\cdot
D_{l}\cdot T_{l}$-complex dimensional IC in the method described in
\ref{subsec:TheTransmissionScheme}. We wish to emphasize that in
this subsection we show that the terms corresponding to $|s|=1$
attain the required optimal DMT since each user uses an optimal
transmission scheme for the point-to-point channel with $M$ transmit
and $N$ receive antennas. However, for the terms corresponding to
$1<|s|\le K$ the effective transmission scheme is no longer optimal
and does not necessarily attain the optimal DMT for a point-to-point
channel with $|s|\cdot M$ transmit and $N$ receive antennas. In fact
it does not even necessarily attain $d^{\ast,|s|\cdot
D_{l}}_{|s|\cdot M,N}\lf(\max \lf(r_{1},\dots,r_{K}\ri)\ri)$. Hence,
the challenge in this subsection is to upper bound the DMT of these
terms and show that, although not optimal for the corresponding
point-to-point channel, they attain the optimal DMT of the
multiple-access channel for $N\ge \lf(K+1\ri)M-1$.

The average decoding error probability equals to the average over
all channel realizations, i.e.,
\begin{equation}
\ol{Pe}(\rho)=E_{H}\lf(\ol{Pe}\lf(H_{\eff}^{(l),K},\rho\ri)\ri).
\end{equation}
Based on Theorem \ref{Th:DirectUpperBoundErrorProb} we get the
following upper bound on the average decoding error probability
\con{\begin{equation}\label{eq:DirectAchieveOptimalDMT13}
\ol{Pe}(\rho)\le \sum_{s\subseteq
\lf\{1,\dots,K\ri\}}E_{H}\lf(D(|s|\cdot D_{l}\cdot
T_{l})\rho^{-T_{l}(|s|D_{l}-\sum_{i\in
s}r_{i})}\cdot|H_{\eff}^{(l),\lf(s\ri)\dagger}\cdot
H_{\eff}^{(l),\lf(s\ri)}|^{-1}\ri).
\end{equation}}{\bal{eq:DirectAchieveOptimalDMT13}
{&\ol{Pe}(\rho)\le \sum_{s\subseteq
\lf\{1,\dots,K\ri\}}D(|s|\cdot D_{l}\cdot
T_{l})\rho^{-T_{l}(|s|D_{l}-\sum_{i\in
s}r_{i})}\times&
\nn\\
&E_{H}\lf(|H_{\eff}^{(l),\lf(s\ri)\dagger}\cdot
H_{\eff}^{(l),\lf(s\ri)}|^{-1}\ri).&}}
Note that
$E_{H}\lf(|H_{\eff}^{(l),\lf(s\ri)\dagger}H_{\eff}^{(l),\lf(s\ri)}|^{-1}\ri)=E_{H}\lf(|H_{\eff}^{(l),|s|\dagger}H_{\eff}^{(l),|s|}|^{-1}\ri)$
for any $|s|=k$, where $k=1,\dots,K$, i.e., the mean value for any
the users equals to the mean value for the first $k$ users.
Therefore, by replacing $H_{\eff}^{(l),\lf(s\ri)}$ with
$H_{\eff}^{(l),|s|}$ we can write
\eqref{eq:DirectAchieveOptimalDMT13} as follows
\con{\begin{equation}\label{eq:DirectAchieveOptimalDMT4}
\ol{Pe}(\rho)\le \sum_{s\subseteq \lf\{1,\dots,K\ri\}}D(|s|\cdot
D_{l}\cdot T_{l})\rho^{-T_{l}(|s|D_{l}-\sum_{i\in s}r_{i})}\cdot
E_{H}\lf(|H_{\eff}^{(l),|s|\dagger}\cdot
H_{\eff}^{(l),|s|}|^{-1}\ri).
\end{equation}}{\bal{eq:DirectAchieveOptimalDMT4}
{&\ol{Pe}(\rho)\le \sum_{s\subseteq \lf\{1,\dots,K\ri\}}D(|s|\cdot
D_{l}\cdot T_{l})\rho^{-T_{l}(|s|D_{l}-\sum_{i\in s}r_{i})}\times&
\nn\\
&E_{H}\lf(|H_{\eff}^{(l),|s|\dagger}\cdot
H_{\eff}^{(l),|s|}|^{-1}\ri).&}}
where $H_{\eff}^{(l),|s|}$ is the effective channel of the first
$|s|$ users, as defined in subsection
\ref{subsec:TheEffectiveChannel}.

The channel matrix $H$ consists of $N\cdot K\cdot M$ i.i.d entries,
where each entry has distribution $h_{i,j}\sim \CN(0,1)$, $1\le i\le
N$ and $1\le j\le K\cdot M$. Without loss of generality we consider
the case where the columns of $H$ are drawn sequentially from left
to right, i.e., $\udl{h}_{1}$ is drawn first, then $\udl{h}_{2}$ is
drawn et cetera. Column $\udl{h}_{j}$ is an $N$-dimensional vector.
Given $\udl{h}_{1},\dots,\udl{h}_{j-1}$, let us denote by
$\udl{\widetilde{h}}_{j}\in\mathbb{C}^{N}$ the elements of the
projection of $\udl{h}_{j}$ on an orthonormal basis that depends on
$\udl{h}_{1},\dots,\udl{h}_{j-1}$. We can write
\begin{equation}
\udl{h}_{j}=\Theta(\udl{h}_{1},\dots,\udl{h}_{j-1})\cdot\udl{\widetilde{h}}_{j}
\end{equation}
where $\Theta(\cdot)$ is an $N\times N$ unitary matrix.
$\Theta(\cdot)$ is chosen such that:
\begin{enumerate}
\item The first element of $\udl{\widetilde{h}}_{j}$, $\widetilde{h}_{1,j}$,  is
in the direction of $\udl{h}_{j-1}$.
\item The second element, $\widetilde{h}_{2,j}$, is in the direction orthogonal to
$\udl{h}_{j-1}$, in the hyperplane spanned by
$\{\udl{h}_{j-1},\udl{h}_{j-2}\}$.
\item Element $\widetilde{h}_{j-1,j}$ is in the direction
orthogonal to the hyperplane spanned by
$\{\udl{h}_{2},\dots,\udl{h}_{j-1}\}$ inside the hyperplane spanned
by $\{\udl{h}_{1},\dots,\udl{h}_{j-1}\}$.
\item The rest of the $N-j+1$
elements are in directions orthogonal to the hyperplane
$\{\udl{h}_{1},\dots,\udl{h}_{j-1}\}$.
\end{enumerate}
Note that $\widetilde{h}_{i,j}$, $1\le i\le N$, $1\le j\le K\cdot M$
are i.i.d random variables with distribution $\CN(0,1)$. Let us
denote by $\udl{h}_{j\perp j-1,\dots,j-k}$ the component of
$\udl{h}_{j}$ which resides in the $N-k$ subspace which is
perpendicular to the space spanned by
$\{\udl{h}_{j-1},\dots,\udl{h}_{j-k}\}$. In this case we get

\begin{equation}\label{eq:DirectAchieveOptimalDMT1}
\lv\udl{h}_{j\perp
j-1,\dots,j-k}\rv^{2}=\sum_{i=k+1}^{N}|\widetilde{h}_{i,j}|^2\quad
1\le k\le j-1.
\end{equation}

If we assign $|\widetilde{h}_{i,j}|^{2}=\rho^{-\xi_{i,j}}$, we get
that the probability density function (PDF) of $\xi_{i,f}$ is
\begin{equation}\label{eq:DirectAchieveOptimalDMT2}
f(\xi_{i,j})=C\cdot\log{\rho}\cdot\rho^{-\xi_{i,j}}\cdot
e^{-\rho^{-\xi_{i,j}}}
\end{equation}
where $C$ is a normalization factor. In our analysis we assume a
very large value for $\rho$. Hence, we can neglect events in which
$\xi_{i,j}<0$ since in this case the PDF
\eqref{eq:DirectAchieveOptimalDMT2} decreases exponentially as a
function of $\rho$. For a very large $\rho$, $\xi_{i,j}\ge 0$, $1\le
i\le N$ and $1\le j\le K\cdot M$, the PDF takes the following form
\begin{equation}\label{eq:DirectAchieveOptimalDMT3}
f(\xi_{i,j})\propto \rho^{-\xi_{i,j}} \qquad \xi_{i,j}\ge 0.
\end{equation}
In this case by assigning in \eqref{eq:DirectAchieveOptimalDMT1} the
vector $\udl{\xi}_{j}=(\xi_{1,j},\dots,\xi_{N,j})^{T}$ with PDF
which is proportional to $\rho^{-\sum_{i=1}^{N}\xi_{i,j}}$, we get

\begin{equation}\label{eq:DirectAchieveOptimalDMT5}
\lv\udl{h}_{j\perp
j-1,\dots,j-k}\rv^{2}\dot{=}\rho^{-\min_{z\in\{k+1,\dots,N\}}\xi_{z,j}}
\end{equation}
where $1\le k\le j-1$. In addition
\begin{equation}\label{eq:DirectAchieveOptimalDMT6}
\lv\udl{h}_{j}\rv^{2}\dot{=}\rho^{-\min_{z\in\{1,\dots,N\}}\xi_{z,j}}.
\end{equation}

As presented in \eqref{eq:DirectAchieveOptimalDMT4}, in order to
calculate the upper bound on the error probability we need to
consider only the effective channel of the first $|s|$ users, $1\le
|s|\le K$. Hence, in order to obtain an upper bound for the error
probability we wish to lower bound the determinant
$|H_{\eff}^{(l),|s|\dagger}\cdot H_{\eff}^{(l),|s|}|$ by lower
bounding the contribution of each column in the channel matrix $H$
to the determinant. The following lemma presents a lower bound on
the determinant.

\begin{lem}\label{lem:DirectLowerBoundDeterminantHeff}
\con{\begin{align*}
|H_{\eff}^{(l),|s|\dagger}\cdot
H_{\eff}^{(l),|s|}|\dot{\ge}\prod_{a=0}^{|s|-1}
\prod_{b=1}^{M}\rho^{-\lf(N-M+1+\min
\lf(M-l-1,M-b\ri)\ri)\cdot\min_{z\in
\lf\{aM+b,\dots,N\ri\}}\xi_{z,aM+b}}
\nonumber\\
\cdot\prod_{b^{'}=2}^{M}\rho^{-\sum_{i=1}^{\min
\lf(M-l-1,b^{'}-1\ri)}\min_
{z\in\lf\{aM+b^{'}-i,\dots,N\ri\}}\xi_{z,aM+b^{'}}}.
\end{align*}}{\baln{}{&|H_{\eff}^{(l),|s|\dagger}\cdot
H_{\eff}^{(l),|s|}|\dot{\ge}\prod_{a=0}^{|s|-1}
\prod_{b=1}^{M}\times&
\nn\\
&\rho^{-\lf(N-M+1+\min
\lf(M-l-1,M-b\ri)\ri)\cdot\min_{z\in
\lf\{aM+b,\dots,N\ri\}}\xi_{z,aM+b}}&
\nn\\
&\times\prod_{b^{'}=2}^{M}\rho^{-\sum_{i=1}^{\min
\lf(M-l-1,b^{'}-1\ri)}\min_
{z\in\lf\{aM+b^{'}-i,\dots,N\ri\}}\xi_{z,aM+b^{'}}}.&
}}
\end{lem}
\begin{proof}
The proof is in appendix
\ref{Append:DirectLowerBoundDeterminantHeff}. Essentially, the term
$\lf(N-M+1+\min \lf(M-l-1,M-b\ri)\ri)\cdot\min_{z\in
\lf\{aM+b,\dots,N\ri\}}\xi_{z,aM+b}$ indicates that in the lower
bound column $\udl{h}_{aM+b}$ occurs $N-M+1+\min \lf(M-l-1,M-b\ri)$
times with $\udl{h}_{1},\dots,\udl{h}_{aM+b-1}$ to its left.
Therefore, only the elements of $\udl{h}_{aM+b}$ which are
orthogonal to this set of columns, $\xi_{z,aM+b}$, where $aM+b\le
z\le N$ contribute to the lower bound.

The term
\beqn{}{\sum_{i=1}^{\min \lf(M-l-1,b^{'}-1\ri)}\cdot\min_
{z\in\lf\{aM+b^{'}-i,\dots,N\ri\}}\xi_{z,aM+b^{'}}}
indicates that
column $\udl{h}_{aM+b^{'}}$ occurs $\min \lf(M-l-1,b^{'}-1\ri)$
times. However, this time we handle the contribution of the
orthogonal elements more carefully. For $1\le i\le\min
\lf(M-l-1,b^{'}-1\ri)$ we consider the elements in
$\udl{h}_{aM+b^{'}}$ which are orthogonal to the set of columns
$\udl{h}_{1},\dots,\udl{h}_{aM+b^{'}-i-1}$.
\end{proof}

Now we are ready to lower bound the transmission scheme DMT, based
on the lower bound on the determinant in Lemma
\ref{lem:DirectLowerBoundDeterminantHeff}. Let us denote the maximal
multiplexing gain by $r_{max}=\max \lf(1,\dots,K\ri)$, and also
assume $l=\lfloor r_{max}\rfloor$.

\begin{theorem}\label{Th:DirectLowerBoundDiversityOrder}
Consider $K$ sequences of ensembles of $D_{l}\cdot T_{l}$-complex
dimensional IC's transmitted over the unconstrained multiple-access
channel, where each user transmits at multiplexing-gain $r_{i}$
using $G_{\lfloor r_{max}\rfloor}^{\lf(i\ri)}$, $i=1,\dots,K$. The
DMT this transmission scheme attains is lower bounded by
$d^{\ast,\lf(FC\ri)}_{M,N}\lf(r_{max}\ri)$.
\end{theorem}
\begin{proof}
We use the upper bound for the error probability derived in Theorem
\ref{Th:DirectUpperBoundErrorProb}, and the lower bound on the
determinant \eqref{eq:DirectAchieveOptimalDMT12} in order to give a
new upper bound on the error probability. We average this upper
bound over the channel realization, and show that for large $\rho$
the diversity order of the most dominant error event is lower
bounded by $d^{\ast,\lf(FC\ri)}_{M,N}\lf(r_{max}\ri)$. The full
proof is in appendix \ref{Append:DirectLowerBoundDiversityOrder}.
\end{proof}

In Theorem \ref{th:ConvThOptimalityandSuboptimalityofICinMACChannel}
we have shown that for $N\ge \lf(K+1\ri)M-1$ the DMT of any IC is
upper bounded by $d^{\ast,\lf(FC\ri)}_{M,N}\lf(r_{max}\ri)$. On the
other hand in Theorem \ref{Th:DirectLowerBoundDiversityOrder} we
have shown that there exist sequences of IC's that attain DMT which
is lower bounded by $d^{\ast,\lf(FC\ri)}_{M,N}\lf(r_{max}\ri)$.
Hence, the transmission scheme must attain the optimal DMT.

In the next theorem we prove the existence of a sequence of lattices
that attains the optimal DMT as in Theorem
\ref{Th:DirectLowerBoundDiversityOrder}.

\begin{theorem}\label{Th:LowerBoundDiversityOrderLattices}
For each tuple of multiplexing gains $\lf(r_{1},\dots,r_{K}\ri)$
there exist $K$ sequences of $2D_{l}\cdot T_{l}$-real dimensional
lattices transmitted over the unconstrained multiple access channel
that attain diversity order of
$d^{\ast,\lf(FC\ri)}_{M,N}\lf(r_{max}\ri)$, when regular lattice
decoder is employed, where $l=\lfloor r_{max}\rfloor$.
\end{theorem}
\begin{proof}
See appendix \ref{Append:LatticesDiversityOrder}
\end{proof}

Now we show that for each segment of the optimal DMT there exists a
sequence of $K$ lattices that attains it, i.e., the optimal DMT
consists of $M$ segments, each in the range $l\le r_{max}\le l+1$
for $l=0,\dots, M-1$, and there are $M$ sequences of lattices that
attain it.
\begin{cor}\label{Cor:SequencesLattocesAttainOptimalDMT}
For $N\ge \lf(K+1\ri)M-1$ each segment of the optimal DMT for the
unconstrained multiple-access channel,
$d^{\ast,\lf(FC\ri)}_{M,N}\lf(r_{max}\ri)$, is attained by a
sequence of $K$, $2D_{\lfloor r_{max}\rfloor}T_{\lfloor
r_{max}\rfloor}$-real dimensional lattices.

\end{cor}
\begin{proof}
See appendix \ref{Append:SequencesLattocesAttainOptimalDMT}.
\end{proof}

\subsection{The Gap from the Upper Bound for $N< \lf(K+1\ri)M-1$}
In section \ref{sec:LowerBoundErrorProb} we presented an upper bound
on the optimal DMT of IC's; We showed that when $N<\lf(K+1\ri)M-1$
IC's can not achieve the optimal DMT of finite constellations.
However, a question that remains open is how tight is the upper
bound in this range. In this subsection we give two examples for the
performance of IC's when $N< \lf(K+1\ri)M-1$, using the transmission
scheme presented in subsection \ref{subsec:TheTransmissionScheme}.
From the examples it follows that there are cases in which IC's
achieve the upper bound for the symmetric case; however in general
the upper bound is not necessarily tight when $N< \lf(K+1\ri)M-1$.

As a first example let us consider the case where $N=M=K=2$, for
which the upper bound on the optimal DMT of IC's in the symmetric
case is
\begin{equation*}
d^{\ast,\lf(IC\ri)}_{2,2,2} \lf(r\ri)=4-4r.
\end{equation*}
It can be shown by using the technique we presented in this section,
that for the transmission matrix
\begin{equation*}
G^{\lf(1,2\ri)} = \lf(\begin{array}{cc} x_{1} &0\\
0 &x_{2}\\
x_{3} &0\\
0 &x_{4}
\end{array}\ri)
\end{equation*}
a random ensemble of IC's can achieve $d^{\ast,\lf(IC\ri)}_{2,2,2}
\lf(r\ri)$. Thus, in this setting the upper bound on the DMT of IC's
is tight in the symmetric case.

We now consider the case where $M=K=2$ and $N=4$. In this case the
upper bound consists of the following three straight lines
\begin{equation*}
d^{\ast,\lf(IC\ri)}_{2,2,4} \lf(r\ri)=\lf\{\begin{array}{cc} 8-5r
&0\le r\le 1\\
7-4r & 1\le r\le\frac{3}{2}\\
4-2r & \frac{3}{2}\le r\le 2
\end{array}\ri.
\end{equation*}
Consider the case where each user uses the optimal transmission
scheme for a point-to-point channel with $M=2$ and $N=4$ by using
the transmission matrix
\begin{equation*}
G^{\lf(1,2\ri)}_{0} = \lf(\begin{array}{ccccc}
x_{1} &x_{3} &x_{5} &x_{7} &0\\
x_{2} &x_{4} &x_{6} &0  &x_{8}\\
x_{9} &x_{11} &x_{13} &x_{15} &0\\
x_{10} &x_{12} &x_{14} &0  &x_{16}
\end{array}\ri)
\end{equation*}
for $0\le r\le 1$, and
\begin{equation*}
G^{\lf(1,2\ri)}_{1} = \lf(\begin{array}{ccc}
x_{1} &x_{3} &x_{5}\\
x_{2} &x_{4} &x_{6}\\
x_{7} &x_{9} &x_{11}\\
x_{8} &x_{10} &x_{12}
\end{array}\ri)
\end{equation*}
when $1\le r\le 2$. The DMT of this transmission scheme
$\frac{16}{3}-\frac{10}{3}r$ for $0\le r\le 1$, and $4-2r$ when
$1\le 1\le 2$, as shown in Figure \ref{fig:SuboptforNsmkpl1timesM}.
Therefore, this transmission scheme DMT coincides with the upper
bound only when $\frac{3}{2}\le r\le 2$. We wish to emphasize that
using this transmission scheme simply provides a lower bound for the
optimal DMT of IC's in this setting, and there may exist other
transmission schemes that attain $d^{\ast,\lf(IC\ri)}_{2,2,4}
\lf(r\ri)$.

\begin{figure}[h]
\centering
\epsfig{figure=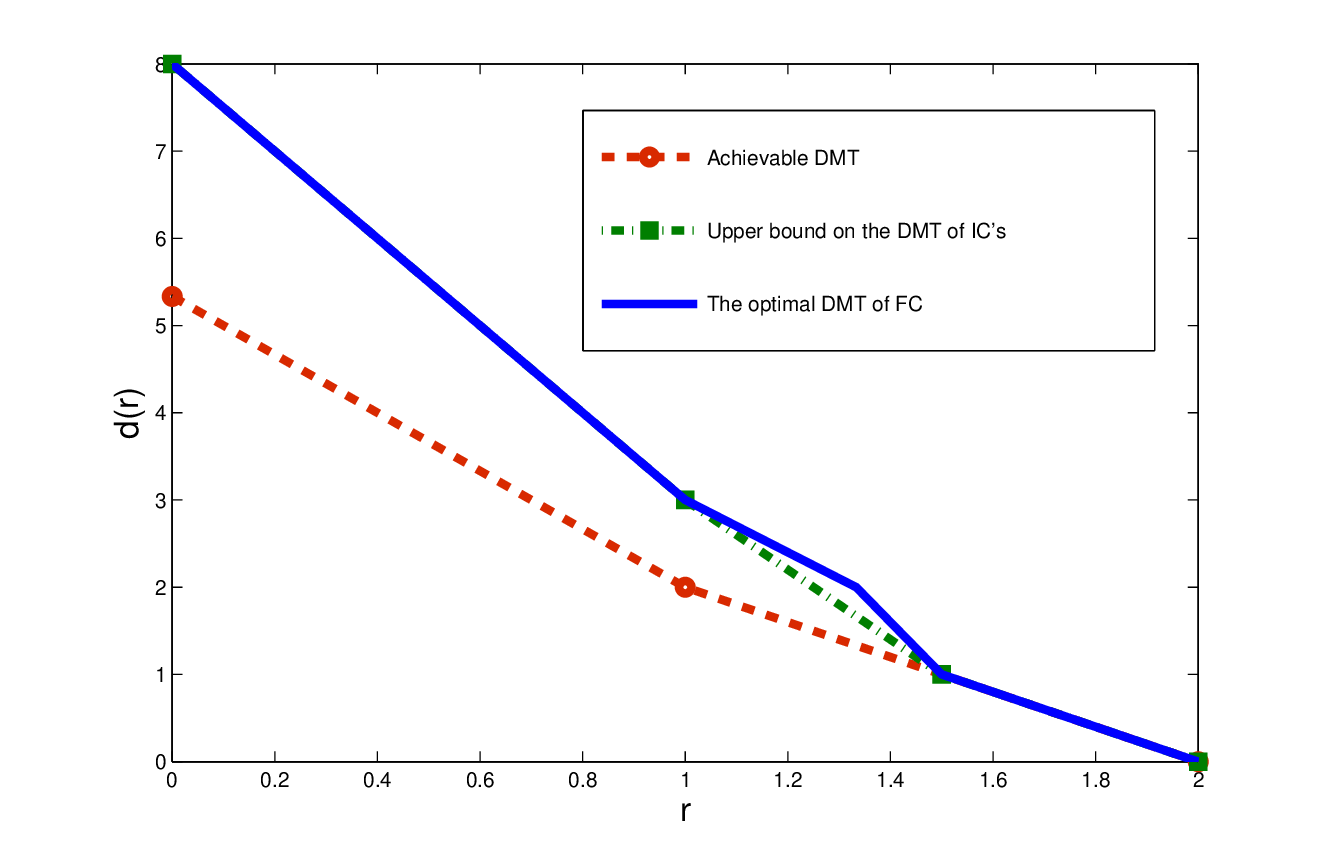,height=5.5cm}\caption{The
gap between the upper bound on the DMT of IC's, and the DMT of the
transmission scheme from subsection
\ref{subsec:TheTransmissionScheme}, for $M=K=2$ and $N=4$.}
\label{fig:SuboptforNsmkpl1timesM}
\end{figure}

In summary, from these examples it follows that when $N<
\lf(K+1\ri)M-1$ the upper bound on the DMT of IC's is not
necessarily tight; nonetheless it enables to show the suboptimality
of IC's in this range.

\section{Discussion}\label{sec:Discussionsec}
In this section we discuss the results presented in the paper. As an
illustrative example we consider the case in which there are two
users, each with two transmit antennas, i.e., $K=M=2$. We consider
the symmetric case in which $r_{1}=r_{2}=r$, and explain based on
Theorem \ref{th:ICOptimalSymmetricDMTUpperNound} why for $N=2,4$
IC's are suboptimal. On the other hand based on Theorem
\ref{Th:DirectUpperBoundErrorProb} and Theorem
\ref{Th:DirectLowerBoundDiversityOrder} we explain why the optimal
DMT is attained for $N\ge 5$. The analysis in this section is
somewhat loosed and we refer the reader to Sections
\ref{sec:LowerBoundErrorProb}, \ref{sec:AttainingtheOptimalDMT} for
the full analysis.

We begin by giving a short reminder to the behavior of lattices in a
point-to-point channel for $M=N=2$, as presented in
\cite{YonaFederICOptimalDMT}. We consider in this discussion
lattices although the results apply to IC's in general. In this
case, the optimal DMT equals
$d^{\ast,\lf(FC\ri)}_{2,2}\lf(r\ri)=4-3r$ in the range $0\le r\le
1$, and in order to attain it the average number of dimensions per
channel use, $D$, must be equal to $\frac{4}{3}$. We wish to explain
why for $D\neq \frac{4}{3}$ the optimal DMT is not attained in the
range $0\le r\le 1$. For lattices, obtaining multiplexing gain $r>0$
requires \emph{scaling} each dimension of the lattice by
$\rho^{-\frac{r}{2D}}$. When $D<\frac{4}{3}$ diversity order of $4$
may be attained for $r=0$. However, the scaling is too strong and
does not enable to attain the optimal DMT for any $r>0$ (there are
not enough degrees of freedom to attain the straight line $4-3r$).
On the other hand when $D>\frac{4}{3}$, the lattice ``fills'' too
much of the space and the \emph{channel} induces error probability
that does not enable to attain diversity order of $4$ for $r=0$, and
therefore does not allow attaining the optimal DMT in the range
$0\le r\le 1$. Hence, choosing $D=\frac{4}{3}$ balances the effect
of the scaling and the channel on the lattice and allows to attain
the optimal DMT in the range $0\le r\le 1$. We now follow this
intuition to discuss the multiple-access channel.

\subsection{Why IC's are Suboptimal for $N< \lf(K+1\ri)M-1$}
The error event for the multiple-access channel can be divided into
the disjoint error events of any subset of the users, as described
in Theorem \ref{Th:DirectUpperBoundErrorProb}. Consider a certain
subset of users $s\subseteq \lf\{1,\dots,K\ri\}$. Due to the
distributed nature of the multiple-access channel, the error
probability for this subset is upper bounded by the error
probability of a point-to-point channel with $|s|\cdot M$ transmit
and $N$ receive antennas, i.e., corresponding to a point-to-point
channel in which the users in $s$ are pulled together. Hence, the
DMT in the multiple-access channel is determined by the most
probable error event. For the unconstrained multiple-access channel
the problem is more involved as each IC has a certain average number
of dimensions per channel use. Assume user $i$ has $D_{i}$ average
number of dimensions per channel use, where $1\le i\le K$. When
considering the error event of users in $s$, we consider an IC with
$\sum_{i\in s}D_{i}$ average number of dimensions per channel use.
The DMT in this error event is upper bounded by $d^{\ast,\sum_{i\in
s}D_{i}}_{|s|\cdot M,N} \lf(|s|\cdot r\ri)$, i.e., the bounds derived
in \cite{YonaFederICOptimalDMT} for the point-to-point channel. In
case the dimensions of \emph{any} subset of the users do not
``align'', i.e., in case a certain subset of the users has average
number of dimensions per channel use that is too large or too small
to attain the optimal DMT, we get sub-optimality. In this subsection
we take as example the case $M=K=2$ and explain why for $N=2,4$ the
dimensions do not align, and therefore the optimal DMT is not
attained.

Let us begin with the case $M=K=N=2$. In this case the optimal DMT
in the symmetric case equals
\con{\begin{equation}
d^{\ast,\lf(FC\ri)}_{K,M,N}\lf(r\ri)=d^{\ast,\lf(FC\ri)}_{2,2,2}\lf(r\ri)=\lf\{\begin{array}{cc}
d^{\ast,\lf(FC\ri)}_{2,2}\lf(r\ri) &0\le r\le \frac{2}{3}\\
d^{\ast,\lf(FC\ri)}_{4,2}\lf(2r\ri) &\frac{2}{3}<r\le 1
\end{array}\ri. =\lf\{\begin{array}{cc}
4-3r &0\le r\le \frac{2}{3}\\
6-6r &\frac{2}{3}<r\le 1
\end{array}\ri..
\end{equation}}{\bal{}{&d^{\ast,\lf(FC\ri)}_{K,M,N}\lf(r\ri)=d^{\ast,\lf(FC\ri)}_{2,2,2}\lf(r\ri)=&
\nn\\
&\lf\{\begin{array}{cc}
d^{\ast,\lf(FC\ri)}_{2,2}\lf(r\ri) &0\le r\le \frac{2}{3}\\
d^{\ast,\lf(FC\ri)}_{4,2}\lf(2r\ri) &\frac{2}{3}<r\le 1
\end{array}\ri. =\lf\{\begin{array}{cc}
4-3r &0\le r\le \frac{2}{3}\\
6-6r &\frac{2}{3}<r\le 1
\end{array}\ri..&}}
On the other hand the optimal DMT of IC's in this case is upper
bounded by $d^{\ast,\lf(IC\ri)}_{2,2,2}\lf(r\ri)=4 \lf(1-r\ri)$,
which is smaller than the optimal DMT for any $0<r<1$. Let us
explain the reason for the sub-optimality. First, note that in the
symmetric case we must choose $D_{1}=D_{2}$ to maximize the IC's
DMT, i.e., the users have the same average number of dimensions per
channel use. Since $N=2$ each user can not transmit more than one
average number of dimensions per channel use, whereas in
\cite{YonaFederICOptimalDMT} it was shown that each user needs to
transmit $\frac{4}{3}$ average number of dimensions per channel use
in order to attain $d^{\ast,\lf(FC\ri)}_{2,2}\lf(r\ri)$ in the range
$0\le r \le \frac{2}{3}$. In addition, the maximal diversity order
each user may attain is $4$ since $M=N=2$, and also
$d^{\ast,1}_{2,2} \lf(r\ri)$ is a straight line. Hence, even when
transmitting one dimension per channel use the DMT must be smaller
than $6-6r$. Therefore, in this case the dimension mismatch manifest
itself in the fact that $N$ is too small even to attain the first
line of $d^{\ast,\lf(FC\ri)}_{2,2}\lf(r\ri)$. This sub-optimality is
presented in Figure
\ref{fig:TheCasewhereDimensionSmallerThanFirstDMTLine}.

For $K=M=2$ and $N=4$ it was shown in Theorem
\ref{th:ICOptimalSymmetricDMTUpperNound} for the symmetric case that
IC's are suboptimal in the range $1<r<\frac{3}{2}$. In this range
the DMT of IC's is upper bounded by $7-4r$, attained at
$D_{1}=D_{2}=\frac{7}{4}$. The dimension mismatch manifests itself
in this example both in error events of a single user, and the error
event of both users. For error events of a single user the optimal
DMT is $d^{\ast,\lf(FC\ri)}_{2,4} \lf(r\ri)$ which is also the
optimal DMT of the multiple-access channel in the range $1\le r \le
\frac{N}{K+1}=\frac{4}{3}$. The average number of dimensions per
channel use required to attain $d^{\ast,\lf(FC\ri)}_{2,4} \lf(r\ri)$
for $1\le r \le 2$ is $2$ which is larger than
$D_{1}=D_{2}=\frac{7}{4}$. Therefore, for the single user error
events the scaling of the IC of each user is too strong and does not
enable to attain the optimal DMT. On the other hand, for the two
users error event the optimal DMT is $d^{\ast,\lf(FC\ri)}_{4,4}
\lf(2r\ri)$ which is also the optimal DMT in the range
$\frac{4}{3}\le r \le 2$. The effective IC of the two users pulled
together has average number of dimensions per channel use
$D_{1}+D_{2}=\frac{7}{2}$, which is too large compared to what is
required to attain $d^{\ast,\lf(FC\ri)}_{2,2} \lf(2r\ri)$ in the
range $1<r<\frac{3}{2}$. Hence, for this error event we get that the
effective IC fills too much of the space and so the channel does not
enable to attain the optimal DMT.

\subsection{Why IC's Attain the Optimal DMT for $N\ge \lf(K+1\ri)M-1$}
For $N\ge \lf(K+1\ri)M-1$ there is no longer a dimension mismatch.
However, the condition that there is no dimension mismatch is merely
a necessary condition in order to attain the optimal DMT. Hence, in
this subsection we will explain why the optimal DMT is attained
based on the transmission scheme presented in subsection
\ref{subsec:TheTransmissionScheme} and on the effective channel
presented in \ref{subsec:TheEffectiveChannel}.

We consider as an example the case $M=K=2$ and $N=5$. We show why
for this case the single user performance $d^{\ast,\lf(FC\ri)}_{2,5}
\lf(r_{\max}\ri)$ is attained. For simplicity we will focus on the
symmetric case. Essentially, we show for this example that IC's
achieve the first DMT line, $10-6r$, which coincides with the
optimal DMT $d^{\ast,\lf(FC\ri)}_{2,5} \lf(r\ri)$ in the range $0\le
r\le 1$. The transmission scheme $G_{0}^{\lf(1,2\ri)}$ is presented
in \eqref{eq:DirectSubsecTransmissiomScheme1}. Note that each user
uses an optimal transmission scheme for the point-to-point channel
with $2$ transmit and $5$ receive antennas. Hence, for the error
event of each of the users, the DMT is upper bounded by $10-6r$
which is the optimal DMT in the range $0\le r\le 1$. Now, it is left
to show for the error event of the two users, that the DMT is also
upper bounded by $10-6r$. For this case we consider the effective
lattice of the two users pulled together, i.e., an error event for a
lattice transmitted over a point-to-point channel with $4$ transmit
and $5$ receive antennas. For this lattice the average number of
dimensions per channel use equals $D_{1}+D_{2}=\frac{10}{3}$. We
will show that at $r=0$ this lattice attains diversity order $10$.
This will lead to DMT $10-6r$ since the DMT of a lattice is a
straight line, and $D_{1}+D_{2}=\frac{10}{3}$.

At the receiver, the effective radius of the lattice of the two
users pulled together at $r=0$ is
\con{\begin{equation}\label{eq:DiscussionDirect1}
r_{eff}^{2}\dot{=}|V|^{\frac{1}{\lf(D_{1}+D_{2}\ri)T}}=\gamma_{\rc}^{-\frac{1}{\lf(D_{1}+D_{2}\ri)T}}\dot{=}|H_{\eff}^{(l=0),K\dagger}H_{\eff}^{(l=0),K}|^{\frac{1}{\lf(D_{1}+D_{2}\ri)T}}
\end{equation}}{\bal{eq:DiscussionDirect1}
{&r_{eff}^{2}\dot{=}|V|^{\frac{1}{\lf(D_{1}+D_{2}\ri)T}}=\gamma_{\rc}^{-\frac{1}{\lf(D_{1}+D_{2}\ri)T}}&
\nn\\
&\dot{=}|H_{\eff}^{(l=0),K\dagger}H_{\eff}^{(l=0),K}|^{\frac{1}{\lf(D_{1}+D_{2}\ri)T}}&}}
where $|V|=\gamma_{\rc}^{-1}$ is the volume of the Voronoi region of
the effective lattice at the receiver. Recall that for lattices
$r_{\eff}\ge r_{packing}=\frac{d_{\min}^{\lf(lattice\ri)}}{2}$,
where $r_{packing}$ and $d_{\min}^{\lf(lattice\ri)}$ are the packing
radius and the minimal distance of the lattice respectively. We are
interested in the event where $r_{\eff}^{2}$ is in the order of the
additive noise variance $\rho^{-1}$. In this case
$\lf(d_{\min}^{\lf(lattice\ri)}\ri)^{2}$ is in the order of the
noise variance or worse, and so the error probability does not
reduce with $\rho$. In subsection
\ref{subsec:DirectAchieveOptimalDMT} it is shown that this event is
the dominant error event in determining the DMT of the transmission
scheme. From \eqref{eq:DiscussionDirect1} we get that
$H_{\eff}^{(l=0),K}$ determines the effective radius at the
receiver. From \eqref{eq:DirectSubsecTransmissiomScheme1} and the
description of the effective channel in subsection
\ref{subsec:TheEffectiveChannel} we get that $H_{\eff}^{(l=0),K}$ is
a block diagonal matrix, where $4$ of its blocks equal
$H\in\mathbb{C}^{5\times 4}$. For large $\rho$, the most probable
error event ($r_{\eff}^{2}\dot{=}\rho^{-1}$) occurs when the
determinant of $H$ reduces with $\rho$, and the determinants of the
rest of the blocks in $H_{\eff}^{(l=0),K}$ remain constant with
$\rho$. Note that if $|H^{\dagger}H|=\rho^{-\alpha}$, then most
likely that the smallest singular value of $H$ equals
$\rho^{-\alpha}$ and the rest of the singular values remain constant
\cite{TseDivMult2003}. In this case we get
$|H^{\dagger}H|\dot{=}\rho^{-\alpha}$ with a PDF which is
proportional to $\rho^{-\lf(5-4+1\ri)\alpha}=\rho^{-2\alpha}$. By
assigning $\lf(D_{1}+D_{2}\ri)T=20$ and
$|H_{\eff}^{(l=0),K\dagger}H_{\eff}^{(l=0),K}|\dot{=}|H^{\dagger}H|^{4}\dot{=}\rho^{-4\alpha}$
in \eqref{eq:DiscussionDirect1} we get that
\begin{equation}
r_{\eff}^{2}\dot{=}|H^{\dagger}H|^{-\frac{4}{20}}\dot{=}\rho^{-\frac{\alpha}{5}}
\end{equation}
with a PDF which is proportional to $\rho^{-2\alpha}$. Hence,
$r_{\eff}^{2}=\rho^{-1}$ at $\alpha=-5$. Based on subsection
\ref{subsec:DirectAchieveOptimalDMT} we get for large $\rho$ that
this is the most dominant error event, and by assigning $\alpha=5$
we get that it happens with probability $\rho^{-10}$. Therefore, in
this case diversity order of $10$ is attained.

For general $N= \lf(K+1\ri)M-1$ each user uses an optimal
transmission scheme for a point-to-point channel with $M$ transmit
and $N$ receive antennas. Since the users do not cooperate, at worst
we get that $H_{\eff}^{(l=0),K}$ has $N-M+1$ blocks that equal
$H\in\mathbb{C}^{N\times K\cdot M}$. For large $\rho$, we get that
$|H^{\dagger}H|=\rho^{-\alpha}$ with PDF proportional to
$\rho^{-\lf(N-K\cdot M+1\ri)\alpha}$. For this case
$\lf(\sum_{i=1}^{K}D_{i}\ri)T=K\cdot M\cdot M$ and so we get
\begin{equation}\label{eq:DiscussionDirect2}
r_{\eff}^{2}\dot{=}|H^{\dagger}H|^{-\frac{N-M+1}{\lf(\sum_{i=1}^{K}D_{i}\ri)T}}\dot{=}\rho^{-\frac{\lf(N-M+1\ri)\alpha}{KMN}}.
\end{equation}
Since $N= \lf(K+1\ri)M-1$, there is a sufficient amount of equations
at the receiver to get $N-M+1=K\cdot M$ and $N-K\cdot M+1=M$. Hence,
by substituting in \eqref{eq:DiscussionDirect2} we get
\begin{equation}
r_{\eff}^{2}\dot{=}\rho^{-\frac{\alpha}{N}}
\end{equation}
with PDF proportional to $\rho^{-\lf(N-KM+1\ri)\cdot
\alpha}=\rho^{-M\cdot \alpha}$. Therefore, at $\alpha=N$ we get that
$r_{\eff}^{2}=\rho^{-1}$ with probability $\rho^{-MN}$, which leads
to diversity order $MN$ at $r=0$. In addition,
$\sum_{i=1}^{K}D_{i}=\frac{KMN}{N-M+1}$ and so the first line of the
optimal DMT is attained. Note that we considered the error event for
the $K$ users pulled together. For any of the other error events,
which considers a subset $s\subseteq \lf(1,\dots,K\ri)$ of the $K$
users, the diversity order is larger or equal to $MN$ at $r=0$.

In summary, since the users do not cooperate we get at worst $N-M+1$
occurrences of $H$ in the blocks of $H_{\eff}^{(l=0),K}$. However,
when $N\ge \lf(K+1\ri)M-1$ there is a sufficient amount of receive
antennas to compensate for the impact of $H$ on $r_{\eff}^{2}$, by
decreasing the probability that $H$ has small determinant.

\section{Summary and Further Research}
This work studies the DMT of the unconstrained multiple-access
channel. For $N\ge \lf(K+1\ri)M-1$ an explicit upper bound on the
optimal DMT of IC's for any multiplexing-gain tuple is presented.
The upper bound coincides with the optimal DMT of finite
constellations, for the multiple-access channel . A transmission
scheme that attains this upper bound is also introduced and
analyzed.

In the case $N< \lf(K+1\ri)M-1$ an upper bound on the optimal DMT of
IC's is derived. For the general case this upper bound remains in
the form of a maximization problem. This maximization problem
depends on $|s|$, the number of IC's pulled together for $1\le
|s|\le K$, and on the average number of dimensions per channel use
for each user. On the other hand for finite constellations the
maximization depends only on the number of users pulled together.
Hence, finding the upper bound on the optimal DMT of IC's is more
involved. In the symmetric case, where all users transmit with the
same multiplexing gain, an explicit upper bound on the optimal DMT
of IC's is presented for $N<\lf(K+1\ri)M-1$. By using this upper
bound, it is shown that IC's are suboptimal compared to finite
constellations in this case.

While this work presents a transmission scheme that attains the
optimal DMT for $N\ge \lf(K+1\ri)M-1$, for the case
$N<\lf(K+1\ri)M-1$ the upper bound on the optimal DMT of IC's is
attained only for some cases. For instance whenever $N=1$,
orthogonalization attains the optimal DMT of IC's for the symmetric
case. Also for $K=2$, $M=2$ and $N=3$, the transmission scheme
presented in this paper attains the upper bound on the optimal DMT
of IC's for the symmetric case. However, finding a transmission
scheme that attains the upper bound on the optimal DMT for all $N<
\lf(K+1\ri)M-1$, remains an open problem even for the symmetric
case.
\begin{appendices}

\section{Proof of Lemma \ref{lem:TheCasewhereSingleUserIstheWorst}}\label{append:TheCasewhereSingleUserIstheWorst}
The proof outline is as follows. First we show that for finite
constellations, the single user DMT is smaller than the contracted
optimal DMT of any number of users (up to $K$) pulled together. Then
we use this relation, together with the anchor points presented in
Corollary \ref{prop:ICP2pDMTAnchorPoints} for the upper bound on
IC's DMT, in order to prove the lemma.

Since $K>1$ and $M$ are positive integers, we get for $N\ge
\lf(K+1\ri)M-1$ that $M\le\frac{N}{i}$, where $1\le i\le K$. Hence
for any $d^{\ast,i\cdot D}_{i\cdot M,N}\lf(i\cdot r\ri)$, the range
of average number of dimensions per channel use per user is $0\le
D\le\min \lf(M,\frac{N}{i}\ri)=M$, where $1\le i\le K$.

We begin by showing that $d^{\ast,\lf(FC\ri)}_{M,N}\lf(r\ri)$ is
smaller or equal to $d^{\ast,\lf(FC\ri)}_{i\cdot M,N}\lf(i\cdot
r\ri)$ for $2\le i\le K$, where $d^{\ast,\lf(FC\ri)}_{i\cdot
M,N}\lf(i\cdot r\ri)$ is the optimal DMT of finite constellations
contracted by $i$, in a point-to-point channel with $i\cdot M$
transmit and $N$ receive antennas. In the case $N> \lf(K+1\ri)M-1$
we get that $\frac{N}{K+1}\ge M$. Hence we also get that
$\frac{N}{i+1}\ge M$ for $1\le i\le K$. Hence, from Theorem
\ref{prop:FCOptDMT} we can see that
\begin{equation}\label{eq:FCDMTBOUNDForGreater}
d^{\ast,\lf(FC\ri)}_{M,N} \lf(r\ri)\le d^{\ast,\lf(FC\ri)}_{i\cdot
M,N} \lf(i\cdot r\ri)\quad 2\le i\le K
\end{equation}
by replacing $K$ with $i$.

For $N= \lf(K+1\ri)M-1$ we still get that $\frac{N}{i+1}\ge M$ for
$1\le i\le K-1$, and again based on Theorem \ref{prop:FCOptDMT}
\begin{equation}
d^{\ast,\lf(FC\ri)}_{M,N} \lf(r\ri)\le d^{\ast,\lf(FC\ri)}_{i\cdot
M,N} \lf(i\cdot r\ri)\quad 2\le i\le K-1.
\end{equation}
For the remaining case of $i=K$, we can see that for
$N=\lf(K+1\ri)M-1$ we get $M-\frac{1}{K}\le\frac{N}{K+1}\le M$.
Hence we get from Theorem \ref{prop:FCOptDMT}
\begin{equation}
d^{\ast,\lf(FC\ri)}_{M,N} \lf(r\ri)\le d^{\ast,\lf(FC\ri)}_{K\cdot
M,N} \lf(K\cdot r\ri) \quad 0\le r\le M-\frac{1}{K}.
\end{equation}
For $M-\frac{1}{K}\le r\le M$ both $d^{\ast,\lf(FC\ri)}_{M,N}
\lf(r\ri)$ and $d^{\ast,\lf(FC\ri)}_{K\cdot M,N} \lf(K\cdot r\ri)$
are on the last straight line of the piecewise linear functions. By
simply assigning $N= \lf(K+1\ri)M-1$ we get for $M-\frac{1}{K}\le
r\le M$
\begin{equation}\label{eq:FCDMTBOUNDForEqual}
d^{\ast,\lf(FC\ri)}_{M,N} \lf(r\ri) = d^{\ast,\lf(FC\ri)}_{K\cdot
M,N} \lf(K\cdot r\ri)= KM\lf(M-r\ri).
\end{equation}
From \eqref{eq:FCDMTBOUNDForGreater}-\eqref{eq:FCDMTBOUNDForEqual}
we get for $N\ge \lf(K+1\ri)M-1$ and $0\le r\le M$ that
\begin{equation}\label{eq:FCDMTRelationFinal}
d^{\ast,\lf(FC\ri)}_{M,N}\lf(r\ri)\le d^{\ast,\lf(FC\ri)}_{i\cdot
M,N}\lf(i\cdot r\ri)\quad 2\le i\le K.
\end{equation}

So far we have proved the relation between the contracted optimal
DMT of finite constellations with different number of users pulled
together. We now use it in order to prove the relation between
$d^{\ast,i\cdot D}_{i\cdot M,N} \lf(i\cdot r\ri)$ for $1\le i\le K$.
In Corollary \ref{prop:RelationBedDandFCd} it was shown that for $0<
D\le \min \lf(M,N\ri)$
\begin{equation}\label{eq:ICOptimalDMTCor2}
d^{\ast,D}_{M,N} \lf(r\ri)\le
d^{\ast,\lf(FC\ri)}_{M,N}\lf(r\ri)\quad 0\le r\le D.
\end{equation}
On the other hand from Corollary \ref{prop:ICP2pDMTAnchorPoints} we
can see that
\begin{equation}\label{eq:ICFCEquationPoint}
d^{\ast,i\cdot D}_{i\cdot M,N} \lf(l\ri)= d^{\ast,
\lf(FC\ri)}_{i\cdot M,N} \lf(l\ri)= \lf(i\cdot M-l\ri)
\lf(N-l\ri)\quad 1\le i\le K
\end{equation}
at $l=0$ when $0\le i\cdot D\le\frac{i\cdot MN}{i\cdot M+N-1}$, and
also for $l=1,\dots,i\cdot M-1$ when $\frac{i\cdot MN-l
\lf(l-1\ri)}{i\cdot M+N-1-2 \lf(l-1\ri) }\le i\cdot D\le
\frac{i\cdot MN-l \lf(l+1\ri)}{i\cdot M+N-1-2l}$. Hence based on
\eqref{eq:FCDMTRelationFinal}-\eqref{eq:ICFCEquationPoint}, and the
fact that $d^{\ast,i\cdot D}_{i\cdot M,N}\lf(i\cdot r\ri)$ is a
contraction of $d^{\ast,i\cdot D}_{i\cdot M,N}\lf(r\ri)$ for $2\le
i\le K$ we get
\begin{equation}\label{eq:ICDMTUpperDifferentUsersSameDFirst}
d^{\ast,i\cdot D}_{i\cdot M,N}\lf(0\ri)\ge
d^{\ast,D}_{M,N}\lf(0\ri)\quad 2\le i\le K
\end{equation}
for $0\le D\le\frac{MN}{i\cdot M+N-1}$, and
\begin{equation}\label{eq:ICDMTUpperDifferentUsersSameDSecond}
d^{\ast,i\cdot D}_{i\cdot M,N}\lf(l\ri)\ge
d^{\ast,D}_{M,N}\lf(\frac{l}{i}\ri)\quad 2\le i\le K
\end{equation}
for $l=1,\dots, i\cdot M-1$ and $\frac{MN-\frac{l}{i}
\lf(l-1\ri)}{i\cdot M+N-1-2 \lf(l-1\ri) }\le D\le
\frac{MN-\frac{l}{i} \lf(l+1\ri)}{i\cdot M+N-1-2l}$. Since
$d^{\ast,i\cdot D}_{i\cdot M,N}\lf(i\cdot r\ri)$, $1\le i\le K$, are
straight lines as a function of $r$, and also all of these straight
lines are equal zero for $r=D$ , i.e., $d^{\ast,i\cdot D}_{i\cdot
M,N} \lf(i\cdot D\ri)=0$ for $1\le i\le K$, the inequalities in
\eqref{eq:ICDMTUpperDifferentUsersSameDFirst},
\eqref{eq:ICDMTUpperDifferentUsersSameDSecond} leads to
$$d^{\ast,D}_{M,N}\lf(r\ri)\le d^{\ast,i\cdot D}_{i\cdot M,N}\lf(i\cdot r\ri)\quad 2\le i\le K$$ for any $0\le D\le M$ and
$0\le r\le D$. This concludes the proof.

\section{Proof of Lemma
\ref{lem:TheCasewhereSingleUserIsNotNecessarilytheWorst}}\label{append:TheCasewhereSingleUserIsNotNecessarilytheWorst}
First note that $\frac{N}{i+1}\ge \frac{L}{K}$ for $1\le i\le K-1$.
Hence from Theorem \ref{prop:FCOptDMT} we get that
\begin{equation}\label{eq:Kmunus1UsersAboveSingleUserSecondLemma}
d^{\ast,\lf(FC\ri)}_{M,N}\lf(r\ri)\le d^{\ast,\lf(FC\ri)}_{i\cdot
M,N}\lf(i\cdot r\ri)\quad 2\le i\le K-1
\end{equation}
for $0\le r\le \frac{L}{K}$. Based on \eqref{eq:ICOptimalDMTCor2},
\eqref{eq:ICFCEquationPoint},
\eqref{eq:Kmunus1UsersAboveSingleUserSecondLemma} and Corollary
\ref{prop:ICP2pDMTAnchorPoints} we get that
\begin{equation}\label{eq:ICDMTUpperDifferentUsersSameDFirstLem2}
d^{\ast,i\cdot D}_{i\cdot M,N}\lf(0\ri)\ge
d^{\ast,D}_{M,N}\lf(0\ri)\quad 2\le i\le K-1
\end{equation}
for $0\le D\le\frac{MN}{i\cdot M+N-1} $, and
\begin{equation}\label{eq:ICDMTUpperDifferentUsersSameDSecondLem2}
d^{\ast,i\cdot D}_{i\cdot M,N}\lf(l\ri)\ge
d^{\ast,D}_{M,N}\lf(\frac{l}{i}\ri)\quad 2\le i\le K-1
\end{equation}
for $l=1,\dots, i\cdot M-1$ and $\frac{MN-\frac{l}{i}
\lf(l-1\ri)}{i\cdot M+N-1-2 \lf(l-1\ri) }\le
D\le\frac{MN-\frac{l}{i} \lf(l+1\ri)}{i\cdot M+N-1-2l}$. Again,
since $d^{\ast,i\cdot D}_{i\cdot M,N}\lf(i\cdot r\ri)$, $1\le i\le
K$, are straight lines as a function of $r$, and also all of these
straight lines are equal to zero for $r=D$, the inequalities in
\eqref{eq:ICDMTUpperDifferentUsersSameDFirstLem2},
\eqref{eq:ICDMTUpperDifferentUsersSameDSecondLem2} lead to
$$d^{\ast,D}_{M,N}\lf(r\ri)\le d^{\ast,i\cdot D}_{i\cdot M,N}\lf(i\cdot r\ri)\quad 2\le i\le K-1$$ for any $0\le D\le
\frac{L}{K}$ and $0\le r\le D$.

\section{Proof of Lemma
\ref{lem:TheCasewhereDimensionSmallerThanFirstDMTLine}}\label{append:TheCasewhereDimensionSmallerThanFirstDMTLine}
Since $M\ge 1$ we get for $N< \lf(K-1\ri)M+1$ that $L=\frac{N}{K}$.
Hence we can consider the range $0\le r\le\frac{N}{K}$. We begin the
proof by showing that for $N< \lf(K-1\ri)M+1$,
$d^{\ast,D}_{M,N}\lf(r\ri)$ is inferior compared to $d^{\ast,K\cdot
D}_{K\cdot M,N} \lf(K\cdot r\ri)$, for any $0\le D\le\frac{N}{K}$.
Then we show that the maximization over $d^{\ast,D}_{M,N} \lf(r\ri)$
yields $M\cdot N-M\cdot K\cdot r$.

We begin by showing that
\begin{equation*}
d^{\ast,D}_{M,N}\lf(r\ri)\le d^{\ast,K\cdot D}_{K\cdot
M,N}\lf(K\cdot r\ri)\quad 0\le D\le\frac{N}{K}
\end{equation*}
for $0\le r\le D$. By assigning $D=\frac{N}{K}$ in $d^{\ast,K\cdot
D}_{K\cdot M,N}\lf(K\cdot r\ri)$ we get
\begin{equation*}
d^{\ast,N}_{K\cdot M,N}\lf(K\cdot r\ri)=\lf(K\cdot M-N+1\ri)\cdot
\lf(N-Kr\ri).
\end{equation*}
Since $N< \lf(K-1\ri)M+1$ we get
\begin{equation}\label{eq:KUserSmallestDiversityOrder}
d^{\ast,N}_{K\cdot M,N}\lf(0\ri)= \lf(K\cdot M-N+1\ri)\cdot N >
M\cdot N.
\end{equation}
It follows from Corollary \ref{prop:ICP2pDMTAnchorPoints} that
\begin{equation}\label{eq:KUserSmallestDiversityOrderLowerBound}
d^{\ast,N}_{K\cdot M,N}\lf(0\ri)\le d^{\ast,K\cdot D}_{K\cdot
M,N}\lf(0\ri) \quad 0\le D\le \frac{N}{K}
\end{equation}
and also
\begin{equation}\label{eq:SingleUserLargestDiversityOrderUpperBound}
d^{\ast,D}_{M,N}\lf(0\ri)\le M\cdot N\quad 0\le D\le\frac{N}{K}.
\end{equation}
Since $d^{\ast,i\cdot D}_{i\cdot M,N} \lf(i\cdot r\ri)$ $1\le i\le
K$ are straight lines as a function of $r$, that equal to zero for
$r=D$, and also based on \eqref{eq:KUserSmallestDiversityOrder},
\eqref{eq:KUserSmallestDiversityOrderLowerBound},
\eqref{eq:SingleUserLargestDiversityOrderUpperBound} and Lemma
\ref{lem:TheCasewhereSingleUserIsNotNecessarilytheWorst} we get
\begin{equation}
d^{\ast,D}_{M,N}\lf(r\ri)\le d^{\ast,i\cdot D}_{i\cdot
M,N}\lf(i\cdot r\ri)\quad 1\le i\le K
\end{equation}
for any $0\le D\le \frac{N}{K}$ and $0\le r\le D$. Hence the
optimization problem takes the following form
\begin{equation}
\max_{D}\min_{1\le i\le K}d^{\ast,i\cdot D}_{i\cdot M,N} \lf(i\cdot
r\ri)= \max_{D}d^{\ast,D}_{M,N} \lf(r\ri)\quad 0\le r\le\frac{N}{K}.
\end{equation}

For $N< \lf(K-1\ri)M+1$ we get that $\frac{N}{K}<\frac{MN}{N+M-1}$.
Also, from Corollary \ref{prop:ICP2pDMTAnchorPoints} we get that
$d^{\ast,D}_{M,N} \lf(0\ri)=M\cdot N$ for $0\le
D\le\frac{MN}{N+M-1}$. Hence, in the range $0\le D\le \frac{N}{K}$
we get a set of straight lines as a function of $r$,
$d^{\ast,D}_{M,N} \lf(r\ri)$, where $d^{\ast,D}_{M,N} \lf(0\ri)=MN$
and $d^{\ast,D}_{M,N} \lf(D\ri)=0$. As a result the maximal value
for each $r$ is attained for $D=\frac{N}{K}$, and equals
\begin{equation}
\max_{D}d^{\ast,D}_{M,N} \lf(r\ri)=d^{\ast,\frac{N}{K}}_{M,N}
\lf(r\ri)=MN-KMr\quad 0\le r\le \frac{N}{K}.
\end{equation}

\section{Proof of Lemma
\ref{lem:TheCasewhereSingleUserandKUserDMTUpperBoundCoincide}}\label{append:TheCasewhereSingleUserandKUserDMTUpperBoundCoincide}
The outline of the proof is as follows. We begin by finding the
straight line that equals $d^{\ast,\lf(FC\ri)}_{M,N}
\lf(\lfloor\frac{l}{2}\rfloor+1\ri)$ at
$r=\lfloor\frac{l}{2}\rfloor+1$, and also equals
$d^{\ast,\lf(FC\ri)}_{K\cdot M,N}
\lf(\lf(K-1\ri)M+\lfloor\frac{l+1}{2}\rfloor\ri)$ for $r=\frac{
\lf(K-1\ri)M+\lfloor\frac{l+1}{2}\rfloor}{K}$; it follows from the
setting in the lemma that $\lfloor\frac{l}{2}\rfloor+1<\min
\lf(M,N\ri)$ and $\lf(K-1\ri)M+\lfloor\frac{l+1}{2}\rfloor<\min
\lf(KM,N\ri)$ for $l=0,\dots,2M-3$. Then we show that the average
number of dimensions per channel use per user,
 $D_{l}$, corresponding to this straight line fulfils Corollary
\ref{prop:ICP2pDMTAnchorPoints}, i.e., for $d^{\ast,D}_{M,N}
\lf(r\ri)$, $D_{l}$ is in the range of average number of dimensions
per channel use that rotate around the anchor point
$d^{\ast,\lf(FC\ri)}_{M,N} \lf(\lfloor\frac{l}{2}\rfloor+1\ri)$, and
also for $d^{\ast,K\cdot D}_{K\cdot M,N} \lf(K\cdot r\ri)$, $D_{l}$
is in the range of average number of dimensions per channel use that
rotate around the anchor point $d^{\ast,\lf(FC\ri)}_{K\cdot M,N}
\lf(K\cdot \frac{\lf(K-1\ri)M+\lfloor\frac{l+1}{2}\rfloor}{K}\ri)$.
By showing that the straight line fulfils Corollary
\ref{prop:ICP2pDMTAnchorPoints} for both cases, we get that the
straight line equals $d^{\ast,D_{l}}_{M,N} \lf(r\ri)$ and also
$d^{\ast,K\cdot D_{l}}_{K\cdot M,N} \lf(K\cdot r\ri)$.

Let us denote the straight line by
\con{\begin{equation*}
d^{\ast}\lf(r\ri) = MN-\lfloor\frac{l}{2}\rfloor\cdot
\lf(\lfloor\frac{l}{2}\rfloor+1\ri)-2\cdot\lf(\lfloor\frac{l}{2}\rfloor+1\ri)\cdot\lf(\frac{l}{2}-\lfloor\frac{l}{2}\rfloor\ri)-
\lf(N+M-1-l\ri)r.
\end{equation*}}{\baln{}{&d^{\ast}\lf(r\ri) = MN-\lfloor\frac{l}{2}\rfloor\cdot
\lf(\lfloor\frac{l}{2}\rfloor+1\ri)-&
\nn\\
&2\cdot\lf(\lfloor\frac{l}{2}\rfloor+1\ri)\cdot\lf(\frac{l}{2}-\lfloor\frac{l}{2}\rfloor\ri)-
\lf(N+M-1-l\ri)r.&}}
First we wish to show that $d^{\ast}
\lf(\lfloor\frac{l}{2}\rfloor+1\ri)=d^{\ast,\lf(FC\ri)}_{M,N}
\lf(\lfloor\frac{l}{2}\rfloor+1\ri)$, and also that
$d^{\ast}\lf(\frac{ \lf(K-1\ri)M+
\lfloor\frac{l+1}{2}\rfloor}{K}\ri)=d^{\ast,\lf(FC\ri)}_{K\cdot
M,N}\lf( \lf(K-1\ri)M +\lfloor\frac{l+1}{2}\rfloor\ri)$. By simply
assigning $r=\lfloor\frac{l}{2}\rfloor+1$ we get
\con{\begin{equation}\label{eq:ConverseStraightLineLem1}
d^{\ast}\lf(\lfloor\frac{l}{2}\rfloor+1\ri)=\lf(N-\lfloor\frac{l}{2}
\rfloor-1\ri)\cdot \lf(M-\lfloor\frac{l}{2} \rfloor-1\ri)=
d^{\ast,\lf(FC\ri)}_{M,N}\lf(\lfloor\frac{l}{2}\rfloor+1\ri).
\end{equation}}{\bal{eq:ConverseStraightLineLem1}
{&d^{\ast}\lf(\lfloor\frac{l}{2}\rfloor+1\ri)=\lf(N-\lfloor\frac{l}{2}
\rfloor-1\ri)\cdot \lf(M-\lfloor\frac{l}{2} \rfloor-1\ri)=&
\nn\\
&d^{\ast,\lf(FC\ri)}_{M,N}\lf(\lfloor\frac{l}{2}\rfloor+1\ri).&}}
For $r=\frac{ \lf(K-1\ri)M+ \lfloor\frac{l+1}{2}\rfloor}{K}$ we
consider two cases. In the first case assume $l=2b$, i.e., $l$ is
even. Under this assumption
$\lfloor\frac{l+1}{2}\rfloor=\lfloor\frac{l}{2}\rfloor=b$, and so
$r=\frac{ \lf(K-1\ri)M+ b}{K}$. By assigning $KM=N+M-1-2b$ in
$d^{\ast} \lf(r\ri)$ we get
\con{\begin{equation*}
d^{\ast}\lf(\frac{ \lf(K-1\ri)M+ b}{K}\ri)=MN-b \lf(b+M+1\ri)-
\lf(K-1\ri)M^{2}=\lf(N- \lf(K-1\ri)M -b\ri)\cdot
\lf(M-b\ri)=d^{\ast,\lf(FC\ri)}_{K\cdot M,N}\lf( \lf(K-1\ri)M
+b\ri).
\end{equation*}}{\baln{}{&d^{\ast}\lf(\frac{ \lf(K-1\ri)M+ b}{K}\ri)=MN-b \lf(b+M+1\ri)-
\lf(K-1\ri)M^{2}&
\nn\\
&=\lf(N- \lf(K-1\ri)M -b\ri)\cdot
\lf(M-b\ri)&
\nn\\
&=d^{\ast,\lf(FC\ri)}_{K\cdot M,N}\lf( \lf(K-1\ri)M
+b\ri).&}}
In the second case $l=2b+1$, i.e., $l$ is odd. In this case we get
$\lfloor\frac{l+1}{2}\rfloor=b+1$, $\lfloor\frac{l}{2}\rfloor=b$ and
$r=\frac{ \lf(K-1\ri)M+ b+1}{K}$. By assigning $KM=N+M-2-2b$ in
$d^{\ast} \lf(r\ri)$ we get
\con{\begin{equation*}
d^{\ast}\lf(\frac{ \lf(K-1\ri)M+ b+1}{K}\ri)=MN- \lf(b+1\ri)\cdot
\lf(b+M+1\ri)- \lf(K-1\ri)M^{2}=d^{\ast,\lf(FC\ri)}_{K\cdot M,N}\lf(
\lf(K-1\ri)M +b+1\ri).
\end{equation*}}{\baln{}{&d^{\ast}\lf(\frac{ \lf(K-1\ri)M+ b+1}{K}\ri)=MN- \lf(b+1\ri)\cdot
\lf(b+M+1\ri)&
\nn\\
&- \lf(K-1\ri)M^{2}=d^{\ast,\lf(FC\ri)}_{K\cdot M,N}\lf(
\lf(K-1\ri)M +b+1\ri).&
}}
Hence from both cases we get
\con{\begin{equation}\label{eq:ConverseStraightLineLem2}
d^{\ast}\lf(\frac{\lf(K-1\ri)M+
\lfloor\frac{l+1}{2}\rfloor}{K}\ri)=d^{\ast,\lf(FC\ri)}_{K\cdot
M,N}\lf( \lf(K-1\ri)M+ \lfloor\frac{l+1}{2}\rfloor\ri).
\end{equation}}{\bal{eq:ConverseStraightLineLem2}
{&d^{\ast}\lf(\frac{\lf(K-1\ri)M+
\lfloor\frac{l+1}{2}\rfloor}{K}\ri)=&
\nn\\
&d^{\ast,\lf(FC\ri)}_{K\cdot
M,N}\lf( \lf(K-1\ri)M+ \lfloor\frac{l+1}{2}\rfloor\ri).&}}

Now we wish to show that $d^{\ast} \lf(r\ri)=d^{\ast,D_{l}}_{M,N}
\lf(r\ri)=d^{\ast,K\cdot D_{l}}_{K\cdot M,N}\lf(K\cdot r\ri)$. We
begin by showing that $d^{\ast} \lf(r\ri)=d^{\ast,D_{l}}_{M,N}
\lf(r\ri)$. First note that
\begin{equation}\label{eq:ConverseStraightLineLem4}
d^{\ast} \lf(D_{l}\ri)=d^{\ast,D_{l}}_{M,N}
\lf(D_{l}\ri)=d^{\ast,K\cdot D_{l}}_{K\cdot M,N} \lf(K\cdot
D_{l}\ri) =0.
\end{equation}
Now let us denote $D_{\lfloor
\frac{l}{2}\rfloor}^{\ast}=\frac{M\cdot
N-\lfloor\frac{l}{2}\rfloor\cdot
\lf(\lfloor\frac{l}{2}\rfloor+1\ri)}{N+M-1-2\lfloor\frac{l}{2}\rfloor}$
and $D_{\lfloor \frac{l}{2}\rfloor+1}^{\ast}=\frac{M\cdot N-
\lf(\lfloor\frac{l}{2}\rfloor+1\ri)\cdot
\lf(\lfloor\frac{l}{2}\rfloor+2\ri)}{N+M-1-
2\lf(\lfloor\frac{l}{2}\rfloor+1\ri)}$; note that $D_{\lfloor
\frac{l}{2}\rfloor+1}^{\ast}>D_{\lfloor \frac{l}{2}\rfloor}^{\ast}$.
We wish to show that
\con{\begin{equation}\label{eq:ConverseStraightLineLem5}
d^{\ast,D_{\lfloor
\frac{l}{2}\rfloor+1}^{\ast}}_{M,N}\lf(0\ri)=M\cdot N- \lf(\lfloor
\frac{l}{2}\rfloor+1\ri) \cdot \lf(\lfloor \frac{l}{2}\rfloor+2\ri)<
d^{\ast} \lf(0\ri)\le M\cdot N-\lfloor \frac{l}{2}\rfloor\cdot
\lf(\lfloor \frac{l}{2}\rfloor+1\ri)=d^{\ast,D_{\lfloor
\frac{l}{2}\rfloor}^{\ast}}_{M,N}\lf(0\ri).
\end{equation}}{\bal{eq:ConverseStraightLineLem5}
{&d^{\ast,D_{\lfloor
\frac{l}{2}\rfloor+1}^{\ast}}_{M,N}\lf(0\ri)=M\cdot N- \lf(\lfloor
\frac{l}{2}\rfloor+1\ri) \cdot \lf(\lfloor \frac{l}{2}\rfloor+2\ri)<&
\nn\\
&d^{\ast} \lf(0\ri)\le M\cdot N-\lfloor \frac{l}{2}\rfloor\cdot
\lf(\lfloor \frac{l}{2}\rfloor+1\ri)=d^{\ast,D_{\lfloor
\frac{l}{2}\rfloor}^{\ast}}_{M,N}\lf(0\ri).&}}
In the first case we take $l=2b$. In this case
\begin{equation*}
d^{\ast} \lf(0\ri) = M\cdot N-b \lf(b+1\ri).
\end{equation*}
On the other hand we also get
\begin{equation*}
M\cdot N-\lfloor \frac{l}{2}\rfloor\cdot \lf(\lfloor
\frac{l}{2}\rfloor+1\ri) = M\cdot N-b\cdot \lf(b+1\ri)=d^{\ast}
\lf(0\ri)
\end{equation*}
which proves \eqref{eq:ConverseStraightLineLem5} for the first case.
In the second case we consider $l=2b+1$. In this case
\begin{equation*}
d^{\ast} \lf(0\ri) = M\cdot N -\lf(b+1\ri)^{2}.
\end{equation*}
For this case we also get $M\cdot N-\lfloor \frac{l}{2}\rfloor\cdot
\lf(\lfloor \frac{l}{2}\rfloor+1\ri)=M\cdot N-b\cdot \lf(b+1\ri)$
and $M\cdot N- \lf(\lfloor \frac{l}{2}\rfloor+1\ri) \cdot
\lf(\lfloor \frac{l}{2}\rfloor+2\ri) = M\cdot
N-\lf(b+1\ri)\cdot\lf(b+2\ri)$. It can be easily shown that for
$b\ge 0$
\con{\begin{equation*}
M\cdot N-\lf(b+1\ri)\cdot\lf(b+2\ri)< d^{\ast} \lf(0\ri)=M\cdot
N-\lf(b+1\ri)^{2} \le  M\cdot N-b\cdot\lf(b+1\ri)
\end{equation*}}{\baln{}{&M\cdot N-\lf(b+1\ri)\cdot\lf(b+2\ri)< d^{\ast} \lf(0\ri)=&
\nn\\
&M\cdot
N-\lf(b+1\ri)^{2} \le  M\cdot N-b\cdot\lf(b+1\ri)&}}
which proves \eqref{eq:ConverseStraightLineLem5} for the second
case. From Corollary \ref{prop:ICP2pDMTAnchorPoints} and
\eqref{eq:ConverseStraightLineLem1} we know that
\con{\begin{equation}\label{eq:ConverseStraightLineLem6}
d^{\ast}\lf(\lfloor\frac{l}{2}\rfloor+1\ri)=d^{\ast,D_{\lfloor
\frac{l}{2}\rfloor}^{\ast}}_{M,N}\lf(\lfloor\frac{l}{2}\rfloor+1\ri)=d^{\ast,D_{\lfloor
\frac{l}{2}\rfloor+1}^{\ast}}_{M,N}\lf(\lfloor\frac{l}{2}\rfloor+1\ri)=d^{\ast,\lf(FC\ri)}_{M,N}\lf(\lfloor\frac{l}{2}\rfloor+1\ri)>0.
\end{equation}}{\bal{eq:ConverseStraightLineLem6}
{&d^{\ast}\lf(\lfloor\frac{l}{2}\rfloor+1\ri)=d^{\ast,D_{\lfloor
\frac{l}{2}\rfloor}^{\ast}}_{M,N}\lf(\lfloor\frac{l}{2}\rfloor+1\ri)=&
\nn\\
&d^{\ast,D_{\lfloor
\frac{l}{2}\rfloor+1}^{\ast}}_{M,N}\lf(\lfloor\frac{l}{2}\rfloor+1\ri)=d^{\ast,\lf(FC\ri)}_{M,N}\lf(\lfloor\frac{l}{2}\rfloor+1\ri)>0.&}}
Since $d^{\ast}\lf(r\ri)$, $d^{\ast,D_{\lfloor
\frac{l}{2}\rfloor}^{\ast}}_{M,N}\lf(r\ri)$ and $d^{\ast,D_{\lfloor
\frac{l}{2}\rfloor+1}^{\ast}}_{M,N}\lf(r\ri)$ are all straight lines
that fulfil \eqref{eq:ConverseStraightLineLem5},
\eqref{eq:ConverseStraightLineLem6} we get for
$r>\lfloor\frac{l}{2}\rfloor+1$
\begin{equation}\label{eq:ConverseStraightLineLem19}
d^{\ast,D_{\lfloor \frac{l}{2}\rfloor}^{\ast}}_{M,N}\lf(r\ri)\le
d^{\ast}\lf(r\ri)<d^{\ast,D_{\lfloor
\frac{l}{2}\rfloor+1}^{\ast}}_{M,N}\lf(r\ri),
\end{equation}
whereas
\begin{equation}\label{eq:ConverseStraightLineLem19_A}
d^{\ast,D_{\lfloor \frac{l}{2}\rfloor}^{\ast}}_{M,N}\lf(D_{\lfloor
\frac{l}{2}\rfloor}^{\ast}\ri)=
d^{\ast}\lf(D_{l}\ri)=d^{\ast,D_{\lfloor
\frac{l}{2}\rfloor+1}^{\ast}}_{M,N}\lf(D_{\lfloor
\frac{l}{2}\rfloor+1}^{\ast}\ri)=0.
\end{equation}
Therefore, it follows from \eqref{eq:ConverseStraightLineLem6},
\eqref{eq:ConverseStraightLineLem19} and
\eqref{eq:ConverseStraightLineLem19_A} that
\begin{equation}\label{eq:ConverseStraightLineLem22}
D_{\lfloor \frac{l}{2}\rfloor}^{\ast}\le D_{l}< D_{\lfloor
\frac{l}{2}\rfloor+1}^{\ast}.
\end{equation}
As a result, from Corollary \ref{prop:ICP2pDMTAnchorPoints} and
\eqref{eq:ConverseStraightLineLem22} we get
\begin{equation}\label{eq:ConverseStraightLineLem7}
d^{\ast,D_{l}}_{M,N}\lf(\lfloor\frac{l}{2}\rfloor+1\ri)=
d^{\ast,\lf(FC\ri)}_{M,N}\lf(\lfloor\frac{l}{2}\rfloor+1\ri).
\end{equation}
Since $d^{\ast}\lf(r\ri)$ and $d^{\ast,D_{l}}_{M,N}\lf(r\ri)$ are
straight lines and based on the equalities in
\eqref{eq:ConverseStraightLineLem1},
\eqref{eq:ConverseStraightLineLem4} and
\eqref{eq:ConverseStraightLineLem7} we get
\begin{equation}\label{eq:ConverseStraightLineLem8}
d^{\ast} \lf(r\ri)=d^{\ast,D_{l}}_{M,N} \lf(r\ri).
\end{equation}

Next we prove $d^{\ast}\lf(r\ri)=d^{\ast,K\cdot D_{l}}_{K\cdot
M,N}\lf(K\cdot r\ri)$. Let us denote
$r_{l}=\frac{\lf(K-1\ri)M+\lfloor \frac{l+1}{2}\rfloor}{K}$ and
$D_{r_{l}}^{\ast}=\frac{MN- \lf(K\cdot r_{l}-1\ri)r_{l}}{K\cdot
M+N-1-2 \lf(K\cdot r_{l}-1\ri) }$. We wish to show
\begin{equation}\label{eq:ConverseStraightLineLem9}
d^{\ast,K\cdot D_{r_{l}+\frac{1}{K}}^{\ast}}_{K\cdot
M,N}\lf(0\ri)\le d^{\ast} \lf(0\ri)<  d^{\ast,K\cdot
D_{r_{l}}^{\ast}}_{K\cdot M,N}\lf(0\ri).
\end{equation}
We consider two cases. For the first case we take $l=2\cdot b$. In
this case we get $r_{2b} = \frac{\lf(K-1\ri)M+b}{K}$,
$d^{\ast}\lf(0\ri)=M\cdot N-b \lf(b+1\ri)$ and $N=
\lf(K-1\ri)M+1+2b$. Hence we get
\con{\begin{equation}\label{eq:ConverseStraightLineLem10}
d^{\ast,K\cdot D^{\ast}_{r_{2b}+\frac{1}{K}}}_{K\cdot
M,N}\lf(0\ri)=KMN-\lf(\lf(K-1\ri)M+b \ri) \lf(N-b\ri)=MN-b
\lf(N-\lf(K-1\ri)M\ri)+b^{2}.
\end{equation}}{\bal{eq:ConverseStraightLineLem10}
{&d^{\ast,K\cdot D^{\ast}_{r_{2b}+\frac{1}{K}}}_{K\cdot
M,N}\lf(0\ri)=KMN-\lf(\lf(K-1\ri)M+b \ri) \lf(N-b\ri)=&
\nn\\
&MN-b\lf(N-\lf(K-1\ri)M\ri)+b^{2}.&}}
Since $N-\lf(K-1\ri)M=1+2b$ we get
\con{\begin{equation}\label{eq:ConverseStraightLineLem11}
MN-b \lf(N-\lf(K-1\ri)M\ri)+b^{2} = MN-b \lf(2b+1\ri)+b^{2}=MN-b
\lf(b+1\ri).
\end{equation}}{\bal{eq:ConverseStraightLineLem11}
{&MN-b \lf(N-\lf(K-1\ri)M\ri)+b^{2} = MN-b \lf(2b+1\ri)+b^{2}&
\nn\\
&=MN-b
\lf(b+1\ri).&}}
From \eqref{eq:ConverseStraightLineLem10} and
\eqref{eq:ConverseStraightLineLem11} we get
$d^{\ast}\lf(0\ri)=d^{\ast,K\cdot
D_{r_{l}+\frac{1}{K}}^{\ast}}_{K\cdot M,N}\lf(0\ri)$, which proves
\eqref{eq:ConverseStraightLineLem9} for the first case. For the
second case we take $l=2b+1$. In this case
$r_{2b+1}=\frac{\lf(K-1\ri)M+b+1}{K}$,
$d^{\ast}\lf(0\ri)=MN-\lf(b+1\ri)^{2}$ and $N=\lf(K-1\ri)M+2b+2$.
For this case we get
\con{\begin{equation}\label{eq:ConverseStraightLineLem12}
d^{\ast,K\cdot D^{\ast}_{r_{2b+1}}}_{K\cdot
M,N}\lf(0\ri)=KMN-\lf(\lf(K-1\ri)M+b \ri)
\lf(N-b-1\ri)=MN+\lf(b+1\ri)\lf(K-1\ri)M-bN+b \lf(b+1\ri).
\end{equation}}{\bal{eq:ConverseStraightLineLem12}
{&d^{\ast,K\cdot D^{\ast}_{r_{2b+1}}}_{K\cdot
M,N}\lf(0\ri)=KMN-\lf(\lf(K-1\ri)M+b \ri)
\lf(N-b-1\ri)&
\nn\\
&=MN+\lf(b+1\ri)\lf(K-1\ri)M-bN+b \lf(b+1\ri).&}}
Hence according to \eqref{eq:ConverseStraightLineLem9} we need to
show
\begin{equation}\label{eq:ConverseStraightLineLem13}
MN+\lf(b+1\ri)\lf(K-1\ri)M-bN+b \lf(b+1\ri)> MN-\lf(b+1\ri)^{2}.
\end{equation}
By assigning $\lf(K-1\ri)M=N-2b-2$ we get from
\eqref{eq:ConverseStraightLineLem13} $N> b+1$. Since $0\le l=2b+1\le
2M-3$, the maximal value of $b$ is $b=M-2$, which gives for
$N=\lf(K-1\ri)M+2b+l$
\begin{equation*}
N>M>M-1\ge b+1.
\end{equation*}
Hence we get
\begin{equation}\label{eq:ConverseStraightLineLem16}
d^{\ast}\lf(0\ri) < d^{\ast,K\cdot D^{\ast}_{r_{2b+1}}}_{K\cdot
M,N}\lf(0\ri)=d^{\ast,K\cdot D^{\ast}_{r_{l}}}_{K\cdot
M,N}\lf(0\ri).
\end{equation}
On the other hand we get
\begin{equation}\label{eq:ConverseStraightLineLem14}
d^{\ast,D^{\ast}_{r_{2b+1}+\frac{1}{K}}}_{K\cdot M,N} \lf(0\ri)
=KMN-\lf(\lf(K-1\ri)M+1+b\ri)\lf(N-b\ri).
\end{equation}
Hence according to \eqref{eq:ConverseStraightLineLem9},
\eqref{eq:ConverseStraightLineLem14} we need to show that
\begin{equation}
MN+b \lf(K-1\ri)M-N \lf(b+1\ri)+b \lf(b+1\ri)\le MN-\lf(b+1\ri)^{2}
\end{equation}
which again leads to $N> b+1$. Hence we get
\begin{equation}\label{eq:ConverseStraightLineLem15}
d^{\ast,K\cdot D^{\ast}_{r_{l}+\frac{1}{K}}}_{K\cdot M,N}
\lf(0\ri)=d^{\ast,K\cdot D^{\ast}_{r_{2b+1}+\frac{1}{K}}}_{K\cdot
M,N} \lf(0\ri)\le d^{\ast}\lf(0\ri).
\end{equation}
From \eqref{eq:ConverseStraightLineLem16} and
\eqref{eq:ConverseStraightLineLem15} we get
\eqref{eq:ConverseStraightLineLem9} for the second case. Hence we
have proved \eqref{eq:ConverseStraightLineLem9}. From Corollary
\ref{prop:ICP2pDMTAnchorPoints} and
\eqref{eq:ConverseStraightLineLem2} we know that
\con{\begin{align}\label{eq:ConverseStraightLineLem21}
d^{\ast}\lf(\frac{\lf(K-1\ri)M+
\lfloor\frac{l+1}{2}\rfloor}{K}\ri)&=d^{\ast,K\cdot
D^{\ast}_{r_{l}+\frac{1}{K}}}_{K\cdot M,N} \lf(\lf(K-1\ri)M+
\lfloor\frac{l+1}{2}\rfloor\ri)\nonumber\\
&=d^{\ast,K\cdot D^{\ast}_{r_{l}}}_{K\cdot M,N} \lf(\lf(K-1\ri)M+
\lfloor\frac{l+1}{2}\rfloor\ri)=d^{\ast,\lf(FC\ri)}_{K\cdot
M,N}\lf(\lf(K-1\ri)M+ \lfloor\frac{l+1}{2}\rfloor\ri).
\end{align}}{\bal{eq:ConverseStraightLineLem21}
{&d^{\ast}\lf(\frac{\lf(K-1\ri)M+
\lfloor\frac{l+1}{2}\rfloor}{K}\ri)=d^{\ast,K\cdot
D^{\ast}_{r_{l}+\frac{1}{K}}}_{K\cdot M,N} \lf(\lf(K-1\ri)M+
\lfloor\frac{l+1}{2}\rfloor\ri)&
\nn\\
&=d^{\ast,K\cdot D^{\ast}_{r_{l}}}_{K\cdot M,N} \lf(\lf(K-1\ri)M+
\lfloor\frac{l+1}{2}\rfloor\ri)&
\nn\\
&=d^{\ast,\lf(FC\ri)}_{K\cdot
M,N}\lf(\lf(K-1\ri)M+ \lfloor\frac{l+1}{2}\rfloor\ri).&}}
Since $d^{\ast}\lf(r\ri)$, $d^{\ast,K\cdot D^{\ast}_{r_{l}}}_{K\cdot
M,N} \lf(K\cdot r\ri)$ and $d^{\ast,K\cdot
D^{\ast}_{r_{l}+\frac{1}{K}}}_{K\cdot M,N} \lf(K\cdot r\ri)$ are all
straight lines that fulfil \eqref{eq:ConverseStraightLineLem9},
\eqref{eq:ConverseStraightLineLem21}, we get similarly to
\eqref{eq:ConverseStraightLineLem22} that
\begin{equation}\label{eq:ConverseStraightLineLem20}
D_{r_{l}}^{\ast}< D_{l}\le D_{r_{l}+\frac{1}{K}}^{\ast}.
\end{equation}
As a result, from Corollary \ref{prop:ICP2pDMTAnchorPoints} and
\eqref{eq:ConverseStraightLineLem20} we get
\con{\begin{equation}\label{eq:ConverseStraightLineLem17}
d^{\ast,K\cdot D_{l}}_{K\cdot M,N}\lf(\lf(K-1\ri)M+
\lfloor\frac{l+1}{2}\rfloor\ri)= d^{\ast,\lf(FC\ri)}_{K\cdot
M,N}\lf(\lf(K-1\ri)M+ \lfloor\frac{l+1}{2}\rfloor\ri).
\end{equation}}{\bal{eq:ConverseStraightLineLem17}
{&d^{\ast,K\cdot D_{l}}_{K\cdot M,N}\lf(\lf(K-1\ri)M+
\lfloor\frac{l+1}{2}\rfloor\ri)&
\nn\\
&= d^{\ast,\lf(FC\ri)}_{K\cdot
M,N}\lf(\lf(K-1\ri)M+ \lfloor\frac{l+1}{2}\rfloor\ri).&}}
Since $d^{\ast}\lf(r\ri)$ and $d^{\ast,K\cdot D_{l}}_{K\cdot
M,N}\lf(K\cdot r\ri)$ are straight lines, and based on the
equalities in \eqref{eq:ConverseStraightLineLem2},
\eqref{eq:ConverseStraightLineLem4} and
\eqref{eq:ConverseStraightLineLem17} we get
\begin{equation}\label{eq:ConverseStraightLineLem18}
d^{\ast} \lf(r\ri)=d^{\ast,K\cdot D_{l}}_{K\cdot M,N} \lf(K\cdot
r\ri).
\end{equation}
From \eqref{eq:ConverseStraightLineLem8},
\eqref{eq:ConverseStraightLineLem18} we get the first part of the
Lemma, whereas from \eqref{eq:ConverseStraightLineLem7},
\eqref{eq:ConverseStraightLineLem17} we get the second part of the
Lemma.

\section{Proof of Theorem
\ref{th:ICOptimalSymmetricDMTUpperNound}}\label{append:ICOptimalSymmetricDMTUpperNound}
We begin by showing that $d^{\ast,\lf(IC\ri)}_{K,M,N}\lf(r\ri)$ is
the solution of the optimization problem in
\eqref{eq:OptimalUpperBoundMACDMTSymmetricEqualDim}, i.e., the case
in which all users have the same average number of dimensions per
channel use, $D$. Then we show that this is also the solution for
\eqref{eq:OptimalUpperBoundMACDMTSymmetric}.

First we find  $\max_{D}\,\min_{1\le i\le K} \lf(d^{\ast,i\cdot
D}_{i\cdot M,N}\lf(i\cdot r\ri)\ri)$, where $0\le r\le \frac{L}{K}$.
In the case $N\ge \lf(K+1\ri)M-1$, we can see from Lemma
\ref{lem:TheCasewhereSingleUserIstheWorst} that
\begin{equation*}
\max_{D}\,\min_{1\le i\le K} \lf(d^{\ast,i\cdot
D}_{i\cdot M,N}\lf(i\cdot
r\ri)\ri)=\max_{D}d^{\ast,D}_{M,N}\lf(r\ri)=d^{\ast,\lf(FC\ri)}_{M,N}\lf(r\ri).
\end{equation*}
For $N<\lf(K-1\ri)M+1$ it was shown in Lemma
\ref{lem:TheCasewhereDimensionSmallerThanFirstDMTLine} that
$d^{\ast,\lf(IC\ri)}_{K,M,N}\lf(r\ri)$ is the optimization problem
solution. For $N=\lf(K-1\ri)M+1+l$ and $l=0,\dots,2M-3$ it follows
from Lemma \ref{lem:TheCasewhereSingleUserIsNotNecessarilytheWorst}
that $d^{\ast,D}_{M,N}\lf(r\ri)$ is smaller than $d^{\ast,i\cdot
D}_{i\cdot M,N}\lf(i\cdot r\ri)$ for $2\le i\le K-1$ and any $0\le
D\le \frac{L}{K}$, $0\le r\le D$. Hence the optimization problem for
this case boils down to
\begin{equation}
\max_{D}\,\min \lf\{d^{\ast,D}_{M,N} \lf(r\ri),d^{\ast,K\cdot
D}_{K\cdot M,N} \lf(K\cdot r\ri) \ri\}
\end{equation}
for $0\le D\le \frac{L}{K}$ and $0\le r\le D$. From Lemma
\ref{lem:TheCasewhereSingleUserandKUserDMTUpperBoundCoincide} we
know that
$d^{\ast,D_{l}}_{M,N}\lf(\lfloor\frac{l}{2}\rfloor+1\ri)=d^{\ast,\lf(FC\ri)}_{M,N}\lf(\lfloor\frac{l}{2}\rfloor+1\ri)$.
As a result, based on Corollary \ref{prop:ICP2pDMTAnchorPoints} we
get that for $0< D\le D_{l}$
\begin{equation*}
d^{\ast,D}_{M,N}\lf(\lfloor\frac{l}{2}\rfloor+1\ri)\le
d^{\ast,\lf(FC\ri)}_{M,N}\lf(\lfloor\frac{l}{2}\rfloor+1\ri)=d^{\ast,D_{l}}_{M,N}\lf(\lfloor\frac{l}{2}\rfloor+1\ri)
\end{equation*}
 and also
\begin{equation*}
d^{\ast,D}_{M,N}\lf(r\ri)=0\le d^{\ast,D_{l}}_{M,N}\lf(r\ri)\quad
 r\ge D.
\end{equation*}
Hence we get for $0<D\le D_{l}$
\begin{equation}\label{eq:ConverseTheoremOptimalSymmetricDMT1}
d^{\ast,D}_{M,N}\lf(r\ri)\le d^{\ast,D_{l}}_{M,N}\lf(r\ri)\quad
\lfloor\frac{l}{2}\rfloor+1\le r\le\frac{L}{K} .
\end{equation}
In a similar manner we also know from Lemma
\ref{lem:TheCasewhereSingleUserandKUserDMTUpperBoundCoincide} that
$d^{\ast,K\cdot D_{l}}_{K\cdot
M,N}\lf(\lf(K-1\ri)M+\lfloor\frac{l+1}{2}\rfloor\ri)=d^{\ast,\lf(FC\ri)}_{K\cdot
M,N}\lf(\lf(K-1\ri)M+\lfloor\frac{l+1}{2}\rfloor\ri)$. As a result,
based on Corollary \ref{prop:ICP2pDMTAnchorPoints} we get that for
$D_{l}\le D\le \frac{L}{K}$
\con{\begin{equation*}
d^{\ast,K\cdot D}_{K\cdot
M,N}\lf(\lf(K-1\ri)M+\lfloor\frac{l+1}{2}\rfloor\ri)\le
d^{\ast,\lf(FC\ri)}_{K\cdot
M,N}\lf(\lf(K-1\ri)M+\lfloor\frac{l+1}{2}\rfloor\ri)=d^{\ast,K\cdot
D_{l}}_{K\cdot M,N}\lf(\lf(K-1\ri)M+\lfloor\frac{l+1}{2}\rfloor\ri)
\end{equation*}}{\baln{}{&d^{\ast,K\cdot D}_{K\cdot
M,N}\lf(\lf(K-1\ri)M+\lfloor\frac{l+1}{2}\rfloor\ri)\le &
\nn\\
&d^{\ast,\lf(FC\ri)}_{K\cdot
M,N}\lf(\lf(K-1\ri)M+\lfloor\frac{l+1}{2}\rfloor\ri)=&
\nn\\
&d^{\ast,K\cdot
D_{l}}_{K\cdot M,N}\lf(\lf(K-1\ri)M+\lfloor\frac{l+1}{2}\rfloor\ri)&}}
 and also
\begin{equation*}
d^{\ast,K\cdot D_{l}}_{K\cdot M,N}\lf(K\cdot r\ri)=0\le
d^{\ast,K\cdot D}_{K\cdot M,N}\lf(K\cdot r\ri)\quad  r\ge D_{l}.
\end{equation*}
Since $D_{l}\ge \frac{\lf(K-1\ri)M+\lfloor\frac{l+1}{2}\rfloor}{K}$
and these are straight lines, we also get for $D_{l}\le D\le
\frac{L}{K}$
\begin{equation}\label{eq:ConverseTheoremOptimalSymmetricDMT2}
d^{\ast,K\cdot D}_{K\cdot M,N}\lf(K\cdot r\ri)\le d^{\ast,K\cdot
D_{l}}_{K\cdot M,N}\lf(K\cdot r\ri)
\end{equation}
where $0\le r\le
\frac{\lf(K-1\ri)M+\lfloor\frac{l+1}{2}\rfloor}{K}$. Hence, based on \eqref{eq:ConverseTheoremOptimalSymmetricDMT1},
\eqref{eq:ConverseTheoremOptimalSymmetricDMT2} and the fact that
$d^{\ast,D_{l}}_{M,N}\lf(r\ri)=d^{\ast,K\cdot D_{l}}_{K\cdot
M,N}\lf(K\cdot r\ri)=d^{\ast}\lf(r\ri)$ (Lemma
\ref{lem:TheCasewhereSingleUserandKUserDMTUpperBoundCoincide}), we
get that
\con{\begin{equation}\label{eq:ConverseTheoremOptimalSymmetricDMT3}
\max_{D}\,\min\lf\{d^{\ast,D}_{M,N}\lf(r\ri),d^{\ast,K\cdot
D}_{K\cdot M,N}\lf(K\cdot r\ri)
\ri\}=d^{\ast}\lf(r\ri)=d^{\ast,\lf(IC\ri)}_{K,M,N}\lf(r\ri)\quad
\lfloor\frac{l}{2}\rfloor+1\le r\le
\frac{\lf(K-1\ri)M+\lfloor\frac{l+1}{2}\rfloor}{K}.
\end{equation}}{\bal{eq:ConverseTheoremOptimalSymmetricDMT3}
{&\max_{D}\,\min\lf\{d^{\ast,D}_{M,N}\lf(r\ri),d^{\ast,K\cdot
D}_{K\cdot M,N}\lf(K\cdot r\ri)
\ri\}=&
\nn\\
&d^{\ast}\lf(r\ri)=d^{\ast,\lf(IC\ri)}_{K,M,N}\lf(r\ri)&}}
for $\lfloor\frac{l}{2}\rfloor+1\le r\le
\frac{\lf(K-1\ri)M+\lfloor\frac{l+1}{2}\rfloor}{K}$.

We now find the solution for $0\le r\le
\lfloor\frac{l}{2}\rfloor+1$. Our starting point is $D=D_{l}$ for
which $d^{\ast,D_{l}}_{M,N}\lf(r\ri)=d^{\ast,K\cdot D_{l}}_{K\cdot
M,N}\lf(K\cdot r\ri)$. Since
$d^{\ast}\lf(\lfloor\frac{l}{2}\rfloor+1\ri)=d^{\ast,\lf(FC\ri)}_{M,N}\lf(\lfloor\frac{l}{2}\rfloor+1\ri)$
we get from Corollary \ref{prop:ICP2pDMTAnchorPoints} and
\eqref{eq:ConverseStraightLineLem22} that
\begin{equation}\label{eq:ConverseTheoremOptimalSymmetricDMT22}
\frac{MN-\lfloor\frac{l}{2}\rfloor
\lf(\lfloor\frac{l}{2}\rfloor+1\ri)}{M+N-1-2\lfloor\frac{l}{2}\rfloor}\le
D_{l}<\frac{MN- \lf(\lfloor\frac{l}{2}\rfloor+1\ri)
\lf(\lfloor\frac{l}{2}\rfloor+2\ri)}{M+N-1-2\lf(\lfloor\frac{l}{2}\rfloor+1\ri)}.
\end{equation}
It follows from Corollary \ref{prop:RelationBedDandFCd} that for
$D_{l}\le D\le\frac{L}{K}$
\begin{equation}\label{eq:ConverseTheoremOptimalSymmetricDMT4}
d^{\ast,D}_{M,N}\lf(r\ri)\le d^{\ast,\lf(FC\ri)}_{M,N}\lf(r\ri).
\end{equation}
In addition it can be easily shown that for $N=\lf(K-1\ri)M+1+l$ and
$l=0,\dots,2M-3$
\begin{equation}\label{eq:ConverseTheoremOptimalSymmetricDMT5}
\lfloor\frac{l}{2}\rfloor+1\le \frac{N}{K+1}\le
\frac{\lf(K-1\ri)M+\lfloor\frac{l+1}{2}\rfloor}{K}
\end{equation}
by considering the cases in which $l$ is even and odd, i.e., the
cases where $l=2b$ and $l=2b+1$. In the case
$\frac{MN-\lfloor\frac{l}{2}\rfloor
\lf(\lfloor\frac{l}{2}\rfloor+1\ri)}{M+N-1-2\lfloor\frac{l}{2}\rfloor}\le
D\le D_{l}$ assume $d^{\ast,K\cdot D}_{K\cdot M,N}\lf(K\cdot r\ri)$
rotates around anchor point with multiplexing gain $m$. In this case
there are two possibilities. The first possibility is
$\lfloor\frac{l}{2}\rfloor+2\le m\le\frac{L}{K}$ where
$m\in\mathbb{Z}$. In this case we get from Corollary
\ref{prop:ICP2pDMTAnchorPoints} that in the range
$\frac{MN-\lfloor\frac{l}{2}\rfloor
\lf(\lfloor\frac{l}{2}\rfloor+1\ri)}{M+N-1-2\lfloor\frac{l}{2}\rfloor}\le
D< D_{l}$
\con{\begin{equation}\label{eq:ConverseTheoremOptimalSymmetricDMT6}
d^{\ast,D}_{M,N}\lf(\lfloor\frac{l}{2}\rfloor+1\ri)=d^{\ast,K\cdot
D_{l}}_{K\cdot M,N}\lf(\lfloor\frac{l}{2}\rfloor+1\ri)\le
d^{\ast,K\cdot D}_{K\cdot M,N}\lf(\lfloor\frac{l}{2}\rfloor+1\ri).
\end{equation}}{\bal{eq:ConverseTheoremOptimalSymmetricDMT6}
{&d^{\ast,D}_{M,N}\lf(\lfloor\frac{l}{2}\rfloor+1\ri)=d^{\ast,K\cdot
D_{l}}_{K\cdot M,N}\lf(\lfloor\frac{l}{2}\rfloor+1\ri)\le&
\nn\\
&d^{\ast,K\cdot D}_{K\cdot M,N}\lf(\lfloor\frac{l}{2}\rfloor+1\ri).&}}
For the second possibility $0\le m\le\lfloor\frac{l}{2}\rfloor+1$ we
get from \eqref{eq:ConverseTheoremOptimalSymmetricDMT5}, Corollary
\ref{prop:RelationBedDandFCd} and Theorem \ref{prop:FCOptDMT} that
\con{\begin{equation}\label{eq:ConverseTheoremOptimalSymmetricDMT9}
d^{\ast,K\cdot D}_{K\cdot M,N}\lf(K\cdot
m\ri)=d^{\ast,\lf(FC\ri)}_{K\cdot M,N}\lf(K\cdot m\ri)\ge
d^{\ast,\lf(FC\ri)}_{M,N}\lf(m\ri)\ge d^{\ast,D}_{M,N}\lf(m\ri).
\end{equation}}{\bal{eq:ConverseTheoremOptimalSymmetricDMT9}
{&d^{\ast,K\cdot D}_{K\cdot M,N}\lf(K\cdot
m\ri)=d^{\ast,\lf(FC\ri)}_{K\cdot M,N}\lf(K\cdot m\ri)\ge&
\nn\\
&d^{\ast,\lf(FC\ri)}_{M,N}\lf(m\ri)\ge d^{\ast,D}_{M,N}\lf(m\ri).&}}
In addition $d^{\ast,D}_{M,N}\lf(D\ri)=d^{\ast,K\cdot D}_{K\cdot
M,N}\lf(K\cdot D\ri)=0$. Since these are straight lines we get in
the range $\frac{MN-\lfloor\frac{l}{2}\rfloor
\lf(\lfloor\frac{l}{2}\rfloor+1\ri)}{M+N-1-2\lfloor\frac{l}{2}\rfloor}\le
D\le D_{l}$
\begin{equation}\label{eq:ConverseTheoremOptimalSymmetricDMT23}
d^{\ast,D}_{M,N}\lf(r\ri)\le d^{\ast,K\cdot D}_{K\cdot
M,N}\lf(K\cdot r\ri).
\end{equation}
By induction, for $\frac{MN-\lf(s-1\ri)s}{M+N-1-2\lf(s-1\ri)}\le
D\le \frac{MN-s \lf(s+1\ri)}{M+N-1-2s}$,
$s=\lfloor\frac{l}{2}\rfloor,\dots,1$, assuming $d^{\ast,K\cdot
D^{\lf(s\ri)}}_{K\cdot M,N}\lf(K\cdot r\ri)\ge
d^{\ast,D^{\lf(s\ri)}}_{M,N}\lf(r\ri)$ at $D^{\lf(s\ri)}=\frac{MN-s
\lf(s+1\ri)}{M+N-1-2s}$, we get from similar arguments to
\eqref{eq:ConverseTheoremOptimalSymmetricDMT5}-\eqref{eq:ConverseTheoremOptimalSymmetricDMT23}
that
\begin{equation}\label{eq:ConverseTheoremOptimalSymmetricDMT24}
d^{\ast,D}_{M,N}\lf(r\ri)\le d^{\ast,K\cdot D}_{K\cdot
M,N}\lf(K\cdot r\ri).
\end{equation}
Finally for $0<D\le\frac{MN}{N+-1}$, from the same arguments as in
\eqref{eq:ConverseTheoremOptimalSymmetricDMT24} we also get
\begin{equation}\label{eq:ConverseTheoremOptimalSymmetricDMT25}
d^{\ast,D}_{M,N}\lf(r\ri)\le d^{\ast,K\cdot D}_{K\cdot
M,N}\lf(K\cdot r\ri).
\end{equation}
Hence, from \eqref{eq:ConverseTheoremOptimalSymmetricDMT23},
\eqref{eq:ConverseTheoremOptimalSymmetricDMT24} and
\eqref{eq:ConverseTheoremOptimalSymmetricDMT25} we get that in the
range $0<D\le D_{l}$
\begin{equation}\label{eq:ConverseTheoremOptimalSymmetricDMT26}
\max_{D}\,\min \lf\{d^{\ast,D}_{M,N}\lf(r\ri),d^{\ast,K\cdot
D}_{K\cdot M,N}\lf(K\cdot
r\ri)\ri\}=\max_{D}d^{\ast,D}_{M,N}\lf(r\ri).
\end{equation}
Since $D_{l}\ge \frac{MN-\lfloor\frac{l}{2}\rfloor
\lf(\lfloor\frac{l}{2}\rfloor+1\ri)}{M+N-1-2\lfloor\frac{l}{2}\rfloor}$
\eqref{eq:ConverseTheoremOptimalSymmetricDMT22}, and also from
\eqref{eq:ConverseTheoremOptimalSymmetricDMT4},
\eqref{eq:ConverseTheoremOptimalSymmetricDMT26} we get based on
Corollary \ref{prop:RelationBedDandFCd}
\con{\begin{equation}\label{eq:ConverseTheoremOptimalSymmetricDMT27}
\max_{D}\,\min \lf\{d^{\ast,D}_{M,N}\lf(r\ri),d^{\ast,K\cdot
D}_{K\cdot M,N}\lf(K\cdot
r\ri)\ri\}=d^{\ast,\lf(FC\ri)}_{M,N}\lf(r\ri)=d^{\ast,\lf(IC\ri)}_{K,M,N}\lf(r\ri)\quad
0\le r\le \lfloor\frac{l}{2}\rfloor+1.
\end{equation}}{\bal{eq:ConverseTheoremOptimalSymmetricDMT27}
{&\max_{D}\,\min \lf\{d^{\ast,D}_{M,N}\lf(r\ri),d^{\ast,K\cdot
D}_{K\cdot M,N}\lf(K\cdot
r\ri)\ri\}=&
\nn\\
&d^{\ast,\lf(FC\ri)}_{M,N}\lf(r\ri)=d^{\ast,\lf(IC\ri)}_{K,M,N}\lf(r\ri)\quad
0\le r\le \lfloor\frac{l}{2}\rfloor+1.&}}

Now we wish to find $d^{\ast,\lf(IC\ri)}_{K,M,N}\lf(r\ri)$ for
$\frac{\lf(K-1\ri)M+\lfloor\frac{l+1}{2}\rfloor}{K}\le r\le
\frac{L}{K}$. Let us denote
$r_{l}=\frac{\lf(K-1\ri)M+\lfloor\frac{l+1}{2}\rfloor}{K}$. Since
\con{\begin{equation*}
d^{\ast,K\cdot D_{l}}_{K\cdot
M,N}\lf(\lf(K-1\ri)M+\lfloor\frac{l+1}{2}\rfloor\ri)=d^{\ast,\lf(FC\ri)}_{K\cdot
M,N}\lf(\lf(K-1\ri)M+\lfloor\frac{l+1}{2}\rfloor\ri)
\end{equation*}}{\baln{}{&d^{\ast,K\cdot D_{l}}_{K\cdot
M,N}\lf(\lf(K-1\ri)M+\lfloor\frac{l+1}{2}\rfloor\ri)=&
\nn\\
&d^{\ast,\lf(FC\ri)}_{K\cdot
M,N}\lf(\lf(K-1\ri)M+\lfloor\frac{l+1}{2}\rfloor\ri)&
}}
we get \eqref{eq:ConverseStraightLineLem20}
\con{\begin{equation}\label{eq:ConverseTheoremOptimalSymmetricDMT28}
\frac{NM-\lf(K\cdot r_{l}-1\ri)r_{l}}{KM+N-1-2\lf(K\cdot
r_{l}-1\ri)}< D_{l}\le \frac{NM-r_{l}\lf(K\cdot
r_{l}+1\ri)}{KM+N-1-2\cdot K\cdot r_{l}}.
\end{equation}}{\bal{eq:ConverseTheoremOptimalSymmetricDMT28}
{&\frac{NM-\lf(K\cdot r_{l}-1\ri)r_{l}}{KM+N-1-2\lf(K\cdot
r_{l}-1\ri)}< D_{l}\le&
\nn\\
&\frac{NM-r_{l}\lf(K\cdot
r_{l}+1\ri)}{KM+N-1-2\cdot K\cdot r_{l}}.&}}
It follows from Corollary \ref{prop:RelationBedDandFCd} that in the
range $0<D\le D_{l}$
\begin{equation}\label{eq:ConverseTheoremOptimalSymmetricDMT15}
d^{\ast,K\cdot D}_{K\cdot M,N}\lf(K\cdot r\ri)\le
d^{\ast,\lf(FC\ri)}_{K\cdot M,N}\lf(K\cdot r\ri).
\end{equation}
For $D_{l}<D\le
\frac{NM-\frac{r_{l}}{K}\lf(r_{l}+1\ri)}{KM+N-1-2r_{l}}$ assume
$d^{\ast,D}_{M,N}\lf(r\ri)$ rotates around anchor point with
multiplexing gain $\frac{m}{K}$, where $m\in\mathbb{Z}$. For $0\le
m<\lf(K-1\ri)M+\lfloor\frac{l+1}{2}\rfloor$, based on Corollary
\ref{prop:ICP2pDMTAnchorPoints} and Lemma
\ref{lem:TheCasewhereSingleUserandKUserDMTUpperBoundCoincide} we get
\con{\begin{align}\label{eq:ConverseTheoremOptimalSymmetricDMT16}
d^{\ast,D}_{M,N}\lf(\frac{\lf(K-1\ri)M+\lfloor\frac{l+1}{2}\rfloor}{K}\ri)&\ge
d^{\ast,D_{l}}_{M,N}\lf(\frac{\lf(K-1\ri)M+\lfloor\frac{l+1}{2}\rfloor}{K}\ri)\nonumber\\
&=d^{\ast,\lf(FC\ri)}_{K\cdot
M,N}\lf(\lf(K-1\ri)M+\lfloor\frac{l+1}{2}\rfloor\ri)\ge
d^{\ast,K\cdot D}_{K\cdot
M,N}\lf(\lf(K-1\ri)M+\lfloor\frac{l+1}{2}\rfloor\ri).
\end{align}}{\bal{eq:ConverseTheoremOptimalSymmetricDMT16}
{&d^{\ast,D}_{M,N}\lf(\frac{\lf(K-1\ri)M+\lfloor\frac{l+1}{2}\rfloor}{K}\ri)\ge
d^{\ast,D_{l}}_{M,N}\lf(\frac{\lf(K-1\ri)M+\lfloor\frac{l+1}{2}\rfloor}{K}\ri)&
\nn\\
&=d^{\ast,\lf(FC\ri)}_{K\cdot
M,N}\lf(\lf(K-1\ri)M+\lfloor\frac{l+1}{2}\rfloor\ri)\ge&
\nn\\
&d^{\ast,K\cdot D}_{K\cdot
M,N}\lf(\lf(K-1\ri)M+\lfloor\frac{l+1}{2}\rfloor\ri).&}}
For $\lf(K-1\ri)M+\lfloor\frac{l+1}{2}\le m\le L$ we get from
\eqref{eq:ConverseTheoremOptimalSymmetricDMT5} and Theorem
\ref{prop:FCOptDMT} that
\con{\begin{equation}\label{eq:ConverseTheoremOptimalSymmetricDMT18}
d^{\ast,D}_{M,N}\lf(m\ri)=d^{\ast,\lf(FC\ri)}_{M,N}\lf(m\ri)\ge
d^{\ast,\lf(FC\ri)}_{K\cdot M,N}\lf(K\cdot m\ri)\ge d^{\ast,K\cdot
D}_{K\cdot M,N}\lf(K\cdot m\ri).
\end{equation}}{\bal{eq:ConverseTheoremOptimalSymmetricDMT18}
{&d^{\ast,D}_{M,N}\lf(m\ri)=d^{\ast,\lf(FC\ri)}_{M,N}\lf(m\ri)\ge
d^{\ast,\lf(FC\ri)}_{K\cdot M,N}\lf(K\cdot m\ri)&
\nn\\
&\ge d^{\ast,K\cdot
D}_{K\cdot M,N}\lf(K\cdot m\ri).&}}
We also get $d^{\ast,D}_{M,N}\lf(D\ri)= d^{\ast,K\cdot D}_{K\cdot
M,N}\lf(K\cdot D\ri)=0$. Since these are straight lines, we get for
$D_{l}<D\le \frac{NM-\frac{r_{l}}{K}\lf(r_{l}+1\ri)}{KM+N-1-2r_{l}}$
\begin{equation}\label{eq:ConverseTheoremOptimalSymmetricDMT19}
d^{\ast,D}_{M,N}\lf(r\ri)\ge d^{\ast,K\cdot D}_{K\cdot
M,N}\lf(K\cdot r\ri).
\end{equation}
Similarly to \eqref{eq:ConverseTheoremOptimalSymmetricDMT24} it can
be shown by induction for
$\frac{MN-\frac{s}{K}\lf(s-1\ri)}{KM+N-1-2\lf(s-1\ri)}\le D\le
\frac{MN-\frac{s}{K}\lf(s+1\ri)}{KM+N-1-2s}$,
$s=\lf(K-1\ri)M+\lfloor\frac{l+1}{2}\rfloor+1,\dots,L-1$, that
\begin{equation}\label{eq:ConverseTheoremOptimalSymmetricDMT21}
d^{\ast,D}_{M,N}\lf(r\ri)\ge d^{\ast,K\cdot D}_{K\cdot
M,N}\lf(K\cdot r\ri).
\end{equation}
Hence, from \eqref{eq:ConverseTheoremOptimalSymmetricDMT15},
\eqref{eq:ConverseTheoremOptimalSymmetricDMT19} and
\eqref{eq:ConverseTheoremOptimalSymmetricDMT21} we get
\con{\begin{equation}
\max_{D}\,\min \lf\{d^{\ast,D}_{M,N}\lf(r\ri),d^{\ast,K\cdot
D}_{K\cdot M,N}\lf(K\cdot r\ri)\ri\}=d^{\ast,\lf(FC\ri)}_{K\cdot
M,N}\lf(K\cdot r\ri)=d^{\ast,\lf(IC\ri)}_{K,M,N}\lf(r\ri)
\end{equation}}{\bal{}{&\max_{D}\,\min \lf\{d^{\ast,D}_{M,N}\lf(r\ri),d^{\ast,K\cdot
D}_{K\cdot M,N}\lf(K\cdot r\ri)\ri\}=&
\nn\\
&d^{\ast,\lf(FC\ri)}_{K\cdot
M,N}\lf(K\cdot r\ri)=d^{\ast,\lf(IC\ri)}_{K,M,N}\lf(r\ri)
&}}
where $\frac{\lf(K-1\ri)M+\lfloor\frac{l+1}{2}\rfloor}{K}\le
r\le\frac{L}{K}$.

The remaining open point for $N=\lf(K-1\ri)M+1+l$, $l=0,\dots,2M-3$
is the case
\begin{equation}\label{eq:ConverseTheoremOptimalSymmetricDMT30}
\lfloor\frac{l}{2}\rfloor+1=\frac{\lf(K-1\ri)M+\lfloor\frac{l+1}{2}\rfloor}{K}.
\end{equation}
First we would like to find when this equality takes place. For this
we consider two cases. First let us consider $l=2b$. For this case
\eqref{eq:ConverseTheoremOptimalSymmetricDMT30} takes the following
form
\begin{equation*}
K\cdot \lf(b+1\ri)=\lf(K-1\ri)M+b
\end{equation*}
which leads to
\begin{equation*}
b=M-\frac{K}{K-1}.
\end{equation*}
Since $b\ge 0$, $M\ge 1$ and $K\ge 2$ are integers, we get that this
equality can only hold at $K=2$. In this case we get $M=b+2$ and
$N=3\lf(b+1\ri)$. Since both $M\ge 1$ and $N\ge 1$, we get that
$b\ge 2$. Hence by assigning $s=b+1$ we get
\eqref{eq:ConverseTheoremOptimalSymmetricDMT30} for $K=2$, $M=s+1$
and $N=3\cdot s$, where $s\ge 1$ is an integer. For the second case
we consider $l=2b+1$. In this case by assigning in
\eqref{eq:ConverseTheoremOptimalSymmetricDMT30} we get $b=M-1$.
However we know that $l=2b+1\le 2M-3$, and so $b\le M-2$. Hence for
$l=2b+1$ \eqref{eq:ConverseTheoremOptimalSymmetricDMT30} can not
take place. From \eqref{eq:ConverseTheoremOptimalSymmetricDMT5},
\eqref{eq:ConverseTheoremOptimalSymmetricDMT30} we get
\begin{equation}\label{eq:ConverseTheoremOptimalSymmetricDMT31}
\lfloor\frac{l}{2}\rfloor+1=\frac{N}{K+1}=\frac{\lf(K-1\ri)M+\lfloor\frac{l+1}{2}\rfloor}{K}.
\end{equation}
In addition, \eqref{eq:ConverseTheoremOptimalSymmetricDMT30} holds
only for $l=2b$. For this case simply by assigning $l=2b$ we get
\begin{equation}\label{eq:ConverseTheoremOptimalSymmetricDMT32}
D^{\ast}_{\lfloor\frac{l}{2}\rfloor}=D_{l}=D^{\ast}_{r_{l}}.
\end{equation}
Hence, we are interested in finding
$d^{\ast,\lf(IC\ri)}_{K,M,N}\lf(r\ri)$ for $K=2$, $M=s+1$ and
$N=3\cdot s$, where $s\ge 1$ is an integer. For $D>D_{l}$ we get
$d^{\ast,D}_{s+1,3\cdot s}\lf(r\ri)\le
d^{\ast,\lf(FC\ri)}_{s+1,3\cdot s}\lf(r\ri)$. On the other hand for
$0<D<D_{l}=D^{\ast}_{\lfloor\frac{l}{2}\rfloor}$ we know from
Corollary \ref{prop:ICP2pDMTAnchorPoints} and
\eqref{eq:ConverseTheoremOptimalSymmetricDMT31} that
$d^{\ast,D}_{s+1,3\cdot s}\lf(r\ri)$ rotates around anchor point at
multiplexing gain $m\le\frac{N}{K+1}$. Hence, by similar arguments
to the ones used in \eqref{eq:ConverseTheoremOptimalSymmetricDMT9}
we get $d^{\ast,D}_{s+1,3\cdot s}\lf(m\ri)\le d^{\ast,2\cdot
D}_{2\cdot \lf(s+1\ri),3\cdot s}\lf(2\cdot m\ri)$, which leads to
$d^{\ast,D}_{s+1,3\cdot s}\lf(r\ri)\le d^{\ast,2\cdot
D}_{2\cdot\lf(s+1\ri),3\cdot s}\lf(2\cdot r\ri)$ for $0<D< D_{l}$.
Hence in the range $0\le r\le \frac{N}{K+1}$ the optimal solution is
$d^{\ast,\lf(FC\ri)}_{s+1,3\cdot s}\lf(r\ri)$. For the same
arguments we get for $\frac{N}{K+1}\le r\le \frac{L}{K}$ that the
optimal solution is $d^{\ast,\lf(FC\ri)}_{2\cdot \lf(s+1\ri),3\cdot
s}\lf(2\cdot r\ri)$. Hence we get
\con{\begin{equation}\label{eq:ConverseTheoremOptimalSymmetricDMT37}
d^{\ast,\lf(IC\ri)}_{K,M,N}\lf(r\ri)=d^{\ast,\lf(IC\ri)}_{2,s+1,3\cdot
s}\lf(r\ri)=\lf\{\begin{array}{lr}
d^{\ast,\lf(FC\ri)}_{s+1,3\cdot s}\lf(r\ri) &0\le r\le \frac{N}{K+1}=s\\
d^{\ast,\lf(FC\ri)}_{2 \lf(s+1\ri),3\cdot s}\lf(2\cdot r\ri) &s\le
r\le 3\cdot s.
\end{array}\ri.
\end{equation}}{\bal{eq:ConverseTheoremOptimalSymmetricDMT37}
{&d^{\ast,\lf(IC\ri)}_{K,M,N}\lf(r\ri)=d^{\ast,\lf(IC\ri)}_{2,s+1,3\cdot
s}\lf(r\ri)=&
\nn\\
&\lf\{\begin{array}{lr}
d^{\ast,\lf(FC\ri)}_{s+1,3\cdot s}\lf(r\ri) &0\le r\le \frac{N}{K+1}=s\\
d^{\ast,\lf(FC\ri)}_{2 \lf(s+1\ri),3\cdot s}\lf(2\cdot r\ri) &s\le
r\le 3\cdot s.
\end{array}\ri.&}}

So far we have shown that
\begin{equation}\label{eq:ConverseTheoremOptimalSymmetricDMT29}
\max_{D}\,\min \lf\{d^{\ast,D}_{M,N}\lf(r\ri),d^{\ast,K\cdot
D}_{K\cdot M,N}\lf(K\cdot
r\ri)\ri\}=d^{\ast,\lf(IC\ri)}_{K,M,N}\lf(r\ri).
\end{equation}
Now we wish to show that this is also the solution of
\eqref{eq:OptimalUpperBoundMACDMTSymmetric}. We begin with the case
for which
$d^{\ast,\lf(IC\ri)}_{K,M,N}\lf(r\ri)=d^{\ast,\lf(FC\ri)}_{M,N}\lf(r\ri)$.
This is the case for $N\ge \lf(K+1\ri)M-1$, and also for
$N=\lf(K-1\ri)M-1+l$, $l=0,\dots,2M-3$ when $0\le
r\le\lfloor\frac{l}{2}\rfloor+1$. As a base line we consider the
case $D_{1}=\dots,D_{K}=D^{\ast}_{r}$, where $D^{\ast}_{r}$ is the
average number of dimensions per channel use per user, that
maximizes the expression in
\eqref{eq:ConverseTheoremOptimalSymmetricDMT29}. Without loss of
generality assume user $i$ has $D_{i}\neq D_{r}^{\ast}$. In this
case based on \eqref{eq:ConverseTheoremOptimalSymmetricDMT29} and
Corollary \ref{prop:RelationBedDandFCd} we get
\con{\begin{align}
\min_{ A\subseteq \lf\{1,\dots,K\ri\}, D_{i}\neq D_{r}^{\ast}}
\lf(d^{\ast,\sum_{a\in A}D_{a}}_{|A|\cdot M,N}\lf(|A|\cdot
r\ri)\ri)\le d^{\ast,D_{i}}_{M,N}\lf(r\ri)\le
d^{\ast,\lf(FC\ri)}_{M,N}\lf(r\ri)=\max_{D}\,\min
\lf\{d^{\ast,D}_{M,N}\lf(r\ri),d^{\ast,K\cdot D}_{K\cdot
M,N}\lf(K\cdot r\ri)\ri\}.
\end{align}}{\bal{}{&\min_{ A\subseteq \lf\{1,\dots,K\ri\}, D_{i}\neq D_{r}^{\ast}}
\lf(d^{\ast,\sum_{a\in A}D_{a}}_{|A|\cdot M,N}\lf(|A|\cdot
r\ri)\ri)\le&
\nn\\
&d^{\ast,D_{i}}_{M,N}\lf(r\ri)\le d^{\ast,\lf(FC\ri)}_{M,N}\lf(r\ri)=&
\nn\\
&\max_{D}\,\min
\lf\{d^{\ast,D}_{M,N}\lf(r\ri),d^{\ast,K\cdot D}_{K\cdot
M,N}\lf(K\cdot r\ri)\ri\}.&}}
Hence the optimal solution must be
$d^{\ast,\lf(IC\ri)}_{K,M,N}\lf(r\ri)$, attained for
$D_{1}=\dots=D_{K}=D_{r}^{\ast}$. We now consider the case in which
$d^{\ast,\lf(IC\ri)}_{K,M,N}\lf(r\ri)=d^{\ast,\lf(FC\ri)}_{K\cdot
M,N}\lf(K\cdot r\ri)$, for which $N=\lf(K-1\ri)M+1+l$, where
$l=0,\dots,2M-3$ and
$\frac{\lf(K-1\ri)M+\lfloor\frac{l+1}{2}\rfloor}{K}\le r
\le\frac{L}{K}$. In this case the optimal solution in
\eqref{eq:ConverseTheoremOptimalSymmetricDMT29} for the $K$ users
pulled together is attained for $K\cdot D_{r}^{\ast}$. Let us assume
that $\sum_{i=1}^{K}D_{i}\neq K\cdot D_{r}^{\ast}$. In this case we
get
\con{\begin{equation}
\min_{ A\subseteq \lf\{1,\dots,K\ri\}, \sum_{i=1}^{K}D_{i}\neq
K\cdot D_{r}^{\ast}} \lf(d^{\ast,\sum_{a\in A}D_{a}}_{|A|\cdot M,N}\lf(|A|\cdot r\ri)\ri)\le d^{\ast,\lf(FC\ri)}_{K\cdot
M,N}\lf(K\cdot r\ri)=\max_{D}\,\min
\lf\{d^{\ast,D}_{M,N}\lf(r\ri),d^{\ast,K\cdot D}_{K\cdot
M,N}\lf(K\cdot r\ri)\ri\}.
\end{equation}}{\bal{}{&\min_{ A\subseteq \lf\{1,\dots,K\ri\}, \sum_{i=1}^{K}D_{i}\neq
K\cdot D_{r}^{\ast}} \lf(d^{\ast,\sum_{a\in A}D_{a}}_{|A|\cdot M,N}\lf(|A|\cdot r\ri)\ri)\le&
\nn\\
&d^{\ast,\lf(FC\ri)}_{K\cdot
M,N}\lf(K\cdot r\ri)=&
\nn\\
&\max_{D}\,\min
\lf\{d^{\ast,D}_{M,N}\lf(r\ri),d^{\ast,K\cdot D}_{K\cdot
M,N}\lf(K\cdot r\ri)\ri\}.&}}
Hence the optimal solution must be
$d^{\ast,\lf(IC\ri)}_{K,M,N}\lf(r\ri)$. Now let us consider the case
$N<\lf(K-1\ri)M+1$. In this case the optimal solution in
\eqref{eq:ConverseTheoremOptimalSymmetricDMT29} is attained for
$D_{r}^{\ast}=\frac{N}{K}$. Without loss of generality assume
$D_{i}<\frac{N}{K}$. In this case we get from Corollary
\ref{prop:RelationBedDandFCd} that
\con{\begin{equation}
\min_{ A\subseteq \lf\{1,\dots,K\ri\}, D_{i}< \frac{N}{K}}
\lf(d^{\ast,\sum_{a\in A}D_{a}}_{|A|\cdot M,N}\lf(|A|\cdot
r\ri)\ri)\le MN-KMr=\max_{D}\,\min
\lf\{d^{\ast,D}_{M,N}\lf(r\ri),d^{\ast,K\cdot D}_{K\cdot
M,N}\lf(K\cdot r\ri)\ri\}.
\end{equation}}{\bal{}{&\min_{ A\subseteq \lf\{1,\dots,K\ri\}, D_{i}< \frac{N}{K}}
\lf(d^{\ast,\sum_{a\in A}D_{a}}_{|A|\cdot M,N}\lf(|A|\cdot
r\ri)\ri)\le MN&
\nn\\
&-KMr=\max_{D}\,\min
\lf\{d^{\ast,D}_{M,N}\lf(r\ri),d^{\ast,K\cdot D}_{K\cdot
M,N}\lf(K\cdot r\ri)\ri\}.&}}
which shows again that $d^{\ast,\lf(IC\ri)}_{K,M,N}\lf(r\ri)$ is the
solution. Finally we consider the case where
$d^{\ast,\lf(IC\ri)}_{K,M,N}\lf(r\ri)=d^{\ast}\lf(r\ri)$, i.e., the
case in which $N=\lf(K-1\ri)M+1+l$, $l=0,\dots,2M-3$ and
$\lfloor\frac{l}{2}\rfloor+1\le
r\le\frac{\lf(K-1\ri)M+\lfloor\frac{l+1}{2}\rfloor}{K}$. Following
Lemma \ref{lem:TheCasewhereSingleUserandKUserDMTUpperBoundCoincide}
and Corollary \ref{prop:ICP2pDMTAnchorPoints} we get without loss of
generality that when $D_{1}<D_{l}$
\con{\begin{equation}
\min_{ A\subseteq \lf\{1,\dots,K\ri\}, D_{1}< D_{l}
}\lf(d^{\ast,\sum_{a\in A}D_{a}}_{|A|\cdot M,N}\lf(|A|\cdot
r\ri)\ri)\le d^{\ast,D_{1}}_{M,N}\lf(r\ri)\le d^{\ast}\lf(r\ri)
=d^{\ast,D_{l}}_{M,N}\lf(r\ri),
\end{equation}}{\bal{}{&\min_{ A\subseteq \lf\{1,\dots,K\ri\}, D_{1}< D_{l}
}\lf(d^{\ast,\sum_{a\in A}D_{a}}_{|A|\cdot M,N}\lf(|A|\cdot
r\ri)\ri)\le&
\nn\\
&d^{\ast,D_{1}}_{M,N}\lf(r\ri)\le d^{\ast}\lf(r\ri)
=d^{\ast,D_{l}}_{M,N}\lf(r\ri),&}}
whereas for $\sum_{i=1}^{K}D_{i}>K\cdot D_{l}$
\con{\begin{equation}
\min_{ A\subseteq \lf\{1,\dots,K\ri\}, \sum_{i=1}^{K}D_{i}> K\cdot
D_{l} }\lf(d^{\ast,\sum_{a\in A}D_{a}}_{|A|\cdot M,N}\lf(|A|\cdot
r\ri)\ri)\le d^{\ast,\sum_{i=1}^{K}D_{i}}_{M,N}\lf(K\cdot r\ri)\le
d^{\ast}\lf(r\ri) =d^{\ast,K\cdot D_{l}}_{M,N}\lf(K\cdot r\ri),
\end{equation}}{\bal{}{&\min_{ A\subseteq \lf\{1,\dots,K\ri\}, \sum_{i=1}^{K}D_{i}> K\cdot
D_{l} }\lf(d^{\ast,\sum_{a\in A}D_{a}}_{|A|\cdot M,N}\lf(|A|\cdot
r\ri)\ri)\le&
\nn\\
&d^{\ast,\sum_{i=1}^{K}D_{i}}_{M,N}\lf(K\cdot r\ri)\le
d^{\ast}\lf(r\ri) =d^{\ast,K\cdot D_{l}}_{M,N}\lf(K\cdot r\ri),&}}
which shows that $d^{\ast,\lf(IC\ri)}_{K,M,N}\lf(r\ri)$ is the
optimal solution. This concludes the proof.

\section{Proof of Lemma \ref{lem:ConLemCompareOptimalSymmetricDMTICandFC}} \label{Append:ConLemCompareOptimalSymmetricDMTICandFC}
For $N\ge \lf(K+1\ri)M-1$ it can be easily shown based on Lemma
\ref{lem:TheCasewhereSingleUserIstheWorst} and Corollary
\ref{prop:ICP2pDMTAnchorPoints} that
\begin{equation}\label{eq:ConverseLemCompareBetweenFCandICOptimalSymmetricDMT1}
d^{\ast,\lf(FC\ri)}_{K,M,N}\lf(r\ri)=d^{\ast,\lf(FC\ri)}_{M,N}\lf(r\ri)=d^{\ast,\lf(IC\ri)}_{K,M,N}\lf(r\ri).
\end{equation}

For $N<\lf(K-1\ri)M+1$ we get $\frac{L}{K}=\frac{N}{K}$. It follows
from \eqref{eq:KUserSmallestDiversityOrder},
\eqref{eq:KUserSmallestDiversityOrderLowerBound},
\eqref{eq:SingleUserLargestDiversityOrderUpperBound} that
\begin{equation*}
d^{\ast,D}_{M,N}\lf(0\ri)<d^{\ast,K\cdot D}_{K\cdot M,N}\lf(0\ri).
\end{equation*}
In addition, $d^{\ast,D}_{M,N}\lf(r\ri)$, $d^{\ast,K\cdot D}_{K\cdot
M,N}\lf(K\cdot r\ri)$ are straight lines, and
$d^{\ast,D}_{M,N}\lf(D\ri)=d^{\ast,K\cdot D}_{K\cdot M,N}\lf(K\cdot
D\ri)=0$. As a consequence we get
\con{\begin{equation}\label{eq:ConverseLemCompareBetweenFCandICOptimalSymmetricDMT2}
d^{\ast,D}_{M,N}\lf(r\ri)<d^{\ast,K\cdot D}_{K\cdot M,N}\lf(K\cdot
r\ri)\le d^{\ast,\lf(FC\ri)}_{KM,N}\lf(K\cdot r\ri)\quad
0<D\le\frac{N}{K}
\end{equation}}{\bal{eq:ConverseLemCompareBetweenFCandICOptimalSymmetricDMT2}{
&d^{\ast,D}_{M,N}\lf(r\ri)<d^{\ast,K\cdot D}_{K\cdot M,N}\lf(K\cdot
r\ri)\le&
\nn\\
&d^{\ast,\lf(FC\ri)}_{KM,N}\lf(K\cdot r\ri)\quad
0<D\le\frac{N}{K}&}} for $0<r<D$, where the second inequality results from Corollary
\ref{prop:RelationBedDandFCd}. In addition, since
$\frac{N}{K}<\frac{MN}{N+M-1}$, $0<D\le \frac{N}{K}$ and
$\lf(N+M-1\ri)<K\cdot M$ we get
\con{\begin{equation}\label{eq:ConverseLemCompareBetweenFCandICOptimalSymmetricDMT3}
d^{\ast,\lf(IC\ri)}_{K,M,N}\lf(r\ri)=MN-KMr<d^{\ast,\lf(FC\ri)}_{M,N}\lf(r\ri)=MN-\lf(N+M-1\ri)r
\end{equation}}{\bal{eq:ConverseLemCompareBetweenFCandICOptimalSymmetricDMT3}
{&d^{\ast,\lf(IC\ri)}_{K,M,N}\lf(r\ri)=MN-KMr<d^{\ast,\lf(FC\ri)}_{M,N}\lf(r\ri)&
\nn\\
&=MN-\lf(N+M-1\ri)r&}}
for $0<r\le\frac{N}{K}$. Since
$d^{\ast,\lf(FC\ri)}_{K,M,N}\lf(r\ri)$ consists of
$d^{\ast,\lf(FC\ri)}_{M,N}\lf(r\ri)$ and
$d^{\ast,\lf(FC\ri)}_{KM,N}\lf(K\cdot r\ri)$ we get from
\eqref{eq:ConverseLemCompareBetweenFCandICOptimalSymmetricDMT2},
\eqref{eq:ConverseLemCompareBetweenFCandICOptimalSymmetricDMT3} that
\begin{equation*}
d^{\ast,\lf(IC\ri)}_{K,M,N}\lf(r\ri)<d^{\ast,\lf(FC\ri)}_{K,M,N}\lf(r\ri)\quad
0<r<\frac{N}{K}.
\end{equation*}

For $N=\lf(K-1\ri)M+1+l$ and $l=0,\dots,2M-3$, recall that we
denoted $D_{l}=\frac{MN-\lfloor\frac{l}{2}\rfloor\cdot
\lf(\lfloor\frac{l}{2}\rfloor+1\ri)-2\cdot\lf(\lfloor\frac{l}{2}\rfloor+1\ri)\cdot\lf(\frac{l}{2}-\lfloor\frac{l}{2}\rfloor\ri)}{N+M-1-l}$
and also $r_{l}=\frac{\lf(K-1\ri)M+\lfloor\frac{l+1}{2}\rfloor}{K}$.
In \eqref{eq:ConverseStraightLineLem22} it was shown that
$D_{l}<\frac{MN-\lf(\lfloor\frac{l}{2}\rfloor+1\ri)\lf(\lfloor\frac{l}{2}\rfloor+2\ri)}{M+N-1-2\lf(\lfloor\frac{l}{2}\rfloor+1\ri)}$;
following the behavior of the straight lines around the anchor
points as presented in Lemma
\ref{lem:TheCasewhereSingleUserandKUserDMTUpperBoundCoincide} and
Corollary \ref{prop:ICP2pDMTAnchorPoints}, it is straightforward to
see that
\begin{equation}\label{eq:ConverseLemCompareBetweenFCandICOptimalSymmetricDMT4}
d^{\ast}\lf(r\ri)=
d^{\ast,D_{l}}_{M,N}\lf(r\ri)<d^{\ast,\lf(FC\ri)}_{M,N}\lf(r\ri)\quad
\lfloor\frac{l}{2}\rfloor+1<r\le\frac{L}{K}.
\end{equation}
On the other hand from \eqref{eq:ConverseStraightLineLem20} we get
$D_{l}>\frac{MN-r_{l}\lf(K\cdot r_{l}-1\ri)}{K\cdot M+N-1-2\lf(\cdot
K\cdot r_{l}-1\ri)}$. From similar arguments to
\eqref{eq:ConverseLemCompareBetweenFCandICOptimalSymmetricDMT4} it
follows that
\begin{equation}\label{eq:ConverseLemCompareBetweenFCandICOptimalSymmetricDMT5}
d^{\ast}\lf(r\ri)=d^{\ast,K\cdot D_{l}}_{K\cdot M,N}\lf(K\cdot
r\ri)<d^{\ast,\lf(FC\ri)}_{K\cdot M,N}\lf(K\cdot r\ri)
\end{equation}
where $0\le
r<\frac{\lf(K-1\ri)M+\lfloor\frac{l+1}{2}\rfloor}{K}$.
Since $d^{\ast,\lf(FC\ri)}_{K,M,N}\lf(r\ri)$ consists of
$d^{\ast,\lf(FC\ri)}_{M,N}\lf(r\ri)$ and
$d^{\ast,\lf(FC\ri)}_{K\cdot M,N}\lf(K\cdot r\ri)$, we get from
\eqref{eq:ConverseLemCompareBetweenFCandICOptimalSymmetricDMT4},
\eqref{eq:ConverseLemCompareBetweenFCandICOptimalSymmetricDMT5}
\begin{equation}
d^{\ast}\lf(r\ri)<d^{\ast,\lf(FC\ri)}_{K,M,N}\lf(r\ri)\quad
\lfloor\frac{l}{2}\rfloor+1<r<\frac{\lf(K-1\ri)M+\lfloor\frac{l+1}{2}\rfloor}{K}.
\end{equation}

The remaining open point for $N=\lf(K-1\ri)M+1+l$ and
$l=0,\dots,2M-3$ is the case
\begin{equation*}
\lfloor\frac{l}{2}\rfloor+1=\frac{\lf(K-1\ri)M+\lfloor\frac{l+1}{2}\rfloor}{K}.
\end{equation*}
In Theorem \ref{th:ICOptimalSymmetricDMTUpperNound} it was shown
(see equation \eqref{eq:ConverseTheoremOptimalSymmetricDMT31}
appendix \ref{append:ICOptimalSymmetricDMTUpperNound}) that we get
equality for $K=2$, $M=s+1$ and $N=3\cdot s$, where $s\ge 1$ is an
integer. According to Theorem \ref{prop:FCOptDMT}, for this case the
optimal DMT of finite constellations equals
\begin{equation*}
d^{\ast,\lf(FC\ri)}_{2,s+1,3\cdot s}\lf(r\ri)=\lf\{\begin{array}{lr}
d^{\ast,\lf(FC\ri)}_{s+1,3\cdot s}\lf(r\ri) &0\le r\le \frac{N}{K+1}=s\\
d^{\ast,\lf(FC\ri)}_{2 \lf(s+1\ri),3\cdot s}\lf(2\cdot r\ri) &s\le
r\le 3\cdot s.
\end{array}\ri.
\end{equation*}
Hence, from \eqref{eq:ConverseTheoremOptimalSymmetricDMT37} we get
$d^{\ast,\lf(FC\ri)}_{2,s+1,3\cdot
s}\lf(r\ri)=d^{\ast,\lf(IC\ri)}_{2,s+1,3\cdot s}\lf(r\ri)$. By
simply assigning we get that in this case $N<\lf(K+1\ri)M-1$. This
concludes the proof.

\section{Proof of Theorem
\ref{th:ConvThOptimalityandSuboptimalityofICinMACChannel}}\label{append:ConvThOptimalityandSuboptimalityofICinMACChannel_1}
We begin by finding for $N\ge \lf(K+1\ri)M-1$ an upper bound on the
DMT of the unconstrained multiple-access channel, that equals to the
optimal DMT of finite constellations
$d^{\ast,\lf(FC\ri)}_{M,N}\lf(\max \lf(r_{1},\dots,r_{K}\ri)  \ri)$.
The proof relies on the upper bound on the optimal DMT in the
symmetric case $d^{\ast,\lf(IC\ri)}_{K,M,N}\lf(r\ri)$. For $N\ge
\lf(K+1\ri)M-1$ it was shown in Lemma
\ref{lem:ConLemCompareOptimalSymmetricDMTICandFC} that
\begin{equation}\label{eq:ConvThOptimalityandSuboptimalityofICinMACChannel1}
d^{\ast,\lf(IC\ri)}_{K,M,N}\lf(r\ri)=d^{\ast,\lf(FC\ri)}_{M,N}\lf(r\ri).
\end{equation}
From Theorem \ref{Th:MDTFormulation} we get that the optimal DMT is
upper bounded by
\begin{equation}\label{eq:ConvThOptimalityandSuboptimalityofICinMACChannel2}
\max_{\lf(D_{1},\dots,D_{K}\ri)\in \RD}\min_{ A\subseteq
\lf\{1,\dots,K\ri\}} d^{\ast,D_{A}}_{|A|\cdot M,N}\lf(R_{A}\ri).
\end{equation}
We wish to solve
\eqref{eq:ConvThOptimalityandSuboptimalityofICinMACChannel2}. We
solve it by finding upper and lower bounds on
\eqref{eq:ConvThOptimalityandSuboptimalityofICinMACChannel2} that
coincide. For the rate tuple $\lf(r_{1},\dots,r_{K}\ri)$ recall the
definition $r_{max}=\max \lf(r_{1},\dots,r_{K}\ri)$. We begin by
lower bounding the optimization problem terms. Based on Lemma
\ref{lem:TheCasewhereSingleUserIstheWorst} and the fact that
$d^{\ast,i\cdot D}_{i\dot M,N}\lf(i\cdot r\ri)$, $i=1,\dots,K$ are
straight lines as a function of $r$ we get
\con{\begin{equation}\label{eq:ConvThOptimalityandSuboptimalityofICinMACChannel3}
d^{\ast,\sum_{a\in A}D_{a}}_{|A|\cdot M,N}\lf(\sum_{a\in
A}r_{a}\ri)\ge d^{\ast,\sum_{a\in A}D_{a}}_{|A|\cdot
M,N}\lf(|A|\cdot r_{max}\ri)\ge d^{\ast,\frac{\sum_{a\in
A}D_{a}}{|A|}}_{M,N}\lf(r_{max}\ri)\quad \forall A\subseteq
\lf\{1,\dots,K\ri\}.
\end{equation}}{\bal{eq:ConvThOptimalityandSuboptimalityofICinMACChannel3}
{&d^{\ast,\sum_{a\in A}D_{a}}_{|A|\cdot M,N}\lf(\sum_{a\in
A}r_{a}\ri)\ge d^{\ast,\sum_{a\in A}D_{a}}_{|A|\cdot
M,N}\lf(|A|\cdot r_{max}\ri)&
\nn\\
&\ge d^{\ast,\frac{\sum_{a\in
A}D_{a}}{|A|}}_{M,N}\lf(r_{max}\ri)\quad \forall A\subseteq
\lf\{1,\dots,K\ri\}.&}}
Hence, we get
\con{\begin{equation}\label{eq:ConvThOptimalityandSuboptimalityofICinMACChannel4}
\min_{ A\subseteq \lf\{1,\dots,K\ri\}} d^{\ast,\sum_{a\in
A}D_{a}}_{|A|\cdot M,N}\lf(\sum_{a\in A}r_{a}\ri)\ge \min_{
A\subseteq \lf\{1,\dots,K\ri\}} d^{\ast,\frac{\sum_{a\in
A}D_{a}}{|A|}}_{M,N}\lf(r_{max}\ri).
\end{equation}}{\bal{eq:ConvThOptimalityandSuboptimalityofICinMACChannel4}
{&\min_{ A\subseteq \lf\{1,\dots,K\ri\}} d^{\ast,\sum_{a\in
A}D_{a}}_{|A|\cdot M,N}\lf(\sum_{a\in A}r_{a}\ri)\ge&
\nn\\
&\min_{A\subseteq \lf\{1,\dots,K\ri\}} d^{\ast,\frac{\sum_{a\in
A}D_{a}}{|A|}}_{M,N}\lf(r_{max}\ri).&}}
From Corollary \ref{prop:RelationBedDandFCd} we know that
\begin{equation}\label{eq:ConvThOptimalityandSuboptimalityofICinMACChannel6}
\max_{D}d^{\ast,D}_{M,N}\lf(r_{max}\ri)=d^{\ast,\lf(FC\ri)}_{M,N}\lf(r_{max}\ri)
\end{equation}
is obtained for
\begin{equation}
D_{max}=\lf\{\begin{array}{cc} \frac{MN-\lfloor r_{max}\rfloor\cdot
\lf(\lfloor r_{max}\rfloor+1\ri)}{N+M-1-2\cdot \lfloor
r_{max}\rfloor} & 0\le r_{max}<M\\
M & r_{max}=M
\end{array}\ri
.
\end{equation}
Hence, from
\eqref{eq:ConvThOptimalityandSuboptimalityofICinMACChannel4},
\eqref{eq:ConvThOptimalityandSuboptimalityofICinMACChannel6} we get
\con{\begin{equation}\label{eq:ConvThOptimalityandSuboptimalityofICinMACChannel7}
\max_{\lf(D_{1},\dots,D_{K}\ri)\in \RD}\min_{ A\subseteq
\lf\{1,\dots,K\ri\}} d^{\ast,D_{A}}_{|A|\cdot M,N}\lf(R_{A}\ri)\ge
\max_{\lf(D_{1},\dots,D_{K}\ri)\in \RD}\min_{ A\subseteq
\lf\{1,\dots,K\ri\}}d^{\ast,\frac{\sum_{a\in
A}D_{a}}{|A|}}_{M,N}\lf(r_{max}\ri)=d^{\ast,\lf(FC\ri)}_{M,N}\lf(r_{max}\ri)
\end{equation}}{\bal{eq:ConvThOptimalityandSuboptimalityofICinMACChannel7}
{&\max_{\lf(D_{1},\dots,D_{K}\ri)\in \RD}\min_{ A\subseteq
\lf\{1,\dots,K\ri\}} d^{\ast,D_{A}}_{|A|\cdot M,N}\lf(R_{A}\ri)\ge&
\nn\\
&\max_{\lf(D_{1},\dots,D_{K}\ri)\in \RD}\min_{ A\subseteq
\lf\{1,\dots,K\ri\}}d^{\ast,\frac{\sum_{a\in
A}D_{a}}{|A|}}_{M,N}\lf(r_{max}\ri)&
\nn\\
&=d^{\ast,\lf(FC\ri)}_{M,N}\lf(r_{max}\ri)&}}
obtained for $D_{1}=\dots=D_{K}=D_{max}$; note that $N\ge
\lf(K+1\ri)M-1$ and so $K\cdot D_{max}\le K\cdot M\le N$. We now
upper bound the optimization problem and show it coincides with the
lower bound. Without loss of generality assume $r_{i}=r_{max}$. In
this case we get
\begin{equation}\label{eq:ConvThOptimalityandSuboptimalityofICinMACChannel8}
\min_{ A\subseteq \lf\{1,\dots,K\ri\}} d^{\ast,\sum_{a\in
A}D_{a}}_{|A|\cdot M,N}\lf(\sum_{a\in A}r_{a}\ri)\le
d^{\ast,D_{i}}_{M,N}\lf(r_{max}\ri).
\end{equation}
From \eqref{eq:ConvThOptimalityandSuboptimalityofICinMACChannel6},
\eqref{eq:ConvThOptimalityandSuboptimalityofICinMACChannel8} we can
write
\con{\begin{equation}\label{eq:ConvThOptimalityandSuboptimalityofICinMACChannel9}
\max_{\lf(D_{1},\dots,D_{K}\ri)\in \RD}\min_{ A\subseteq
\lf\{1,\dots,K\ri\}} d^{\ast,D_{A}}_{|A|\cdot M,N}\lf(R_{A}\ri)\le
\max_{D_{i}}d^{\ast,D_{i}}_{M,N}\lf(r_{max}\ri) =
d^{\ast,\lf(FC\ri)}_{M,N}\lf(r_{max}\ri)
\end{equation}}{\bal{eq:ConvThOptimalityandSuboptimalityofICinMACChannel9}
{&\max_{\lf(D_{1},\dots,D_{K}\ri)\in \RD}\min_{ A\subseteq
\lf\{1,\dots,K\ri\}} d^{\ast,D_{A}}_{|A|\cdot M,N}\lf(R_{A}\ri)\le&
\nn\\
&\max_{D_{i}}d^{\ast,D_{i}}_{M,N}\lf(r_{max}\ri) =
d^{\ast,\lf(FC\ri)}_{M,N}\lf(r_{max}\ri)&}}
obtained for $D_{i}=D_{max}$. Hence, from
\eqref{eq:ConvThOptimalityandSuboptimalityofICinMACChannel7},
\eqref{eq:ConvThOptimalityandSuboptimalityofICinMACChannel9} we get
\con{\begin{equation}
\max_{\lf(D_{1},\dots,D_{K}\ri)\in \RD}\min_{ A\subseteq
\lf\{1,\dots,K\ri\}} d^{\ast,D_{A}}_{|A|\cdot
M,N}\lf(R_{A}\ri)=d^{\ast,\lf(FC\ri)}_{M,N}\lf(r_{max}\ri)
\end{equation}}{\bal{}{&\max_{\lf(D_{1},\dots,D_{K}\ri)\in \RD}\min_{ A\subseteq
\lf\{1,\dots,K\ri\}} d^{\ast,D_{A}}_{|A|\cdot
M,N}\lf(R_{A}\ri)=&
\nn\\
&d^{\ast,\lf(FC\ri)}_{M,N}\lf(r_{max}\ri)&}}
which is the optimal DMT of finite constellations.

Now we show for $N<\lf(K+1\ri)M-1$ that the optimal DMT of the
unconstrained multiple-access channel is suboptimal compared to the
optimal DMT of finite constellations. We do that by showing that
there exists a set $B$ of multiplexing gain tuples
$\lf(r_{1},\dots,r_{K}\ri)$ for which
\con{\begin{equation*}
\max_{\lf(D_{1},\dots,D_{K}\ri)\in \RD}\min_{ A\subseteq
\lf\{1,\dots,K\ri\}} d^{\ast,D_{A}}_{|A|\cdot
M,N}\lf(R_{A}\ri)<d^{\ast,\lf(FC\ri)}_{K,M,N}
\lf(r_{1},\dots,r_{K}\ri)\quad \forall\lf(r_{1},\dots,r_{K}\ri)\in B
\end{equation*}}{\baln{}{&\max_{\lf(D_{1},\dots,D_{K}\ri)\in \RD}\min_{ A\subseteq
\lf\{1,\dots,K\ri\}} d^{\ast,D_{A}}_{|A|\cdot
M,N}\lf(R_{A}\ri)<&
\nn\\
&d^{\ast,\lf(FC\ri)}_{K,M,N}
\lf(r_{1},\dots,r_{K}\ri)\quad \forall\lf(r_{1},\dots,r_{K}\ri)\in B&}}
where $d^{\ast,\lf(FC\ri)}_{K,M,N}\lf(r_{1},\dots,r_{K}\ri)$ is the
optimal DMT of finite constellations. We divide the sub-optimality
proof of $N<\lf(K+1\ri)M-1$ to several cases. We begin with the case
$N<\lf(K-1\ri)M+1$. For this case we show the sub-optimality by
considering symmetric multiplexing gain tuples, i.e.,
$r_{1}=\dots=r_{K}=r$. In this case the optimization problem
\eqref{eq:ConvThOptimalityandSuboptimalityofICinMACChannel2}
solution equals $d^{\ast,\lf(IC\ri)}_{K,M,N}\lf(r\ri)$. From Lemma
\ref{lem:ConLemCompareOptimalSymmetricDMTICandFC} we get that
\begin{equation*}
d^{\ast,\lf(IC\ri)}_{K,M,N}\lf(r\ri)<d^{\ast,\lf(FC\ri)}_{K,M,N}\lf(r\ri)=d^{\ast,\lf(FC\ri)}_{K,M,N}\lf(r,\dots,r\ri)
\end{equation*}
for $0<r<\frac{N}{K}$. Hence, in this case we have proved the sub-optimality based on the
optimal DMT in the symmetric case. We now prove the sub-optimality
for $N=\lf(K-1\ri)M+1+l$, where $l=0,\dots,2M-3$. In Lemma
\ref{lem:ConLemCompareOptimalSymmetricDMTICandFC} we have showed for
$r_{1}=\dots=r_{K}=r$ that
\begin{equation}\label{eq:ConvThOptimalityandSuboptimalityofICinMACChannel10}
d^{\ast,\lf(IC\ri)}_{K,M,N}\lf(r\ri)<d^{\ast,\lf(FC\ri)}_{K,M,N}\lf(r\ri)
\end{equation}
$\lfloor\frac{l}{2}\rfloor+1<r<\frac{\lf(K-1\ri)M+\lfloor\frac{l+1}{2}\rfloor}{K}$. Hence, for $\lfloor\frac{l}{2}\rfloor+1\neq
\frac{\lf(K-1\ri)M+\lfloor\frac{l+1}{2}\rfloor}{K}$ this shows the
sub-optimality of any IC's DMT. Therefore, in order to complete the
sub-optimality proof we are left only with the case
$\lfloor\frac{l}{2}\rfloor+1=
\frac{\lf(K-1\ri)M+\lfloor\frac{l+1}{2}\rfloor}{K}$.

In Theorem \ref{th:ICOptimalSymmetricDMTUpperNound} we have shown
that $\lfloor\frac{l}{2}\rfloor+1=
\frac{\lf(K-1\ri)M+\lfloor\frac{l+1}{2}\rfloor}{K}$ only at $K=2$,
$M=s+1$ and $N=3\cdot s$, where $s\ge 1$ is an integer. Note that in
this case the upper bound on the optimal DMT of IC's in the
symmetric case equals to the optimal DMT of finite constellations.
Hence, in this case we can not obtain the sub-optimality from the
symmetric case and we need to find a set of multiplexing gain tuples
$B$ for which
\con{\begin{equation}\label{eq:ConvThOptimalityandSuboptimalityofICinMACChannel15}
\max_{\lf(D_{1},D_{2}\ri)}\min\lf(d^{\ast,D_{1}}_{s+1,3\cdot
s}\lf(r_{1}\ri),d^{\ast,D_{1}+D_{2}}_{2\lf(s+1\ri),3\cdot
s}\lf(r_{1}+r_{2}\ri),d^{\ast,D_{2}}_{s+1,3\cdot
s}\lf(r_{2}\ri)\ri)<d^{\ast,\lf(FC\ri)}_{2,s+1,3\cdot s}
\lf(r_{1},r_{2}\ri)\quad \lf(r_{1},r_{2}\ri)\in B.
\end{equation}}{\bal{eq:ConvThOptimalityandSuboptimalityofICinMACChannel15}
{&\max_{\lf(D_{1},D_{2}\ri)}&
\nn\\
&\min\lf(d^{\ast,D_{1}}_{s+1,3\cdot
s}\lf(r_{1}\ri),d^{\ast,D_{1}+D_{2}}_{2\lf(s+1\ri),3\cdot
s}\lf(r_{1}+r_{2}\ri),d^{\ast,D_{2}}_{s+1,3\cdot
s}\lf(r_{2}\ri)\ri)&
\nn\\
&<d^{\ast,\lf(FC\ri)}_{2,s+1,3\cdot s}
\lf(r_{1},r_{2}\ri)\quad \lf(r_{1},r_{2}\ri)\in B.&}}
We defer the proof of
\eqref{eq:ConvThOptimalityandSuboptimalityofICinMACChannel15} to
appendix
\ref{append:ConvThOptimalityandSuboptimalityofICinMACChannel}. In a
nutshell we are interested in finding a set such that the optimal
DMT of finite constellations equals to the two user optimal DMT,
i.e., $d^{\ast,D_{1}+D_{2}}_{2\lf(s+1\ri),3\cdot
s}\lf(r_{1}+r_{2}\ri)=d^{\ast,\lf(FC\ri)}_{2,s+1,3\cdot s}
\lf(r_{1},r_{2}\ri)$, whereas the IC's single user expressions
$d^{\ast,D_{1}}_{s+1,3\cdot s}\lf(r_{1}\ri)$ or
$d^{\ast,D_{2}}_{s+1,3\cdot s}\lf(r_{2}\ri)$  will be smaller than
$d^{\ast,\lf(FC\ri)}_{2,s+1,3\cdot s} \lf(r_{1},r_{2}\ri)$ for any
$D_{1}$, $D_{2}$ for which
$d^{\ast,D_{1}+D_{2}}_{2\lf(s+1\ri),3\cdot
s}\lf(r_{1}+r_{2}\ri)=d^{\ast,\lf(FC\ri)}_{2,s+1,3\cdot s}
\lf(r_{1},r_{2}\ri)$. Figure
\ref{fig:ConvThOptimalityandSuboptimalityofICinMACChannel} shows the
optimal DMT of finite constellations for the case $K=2$, $M=3$ and
$N=6$, and Figure
\ref{fig:ConvThOptimalityandSuboptimalityofICinMACChannel_b}
illustrates the aforementioned description of the proof method for
the same setting.

\section{Final Part of the Proof of Theorem
\ref{th:ConvThOptimalityandSuboptimalityofICinMACChannel}}\label{append:ConvThOptimalityandSuboptimalityofICinMACChannel}
In order to find the set $B$ we first present several properties of
$d^{\ast,\lf(IC\ri)}_{2,s+1,3\cdot s}\lf(r\ri)$, i.e., the optimal
DMT of IC's in the symmetric case, for this case. First note that
from Theorem \ref{th:ICOptimalSymmetricDMTUpperNound} we get
\con{\begin{equation*}
d^{\ast,\lf(IC\ri)}_{2,s+1,3\cdot s}\lf(r\ri)=\lf\{\begin{array}{ll}
d^{\ast,\lf(FC\ri)}_{s+1,3\cdot s}\lf(r\ri) &0\le r\le \frac{N}{K+1}=s\\
d^{\ast,\lf(FC\ri)}_{2 \lf(s+1\ri),3\cdot s}\lf(2\cdot r\ri) &s\le
r\le \min \lf(s+1,\frac{3}{2}s\ri)
\end{array}
\ri. =d^{\ast,\lf(FC\ri)}_{2,s+1,3\cdot s}\lf(r\ri) .
\end{equation*}}{\baln{}{&d^{\ast,\lf(IC\ri)}_{2,s+1,3\cdot s}\lf(r\ri)=\lf\{\begin{array}{ll}
d^{\ast,\lf(FC\ri)}_{s+1,3\cdot s}\lf(r\ri) &0\le r\le \frac{N}{K+1}=s\\
d^{\ast,\lf(FC\ri)}_{2 \lf(s+1\ri),3\cdot s}\lf(2\cdot r\ri) &s\le
r\le \min \lf(s+1,\frac{3}{2}s\ri)
\end{array}
\ri.&
\nn\\
&=d^{\ast,\lf(FC\ri)}_{2,s+1,3\cdot s}\lf(r\ri) .&}}
An example of $d^{\ast,\lf(IC\ri)}_{2,s+1,3\cdot s}\lf(r\ri)$ for
$M=3$, $N=6$ and $K=2$, i.e., $s=2$, is given in Figure
\ref{fig:ConvThOptimalityandSuboptimalityofICinMACChannel}.

From simple assignment of the values of $M$, $N$ and $K$ we get that
$l=2 \lf(s-1\ri)$. We know from Lemma
\ref{lem:TheCasewhereSingleUserandKUserDMTUpperBoundCoincide},
Theorem \ref{prop:FCOptDMT} and
\eqref{eq:ConverseTheoremOptimalSymmetricDMT31} that
\con{\begin{equation}\label{eq:ConvThOptimalityandSuboptimalityofICinMACChannel11}
d^{\ast,D_{l}}_{s+1,3\cdot
s}\lf(\frac{N}{K+1}\ri)=d^{\ast,\lf(FC\ri)}_{s+1,3\cdot
s}\lf(\frac{N}{K+1}\ri)=d^{\ast,\lf(FC\ri)}_{2\cdot
\lf(s+1\ri),3\cdot s}\lf(\frac{K\cdot N}{K+1}\ri)=d^{\ast,2\cdot
D_{l}}_{2 \lf(s+1\ri),3\cdot s}\lf(\frac{K\cdot N}{K+1}\ri).
\end{equation}}{\bal{eq:ConvThOptimalityandSuboptimalityofICinMACChannel11}
{&d^{\ast,D_{l}}_{s+1,3\cdot
s}\lf(\frac{N}{K+1}\ri)=d^{\ast,\lf(FC\ri)}_{s+1,3\cdot
s}\lf(\frac{N}{K+1}\ri)&
\nn\\
&=d^{\ast,\lf(FC\ri)}_{2\cdot
\lf(s+1\ri),3\cdot s}\lf(\frac{K\cdot N}{K+1}\ri)=d^{\ast,2\cdot
D_{l}}_{2 \lf(s+1\ri),3\cdot s}\lf(\frac{K\cdot N}{K+1}\ri).&}}
Hence, from \eqref{eq:ConverseTheoremOptimalSymmetricDMT32} and
\eqref{eq:ConvThOptimalityandSuboptimalityofICinMACChannel11} we get
\begin{equation}\label{eq:ConvThOptimalityandSuboptimalityofICinMACChannel12}
d^{\ast,D_{\lfloor\frac{l}{2}\rfloor}^{\ast}}_{s+1,3\cdot
s}\lf(r\ri)=d^{\ast,2\cdot D_{r_{l}}^{\ast}}_{2 \lf(s+1\ri),3\cdot
s}\lf(2\cdot r\ri).
\end{equation}
Finally, it follows from Corollary \ref{prop:ICP2pDMTAnchorPoints}
that at $D_{\lfloor\frac{l}{2}\rfloor}^{\ast}$
\begin{equation}\label{eq:ConvThOptimalityandSuboptimalityofICinMACChannel12_A}
d^{\ast,D_{\lfloor\frac{l}{2}\rfloor}^{\ast}}_{s+1,3\cdot
s}\lf(s-1\ri)=d^{\ast,\lf(FC\ri)}_{s+1,3\cdot s}\lf(s-1\ri)
\end{equation}
and therefore from \eqref{eq:ConverseTheoremOptimalSymmetricDMT32},
\eqref{eq:ConvThOptimalityandSuboptimalityofICinMACChannel11},
\eqref{eq:ConvThOptimalityandSuboptimalityofICinMACChannel12},
\eqref{eq:ConvThOptimalityandSuboptimalityofICinMACChannel12_A} and
the fact that $d^{\ast,\lf(FC\ri)}_{s+1,3\cdot s}\lf(s-1\ri)$ is a
straight line in the range $s-1\le r\le s$ we get
\begin{equation}\label{eq:ConvThOptimalityandSuboptimalityofICinMACChannel13}
d^{\ast,D_{\lfloor\frac{l}{2}\rfloor}^{\ast}}_{s+1,3\cdot
s}\lf(r\ri)=d^{\ast,2\cdot D_{r_{l}}^{\ast}}_{2\cdot
\lf(s+1\ri),3\cdot s}\lf(2\cdot
r\ri)=d^{\ast,\lf(FC\ri)}_{s+1,3\cdot s}\lf(r\ri)
\end{equation}
where $s-1\le
r\le\frac{N}{K+1}=s$. From similar arguments we get
\begin{equation}\label{eq:ConvThOptimalityandSuboptimalityofICinMACChannel14}
d^{\ast,D_{\lfloor\frac{l}{2}\rfloor}^{\ast}}_{s+1,3\cdot
s}\lf(r\ri)=d^{\ast,2\cdot D_{r_{l}}^{\ast}}_{2\cdot
\lf(s+1\ri),3\cdot s}\lf(2\cdot r\ri)=d^{\ast,\lf(FC\ri)}_{2\cdot
\lf(s+1\ri),3\cdot s}\lf(2\cdot r\ri)
\end{equation}
where $s\le r\le s+\frac{1}{2}$, i.e., The last line of $d^{\ast,\lf(FC\ri)}_{s+1,3\cdot s}\lf(r\ri)$
before $\frac{N}{K+1}=s$, and the first line of
$d^{\ast,\lf(FC\ri)}_{2 \lf(s+1\ri),3\cdot s}\lf(2r\ri)$ after $s$
are equal. To sum up, for $\lfloor\frac{l}{2}\rfloor+1=
\frac{\lf(K-1\ri)M+\lfloor\frac{l+1}{2}\rfloor}{K}$ the optimal DMT
of IC's in the symmetric case is upper bounded by a piecewise linear
function as expected, and we have found the straight line coincide
with it for $s-1\le r\le s+\frac{1}{2}$. We are interested in
finding a set of multiplexing gain tuples $B$, for which
\eqref{eq:ConvThOptimalityandSuboptimalityofICinMACChannel15} is
fulfilled. In a nutshell we are interested in finding a set such
that the optimal DMT of finite constellations equals to the two user
optimal DMT, whereas IC's single user expressions will be smaller
than the optimal DMT of finite constellations for any $D_{1}$,
$D_{2}$ for which the IC's two users expression equals to the
optimal DMT of finite constellations. Figure
\ref{fig:ConvThOptimalityandSuboptimalityofICinMACChannel_b}
illustrates the aforementioned description of the proof method.

From Corollary \ref{prop:RelationBedDandFCd} we know that
\begin{equation}\label{eq:ConvThOptimalityandSuboptimalityofICinMACChannel16}
d^{\ast,\lf(FC\ri)}_{s+1,3\cdot
s}\lf(r\ri)=d^{\ast,D_{\lfloor\frac{l}{2}\rfloor+1}^{\ast}}_{s+1,3\cdot
s}\lf(r\ri)\quad s\le r\le s+1.
\end{equation}
Hence, for certain $s<r_{0}<s+\frac{1}{2}$, we are interested in the
set for which $r_{1}=r_{0}+\epsilon$, $r_{2}=r_{0}-\epsilon$ such
that $s<r_{0}+\epsilon<s+\frac{1}{2}$ and also
\con{\begin{equation}\label{eq:ConvThOptimalityandSuboptimalityofICinMACChannel17}
d^{\ast,D_{\lfloor\frac{l}{2}\rfloor}^{\ast}}_{s+1,3\cdot
s}\lf(r_{0}\ri)= d^{\ast,\lf(FC\ri)}_{2\lf(s+1\ri),3\cdot
s}\lf(2r_{0}\ri)< d^{\ast,\lf(FC\ri)}_{s+1,3\cdot
s}\lf(r_{0}+\epsilon\ri)=d^{\ast,D_{\lfloor\frac{l}{2}\rfloor+1}^{\ast}}_{s+1,3\cdot
s}\lf(r_{0}+\epsilon\ri)
\end{equation}}{\bal{eq:ConvThOptimalityandSuboptimalityofICinMACChannel17}
{&d^{\ast,D_{\lfloor\frac{l}{2}\rfloor}^{\ast}}_{s+1,3\cdot
s}\lf(r_{0}\ri)= d^{\ast,\lf(FC\ri)}_{2\lf(s+1\ri),3\cdot
s}\lf(2r_{0}\ri)<&
\nn\\
&d^{\ast,\lf(FC\ri)}_{s+1,3\cdot
s}\lf(r_{0}+\epsilon\ri)=d^{\ast,D_{\lfloor\frac{l}{2}\rfloor+1}^{\ast}}_{s+1,3\cdot
s}\lf(r_{0}+\epsilon\ri)&}}
where the first equality results from
\eqref{eq:ConvThOptimalityandSuboptimalityofICinMACChannel14}. Note
that the inequality in
\eqref{eq:ConvThOptimalityandSuboptimalityofICinMACChannel17} holds
as, based on Corollary \ref{prop:ICP2pDMTAnchorPoints} and Corollary
\ref{prop:RelationBedDandFCd},
$d^{\ast,D_{\lfloor\frac{l}{2}\rfloor}^{\ast}}_{s+1,3\cdot
s}\lf(r\ri)<d^{\ast,D_{\lfloor\frac{l}{2}\rfloor+1}^{\ast}}_{s+1,3\cdot
s}\lf(r\ri)$ for $r>s$. In order to translate this condition to
$\epsilon$ we write the following inequality
\con{\begin{align}\label{eq:ConvThOptimalityandSuboptimalityofICinMACChannel18}
d^{\ast,D_{\lfloor\frac{l}{2}\rfloor+1}^{\ast}}_{s+1,3\cdot s}\lf(r_{0}+\epsilon\ri)=MN &-\lf(\lfloor\frac{l}{2}\rfloor+1\ri)\cdot \lf(\lfloor\frac{l}{2}\rfloor+2\ri)-\lf(N+M-1-2\cdot \lf(\lfloor\frac{l}{2}\rfloor+1\ri)\ri) \lf(r_{0}+\epsilon\ri)>\nonumber\\
&MN-\lfloor\frac{l}{2}\rfloor\cdot
\lf(\lfloor\frac{l}{2}\rfloor+1\ri)-\lf(N+M-1-2\cdot
\lfloor\frac{l}{2}\rfloor
\ri)r_{0}=d^{\ast,D_{\lfloor\frac{l}{2}\rfloor}^{\ast}}_{s+1,3\cdot
s}\lf(r_{0}\ri)
\end{align}}{\bal{eq:ConvThOptimalityandSuboptimalityofICinMACChannel18}
{&d^{\ast,D_{\lfloor\frac{l}{2}\rfloor+1}^{\ast}}_{s+1,3\cdot s}\lf(r_{0}+\epsilon\ri)=MN -\lf(\lfloor\frac{l}{2}\rfloor+1\ri)\cdot \lf(\lfloor\frac{l}{2}\rfloor+2\ri)-&
\nn\\
&\lf(N+M-1-2\cdot \lf(\lfloor\frac{l}{2}\rfloor+1\ri)\ri) \lf(r_{0}+\epsilon\ri)>
MN-&
\nn\\
&\lfloor\frac{l}{2}\rfloor\cdot\lf(\lfloor\frac{l}{2}\rfloor+1\ri)-\lf(N+M-1-2\cdot
\lfloor\frac{l}{2}\rfloor
\ri)r_{0}=&
\nn\\
&d^{\ast,D_{\lfloor\frac{l}{2}\rfloor}^{\ast}}_{s+1,3\cdot
s}\lf(r_{0}\ri)&}}
for $K=2$, $M=s+1$ and $N=3\cdot s$ we get
\begin{equation}
\epsilon<\frac{r_{0}}{s}-1.
\end{equation}
Hence, the set of multiplexing gain tuples we are considering is
\con{\begin{equation}\label{eq:ConvThOptimalityandSuboptimalityofICinMACChannel19}
B_{r_{0}}=\lf\{r_{1},r_{2}|r_{1}=r_{0}+\epsilon,r_{2}=r_{0}-\epsilon,0<\epsilon<\min
\lf(r_{0}+\frac{r_{0}}{s}-1,s+\frac{1}{2}\ri)-r_{0}\ri\}
\end{equation}}{\bals{eq:ConvThOptimalityandSuboptimalityofICinMACChannel19}
{&B_{r_{0}}=\lf\{r_{1},r_{2}|r_{1}=r_{0}+\epsilon, r_{2}=r_{0}-\epsilon,\ri.&
\nn\\
&\lf. 0<\epsilon<\min
\lf(r_{0}+\frac{r_{0}}{s}-1,s+\frac{1}{2}\ri)-r_{0}\ri\}&}}
where $s<r_{0}<s+\frac{1}{2}$ is a parameter determining the set.
From \cite[Lemma 7]{ZhengTseMACDMT2004} we get that the optimal DMT
of finite constellations equals
\con{\begin{equation}
d^{\ast,\lf(FC\ri)}_{2,s+1,3\cdot s}\lf(r_{1},r_{2}\ri)=\min
\lf(d^{\ast,\lf(FC\ri)}_{s+1,3\cdot
s}\lf(r_{1}\ri),d^{\ast,\lf(FC\ri)}_{s+1,3\cdot
s}\lf(r_{2}\ri),d^{\ast,\lf(FC\ri)}_{2 \lf(s+1\ri),3\cdot
s}\lf(r_{1}+r_{2}\ri)\ri).
\end{equation}}{\bal{}{&d^{\ast,\lf(FC\ri)}_{2,s+1,3\cdot s}\lf(r_{1},r_{2}\ri)=\min\lf(d^{\ast,\lf(FC\ri)}_{s+1,3\cdot
s}\lf(r_{1}\ri),\ri.&
\nn\\
&\lf. d^{\ast,\lf(FC\ri)}_{s+1,3\cdot
s}\lf(r_{2}\ri),d^{\ast,\lf(FC\ri)}_{2 \lf(s+1\ri),3\cdot
s}\lf(r_{1}+r_{2}\ri)\ri).&}}
Considering $\lf(r_{1},r_{2}\ri)\in B_{r_{0}}$, based on
\eqref{eq:ConvThOptimalityandSuboptimalityofICinMACChannel17},
\eqref{eq:ConvThOptimalityandSuboptimalityofICinMACChannel19} and
the fact that $d^{\ast,\lf(FC\ri)}_{s+1,3\cdot s}\lf(r\ri)$ is a
straight line, we get
\con{\begin{align}\label{eq:ConvThOptimalityandSuboptimalityofICinMACChannel20}
d^{\ast,\lf(FC\ri)}_{2,s+1,3\cdot s}\lf(r_{1},r_{2}\ri)=\min \lf(d^{\ast,\lf(FC\ri)}_{s+1,3\cdot s}\lf(r_{0}+\epsilon\ri),d^{\ast,\lf(FC\ri)}_{s+1,3\cdot s}\lf(r_{0}-\epsilon\ri),d^{\ast,\lf(FC\ri)}_{2 \lf(s+1\ri),3\cdot s}\lf(2r_{0}\ri)\ri)\nonumber\\
=d^{\ast,\lf(FC\ri)}_{2 \lf(s+1\ri),3\cdot s}\lf(2r_{0}\ri)
\end{align}}{\bal{eq:ConvThOptimalityandSuboptimalityofICinMACChannel20}
{&d^{\ast,\lf(FC\ri)}_{2,s+1,3\cdot s}\lf(r_{1},r_{2}\ri)=&
\nn\\
&\min \lf(d^{\ast,\lf(FC\ri)}_{s+1,3\cdot s}\lf(r_{0}+\epsilon\ri),d^{\ast,\lf(FC\ri)}_{s+1,3\cdot s}\lf(r_{0}-\epsilon\ri),d^{\ast,\lf(FC\ri)}_{2 \lf(s+1\ri),3\cdot s}\lf(2r_{0}\ri)\ri)&
\nn\\
&=d^{\ast,\lf(FC\ri)}_{2 \lf(s+1\ri),3\cdot s}\lf(2r_{0}\ri)&}}
where $0<\epsilon<\min \lf(r_{0}+\frac{r_{0}}{s}-1,s+\frac{1}{2}\ri)-r_{0}$. Hence, in order to prove
\eqref{eq:ConvThOptimalityandSuboptimalityofICinMACChannel15} we
need to show for certain $0<r_{0}<s+\frac{1}{2}$ that
\con{\begin{equation}\label{eq:ConvThOptimalityandSuboptimalityofICinMACChannel21}
\max_{\lf(D_{1},D_{2}\ri)}\min\lf(d^{\ast,D_{1}}_{s+1,3\cdot
s}\lf(r_{0}+\epsilon\ri),d^{\ast,D_{1}+D_{2}}_{2\lf(s+1\ri),3\cdot
s}\lf(2r_{0}\ri),d^{\ast,D_{2}}_{s+1,3\cdot
s}\lf(r_{0}-\epsilon\ri)\ri)<d^{\ast,\lf(FC\ri)}_{2\lf(s+1\ri),3\cdot
s}\lf(2r_{0}\ri)
\end{equation}}{\bal{eq:ConvThOptimalityandSuboptimalityofICinMACChannel21}
{&\max_{\lf(D_{1},D_{2}\ri)}\min\lf(d^{\ast,D_{1}}_{s+1,3\cdot
s}\lf(r_{0}+\epsilon\ri),\ri.
\nn\\
&\lf.d^{\ast,D_{1}+D_{2}}_{2\lf(s+1\ri),3\cdot
s}\lf(2r_{0}\ri),d^{\ast,D_{2}}_{s+1,3\cdot
s}\lf(r_{0}-\epsilon\ri)\ri)<&
\nn\\
&d^{\ast,\lf(FC\ri)}_{2\lf(s+1\ri),3\cdot
s}\lf(2r_{0}\ri)&}}
where $0<\epsilon<\min
\lf(r_{0}+\frac{r_{0}}{s}-1,s+\frac{1}{2}\ri)-r_{0}$. We begin the
proof by taking the symmetric case, i.e., $D_{1}=D_{2}$, as a
baseline. We assign
$D_{1}=D_{2}=D_{r_{l}}^{\ast}=D_{\lfloor\frac{l}{2}\rfloor}^{\ast}$.
From \eqref{eq:ConvThOptimalityandSuboptimalityofICinMACChannel14}
we get that $d^{\ast,2 D^{\ast}_{r_{l}}}_{2\lf(s+1\ri),3\cdot
s}\lf(2
r_{0}\ri)=d^{\ast,D_{\lfloor\frac{l}{2}\rfloor}^{\ast}}_{s+1,3\cdot
s}\lf(r_{0}\ri)=d^{\ast,\lf(FC\ri)}_{2\lf(s+1\ri),3\cdot
s}\lf(2r_{0}\ri)$. Hence for the symmetric case we get
\con{\begin{equation}\label{eq:ConvThOptimalityandSuboptimalityofICinMACChannel23}
\min\lf(d^{\ast,D_{\lfloor\frac{l}{2}\rfloor}^{\ast}}_{s+1,3\cdot
s}\lf(r_{0}+\epsilon\ri),d^{\ast,D_{\lfloor\frac{l}{2}\rfloor}^{\ast}}_{s+1,3\cdot
s}\lf(r_{0}-\epsilon\ri),d^{\ast,D_{\lfloor\frac{l}{2}\rfloor}^{\ast}}_{s+1,3\cdot
s}\lf(r_{0}\ri)\ri)
=d^{\ast,D_{\lfloor\frac{l}{2}\rfloor}^{\ast}}_{s+1,3\cdot
s}\lf(r_{0}+\epsilon\ri)<d^{\ast,\lf(FC\ri)}_{2\lf(s+1\ri),3\cdot
s}\lf(2r_{0}\ri).
\end{equation}}{\bal{eq:ConvThOptimalityandSuboptimalityofICinMACChannel23}
{&\min\lf(d^{\ast,D_{\lfloor\frac{l}{2}\rfloor}^{\ast}}_{s+1,3\cdot
s}\lf(r_{0}+\epsilon\ri),d^{\ast,D_{\lfloor\frac{l}{2}\rfloor}^{\ast}}_{s+1,3\cdot
s}\lf(r_{0}-\epsilon\ri),d^{\ast,D_{\lfloor\frac{l}{2}\rfloor}^{\ast}}_{s+1,3\cdot
s}\lf(r_{0}\ri)\ri)
=&
\nn\\
&d^{\ast,D_{\lfloor\frac{l}{2}\rfloor}^{\ast}}_{s+1,3\cdot
s}\lf(r_{0}+\epsilon\ri)<d^{\ast,\lf(FC\ri)}_{2\lf(s+1\ri),3\cdot
s}\lf(2r_{0}\ri).&}}
Since $s<r_{0}<s+\frac{1}{2}$ is not an anchor point, we get from
\eqref{eq:ConvThOptimalityandSuboptimalityofICinMACChannel14} and
the anchor point behavior presented in Corollary
\ref{prop:ICP2pDMTAnchorPoints} that
$d^{\ast,D_{1}+D_{2}}_{2\lf(s+1\ri),3\cdot
s}\lf(2r_{0}\ri)=d^{\ast,\lf(FC\ri)}_{2\lf(s+1\ri),3\cdot
s}\lf(2r_{0}\ri)$ if and only if
$D_{1}+D_{2}=2D_{r_{l}}^{\ast}=2D_{\lfloor\frac{l}{2}\rfloor}^{\ast}$.
Hence, in order for $d^{\ast,D_{1}+D_{2}}_{2\lf(s+1\ri),3\cdot
s}\lf(2r_{0}\ri)$
\eqref{eq:ConvThOptimalityandSuboptimalityofICinMACChannel21} to
attain the optimal DMT of finite constellations, we must choose
\begin{equation}\label{eq:ConvThOptimalityandSuboptimalityofICinMACChannel25}
 D_{1}+D_{2}=2D_{\lfloor\frac{l}{2}\rfloor}^{\ast}.
\end{equation}
From \eqref{eq:ConvThOptimalityandSuboptimalityofICinMACChannel17},
\eqref{eq:ConvThOptimalityandSuboptimalityofICinMACChannel23} we
know that
\begin{equation}
d^{\ast,D_{\lfloor\frac{l}{2}\rfloor}^{\ast}}_{s+1,3\cdot
s}\lf(r_{0}+\epsilon\ri)<d^{\ast,\lf(FC\ri)}_{2\lf(s+1\ri),3\cdot
s}\lf(2
r_{0}\ri)<d^{\ast,D_{\lfloor\frac{l}{2}\rfloor+1}^{\ast}}_{s+1,3\cdot
s}\lf(r_{0}+\epsilon\ri).
\end{equation}
Since $s<r_{0}<s+\frac{1}{2}$, and based on the anchor points
behavior presented in Corollary \ref{prop:ICP2pDMTAnchorPoints},
from which we know that for
$D_{\lfloor\frac{l}{2}\rfloor}^{\ast}<D<D_{\lfloor\frac{l}{2}\rfloor+1}^{\ast}$
there is an anchor point at $r=s$, we can see that there must exist
$D^{'}=D_{\lfloor\frac{l}{2}\rfloor}^{\ast}+\epsilon^{'}$, where
$0<\epsilon^{'}<D_{\lfloor\frac{l}{2}\rfloor+1}^{\ast}-D_{\lfloor\frac{l}{2}\rfloor}^{\ast}$,
such that
\begin{equation}\label{eq:ConvThOptimalityandSuboptimalityofICinMACChannel26}
d^{\ast,D^{'}}_{s+1,3\cdot
s}\lf(r_{0}+\epsilon\ri)=d^{\ast,\lf(FC\ri)}_{2\lf(s+1\ri),3\cdot
s}\lf(2 r_{0}\ri).
\end{equation}
We divide the assignment of $D_{1}$ into several cases. In the range
$0<D_{1}<D^{'}$ following the anchor point behavior of the straight
lines presented in Corollary \ref{prop:ICP2pDMTAnchorPoints}, and
also since $s<r_{0}+\epsilon<s+\frac{1}{2}$ is not an anchor point
we get
\begin{equation}\label{eq:ConvThOptimalityandSuboptimalityofICinMACChannel27}
d^{\ast,D_{1}}_{s+1,3\cdot
s}\lf(r_{0}+\epsilon\ri)<d^{\ast,D^{'}}_{s+1,3\cdot
s}\lf(r_{0}+\epsilon\ri)=d^{\ast,\lf(FC\ri)}_{2\lf(s+1\ri),3\cdot
s}\lf(2 r_{0}\ri).
\end{equation}
Hence in this range the optimal DMT of finite constellations is not
obtained. For
$D_{1}=D^{'}=D_{\lfloor\frac{l}{2}\rfloor}^{\ast}+\epsilon^{'}$, we
have shown
\eqref{eq:ConvThOptimalityandSuboptimalityofICinMACChannel26} that
$d^{\ast,D^{'}}_{s-1,3\cdot s}\lf(r_{0}+\epsilon\ri)$ equals to the
optimal DMT of finite constellations. According to
\eqref{eq:ConvThOptimalityandSuboptimalityofICinMACChannel25} we
need to assign
$D_{2}=D^{''}=D_{\lfloor\frac{l}{2}\rfloor}^{\ast}-\epsilon^{'}$ in
order to get $D_{1}+D_{2}=2D_{\lfloor\frac{l}{2}\rfloor}^{\ast}$ and
as a consequence
\begin{equation*}
d^{\ast,D^{'}}_{s+1,3\cdot
s}\lf(r_{0}+\epsilon\ri)=d^{\ast,2D_{\lfloor\frac{l}{2}\rfloor}^{\ast}}_{2\lf(s+1\ri),3\cdot
s}\lf(2r_{0}\ri)=d^{\ast,\lf(FC\ri)}_{2\lf(s+1\ri),3\cdot s}\lf(2
r_{0}\ri).
\end{equation*}
So far we have shown that the first two terms in the left side of
\eqref{eq:ConvThOptimalityandSuboptimalityofICinMACChannel21} can
attain the optimal DMT of finite constellations for $D_{1}=D^{'}$.
We are left with the third term that equals to the straight line
$d^{\ast,D^{''}}_{s+1,3\cdot s}\lf(r\ri)$. We consider two cases. In
the first case we assume $D^{''}\le r_{0}-\epsilon$ for which we get
\begin{equation}\label{eq:ConvThOptimalityandSuboptimalityofICinMACChannel29}
d^{\ast,D^{''}}_{s+1,3\cdot
s}\lf(r_{0}-\epsilon\ri)=0<d^{\ast,\lf(FC\ri)}_{2\lf(s+1\ri),3\cdot
s}\lf(2r_{0}\ri).
\end{equation}
In the second case we assume $D^{''}>r_{0}-\epsilon$. From symmetry
considerations it can be easily shown that the straight line
$d^{'}\lf(r\ri)$ that fulfils
$d^{'}\lf(s\ri)=d^{\ast,\lf(FC\ri)}_{s+1,3\cdot
s}\lf(s\ri)=d^{\ast,D^{'}}_{s+1,3\cdot s}\lf(s\ri)$ and
$d^{'}\lf(D^{''}\ri)=0$, also fulfills
\begin{equation}\label{eq:ConvThOptimalityandSuboptimalityofICinMACChannel28}
d^{'}\lf(r_{0}-\epsilon\ri)=d^{\ast,D^{'}}_{s+1,3\cdot
s}\lf(r_{0}+\epsilon\ri)=d^{\ast,\lf(FC\ri)}_{2\lf(s+1\ri),3\cdot
s}\lf(2r_{0}\ri).
\end{equation}
Since $D^{''}<D_{\lfloor\frac{l}{2}\rfloor}^{\ast}$, we get from
Corollary \ref{prop:ICP2pDMTAnchorPoints} that the anchor point of
the straight line $d^{\ast,D^{''}}_{s+1,3\cdot s}\lf(s\ri)$ is
smaller than $s$ and so
\begin{equation}
d^{\ast,D^{''}}_{s+1,3\cdot
s}\lf(s\ri)<d^{\ast,D_{\lfloor\frac{l}{2}\rfloor}^{\ast}}_{s+1,3\cdot
s}\lf(s\ri)=d^{'}\lf(s\ri).
\end{equation}
Since $d^{\ast,D^{''}}_{s+1,3\cdot
s}\lf(D^{''}\ri)=d^{'}\lf(D^{''}\ri)=0$ and these are straight lines
we get
\begin{equation}
d^{\ast,D^{''}}_{s+1,3\cdot s}\lf(r\ri)<d^{'}\lf(r\ri)\quad
0<r<D^{''}
\end{equation}
and so from
\eqref{eq:ConvThOptimalityandSuboptimalityofICinMACChannel28}
\begin{equation}\label{eq:ConvThOptimalityandSuboptimalityofICinMACChannel30}
d^{\ast,D^{''}}_{s+1,3\cdot
s}\lf(r_{0}-\epsilon\ri)<d^{'}\lf(r_{0}-\epsilon\ri)=d^{\ast,\lf(FC\ri)}_{2\lf(s+1\ri),3\cdot
s}\lf(2r_{0}\ri).
\end{equation}
Thus, the third term in the left side of
\eqref{eq:ConvThOptimalityandSuboptimalityofICinMACChannel21}
$d^{\ast,D_{2}}_{s+1,3\cdot s}\lf(r_{0}-\epsilon\ri)$ is smaller
than the optimal DMT of finite constellations. Finally, we consider
the case $D_{1}>D^{'}$. For this case we get
$D_{2}<D^{''}<D_{\lfloor\frac{l}{2}\rfloor}^{\ast}$, which based on
the anchor points behavior in Corollary
\ref{prop:ICP2pDMTAnchorPoints}, and similarly to the previously
mentioned arguments  leads to
\begin{equation}\label{eq:ConvThOptimalityandSuboptimalityofICinMACChannel31}
d^{\ast,D_{2}}_{s+1,3\cdot
s}\lf(r_{0}-\epsilon\ri)<d^{\ast,D^{''}}_{s+1,3\cdot
s}\lf(r_{0}-\epsilon\ri)<d^{\ast,\lf(FC\ri)}_{2\lf(s+1\ri),3\cdot
s}\lf(2r_{0}\ri).
\end{equation}
From
\eqref{eq:ConvThOptimalityandSuboptimalityofICinMACChannel27},\eqref{eq:ConvThOptimalityandSuboptimalityofICinMACChannel29},
\eqref{eq:ConvThOptimalityandSuboptimalityofICinMACChannel30} and
\eqref{eq:ConvThOptimalityandSuboptimalityofICinMACChannel31} we
have proved that
\con{\begin{equation}
\max_{\lf(D_{1},D_{2}\ri)}\min\lf(d^{\ast,D_{1}}_{s+1,3\cdot
s}\lf(r_{0}+\epsilon\ri),d^{\ast,D_{1}+D_{2}}_{2\lf(s+1\ri),3\cdot
s}\lf(2r_{0}\ri),d^{\ast,D_{2}}_{s+1,3\cdot
s}\lf(r_{0}-\epsilon\ri)\ri)<
d^{\ast,\lf(FC\ri)}_{2\lf(s+1\ri),3\cdot s}\lf(2r_{0}\ri).
\end{equation}}{\bal{}{&\max_{\lf(D_{1},D_{2}\ri)}\min\lf(d^{\ast,D_{1}}_{s+1,3\cdot
s}\lf(r_{0}+\epsilon\ri),\ri.&
\nn\\
&\lf. d^{\ast,D_{1}+D_{2}}_{2\lf(s+1\ri),3\cdot
s}\lf(2r_{0}\ri),d^{\ast,D_{2}}_{s+1,3\cdot
s}\lf(r_{0}-\epsilon\ri)\ri)<&
\nn\\
&d^{\ast,\lf(FC\ri)}_{2\lf(s+1\ri),3\cdot s}\lf(2r_{0}\ri).&}}
This concludes the proof.

\section{Proof of Theorem \ref{Th:DirectUpperBoundErrorProb}}\label{Append:DirectUpperBoundErrorProb}
We base our proof on the techniques developed by Poltyrev
\cite{PoltirevJournal} for the AWGN channel and extended in
\cite{YonaFederICOptimalDMT} to colored channels in the
point-to-point case. We begin by partitioning the error event into
several disjoint events of errors for subsets of the users. We
relate each of these error events to the point-to-point channel of
the relevant users pulled together. Then we use the bounds derived
in \cite{YonaFederICOptimalDMT} to upper bound each of the error
events probabilities.

When the ML decoder makes an error it means that the decoded word is
different from the transmitted signal for at least one of the users.
Hence, we can break the error probability into the following sum of
disjoint events
\begin{equation}\label{eq:DirectUpperBoundErrorProb1}
\ol{Pe}(H_{\eff}^{(l),K},\rho)=\sum_{ s\subseteq
\lf\{1,\dots,K\ri\}}\ol{Pe}(H_{\eff}^{(l),\lf(s\ri)},\rho)
\end{equation}
where $\ol{Pe}(H_{\eff}^{(l),\lf(s\ri)},\rho)$ is the probability of
error to words that induce error on the users in $s$. Note that the
event of error to users in $s$ depends only on
$H_{\eff}^{(l),\lf(s\ri)}$ and not on
$H_{\eff}^{(l),\lf(1,\dots,K\ri)}$. We wish to upper bound
$\ol{Pe}(H_{\eff}^{(l),\lf(s\ri)},\rho)$ for any $s\subseteq
\lf\{1,\dots,K\ri\}$.

Based on \cite{PoltirevJournal} we get the following upper bound on
the error probability of the joint ML decoder when transmitting
$\udl{x}^{'}\in S_{K\cdot D_{l}\cdot T_{l}}$
\ifthenelse{\equal{\singlecolumntype}{1}}
{\begin{equation}\label{eq:DirectUpperBoundErrorProb2}
Pe(\underline{x}^{'})\le
Pr(\lv\underline{\tilde{n}}_{\mathrm{ex}}\rv\ge R)+
\sum_{\underline{l}\in Ball(\underline{x}^{'},2R)\bigcap S_{K\cdot
D_{l}\cdot T_{l}}, \udl{l}\ne \udl{x}^{'}} Pr(\lv
\udl{l}-\udl{x}^{'}-\udl{\tilde{n}}_{ex}\rv <
\lv\udl{\tilde{n}}_{ex}\rv)
\end{equation}}
{\begin{align}\label{eq:DirectUpperBoundErrorProb2}
&Pe(\underline{x}^{'})\le
Pr(\lv\underline{\tilde{n}}_{\mathrm{ex}}\rv\ge R)+
\nonumber\\
&\sum_{\underline{l}\in Ball(\underline{x}^{'},2R)\bigcap
S_{K_{l}T_{l}}, \udl{l}\ne \udl{x}^{'}} Pr(\lv
\udl{l}-\udl{x}^{'}-\udl{\tilde{n}}_{ex}\rv <
\lv\udl{\tilde{n}}_{ex}\rv)
\end{align}}
where $S_{K\cdot D_{l}\cdot T_{l}}$ is the $K\cdot D_{l}\cdot
T_{l}$-complex dimensional effective IC of the $K$ users,
$Ball(\udl{x}^{'},2R)$ is a $K\cdot D_{l}\cdot T_{l}$-complex
dimensional ball of radius $2R$ centered around $\udl{x}^{'}$, and
$\udl{\tilde{n}}_{\ex}$ is the effective noise in the $K\cdot
D_{l}\cdot T_{l}$-complex dimensional hyperplane in which the
effective IC resides. Instead of calculating
\eqref{eq:DirectUpperBoundErrorProb2}, we focus on upper bounding
the probability of decoding words that lead to an error only for the
users in $s\subseteq \lf\{1,\dots,K\ri\}$
\eqref{eq:DirectUpperBoundErrorProb1}. This will lead to an upper
bound on the error probability. Hence, we begin by considering the
error probability of $x^{'}$ to words that are different from
$\udl{x}^{'}$ only in the entries of the users in $s$. Based on our
ensemble, this is the error event of users in $s$ almost surely
(with probability 1). This error event is equivalent to the error
event of a word $\udl{x}^{''}$, which is a vector of length
$|s|\cdot D_{l}\cdot T_{l}$ that resides within an $|s|\cdot
D_{l}\cdot T_{l}$-complex dimensional IC $S_{|s|\cdot D_{l}\cdot
T_{l}}$, when $\udl{x}^{''}$ equals to $\udl{x}^{'}$ in the entries
of the users in $s$, and the other words in $S_{|s|\cdot D_{l}\cdot
T_{l}}$ are equal, in the entries of the users in $s$, to words in
$S_{K\cdot D_{l}\cdot T_{l}}$, that lead to an error for the users
in $s$. Hence, we wish to upper bound the error probability of
$\udl{x}^{''}\in S_{|s|\cdot D_{l}\cdot T_{l}}$. Based on the
expressions in \eqref{eq:DirectUpperBoundErrorProb2} we get that
this upper bound can be written as
\con{\begin{equation}\label{eq:DirectUpperBoundErrorProb3}
Pr(\lv\underline{\tilde{n}}^{'}_{\mathrm{ex}}\rv\ge R^{'})+
\sum_{\underline{l}\in Ball(\underline{x}^{''},2R^{'})\bigcap
S_{|s|\cdot D_{l}\cdot T_{l}}, \udl{l}\ne \udl{x}^{''}} Pr(\lv
\udl{l}-\udl{x}^{''}-\udl{\tilde{n}}^{'}_{ex}\rv <
\lv\udl{\tilde{n}}^{'}_{ex}\rv)
\end{equation}}{\bal{eq:DirectUpperBoundErrorProb3}
{&Pr(\lv\underline{\tilde{n}}^{'}_{\mathrm{ex}}\rv\ge R^{'})+&
\nn\\
&\sum_{\underline{l}\in Ball(\underline{x}^{''},2R^{'})\bigcap
S_{|s|\cdot D_{l}\cdot T_{l}}, \udl{l}\ne \udl{x}^{''}} Pr(\lv
\udl{l}-\udl{x}^{''}-\udl{\tilde{n}}^{'}_{ex}\rv <&
\nn\\
&\lv\udl{\tilde{n}}^{'}_{ex}\rv)&}}
where $Ball(\udl{x}^{''},2R^{'})$ is a $|s|\cdot D_{l}\cdot
T_{l}$-complex dimensional ball of radius $2R^{'}$ centered around
$\udl{x}^{''}$, and $\udl{\tilde{n}}^{'}_{\ex}$ is the effective
noise in the $|s|\cdot D_{l}\cdot T_{l}$-complex dimensional
hyperplane where $S_{|s|\cdot D_{l}\cdot T_{l}}$ resides.

Next we upper bound the average decoding error probability of an
ensemble of finite constellations, which later we will extend to
ensemble of IC's. Note that the upper bounds on the error
probability of IC's in \eqref{eq:DirectUpperBoundErrorProb1},
\eqref{eq:DirectUpperBoundErrorProb2} also apply to finite
constellations. Assume user $j$ code-book contains $\lfloor
\gamma_{\mathrm{tr}}^{\lf(j\ri)}b^{2D_{l}\cdot T_{l}}\rfloor$ words,
where each word is drawn independently and uniformly within
$cube_{D_{l}\cdot T_{l}}(b)$, $j=1,\dots,K$. Recall from
\ref{sec:BasicDefinitions} that
$\gamma_{\mathrm{tr}}^{\lf(j\ri)}=\rho^{T r_{j}}$. The $K$ users
constitute together an ensemble of $\prod_{j=1}^{K}\lfloor
\gamma_{\mathrm{tr}}^{\lf(j\ri)}b^{2D_{l}\cdot T_{l}}\rfloor$ words,
where a word in the ensemble is sampled from a uniform distribution
in $\cube_{K\cdot D_{l}\cdot T_{l}}\lf(b\ri)$ (not all words are
drawn independently). In fact any subset of the users $s\subseteq
\lf\{1,\dots,K\ri\}$ corresponds to an ensemble of $\prod_{i\in
s}\lfloor \gamma_{\mathrm{tr}}^{\lf(i\ri)}b^{2D_{l}\cdot
T_{l}}\rfloor$ words, where a word in the ensemble is sampled from a
uniform distribution, this time in $\cube_{|s|\cdot D_{l}\cdot
T_{l}}\lf(b\ri)$. Hence, the number of codewords that are different
in the entries of the users in $s$ is upper bounded by $\prod_{i\in
s}\lfloor \gamma_{\mathrm{tr}}^{\lf(i\ri)}b^{2D_{l}\cdot
T_{l}}\rfloor$. These words are in fact drawn independently in the
entries of the users in $s$. Based on these arguments and since the
ML decoder decides on the word with minimal Euclidean distance from
the observation, we get for each word in the ensemble that the
probability of error for users in $s\subseteq \lf\{1,\dots,K\ri\}$
is upper bounded by the average decoding error probability of an
ensemble consisting of $\prod_{i\in s}\lfloor
\gamma_{\mathrm{tr}}^{\lf(i\ri)}b^{2D_{l}\cdot T_{l}}\rfloor$ words
drawn independently and uniformly within $\cube_{|s|\cdot D_{l}\cdot
T_{l}}\lf(b\ri)$, with effective channel
$H_{eff}^{\lf(l\ri),\lf(s\ri)}$. In \cite[Theorem
3]{YonaFederICOptimalDMT} an upper bound on the average decoding
error probability of this ensemble was derived. By choosing for any
$s\subseteq \lf\{1,\dots,K\ri\}$
\begin{equation*}
R^{2}_{\lf(s\ri)}=R_{\eff}^{2}=\frac{2|s|\cdot D_{l}\cdot
T_{l}}{2\pi e}\rho^{-\frac{\sum_{i\in s}r_{i}}{|s|\cdot
D_{l}}-\sum_{i=1}^{|s|\cdot D_{l}\cdot
T_{l}}\frac{\eta_{i}^{\lf(s\ri)}}{|s|\cdot D_{l}\cdot T_{l}}}.
\end{equation*}
we get for the ensemble the following upper bound on the probability
of error for users in $s$
\con{\begin{equation}\label{eq:DirectUpperBoundErrorProb4}
\ol{P_{e}^{FC}}^{\lf(s\ri)}(\rho,\udl{\eta}^{\lf(s\ri)})\le
D^{'}(|s|\cdot D_{l}\cdot T_{l})\rho^{-T_{l}(|s|\cdot
D_{l}-\sum_{i\in s}r_{i})+\sum_{i=1}^{|s|\cdot D_{l}\cdot
T_{l}}\eta_{i}^{\lf(s\ri)}} \quad\forall s\subseteq
\lf\{1,\dots,K\ri\}
\end{equation}}{\bal{eq:DirectUpperBoundErrorProb4}
{&\ol{P_{e}^{FC}}^{\lf(s\ri)}(\rho,\udl{\eta}^{\lf(s\ri)})\le D^{'}(|s|\cdot D_{l}\cdot T_{l})\times&
\nn\\
&\rho^{-T_{l}(|s|\cdot
D_{l}-\sum_{i\in s}r_{i})+\sum_{i=1}^{|s|\cdot D_{l}\cdot
T_{l}}\eta_{i}^{\lf(s\ri)}}&
\nn\\
&\forall s\subseteq\lf\{1,\dots,K\ri\}&}}
where $D^{'}(|s|\cdot D_{l}\cdot T_{l})\ge 1$ and
$\eta_{i}^{\lf(s\ri)}\ge 0$, $i=1,\dots, |s|\cdot D_{l}\cdot T_{l}$.

So far we have upper bounded the probability of error of users in
$s$, in an ensemble of \emph{finite} constellations, for any
$s\subseteq \lf\{1,\dots,K\ri\}$. We now extend this ensemble of
finite constellations into an ensemble of IC's with density
$\gamma_{tr}^{\lf(j\ri)}$ for user $j$, where $j=1,\dots,K$. We show
that extending the ensemble of finite constellations to ensemble of
IC's does not change the upper bound on the error probability. Let
us consider for user $j$ a certain finite constellation from the
ensemble $C_{0}^{j}(\rho,b)\subset cube_{D_{l}\cdot T_{l}}(b)$. In
accordance, for the ensemble of users relates to $s$ let us denote a
certain finite constellation from the effective ensemble by
$C_{0}^{\lf(s\ri)}(\rho,b)\subset cube_{|s| \cdot D_{l}\cdot
T_{l}}(b)$. We extend each finite constellation into IC by extending
each user finite constellation in the following manner
\begin{equation}\label{eq:DirectUpperBoundErrorProb5}
IC^{j}(\rho,D_{l}\cdot
T_{l})=C_{0}^{j}(\rho,b)+(b+b^{'})\cdot\mathbb{Z}^{2D_{l}\cdot
T_{l}}
\end{equation}
where without loss of generality \footnote{In case $cube_{D_{l}\cdot
T_{l}}(b)$ is a rotated cube within $\mathbb{C}^{M\cdot T_{l}}$,
then the replication is done according the corresponding $M\cdot
T_{l}\times D_{l}\cdot T_{l}$ matrix with orthonormal columns.} we
assumed that $cube_{D_{l}\cdot T_{l}}(b)\in\mathbb{C}^{D_{l}\cdot
T_{l}}$. Therefore for the users in $s\subseteq \lf\{1,\dots,K\ri\}$
we get an effective IC
\begin{equation}\label{eq:DirectUpperBoundErrorProb6}
IC^{\lf(s\ri)}(\rho,|s|\cdot D_{l}\cdot
T_{l})=C_{0}^{\lf(s\ri)}(\rho,b)+(b+b^{'})\cdot\mathbb{Z}^{2|s|\cdot
D_{l}\cdot T_{l}}.
\end{equation}
At the receiver we get
\con{\begin{equation}\label{eq:DirectUpperBoundErrorProb7}
IC^{\lf(s\ri)}(\rho,|s|\cdot D_{l}\cdot
T_{l},H_{\eff}^{(l),\lf(s\ri)})=H_{\eff}^{(l),\lf(s\ri)}\cdot
C_{0}(\rho,b)+(b+b^{'})H_{\eff}^{(l),\lf(s\ri)}\cdot\mathbb{Z}^{2|s|\cdot
D_{l}\cdot T_{l}}.
\end{equation}}{\bal{eq:DirectUpperBoundErrorProb7}
{&IC^{\lf(s\ri)}(\rho,|s|\cdot D_{l}\cdot
T_{l},H_{\eff}^{(l),\lf(s\ri)})=&\nn\\
&H_{\eff}^{(l),\lf(s\ri)}\cdot
C_{0}(\rho,b)+(b+b^{'})H_{\eff}^{(l),\lf(s\ri)}\cdot\mathbb{Z}^{2|s|\cdot
D_{l}\cdot T_{l}}.&}}
By extending each finite constellation in the ensemble into an IC
according to the method presented in
\eqref{eq:DirectUpperBoundErrorProb6},
\eqref{eq:DirectUpperBoundErrorProb7} we get a new ensemble of IC's.
We would like to set $b$ and $b^{'}$ to be large enough such that
the ensemble average decoding error probability has the same upper
bound as in \eqref{eq:DirectUpperBoundErrorProb4}, and the users
densities are equal to $\gamma_{tr}^{\lf(j\ri)}$ up to a
coefficient, where $j=1,\dots,K$. First we would like to set a value
for $b^{'}$. For a word within the set
$\{H_{\eff}^{(l),\lf(s\ri)}\cdot C_{0}^{\lf(s\ri)}(\rho,b)\}$,
increasing $b^{'}$ decreases the error probability inflicted by the
codewords outside the set $\{H_{\eff}^{(l),\lf(s\ri)}\cdot
C_{0}^{\lf(s\ri)}(\rho,b)\}$, for any $s\subseteq
\lf\{1,\dots,K\ri\}$. In \cite[Theorem 3]{YonaFederICOptimalDMT} we
have shown that for any $\eta_{i}^{\lf(s\ri)}\ge 0$, by choosing
$b^{'}=\sqrt{\frac{|s|\cdot D_{l}\cdot T_{l}}{\pi
e}}\rho^{\frac{T_{l}}{2}(|s|\cdot D_{l}-\sum_{i\in
s}r_{i})+\epsilon}$, where $\epsilon >0$, we get for $\rho\ge 1$
\con{\begin{align}\label{eq:DirectUpperBoundErrorProb8}
\ol{Pe}(H_{eff}^{\lf(l\ri),\lf(s\ri)},\rho)=E_{C_{0}}\big(P_{e}^{IC}(H_{\eff}^{(l),\lf(s\ri)}\cdot
C_{0})\big)\le D(|s|\cdot D_{l}\cdot T_{l})\rho^{-T_{l}(|s|\cdot
D_{l}-\sum_{i\in s}r_{i})+\sum_{i=1}^{|s|\cdot D_{l}\cdot
T_{l}}\eta_{i}^{\lf(s\ri)}}
\end{align}}{\bal{eq:DirectUpperBoundErrorProb8}
{&\ol{Pe}(H_{eff}^{\lf(l\ri),\lf(s\ri)},\rho)=E_{C_{0}}\big(P_{e}^{IC}(H_{\eff}^{(l),\lf(s\ri)}\cdot
C_{0})\big)\le&
\nn\\
&D(|s|\cdot D_{l}\cdot T_{l})\rho^{-T_{l}(|s|\cdot
D_{l}-\sum_{i\in s}r_{i})+\sum_{i=1}^{|s|\cdot D_{l}\cdot
T_{l}}\eta_{i}^{\lf(s\ri)}}&}}
where $E_{C_{0}}\big(P_{e}^{IC}(H_{\eff}^{(l),\lf(s\ri)}\cdot
C_{0})\big)$ is the average decoding error probability of the
ensemble of IC's defined in \eqref{eq:DirectUpperBoundErrorProb7},
and $D \lf(|s|\cdot D_{l}\cdot T_{l}\ri)\ge D^{'}\lf(|s|\cdot
D_{l}\cdot T_{l}\ri)$. Hence, choosing $b^{'}$ to be the maximal
value between $\sqrt{\frac{|s|\cdot D_{l}\cdot T_{l}}{\pi
e}}\rho^{\frac{T_{l}}{2}(|s|\cdot D_{l}-\sum_{i\in
s}r_{i})+\epsilon}$, where $s\subseteq \lf\{1,\dots,K\ri\}$ will
enable to satisfy \eqref{eq:DirectUpperBoundErrorProb8} for any $s$.
$s$.

Next, we set the value of $b$ to be large enough such that for each
user, each IC density from the ensemble in
\eqref{eq:DirectUpperBoundErrorProb7}, $\gamma_{rc}^{',\lf(j\ri)}$,
equals $\gamma_{rc}^{\lf(j\ri)}$ up to a factor of 2, where
$j=1,\dots,K$. By choosing $b=b^{'}\cdot\rho^{\epsilon}$ we get
\begin{equation*}
\gamma_{tr}^{',\lf(j\ri)}=\gamma_{tr}^{\lf(j\ri)}\cdot
(\frac{b}{b+b^{'}})^{2D_{l}\cdot T_{l}}=\gamma_{tr}^{\lf(j\ri)}\cdot
\frac{1}{1+\rho^{-\epsilon}}.
\end{equation*}
Hence, for $\rho\ge 1$ we get
\begin{equation}\label{eq:DirectUpperBoundErrorProb9}
\frac{1}{2}\gamma_{tr}^{\lf(j\ri)}\le\gamma_{tr}^{',\lf(j\ri)}\le\gamma_{tr}^{\lf(j\ri)}.
\end{equation}
As a result we also get
\begin{equation*}
\mu_{tr}^{\lf(j\ri)}\le\mu_{tr}^{',\lf(j\ri)}=\frac{(\gamma_{tr}^{',\lf(j\ri)})^{-\frac{1}{D_{l}T_{j}}}}{2\pi
e\sigma^{2}}\le 2\mu_{tr}^{\lf(j\ri)}.
\end{equation*}
Hence, from \eqref{eq:DirectUpperBoundErrorProb1} and
\eqref{eq:DirectUpperBoundErrorProb8} we get that
\con{\begin{equation}\label{eq:DirectUpperBoundErrorProb10}
\ol{Pe}(H_{\eff}^{(l),K},\rho)\le\sum_{s\subseteq
\lf\{1,\dots,K\ri\}}D(|s|\cdot D_{l}\cdot
T_{l})\rho^{-T_{l}(|s|D_{l}-\sum_{i\in
s}r_{i})}\cdot|H_{\eff}^{(l),\lf(s\ri)\dagger}H_{\eff}^{(l),\lf(s\ri)}|^{-1}
\end{equation}}{\bal{eq:DirectUpperBoundErrorProb10}
{&\ol{Pe}(H_{\eff}^{(l),K},\rho)\le\sum_{s\subseteq
\lf\{1,\dots,K\ri\}}D(|s|\cdot D_{l}\cdot
T_{l})\times&
\nn\\
&\rho^{-T_{l}(|s|D_{l}-\sum_{i\in
s}r_{i})}\cdot |H_{\eff}^{(l),\lf(s\ri)\dagger}H_{\eff}^{(l),\lf(s\ri)}|^{-1}&}}
and from \eqref{eq:DirectUpperBoundErrorProb9} we get that user $j$
has multiplexing gain $r_{j}$ as required, where $j=1,\dots,K$. This
concludes the proof.

\section{Proof of Lemma \ref{lem:DirectLowerBoundDeterminantHeff}}\label{Append:DirectLowerBoundDeterminantHeff}
$H_{\eff}^{(l),|s|}$ is a block diagonal matrix. Hence the
determinant of $|H_{\eff}^{(l),|s|\dagger}\cdot H_{\eff}^{(l),|s|}|$
can be expressed as
\begin{equation}
|H_{\eff}^{(l),|s|\dagger}\cdot
H_{\eff}^{(l),|s|}|=\prod_{i=1}^{T_{l}}|\widehat{H}_{i}^{\dagger}\cdot\widehat{H}_{i}|.
\end{equation}
Assume
$\widehat{H}_{i}=(\udl{\widehat{h}}_{1},\dots,\udl{\widehat{h}}_{m})$,
i.e., $\widehat{H}_{i}$ has $m$ columns. In this case we can state
that the determinant
$$|\widehat{H}_{i}^{\dagger}\cdot\widehat{H}_{i}|=\lv\udl{\widehat{h}}_{1}\rv^{2}\lv\udl{\widehat{h}}_{2\perp
1}\rv^{2}\dots\lv\udl{\widehat{h}}_{m\perp m-1,\dots,1}\rv^{2}.$$
Note that $\widehat{H}_{i}$ has more rows than columns. The columns
of $\widehat{H}_{i}$ are subset of the columns of the channel matrix
$H$. Hence, in order to quantify the contribution of a certain
column of $H$, $\udl{h}_{j}$, $j=1,\dots, K\cdot M$, to the
determinant we need to consider the blocks where it occurs. We know
that the contribution of $\udl{h}_{j}$ to these determinants can be
quantified by taking into account the columns to its left in each
block, i.e., by taking into account
$\lf\{\udl{h}_{1},\dots,\udl{h}_{j-1}\ri\}$.

Based on \eqref{eq:DirectAchieveOptimalDMT5} and
\eqref{eq:DirectAchieveOptimalDMT6} we can quantify the contribution
of $\udl{h}_{j}$ to $|H_{\eff}^{(l),|s|\dagger}\cdot
H_{\eff}^{(l),|s|}|$ by \ifthenelse{\equal{\singlecolumntype}{1}}
{\begin{equation}
\lv\udl{h}_{j}\rv^{2b_{j}^{\lf(|s|\ri)}(0)}\prod_{k=1}^{j-1}\lv\udl{h}_{j\perp
j-1,\dots,j-k}\rv^{2b_{j}^{\lf(|s|\ri)}(k)}\dot{=}
\rho^{-\sum_{k=0}^{j-1}b_{j}^{\lf(|s|\ri)}(k)\cdot \min_{z\in
\lf(k+1,\dots,N\ri)}\xi_{z,j}}
\end{equation}}
{\begin{align}
\lv\udl{h}_{j}\rv^{2b_{j}(0)}\prod_{k=1}^{\min(j,L)-1}\lv\udl{h}_{j\perp j-1,\dots,j-k}\rv^{2b_{j}(k)}\dot{=}\nonumber\\
\rho^{-\sum_{k=0}^{\min(j,L)-1}b_{j}(k)a(k,\udl{\xi}_{j})}
\end{align}}
where $b_{j}^{\lf(|s|\ri)}(k)$ is the number of occurrences of
$\udl{h}_{j}$ in the blocks of $H_{\eff}^{(l),|s|}$, with only
$\{\udl{h}_{j-1},\dots,\udl{h}_{j-k}\}$ to its left.
$b_{j}^{\lf(|s|\ri)}(0)$ is the number of occurrences of
$\udl{h}_{j}$ with no columns to its left. Hence, the determinant is
obtained by multiplying the contribution of each column in
$H_{\eff}^{(l),|s|}$
\con{\begin{equation}\label{eq:DirectAchieveOptimalDMT7}
|H_{\eff}^{(l),|s|\dagger}\cdot
H_{\eff}^{(l),|s|}|=\prod_{j=1}^{|s|\cdot
M}\lv\udl{h}_{j}\rv^{2b_{j}^{\lf(|s|\ri)}(0)}\prod_{k=1}^{j-1}\lv\udl{h}_{j\perp
j-1,\dots,j-k}\rv^{2b_{j}^{\lf(|s|\ri)}(k)}\dot{=}
\rho^{-\sum_{k=0}^{j-1}b_{j}^{\lf(|s|\ri)}(k)\cdot \min_{z\in
\lf(k+1,\dots,N\ri)}\xi_{z,j}}.
\end{equation}}{\bal{eq:DirectAchieveOptimalDMT7}
{&|H_{\eff}^{(l),|s|\dagger}\cdot
H_{\eff}^{(l),|s|}|=\prod_{j=1}^{|s|\cdot
M}\lv\udl{h}_{j}\rv^{2b_{j}^{\lf(|s|\ri)}(0)}\times&
\nn\\
&\prod_{k=1}^{j-1}\lv\udl{h}_{j\perp
j-1,\dots,j-k}\rv^{2b_{j}^{\lf(|s|\ri)}(k)}\dot{=}&
\nn\\
&\rho^{-\sum_{k=0}^{j-1}b_{j}^{\lf(|s|\ri)}(k)\cdot \min_{z\in
\lf(k+1,\dots,N\ri)}\xi_{z,j}}.&}}

We now lower bound the determinant
\eqref{eq:DirectAchieveOptimalDMT7} by lower bounding the
contribution of each column. Let us consider column $\udl{h}_{a\cdot
M+b}$, $a=0,\dots,|s|-1$, $b=1,\dots,M$. From Lemma
\ref{lem:DirectSubsecTheEffectiveChannel1} we know that
$\udl{h}_{a\cdot M+b}$ occurs $N-M+1$ times with
$\lf\{\udl{h}_{1},\dots,\udl{h}_{a\cdot M+b-1}\ri\}$ to its left,
i.e., $b_{a\cdot M+b}^{\lf(|s|\ri)}\lf(a\cdot M+b-1\ri)=N-M+1$. In
addition, $\udl{h}_{a\cdot M+b}$ occurs in $\widehat{H}_{N-M+2v+1}$,
$v=1,\dots, \min \lf(M-l-1,b-1\ri)$, with
\begin{equation}\label{eq:DirectAchieveOptimalDMT8}
\lf\{\udl{h}_{1},\dots,\udl{h}_{a\cdot
M+b-1}\ri\}\setminus\lf\{\bigcup_{z=0}^{a}\udl{h}_{z\cdot
M+1},\dots,\udl{h}_{z\cdot M+v}\ri\}
\end{equation}
to its left, i.e., when $v$ is increased by one the number of columns
to its left reduces by $a+1$. Finally, $\udl{h}_{a\cdot M+b}$ occurs
in $\widehat{H}_{N-M+2v}$, $v=1,\dots,\min \lf(M-l-1,M-b\ri)$, with
\begin{equation}\label{eq:DirectAchieveOptimalDMT9}
\lf\{\udl{h}_{1},\dots,\udl{h}_{a\cdot
M+b-1}\ri\}\setminus\lf\{\bigcup_{z=1}^{a}\udl{h}_{z\cdot
M-v+1},\dots,\udl{h}_{z\cdot M}\ri\}.
\end{equation}
to its left (for $a=0$ it occurs with
$\lf\{\udl{h}_{1},\dots,\udl{h}_{b-1}\ri\}$ to its left), i.e., when
$v$ is increased by one the number of columns to its left reduces by
$a$. We wish to quantify the change in the determinant when reducing
columns, and relate it to the PDF in
\eqref{eq:DirectAchieveOptimalDMT3}. In order to analyze the
performance we would like the set of columns in
\eqref{eq:DirectAchieveOptimalDMT8} to be a subset of the set of
columns in \eqref{eq:DirectAchieveOptimalDMT9}, which is not the
case. Hence, we assume a columns reduction that gives a lower bound
on the determinant induced by the reduction in
\eqref{eq:DirectAchieveOptimalDMT8} and
\eqref{eq:DirectAchieveOptimalDMT9}. We assume for
$\widehat{H}_{N-M+2v}$, $v=1,\dots,\min \lf(M-l-1,M-b\ri)$ that
$\udl{h}_{aM+b}$ occurs with $\lf\{\udl{h}_{1},\dots,
\udl{h}_{aM+b-1} \ri\}$ to its left instead of
\eqref{eq:DirectAchieveOptimalDMT9}. In this case, by adding columns
to \eqref{eq:DirectAchieveOptimalDMT9} we get a lower bound on the
contribution of $\udl{h}_{a\cdot M+b}$ to the determinant in each of
its occurrences, that equals to
\begin{equation}\label{eq:DirectAchieveOptimalDMT10}
\rho^{-\min_{z\in \lf\{aM+b,\dots,N\ri\}}\xi_{z,aM+b}}.
\end{equation}
for any $v=1\dots,\min \lf(M-l-1,M-b\ri)$. On the other hand for
\eqref{eq:DirectAchieveOptimalDMT8} we assume that only the left
most column is reduced when increasing $v$, instead of the $a+1$
columns. This leads to lower bound to the contribution of
\eqref{eq:DirectAchieveOptimalDMT8} to the determinant that equals
to
\begin{equation}\label{eq:DirectAchieveOptimalDMT11}
\rho^{-\min_{z\in \lf\{aM+b-v,\dots,N\ri\}}\xi_{z,aM+b}}
\end{equation}
where $v=1\dots,\min \lf(M-l-1,b-1\ri)$. Hence, we get that the set of columns corresponding to
\eqref{eq:DirectAchieveOptimalDMT11} is a subset of the set of
columns corresponding to \eqref{eq:DirectAchieveOptimalDMT10}. Thus,
from
\eqref{eq:DirectAchieveOptimalDMT10},\eqref{eq:DirectAchieveOptimalDMT11}
we get the following lower bound on the determinant
\con{\begin{align}\label{eq:DirectAchieveOptimalDMT12}
|H_{\eff}^{(l),|s|\dagger}\cdot
H_{\eff}^{(l),|s|}|\dot{\ge}\prod_{a=0}^{|s|-1}\prod_{b=1}^{M}\rho^{-\lf(N-M+1+\min
\lf(M-l-1,M-b\ri)\ri)\min_{z\in \lf\{aM+b,\dots,N\ri\}}\xi_{z,aM+b}}
\nonumber\\
\cdot\prod_{b^{'}=2}^{M}\rho^{-\sum_{i=1}^{\min
\lf(M-l-1,b^{'}-1\ri)}\min_
{z\in\lf\{aM+b^{'}-i,\dots,N\ri\}}\xi_{z,aM+b^{'}}}.
\end{align}}{\bal{eq:DirectAchieveOptimalDMT12}
{&|H_{\eff}^{(l),|s|\dagger}\cdot
H_{\eff}^{(l),|s|}|\dot{\ge}&
\nn\\
&\prod_{a=0}^{|s|-1}\prod_{b=1}^{M}\times&
\nn\\
&\rho^{-\lf(N-M+1+\min
\lf(M-l-1,M-b\ri)\ri)\min_{z\in \lf\{aM+b,\dots,N\ri\}}\xi_{z,aM+b}}\times&
\nn\\
&\prod_{b^{'}=2}^{M}\times&
\nn\\
&\rho^{-\sum_{i=1}^{\min
\lf(M-l-1,b^{'}-1\ri)}\min_
{z\in\lf\{aM+b^{'}-i,\dots,N\ri\}}\xi_{z,aM+b^{'}}}.&}}

\section{Proof of theorem \ref{Th:DirectLowerBoundDiversityOrder}}\label{Append:DirectLowerBoundDiversityOrder}
 In order to lower bound
the DMT of the transmission scheme we use the upper bound on the
average decoding error probability from Theorem
\ref{Th:DirectUpperBoundErrorProb} and the lower bound on the
determinant of $|H_{\eff}^{(l),|s|\dagger}H_{\eff}^{(l),|s|}|$
\eqref{eq:DirectAchieveOptimalDMT12}, to get a new upper bound on
the error probability. We average the new upper bound on the
realizations of $H$ to obtain the transmission scheme DMT.

First let us denote $l=\lfloor r_{max}\rfloor$. Recall from Theorem
\ref{Th:DirectUpperBoundErrorProb} that the upper bound on the error
probability applies to $\eta_{i}^{\lf(s\ri)}\ge 0$, for every
$i=0,\dots, |s|\cdot D_{l}\cdot T_{l}$ and for any $s\subseteq
\lf(1,\dots,K\ri)$. In our analysis we assume that $\xi_{i,j}\ge 0$
for $i=1,\dots,N$, $j=1,\dots, K\cdot M$. We wish to show that it
leads to $\eta_{i}^{\lf(s\ri)}\ge 0$, i.e., we can use the upper
bound on the error probability. We know that
$H_{eff}^{\lf(l\ri),\lf(s\ri)}$ is a block diagonal matrix, where
the set of columns of each block is a subset of
$\lf\{\udl{h}_{1},\dots,\udl{h}_{K\cdot M}\ri\}$. Let us denote the
set of indices of the columns of $H$ that take place in
$H_{eff}^{\lf(l\ri),\lf(s\ri)}$ by $a \lf(s\ri)$. In this case we
get from trace considerations
\begin{equation}\label{eq:DirectLowerBoundDiversityOrder1_A}
\sum_{i=1}^{N}\sum_{j\in a
\lf(s\ri)}\rho^{-\xi_{i,j}}\le\sum_{i=1}^{|s|\cdot D_{l}\cdot
T_{l}}\rho^{-\eta_{i}^{\lf(s\ri)}}\quad \forall s\subseteq
\lf\{1,\dots,K\ri\}.
\end{equation}
The inequality results from the fact that $a \lf(s\ri)$ represents
the indices of columns that take place in
$H_{eff}^{\lf(l\ri),\lf(s\ri)}$, whereas some of the columns may
appear more than once in $H_{eff}^{\lf(l\ri),\lf(s\ri)}$. However,
the number of appearances of each column is bounded, and so the
inequality in \eqref{eq:DirectLowerBoundDiversityOrder1_A} is up to
a constant. Therefore, we get the following exponential equality
(for large $\rho$)
\begin{equation}\label{eq:DirectLowerBoundDiversityOrder1}
\sum_{i=1}^{N}\sum_{j\in a
\lf(s\ri)}\rho^{-\xi_{i,j}}\dot{=}\sum_{i=1}^{|s|\cdot D_{l}\cdot
T_{l}}\rho^{-\eta_{i}^{\lf(s\ri)}}\quad \forall s\subseteq
\lf\{1,\dots,K\ri\}.
\end{equation}
From \eqref{eq:DirectLowerBoundDiversityOrder1} we get that
$\xi_{i,j}\ge 0$ for $i=1,\dots,N$, $j=1,\dots, K\cdot M$ if and
only if $\eta_{i}^{\lf(s\ri)}\ge 0$ for any $s\subseteq
\lf\{1,\dots,K\ri\}$ and $i=1,\dots,|s|\cdot D_{l}\cdot T_{l}$. It
follows that we can use the upper bound in Theorem
\ref{Th:DirectUpperBoundErrorProb}.

The upper bound on the error probability consists of the sum of
$\ol{Pe}(\udl{\eta}^{\lf(s\ri)},\rho)$ for all $s\subseteq
\lf\{1,\dots,K\ri\}$. We wish to show that the DMT of each of the
terms is lower bounded by
$d^{\ast,\lf(FC\ri)}_{M,N}\lf(r_{max}\ri)$. First note that $\forall
s\subseteq \lf\{1,\dots,K\ri\}$ we can write
\con{\begin{align}
\ol{Pe}(\udl{\eta}^{\lf(s\ri)},\rho)=\min&\lf(1,D \lf(|s|\cdot
D_{l}\cdot T_{l}\ri)\rho^{-T_{l}(|s|D_{l}-\sum_{i\in
s}r_{i})}\cdot|H_{\eff}^{(l),\lf(s\ri)\dagger}H_{\eff}^{(l),\lf(s\ri)}|^{-1}\ri)\nonumber\\
&\le \min\lf(1,D \lf(|s|\cdot D_{l}\cdot T_{l}\ri)\rho^{-|s|\cdot
T_{l}(D_{l}-r_{max})}\cdot|H_{\eff}^{(l),\lf(s\ri)\dagger}H_{\eff}^{(l),\lf(s\ri)}|^{-1}\ri)
\end{align}}{\bal{}{&\ol{Pe}(\udl{\eta}^{\lf(s\ri)},\rho)=\min\lf(1,\ri.&
\nn\\
&\lf. D \lf(|s|\cdot
D_{l}\cdot T_{l}\ri)\rho^{-T_{l}(|s|D_{l}-\sum_{i\in
s}r_{i})}\cdot|H_{\eff}^{(l),\lf(s\ri)\dagger}H_{\eff}^{(l),\lf(s\ri)}|^{-1}\ri)
\nn\\
&\le\min\lf(1,D \lf(|s|\cdot D_{l}\cdot T_{l}\ri)\rho^{-|s|\cdot
T_{l}(D_{l}-r_{max})}\times\ri.&
\nn\\
&\lf. |H_{\eff}^{(l),\lf(s\ri)\dagger}H_{\eff}^{(l),\lf(s\ri)}|^{-1}\ri) &}}
where the inequality comes from the fact that assuming all users
transmit at the maximal multiplexing gain increases the error
probability. By assigning $D_{l}=\frac{MN-l\cdot
\lf(l+1\ri)}{N+M-1-2\cdot l}$ and $T_{l}=N+M-1-2\cdot l$ we get
\con{\begin{equation}\label{eq:DirectLowerBoundDiversityOrder2}
\ol{Pe}(\udl{\eta}^{\lf(s\ri)},\rho)\le \min\lf(1,D \lf(|s|\cdot
D_{l}\cdot T_{l}\ri)\rho^{-|s|\cdot (MN-l\cdot
\lf(l+1\ri)-\lf(N+M-1-2l\ri)\cdot
r_{max})}\cdot|H_{\eff}^{(l),\lf(s\ri)\dagger}H_{\eff}^{(l),\lf(s\ri)}|^{-1}\ri).
\end{equation}}{\bal{eq:DirectLowerBoundDiversityOrder2}
{&\ol{Pe}(\udl{\eta}^{\lf(s\ri)},\rho)\le \min\lf(1,D \lf(|s|\cdot
D_{l}\cdot T_{l}\ri)\times\ri.&
\nn\\
&\lf.\rho^{-|s|\cdot (MN-l\cdot
\lf(l+1\ri)-\lf(N+M-1-2l\ri)\cdot
r_{max})}\times\ri.&
\nn\\
&\lf. |H_{\eff}^{(l),\lf(s\ri)\dagger}H_{\eff}^{(l),\lf(s\ri)}|^{-1}\ri).&}}
From \eqref{eq:DirectAchieveOptimalDMT4} we know that
$E_{H}\lf(\ol{Pe}(\udl{\eta}^{\lf(s\ri)},\rho)\ri)=E_{H}\lf(\ol{Pe}(\udl{\eta}^{\lf(1,\dots,|s|\ri)},\rho)\ri)$,
i,e, the term corresponding to the first $|s|$ users. Hence, for all
terms with the same $|s|$ we can consider
\con{\begin{equation}\label{eq:DirectLowerBoundDiversityOrder3}
\ol{Pe}(\udl{\eta}^{\lf(1,\dots,|s|\ri)},\rho)\le \min\lf(1,D
\lf(|s|\cdot D_{l}\cdot T_{l}\ri)\rho^{-|s|\cdot (MN-l\cdot
\lf(l+1\ri)-\lf(N+M-1-2l\ri)\cdot
r_{max})}\cdot|H_{\eff}^{(l),|s|\dagger}H_{\eff}^{(l),|s|}|^{-1}\ri).
\end{equation}}{\bal{eq:DirectLowerBoundDiversityOrder3}
{&\ol{Pe}(\udl{\eta}^{\lf(1,\dots,|s|\ri)},\rho)\le \min\lf(1,D
\lf(|s|\cdot D_{l}\cdot T_{l}\ri)\times\ri.&
\nn\\
&\lf.\rho^{-|s|\cdot (MN-l\cdot
\lf(l+1\ri)-\lf(N+M-1-2l\ri)\cdot
r_{max})}\times\ri.
\nn\\
&\lf.|H_{\eff}^{(l),|s|\dagger}H_{\eff}^{(l),|s|}|^{-1}\ri).&}}
Based on \eqref{eq:DirectAchieveOptimalDMT12} let us define
\begin{equation}
A \lf(a\cdot M+b,l \ri)= \lf(N-b+1\ri)\min_{z\in
\lf\{aM+b,\dots,N\ri\}}\xi_{z,aM+b}
\end{equation}
for $b=1$, $a=0,\dots,|s|-1$, and
\con{\begin{equation}
A \lf(a\cdot M+b,l \ri)= \lf(N-b+1\ri)\min_{z\in
\lf\{aM+b,\dots,N\ri\}}\xi_{z,aM+b}+\sum_{i=1}^{\min
\lf(M-l-1,b-1\ri)}\min_ {z\in\lf\{aM+b-i,\dots,N\ri\}}\xi_{z,aM+b}
\end{equation}}{\bal{}{&A \lf(a\cdot M+b,l \ri)= \lf(N-b+1\ri)\min_{z\in
\lf\{aM+b,\dots,N\ri\}}\xi_{z,aM+b}&
\nn\\
&+\sum_{i=1}^{\min\lf(M-l-1,b-1\ri)}\min_ {z\in\lf\{aM+b-i,\dots,N\ri\}}\xi_{z,aM+b}&}}
for $b=2,\dots,M$ and $a=0,\dots, |s|-1$. From the bounds in
\eqref{eq:DirectAchieveOptimalDMT10},
\eqref{eq:DirectAchieveOptimalDMT11},
\eqref{eq:DirectAchieveOptimalDMT12} and also since $N-M+1+\min
\lf(M-l-1,M-b\ri)\le N-b+1$, we get that $\rho^{-A \lf(a\cdot
M+b,l\ri)}$ gives a lower bound on the contribution of
$\udl{h}_{a\cdot M+b}$ to the determinant. As a result we get the
following upper bound
\begin{equation}
|H_{\eff}^{(l),|s|\dagger}H_{\eff}^{(l),|s|}|^{-1}\dot{\le}\prod_{a=0}^{|s|-1}\prod_{b=1}^{M}\rho^{A
\lf(a\cdot M+b,l\ri)}.
\end{equation}
By assigning in the bound from
\eqref{eq:DirectLowerBoundDiversityOrder3} we get
\con{\begin{equation}\label{eq:DirectLowerBoundDiversityOrder4}
\ol{Pe}(\udl{\eta}^{\lf(1,\dots,|s|\ri)},\rho)\dot{\le}\rho^{-\lf(|s|\cdot
\lf(MN-l \lf(l+1\ri)-\lf(N+M-1-2l\ri)r_{max}\ri)-\sum_{i=1}^{|s|M}A
\lf(i,l\ri)\ri)^{+}}
\end{equation}}{\bal{eq:DirectLowerBoundDiversityOrder4}
{&\ol{Pe}(\udl{\eta}^{\lf(1,\dots,|s|\ri)},\rho)\dot{\le}&
\nn\\
&\rho^{-\lf(|s|\cdot
\lf(MN-l \lf(l+1\ri)-\lf(N+M-1-2l\ri)r_{max}\ri)-\sum_{i=1}^{|s|M}A
\lf(i,l\ri)\ri)^{+}}&}}
where $\lf(x\ri)^{+}$ equals $x$ for  $x\ge 0$ and $0$ else; we omit
the constant $\min\lf(1,D \lf(|s|\cdot D_{l}\cdot T_{l}\ri)\ri)$ as
we consider the equality for asymptotically large $\rho$ in
\eqref{eq:DirectLowerBoundDiversityOrder4}.

Based on \eqref{eq:DirectLowerBoundDiversityOrder4} the average over
the channel realizations can be upper bounded by
\con{\begin{align}
E_{H}\lf(\ol{Pe}(\udl{\eta}^{\lf(s\ri)},\rho)\ri)=E_{H}&\lf(\ol{Pe}(\udl{\eta}^{\lf(1,\dots,|s|\ri)},\rho)\ri)\nonumber\\
&\dot{\le} \int_{\xi_{\udl{i},\udl{j}}\ge 0}\rho^{-\lf(|s|\cdot
\lf(MN-l \lf(l+1\ri)-\lf(N+M-1-2l\ri)r_{max}\ri)-\sum_{i=1}^{|s|M}A
\lf(i,l\ri)\ri)^{+}-\sum_{i=1}^{N}\sum_{j=1}^{K\cdot
M}\xi_{i,j}}d\xi_{\udl{i},\udl{j}}.
\end{align}}{\bal{}{&E_{H}\lf(\ol{Pe}(\udl{\eta}^{\lf(s\ri)},\rho)\ri)=E_{H}\lf(\ol{Pe}(\udl{\eta}^{\lf(1,\dots,|s|\ri)},\rho)\ri)\dot{\le}&
\nn\\
& \int_{\xi_{\udl{i},\udl{j}}\ge 0}\rho^{-\lf(|s|\cdot
\lf(MN-l \lf(l+1\ri)-\lf(N+M-1-2l\ri)r_{max}\ri)-\sum_{i=1}^{|s|M}A
\lf(i,l\ri)\ri)^{+}}&
\nn\\
&\times\rho^{-\sum_{i=1}^{N}\sum_{j=1}^{K\cdot
M}\xi_{i,j}}d\xi_{\udl{i},\udl{j}}.&}}
where $\xi_{\udl{i},\udl{j}}\ge 0$ means $\xi_{i,j}\ge 0$ for
$i=1,\dots,N$ and $j=1,\dots,K\cdot M$. We divide the integration
range to two sets
\con{\begin{equation}\label{eq:DirectLowerBoundDiversityOrder5}
\int_{\xi_{\udl{i},\udl{j}}\in \mathcal{A}}\rho^{-\lf(|s|\cdot
\lf(MN-l \lf(l+1\ri)-\lf(N+M-1-2l\ri)r_{max}\ri)-\sum_{i=1}^{|s|M}A
\lf(i,l\ri)\ri)^{+}-\sum_{i=1}^{N}\sum_{j=1}^{K\cdot
M}\xi_{i,j}}d\xi_{\udl{i},\udl{j}}+ \int_{\xi_{\udl{i},\udl{j}}\in
\ol{\mathcal{A}}}1\cdot \rho^{-\sum_{i=1}^{N}\sum_{j=1}^{K\cdot
M}\xi_{i,j}}d\xi_{\udl{i},\udl{j}}
\end{equation}}{\bal{eq:DirectLowerBoundDiversityOrder5}
{&\int_{\xi_{\udl{i},\udl{j}}\in \mathcal{A}}\rho^{-\lf(|s|\cdot
\lf(MN-l \lf(l+1\ri)-\lf(N+M-1-2l\ri)r_{max}\ri)-\sum_{i=1}^{|s|M}A
\lf(i,l\ri)\ri)^{+}}&
\nn\\\
&\times \rho^{-\sum_{i=1}^{N}\sum_{j=1}^{K\cdot
M}\xi_{i,j}}d\xi_{\udl{i},\udl{j}}&
\nn\\
&+ \int_{\xi_{\udl{i},\udl{j}}\in
\ol{\mathcal{A}}}1\cdot \rho^{-\sum_{i=1}^{N}\sum_{j=1}^{K\cdot
M}\xi_{i,j}}d\xi_{\udl{i},\udl{j}}&}}
where $\mathcal{A}=\lf\{\bigcap_{i=1}^{N}\bigcap_{j=1}^{K\cdot
M}0\le \xi_{i,j}\le K\cdot M\cdot N\ri\}$,
$\ol{\mathcal{A}}=\lf\{\bigcup_{i=1}^{N}\bigcup_{j=1}^{K\cdot
M}\xi_{i,j}> K\cdot M\cdot N\ri\}$, and for the second term in
\eqref{eq:DirectLowerBoundDiversityOrder5} we upper bounded the
error probability per channel realization by 1.

We begin by lower bounding the DMT of the first term in
\eqref{eq:DirectLowerBoundDiversityOrder5}. In a similar manner to
\cite{TseDivMult2003}, \cite{YonaFederICOptimalDMT}, for very large
$\rho$ and finite integration range, we can approximate the integral
by finding the most dominant exponential term. Hence, for large
$\rho$ the first term in \eqref{eq:DirectLowerBoundDiversityOrder5}
equals
\con{\begin{equation}
\rho^{-\min_{\xi_{\udl{i},\udl{j}}\in\mathcal{A}}\lf(\lf(|s|\cdot
\lf(MN-l \lf(l+1\ri)-\lf(N+M-1-2l\ri)r_{max}\ri)-\sum_{i=1}^{|s|M}A
\lf(i,l\ri)\ri)^{+}+\sum_{i=1}^{N}\sum_{j=1}^{K\cdot
M}\xi_{i,j}\ri)}.
\end{equation}}{\bal{}{&\max_{\xi_{\udl{i},\udl{j}}\in\mathcal{A}}&
\nn\\
&\rho^{-\lf(|s|\cdot
\lf(MN-l \lf(l+1\ri)-\lf(N+M-1-2l\ri)r_{max}\ri)-\sum_{i=1}^{|s|M}A
\lf(i,l\ri)\ri)^{+}}&
\nn\\
&\times\rho^{-\sum_{i=1}^{N}\sum_{j=1}^{K\cdot
M}\xi_{i,j}}.&}}
Hence, by showing that
\con{\begin{align}\label{eq:DirectLowerBoundDiversityOrder6}
\min_{\xi_{\udl{i},\udl{j}}\in\mathcal{A}}\lf(|s|\cdot \lf(MN-l
\lf(l+1\ri)-\lf(N+M-1-2l\ri)r_{max}\ri)-\sum_{i=1}^{|s|M}A
\lf(i,l\ri)\ri)^{+}+\sum_{i=1}^{N}\sum_{j=1}^{K\cdot
M}\xi_{i,j}\nonumber\\ \ge MN-l \lf(l+1\ri)-\lf(N+M-1-2l\ri)r_{max}
\end{align}}{\bal{eq:DirectLowerBoundDiversityOrder6}
{&\min_{\xi_{\udl{i},\udl{j}}\in\mathcal{A}}\lf(|s|\cdot \lf(MN-l
\lf(l+1\ri)\ri.\ri.&
\nn\\
&\lf.\lf. -\lf(N+M-1-2l\ri)r_{max}\ri)-\sum_{i=1}^{|s|M}A
\lf(i,l\ri)\ri)^{+}&
\nn\\
&+\sum_{i=1}^{N}\sum_{j=1}^{K\cdot
M}\xi_{i,j}&
\nn\\
&\ge MN-l \lf(l+1\ri)-\lf(N+M-1-2l\ri)r_{max}&}}
we get that the first term attains DMT which is lower bounded by
$d^{\ast,\lf(FC\ri)}_{M,N}\lf(r_{max}\ri)$. In order to show
\eqref{eq:DirectLowerBoundDiversityOrder6} we use the following
lemma.

\begin{lem}\label{lem:DirectLowerBoundDiversityOrderEquivalentOptimization}
The  solution for the minimization problem
\con{\begin{equation*}
\min_{\xi_{\udl{i},\udl{j}}\in\mathcal{A}}\lf(|s|\cdot \lf(MN-l
\lf(l+1\ri)-\lf(N+M-1-2l\ri)r_{max}\ri)-\sum_{i=1}^{|s|M}A
\lf(i,l\ri)\ri)^{+}+\sum_{i=1}^{N}\sum_{j=1}^{K\cdot M}\xi_{i,j}
\end{equation*}}{\baln{}{&\min_{\xi_{\udl{i},\udl{j}}\in\mathcal{A}}\lf(|s|\cdot \lf(MN-l
\lf(l+1\ri)\ri.\ri.
\nn\\
&\lf.\lf.-\lf(N+M-1-2l\ri)r_{max}\ri)-\sum_{i=1}^{|s|M}A
\lf(i,l\ri)\ri)^{+}&
\nn\\
&+\sum_{i=1}^{N}\sum_{j=1}^{K\cdot M}\xi_{i,j}&}}
equals to the solution of the following minimization problem
\begin{equation*}
\min_{\udl{\alpha}\in\mathcal{A}^{'}}\sum_{i=1}^{|s|\cdot
M}\lf(N-i+1\ri)\alpha_{i}
\end{equation*}
where $\udl{\alpha}=\lf(\alpha_{1},\dots,\alpha_{|s|\cdot
M}\ri)^{T}$, and the set $\mathcal{A}^{'}$ fulfils the following two
conditions: $0\le \alpha_{i}\le K\cdot M\cdot N$ for $i=1,\dots,
|s|\cdot M$ and also
\con{\begin{equation*}
\sum_{a=0}^{|s|-1}\sum_{b=1}^{M}\lf(N-b+1\ri)\alpha_{a\cdot M+b}=|s|
\lf(MN-l \lf(l+1\ri)-\lf(N+M-1-2l\ri)r_{max}\ri).
\end{equation*}}{\baln{}{&\sum_{a=0}^{|s|-1}\sum_{b=1}^{M}\lf(N-b+1\ri)\alpha_{a\cdot M+b}=&
\nn\\
&|s|\lf(MN-l \lf(l+1\ri)-\lf(N+M-1-2l\ri)r_{max}\ri).&}}
\end{lem}
\begin{proof}
The proof is in appendix
\ref{Append:DirectLowerBoundDiversityOrderEquivalentOptimization}.
\end{proof}

Based on Lemma
\ref{lem:DirectLowerBoundDiversityOrderEquivalentOptimization} we
can see that by proving
\con{\begin{equation}\label{eq:DirectLowerBoundDiversityOrder7}
\min_{\udl{\alpha}\in\mathcal{A}^{'}}\sum_{i=1}^{|s|\cdot
M}\lf(N-i+1\ri)\alpha_{i}\ge MN-l
\lf(l+1\ri)-\lf(N+M-1-2l\ri)r_{max}
\end{equation}}{\bal{eq:DirectLowerBoundDiversityOrder7}
{&\min_{\udl{\alpha}\in\mathcal{A}^{'}}\sum_{i=1}^{|s|\cdot
M}\lf(N-i+1\ri)\alpha_{i}\ge&
\nn\\
&MN-l\lf(l+1\ri)-\lf(N+M-1-2l\ri)r_{max}&}}
we also prove \eqref{eq:DirectLowerBoundDiversityOrder6}. Therefore,
we wish to show that any vector $\udl{\alpha}\in\mathcal{A}^{'}$
fulfils this inequality. Consider a certain vector
$\udl{\alpha}\in\mathcal{A}^{'}$. We define $\beta_{a\cdot
M+b}=\frac{\lf(N+1-b\ri)\cdot \alpha_{a\cdot M+b}}{|s|}$ for
$a=0,\dots, |s|-1$, $b=1,\dots,M$. From this definition we get
\con{\begin{equation}\label{eq:DirectLowerBoundDiversityOrder8}
\sum_{a=0}^{|s|-1}\sum_{b=1}^{M}\beta_{a\cdot
M+b}=\sum_{a=0}^{|s|-1}\sum_{b=1}^{M}\frac{\lf(N-b+1\ri)\alpha_{a\cdot
M+b}}{|s|}=MN-l \lf(l+1\ri)-\lf(N+M-1-2l\ri)r_{max}.
\end{equation}}{\bal{eq:DirectLowerBoundDiversityOrder8}
{&\sum_{a=0}^{|s|-1}\sum_{b=1}^{M}\beta_{a\cdot
M+b}=\sum_{a=0}^{|s|-1}\sum_{b=1}^{M}\frac{\lf(N-b+1\ri)\alpha_{a\cdot
M+b}}{|s|}=&
\nn\\
&MN-l \lf(l+1\ri)-\lf(N+M-1-2l\ri)r_{max}.&}}
By assigning in \eqref{eq:DirectLowerBoundDiversityOrder7} we get
\con{\begin{equation}\label{eq:DirectLowerBoundDiversityOrder9}
\sum_{a=0}^{|s|-1}\sum_{b=1}^{M}\lf(N-a\cdot M-b+1\ri)\alpha_{a\cdot
M+b}=\sum_{a=0}^{|s|-1}\sum_{b=1}^{M}\frac{|s|\lf(N-a\cdot
M-b+1\ri)\beta_{a\cdot M+b}}{N-b+1}.
\end{equation}}{\bal{eq:DirectLowerBoundDiversityOrder9}
{&\sum_{a=0}^{|s|-1}\sum_{b=1}^{M}\lf(N-a\cdot M-b+1\ri)\alpha_{a\cdot
M+b}=&
\nn\\
&\sum_{a=0}^{|s|-1}\sum_{b=1}^{M}\frac{|s|\lf(N-a\cdot
M-b+1\ri)\beta_{a\cdot M+b}}{N-b+1}.&}}
We use the following lemma to prove
\eqref{eq:DirectLowerBoundDiversityOrder7}.
\begin{lem}\label{lem:DirectLowerBoundDiversityOrderBetaLowerBound}
Consider $N\ge \lf(|s|+1\ri)M-1$, we get for any $a=0\dots, |s|-1$
and $b=1,\dots,M$
\begin{equation*}
\frac{|s|\lf(N-\lf(a\cdot M+b\ri)+1\ri)}{N-b+1}\ge 1.
\end{equation*}
\end{lem}
\begin{proof}
The proof is in appendix
\ref{Append:DirectLowerBoundDiversityOrderBetaLowerBound}.
\end{proof}
Since $K\ge |s|$ and $N\ge \lf(K+1\ri)M-1$ we can assign the
inequality of Lemma
\ref{lem:DirectLowerBoundDiversityOrderBetaLowerBound} in
\eqref{eq:DirectLowerBoundDiversityOrder9} to get
\con{\begin{equation}
\sum_{a=0}^{|s|-1}\sum_{b=1}^{M}\lf(N-a\cdot M-b+1\ri)\alpha_{a\cdot
M+b}\ge \sum_{a=0}^{|s|-1}\sum_{b=1}^{M}\beta_{a\cdot M+b}=MN-l
\lf(l+1\ri)-\lf(N+M-1-2l\ri)r_{max}
\end{equation}}{\bal{}{&\sum_{a=0}^{|s|-1}\sum_{b=1}^{M}\lf(N-a\cdot M-b+1\ri)\alpha_{a\cdot
M+b}\ge&
\nn\\
&\sum_{a=0}^{|s|-1}\sum_{b=1}^{M}\beta_{a\cdot M+b}=&
\nn\\
&MN-l\lf(l+1\ri)-\lf(N+M-1-2l\ri)r_{max}&}}
where the equality results from
\eqref{eq:DirectLowerBoundDiversityOrder8}. This proves
\eqref{eq:DirectLowerBoundDiversityOrder7} and so proves that the
DMT of the first term in \eqref{eq:DirectLowerBoundDiversityOrder5}
is lower bounded by $d^{\ast,\lf(FC\ri)}_{M,N}\lf(r_{max}\ri)$.

Now let us show that the second term in
\eqref{eq:DirectLowerBoundDiversityOrder5} is also lower bounded by
$d^{\ast,\lf(FC\ri)}_{M,N}\lf(r_{max}\ri)$.
\con{\begin{equation*}
\int_{\xi_{\udl{i},\udl{j}}\in \ol{\mathcal{A}}}1\cdot
\rho^{-\sum_{i=1}^{N}\sum_{j=1}^{K\cdot
M}\xi_{i,j}}d\xi_{\udl{i},\udl{j}}\le\int_{\xi_{1,1}>K\cdot M\cdot
N}\rho^{-\xi_{1,1}}\dot{=}\rho^{-K\cdot M\cdot N}.
\end{equation*}}{\baln{}{&\int_{\xi_{\udl{i},\udl{j}}\in \ol{\mathcal{A}}}1\cdot
\rho^{-\sum_{i=1}^{N}\sum_{j=1}^{K\cdot
M}\xi_{i,j}}d\xi_{\udl{i},\udl{j}}\le\int_{\xi_{1,1}>K\cdot M\cdot N}\rho^{-\xi_{1,1}}\dot{=}&
\nn\\
&\rho^{-K\cdot M\cdot N}.&}}
Since $d^{\ast,\lf(FC\ri)}_{M,N}\lf(r_{max}\ri)\le K\cdot M\cdot N$
the DMT of the second term in
\eqref{eq:DirectLowerBoundDiversityOrder5} is also lower bounded by
$d^{\ast,\lf(FC\ri)}_{M,N}\lf(r_{max}\ri)$.

We have shown that for $l=\lfloor r_{max}\rfloor$ the DMT of
$E_{H}\lf(\ol{Pe}(\udl{\eta}^{\lf(s\ri)},\rho)\ri)$ is lower bounded
by $d^{\ast,\lf(FC\ri)}_{M,N}\lf(r_{max}\ri)=MN-l
\lf(l-1\ri)-\lf(M+N-1-2l\ri)r_{max}$ for any $s\subseteq
\lf\{1,\dots,K\ri\}$. Since
\begin{equation*}
\ol{Pe}(H_{\eff}^{(l),K},\rho)\le \sum_{s\subseteq
\lf\{1,\dots,K\ri\}}\ol{Pe}(\udl{\eta}^{\lf(s\ri)},\rho)
\end{equation*}
we get that the DMT of the $K$ sequences of IC's is also lower
bounded by $d^{\ast,\lf(FC\ri)}_{M,N}\lf(r_{max}\ri)$. This
concludes the proof.

\section{Proof of Lemma \ref{lem:DirectLowerBoundDiversityOrderEquivalentOptimization}}\label{Append:DirectLowerBoundDiversityOrderEquivalentOptimization}
Recall that the optimization problem
\con{\begin{equation}\label{eq:DirectLemReducingMinProblem0}
\min_{\xi_{\udl{i},\udl{j}}\in\mathcal{A}}\lf(|s|\cdot \lf(MN-l
\lf(l+1\ri)-\lf(N+M-1-2l\ri)r_{max}\ri)-\sum_{i=1}^{|s|M}A
\lf(i,l\ri)\ri)^{+}+\sum_{i=1}^{N}\sum_{j=1}^{K\cdot M}\xi_{i,j}
\end{equation}}{\bal{eq:DirectLemReducingMinProblem0}
{&\min_{\xi_{\udl{i},\udl{j}}\in\mathcal{A}}\lf(|s|\cdot \lf(MN-l
\lf(l+1\ri)-\ri.\ri.&
\nn\\
&\lf.\lf.\lf(N+M-1-2l\ri)r_{max}\ri)-\sum_{i=1}^{|s|M}A
\lf(i,l\ri)\ri)^{+}&
\nn\\
&+\sum_{i=1}^{N}\sum_{j=1}^{K\cdot M}\xi_{i,j}&}}
where
\begin{equation}\label{eq:DirectLemReducingMinProblem4}
A \lf(a\cdot M+b,l \ri)= \lf(N-b+1\ri)\min_{z\in
\lf\{aM+b,\dots,N\ri\}}\xi_{z,aM+b}
\end{equation}
for $b=1$ and $a=0,\dots,|s|-1$, and
\con{\begin{equation}\label{eq:DirectLemReducingMinProblem1}
A \lf(a\cdot M+b,l \ri)= \lf(N-b+1\ri)\min_{z\in
\lf\{aM+b,\dots,N\ri\}}\xi_{z,aM+b}+\sum_{i=1}^{\min
\lf(M-l-1,b-1\ri)}\min_ {z\in\lf\{aM+b-i,\dots,N\ri\}}\xi_{z,aM+b}
\end{equation}}{\bal{eq:DirectLemReducingMinProblem1}
{&A \lf(a\cdot M+b,l \ri)=&
\nn\\
&\lf(N-b+1\ri)\min_{z\in
\lf\{aM+b,\dots,N\ri\}}\xi_{z,aM+b}+&
\nn\\
&\sum_{i=1}^{\min
\lf(M-l-1,b-1\ri)}\min_ {z\in\lf\{aM+b-i,\dots,N\ri\}}\xi_{z,aM+b}&}}
for $b=2,\dots,M$ and $a=0,\dots, |s|-1$. For $|s|\cdot M+1\le j\le
K\cdot M$ and $1\le i\le N$, we get that $\xi_{i,j}$ occurs only in
the term $\sum_{i=1}^{N}\sum_{j=1}^{K\cdot M}\xi_{i,j}$ in
\eqref{eq:DirectLemReducingMinProblem0}, where $\xi_{i,j}\ge 0$.
Thus, the minimization is obtained for
\begin{equation}\label{eq:DirectLemReducingMinProblem2}
\xi_{i,j}=0\qquad |s|\cdot M+1\le j\le K\cdot M,\ \ 1\le i\le N.
\end{equation}
Therefore, we can rewrite the optimization problem
\con{\begin{equation}\label{eq:DirectLemReducingMinProblem3}
\min_{\xi_{\udl{i},\udl{j}}\in\mathcal{A}}\lf(|s|\cdot \lf(MN-l
\lf(l+1\ri)-\lf(N+M-1-2l\ri)r_{max}\ri)-\sum_{i=1}^{|s|M}A
\lf(i,l\ri)\ri)^{+}+\sum_{i=1}^{N}\sum_{j=1}^{|s|\cdot M}\xi_{i,j}.
\end{equation}}{\bal{eq:DirectLemReducingMinProblem3}
{&\min_{\xi_{\udl{i},\udl{j}}\in\mathcal{A}}\lf(|s|\cdot \lf(MN-l
\lf(l+1\ri)-\lf(N+M-1-2l\ri)\ri.\ri.&
\nn\\
&\lf.\lf.r_{max}\ri)-\sum_{i=1}^{|s|M}A
\lf(i,l\ri)\ri)^{+}+\sum_{i=1}^{N}\sum_{j=1}^{|s|\cdot M}\xi_{i,j}.&}}

We now wish to show that $\xi_{i,j}=0$ for $j=1,\dots, |s|\cdot M$
and $i=1,\dots,j-1$. Essentially, we show for $i<j$ that reducing
$\xi_{i,j}$ affects \eqref{eq:DirectLemReducingMinProblem3} more
than $-\min_{z\in \lf\{i,\dots,N\ri\}}\xi_{z,j}$ does. First let us
observe $\xi_{i,a\cdot M+b}$ for $i=1,\dots,a\cdot M+b-\min
\lf(M-l-1,b-1\ri)-1$, where $a=0,\dots, |s|-1$ , $b=1,\dots,M$. Note
that this values do not have any representation in $A \lf(a\cdot
M+b,l\ri)$. Therefore, they do not affect $\lf(\cdot\ri)^{+}$ and
only affect $\sum_{i=1}^{N}\sum_{j=1}^{|s|\cdot M}\xi_{i,j}$. Thus,
in order to obtain the minimum we must choose
\begin{equation*}
\xi_{i,a\cdot M+b}=0\quad i=1,\dots,a\cdot M+b-\min
\lf(M-l-1,b-1\ri)-1
\end{equation*}
for any $a=0,\dots, |s|-1$ and $b=1,\dots,M$. Note that the function
in \eqref{eq:DirectLemReducingMinProblem3} is continues. In the case
$\lf(\cdot\ri)^{+}=0$ the function in
\eqref{eq:DirectLemReducingMinProblem3} can be written as
\begin{equation}\label{eq:DirectLemReducingMinProblem4}
\sum_{a=0}^{|s|-1}\sum_{b=1}^{M}\sum_{i=a\cdot M+b-\min
\lf(M-l-1,b-1\ri)}^{N}\xi_{i,a\cdot M+b}
\end{equation}
In this case as long as $\lf(\cdot\ri)^{+}=0$ reducing
$\xi_{i,a\cdot M+b}$ for $a\cdot M+b- \min \lf(M-l-1,b-1\ri)\le i\le
a\cdot M+b-1$ and $a=0\dots,|s|-1$, $b=2,\dots,M$ also reduces
\eqref{eq:DirectLemReducingMinProblem4}. For $\lf(\cdot\ri)^{+}>0$
\eqref{eq:DirectLemReducingMinProblem3} can be written as
\con{\begin{align}\label{eq:DirectLemReducingMinProblem5}
|s|\cdot \lf(MN-l
\lf(l+1\ri)-\lf(N+M-1-2l\ri)r_{max}\ri)+\sum_{a=0}^{|s|-1}\sum_{b=2}^{M}\sum_{i=1}^{\min
\lf(M-l-1,b-1\ri)} \lf(\xi_{a\cdot M+b-i,a\cdot M+b}-\min_{z\in
\lf\{a\cdot M+b-i,\dots,N\ri\}}\xi_{z,a\cdot M+b}\ri)\nonumber\\
+\sum_{a=0}^{|s|-1}\sum_{b=1}^{M}\lf(\sum_{z=a\cdot
M+b}^{N}\xi_{z,a\cdot M+b}-\lf(N-b+1\ri)\min_{z\in \lf\{a\cdot
M+b,\dots,N\ri\}}\xi_{z,a\cdot M+b}\ri).
\end{align}}{\bal{eq:DirectLemReducingMinProblem5}
{&|s|\cdot \lf(MN-l
\lf(l+1\ri)-\lf(N+M-1-2l\ri)r_{max}\ri)+&
\nn\\
&\sum_{a=0}^{|s|-1}\sum_{b=2}^{M}\sum_{i=1}^{\min
\lf(M-l-1,b-1\ri)} \lf(\xi_{a\cdot M+b-i,a\cdot M+b}-\ri.&
\nn\\
&\lf.\min_{z\in
\lf\{a\cdot M+b-i,\dots,N\ri\}}\xi_{z,a\cdot M+b}\ri)&
\nn\\
&+\sum_{a=0}^{|s|-1}\sum_{b=1}^{M}\lf(\sum_{z=a\cdot
M+b}^{N}\xi_{z,a\cdot M+b}-\ri.&
\nn\\
&\lf.\lf(N-b+1\ri)\min_{z\in \lf\{a\cdot
M+b,\dots,N\ri\}}\xi_{z,a\cdot M+b}\ri).&}}
Since $\xi_{a\cdot M+b-i,a\cdot M+b}\ge \min_{z\in \lf\{a\cdot
M+b-i,\dots,N\ri\}}\xi_{z,a\cdot M+b}$, reducing $\xi_{a\cdot
M+b-i,a\cdot M+b}$ also reduces
\eqref{eq:DirectLemReducingMinProblem5}. Since the function is
continues, considering these two cases is sufficient in order to
state that the minimum is obtained when
\begin{equation}\label{eq:DirectLemReducingMinProblem6}
\xi_{i,j}=0\quad j=1,\dots,|s|\cdot M, \ i=1\dots,j-1.
\end{equation}
This is due to the fact that for any value of $\xi_{z,a\cdot M+b}\ge
0$, $a=0,\dots,|s|-1$, $b=1,\dots,M$ and $z=a\cdot M+b,\dots, N$ the
terms in
\eqref{eq:DirectLemReducingMinProblem4},\eqref{eq:DirectLemReducingMinProblem5}
are reduced when decreasing $\lf\{\xi_{a\cdot M+b-i,a\cdot
M+b}\ri\}_{i=1}^{\min \lf(M-l-1,b-1\ri)}$, and also since the
function is continues. Note that from
\eqref{eq:DirectLemReducingMinProblem5} we can see that decreasing
$\sum_{z=a\cdot M+b}^{N}\xi_{z,a\cdot M+b}$ does not necessarily
decrease the function. This is due to the fact that $N-b+1\ge
N-\lf(a\cdot M+b\ri)+1$, and so the contribution of
$\lf(N-b+1\ri)\min_{z\in \lf\{a\cdot M+b,\dots,N\ri\}}\xi_{z,a\cdot
M+b}$ may be more significant than $\sum_{z=a\cdot
M+b}^{N}\xi_{z,a\cdot M+b}$.

Based on \eqref{eq:DirectLemReducingMinProblem6} we can rewrite the
function in the following manner
\con{\begin{equation}\label{eq:DirectLemReducingMinProblem7}
\lf(|s|\cdot \lf(MN-l
\lf(l+1\ri)-\lf(N+M-1-2l\ri)r_{max}\ri)-\sum_{a=0}^{|s|-1}\sum_{b=1}^{M}\lf(N-b+1\ri)\min_{z\in
\lf\{a\cdot M+b,\dots,N\ri\}}\xi_{z,a\cdot
M+b}\ri)^{+}+\sum_{a=0}^{|s|-1}\sum_{b=1}^{M}\sum_{z=a\cdot
M+b}^{N}\xi_{z,a\cdot M+b}.
\end{equation}}{\bal{eq:DirectLemReducingMinProblem7}
{&\lf(|s|\cdot \lf(MN-l
\lf(l+1\ri)-\lf(N+M-1-2l\ri)r_{max}\ri)-\ri.&
\nn\\
&\lf.\sum_{a=0}^{|s|-1}\sum_{b=1}^{M}\lf(N-b+1\ri)\min_{z\in
\lf\{a\cdot M+b,\dots,N\ri\}}\xi_{z,a\cdot
M+b}\ri)^{+}&
\nn\\
&+\sum_{a=0}^{|s|-1}\sum_{b=1}^{M}\sum_{z=a\cdot
M+b}^{N}\xi_{z,a\cdot M+b}.&}}
From \eqref{eq:DirectLemReducingMinProblem7} we can see that the
minimum is obtained when
\begin{equation}
\xi_{z,a\cdot M+b}=\alpha_{a\cdot M+b}\quad a\cdot M+b\le z\le N
\end{equation}
for $a=0,\dots,|s|-1$, $b=1,\dots, M$. This is due to the fact that
when the values are not equal, reducing the values to the minimal
value will reduce $\sum_{z=a\cdot M+b}^{N}\xi_{z,a\cdot M+b}$ while
not changing $\min_{z\in \lf\{a\cdot M+b,\dots,N\ri\}}\xi_{z,a\cdot
M+b}$. Therefore, we can write
\eqref{eq:DirectLemReducingMinProblem7} as follows
\con{\begin{equation}
\lf(|s|\cdot \lf(MN-l
\lf(l+1\ri)-\lf(N+M-1-2l\ri)r_{max}\ri)-\sum_{a=0}^{|s|-1}\sum_{b=1}^{M}\lf(N-b+1\ri)\alpha_{a\cdot
M+b}\ri)^{+}+\sum_{a=0}^{|s|-1}\sum_{b=1}^{M} \lf(N-\lf(a\cdot
M+b\ri)+1\ri)\alpha_{a\cdot M+b}
\end{equation}}{\bal{}{&\lf(|s|\cdot \lf(MN-l
\lf(l+1\ri)-\lf(N+M-1-2l\ri)r_{max}\ri)\ri.&
\nn\\
&\lf.-\sum_{a=0}^{|s|-1}\sum_{b=1}^{M}\lf(N-b+1\ri)\alpha_{a\cdot
M+b}\ri)^{+}+&
\nn\\
&\sum_{a=0}^{|s|-1}\sum_{b=1}^{M} \lf(N-\lf(a\cdot
M+b\ri)+1\ri)\alpha_{a\cdot M+b}&}}
where $0\le \alpha_{i}\le K\cdot M\cdot N$, $i=1,\dots,|s|\cdot M$.

We wish to show that the minimum is obtained for
\con{\begin{equation*}
\sum_{a=0}^{|s|-1}\sum_{b=1}^{M}\lf(N-b+1\ri)\alpha_{a\cdot
M+b}=|s|\cdot \lf(MN-l \lf(l+1\ri)-\lf(N+M-1-2l\ri)r_{max}\ri).
\end{equation*}}{\baln{}{&\sum_{a=0}^{|s|-1}\sum_{b=1}^{M}\lf(N-b+1\ri)\alpha_{a\cdot
M+b}=&
\nn\\
&|s|\cdot \lf(MN-l \lf(l+1\ri)-\lf(N+M-1-2l\ri)r_{max}\ri).&
}}
Again, note that the function is continues. For
$\lf(\cdot\ri)^{+}=0$ we get
\begin{equation}\label{eq:DirectLemReducingMinProblem8}
\sum_{a=0}^{|s|-1}\sum_{b=1}^{M} \lf(N-\lf(a\cdot
M+b\ri)+1\ri)\alpha_{a\cdot M+b}.
\end{equation}
This is attained for
$\sum_{a=0}^{|s|-1}\sum_{b=1}^{M}\lf(N-b+1\ri)\alpha_{a\cdot
M+b}\ge|s|\cdot \lf(MN-l \lf(l+1\ri)-\lf(N+M-1-2l\ri)r_{max}\ri)$.
Evidently for this case the minimal values occur at
$\sum_{a=0}^{|s|-1}\sum_{b=1}^{M}\lf(N-b+1\ri)\alpha_{a\cdot
M+b}=|s|\cdot \lf(MN-l \lf(l+1\ri)-\lf(N+M-1-2l\ri)r_{max}\ri)$. On
the other hand for $\lf(\cdot\ri)^{+}>0$ we get
\con{\begin{equation}\label{eq:DirectLemReducingMinProblem9}
|s|\cdot \lf(MN-l
\lf(l+1\ri)-\lf(N+M-1-2l\ri)r_{max}\ri)-\sum_{a=0}^{|s|-1}\sum_{b=1}^{M}\lf(a\cdot
M\ri)\alpha_{a\cdot M+b}.
\end{equation}}{\bal{eq:DirectLemReducingMinProblem9}
{&|s|\cdot \lf(MN-l
\lf(l+1\ri)-\lf(N+M-1-2l\ri)r_{max}\ri)-&
\nn\\
&\sum_{a=0}^{|s|-1}\sum_{b=1}^{M}\lf(a\cdot
M\ri)\alpha_{a\cdot M+b}.&}}
Hence increasing $\sum_{a=0}^{|s|-1}\sum_{b=1}^{M}\lf(a\cdot
M\ri)\alpha_{a\cdot M+b}$ decreases the function as long as
$\lf(\cdot\ri)^{+}>0$ which means
\con{\begin{equation*}
\sum_{a=0}^{|s|-1}\sum_{b=1}^{M}\lf(N-b+1\ri)\alpha_{a\cdot
M+b}<|s|\cdot \lf(MN-l\lf(l+1\ri)-\lf(N+M-1-2l\ri)r_{max}\ri).
\end{equation*}}{\baln{}{&\sum_{a=0}^{|s|-1}\sum_{b=1}^{M}\lf(N-b+1\ri)\alpha_{a\cdot
M+b}<&
\nn\\
&|s|\cdot \lf(MN-l\lf(l+1\ri)-\lf(N+M-1-2l\ri)r_{max}\ri).&}}
Hence, based on the fact that the function is continues we get again
that for this case the minimal values occur at
\con{\begin{equation*}
\sum_{a=0}^{|s|-1}\sum_{b=1}^{M}\lf(N-b+1\ri)\alpha_{a\cdot
M+b}=|s|\cdot \lf(MN-l \lf(l+1\ri)-\lf(N+M-1-2l\ri)r_{max}\ri).
\end{equation*}}{\baln{}{&\sum_{a=0}^{|s|-1}\sum_{b=1}^{M}\lf(N-b+1\ri)\alpha_{a\cdot
M+b}=&
\nn\\
&|s|\cdot \lf(MN-l \lf(l+1\ri)-\lf(N+M-1-2l\ri)r_{max}\ri).&}}

The event
$\sum_{a=0}^{|s|-1}\sum_{b=1}^{M}\lf(N-b+1\ri)\alpha_{a\cdot
M+b}=|s|\cdot \lf(MN-l \lf(l+1\ri)-\lf(N+M-1-2l\ri)r_{max}\ri)$,
where $\alpha_{i}\ge 0$, $i=1,\dots,|s|\cdot M$, is within the range
$0\le \alpha_{i}\le K\cdot M\cdot N$, $i=1,\dots,|s|\cdot M$. This
is because in order to fulfil the equality we get
\con{\begin{equation*}
\max \lf(\alpha_{1},\dots,\alpha_{|s|\cdot M}\ri)\le \frac{|s|\cdot
\lf(MN-l \lf(l+1\ri)-\lf(N+M-1-2l\ri)r_{max}\ri)}{N-b+1}\le K\cdot
M\cdot N.
\end{equation*}}{\baln{}{&\max \lf(\alpha_{1},\dots,\alpha_{|s|\cdot M}\ri)\le&
\nn\\
&\frac{|s|\cdot\lf(MN-l \lf(l+1\ri)-\lf(N+M-1-2l\ri)r_{max}\ri)}{N-b+1}\le K\cdot
M\cdot N.&}}
Therefore, the minimization problem solution is obtained for
\begin{equation*}
\min_{\udl{\alpha}\in\mathcal{A}^{'}}\sum_{a=0}^{|s|-1}\sum_{b=1}^{M}
\lf(N-\lf(a\cdot M+b\ri)+1\ri)\alpha_{a\cdot M+b}
\end{equation*}
where the set $\mathcal{A}^{'}$ is defined by the following two
conditions: $0\le \alpha_{i}\le K\cdot M\cdot N$,
$i=1,\dots,|s|\cdot M$, and
\con{\begin{equation*}
\sum_{a=0}^{|s|-1}\sum_{b=1}^{M}\lf(N-b+1\ri)\alpha_{a\cdot
M+b}=|s|\cdot \lf(MN-l \lf(l+1\ri)-\lf(N+M-1-2l\ri)r_{max}\ri).
\end{equation*}}{\baln{}{&\sum_{a=0}^{|s|-1}\sum_{b=1}^{M}\lf(N-b+1\ri)\alpha_{a\cdot
M+b}=&
\nn\\
&|s|\cdot \lf(MN-l \lf(l+1\ri)-\lf(N+M-1-2l\ri)r_{max}\ri).&}}

\section{Proof of Lemma \ref{lem:DirectLowerBoundDiversityOrderBetaLowerBound}}\label{Append:DirectLowerBoundDiversityOrderBetaLowerBound}
We begin by analyzing the case $a=|s|-1$ and $b=M$. For this case
let us consider $N=\lf(|s|+1\ri)M-1$. In this case we get
\begin{equation}\label{eq:DirectLemLowerBoundDiversityOrderBetaLowerBound1}
\frac{|s|\lf(N-|s|\cdot M+1\ri)}{N-M+1}=\frac{|s|\lf(M\ri)}{|s|M}=
1.
\end{equation}
Note that for $c\ge d\ge 0$ and $x_{2}>x_{1}\ge c$ we get
\begin{equation}\label{eq:DirectLemLowerBoundDiversityOrderBetaLowerBound2}
\frac{x_{2}-c}{x_{2}-d}\ge\frac{x_{1}-c}{x_{1}-d}.
\end{equation}
Hence, based on
\eqref{eq:DirectLemLowerBoundDiversityOrderBetaLowerBound2},
\eqref{eq:DirectLemLowerBoundDiversityOrderBetaLowerBound1}, we get
for $N>\lf(|s|+1\ri)M-1$
\begin{equation}\label{eq:DirectLemLowerBoundDiversityOrderBetaLowerBound3}
\frac{|s|\lf(N-\lf(|s|\cdot M-1\ri)\ri)}{N-\lf(M-1\ri)}\ge
\frac{|s|\lf(M\ri)}{|s|M}= 1.
\end{equation}

So far we have proved the lemma for $a=|s|-1$, $b=M$ and $N\ge
\lf(|s|+1\ri)M-1$. For the general case we consider
$\frac{|s|\lf(N-\lf(a\cdot M+b-1\ri)\ri)}{N-\lf(b-1\ri)}$. In this
case we get
\con{\begin{equation}\label{eq:DirectLemLowerBoundDiversityOrderBetaLowerBound4}
\frac{|s|\lf(N-\lf(a\cdot
M+b-1\ri)\ri)}{N-\lf(b-1\ri)}=|s|\frac{\lf(N+|s|M-a\cdot
M-b\ri)-\lf(|s|M-1\ri)}{\lf(N+M-b\ri)-\lf(M-1\ri)}\ge
|s|\frac{\lf(N+|s|M-a\cdot M-b\ri)-\lf(|s|M-1\ri)}{\lf(N+|s|M-a\cdot
M-b\ri)-\lf(M-1\ri)}
\end{equation}}{\bal{eq:DirectLemLowerBoundDiversityOrderBetaLowerBound4}
{&\frac{|s|\lf(N-\lf(a\cdot
M+b-1\ri)\ri)}{N-\lf(b-1\ri)}=&
\nn\\
&|s|\frac{\lf(N+|s|M-a\cdot
M-b\ri)-\lf(|s|M-1\ri)}{\lf(N+M-b\ri)-\lf(M-1\ri)}\ge&
\nn\\
&|s|\frac{\lf(N+|s|M-a\cdot M-b\ri)-\lf(|s|M-1\ri)}{\lf(N+|s|M-a\cdot
M-b\ri)-\lf(M-1\ri)}&}}
where the inequality results from the fact that $M-b\le |s|M-a\cdot
M-b$. From
\eqref{eq:DirectLemLowerBoundDiversityOrderBetaLowerBound2} and
\eqref{eq:DirectLemLowerBoundDiversityOrderBetaLowerBound3} we get
that
\con{\begin{equation}\label{eq:DirectLemLowerBoundDiversityOrderBetaLowerBound5}
|s|\frac{\lf(N+|s|M-a\cdot M-b\ri)-\lf(|s|M-1\ri)}{\lf(N+|s|M-a\cdot
M-b\ri)-\lf(M-1\ri)}\ge |s|\frac{N-\lf(|s|M-1\ri)}{N-\lf(M-1\ri)}\ge
1.
\end{equation}}{\bal{eq:DirectLemLowerBoundDiversityOrderBetaLowerBound5}
{&|s|\frac{\lf(N+|s|M-a\cdot M-b\ri)-\lf(|s|M-1\ri)}{\lf(N+|s|M-a\cdot
M-b\ri)-\lf(M-1\ri)}\ge&
\nn\\
&|s|\frac{N-\lf(|s|M-1\ri)}{N-\lf(M-1\ri)}\ge
1.&}}
From \eqref{eq:DirectLemLowerBoundDiversityOrderBetaLowerBound4},
\eqref{eq:DirectLemLowerBoundDiversityOrderBetaLowerBound5} we get
the proof of the lemma also for any $a=0,\dots,|s|-1$ and
$b=1,\dots,M$. This concludes the proof.

\section{Proof of Theorem \ref{Th:LowerBoundDiversityOrderLattices}}\label{Append:LatticesDiversityOrder}
We prove that there exists $K$ sequences of $2\cdot D_{l}\cdot
T_{l}$-real dimensional lattices (as a function of $\rho$) that
attains the optimal DMT for $N\ge \lf(K+1\ri)M-1$. We rely on the
extension of the \emph{Minkowski-Hlawaka} Theorem to the
multiple-access channel presented in \cite[Theorem
2]{NamElGamalOptimalDMT2007}. We upper bound the error probability
of the ensemble of lattices for each channel realization, and
average the upper bound over all channel realizations to obtain the
optimal DMT.

We consider $K$ ensembles of $2\cdot D_{l}\cdot T_{l}$-real
dimensional lattices, one for each user, transmitted using
$G_{l}^{\lf(1,\dots,K\ri)}$ defined in
\ref{subsec:TheTransmissionScheme}. For user $i$, the first
$D_{l}\cdot T_{l}$ dimensions of the lattice are spread on the real
part of the non-zero entries of $G_{l}^{\lf(i\ri)}$, and the other
$D_{l}\cdot T_{l}$ dimensions of the lattice on the imaginary part
of the non-zero entries of $G_{l}^{\lf(i\ri)}$. The volume of the
Voronoi region of the lattice of user $i$ equals
$V_{f}^{\lf(i\ri)}=\lf(\gamma_{tr}^{\lf(i\ri)}\ri)^{-1}=\rho^{-r_{i}T_{l}}$,
i.e., multiplexing gain $r_{i}$. Since the users are distributed, the
effective lattice at the transmitter can be written as
$\Lambda_{tr}=\Lambda_{1}\times\Lambda_{2}\times\dots\times\Lambda_{K}$,
where $\Lambda_{i}$ is the lattice transmitted by user $i$. At the
receiver the channel induces a new lattice
$H_{eff}^{(l),K}\cdot\udl{x}^{'}$, where
$\udl{x}^{'}\in\Lambda_{tr}$.  For lattices with regular lattice
decoding, the error probability is equal among all codewords. Hence,
it is sufficient to analyze the lattice's zero codeword error
probability. Without loss of generality let us assume that the
receiver rotates $\udl{y}_{\ex}$ such that the channel can be
rewritten as
\begin{equation}
\udl{y}_{\ex}=B\cdot\udl{x}+\tilde{\udl{n}}_{\mathrm{ex}}
\end{equation}
where
$B^{\dagger}B=H_{eff}^{\lf(l\ri),K\dagger}H_{eff}^{\lf(l\ri),K}$,
and $\tilde{\udl{n}}_{\ex}\sim
\CN(\underline{0},\rho^{-1}\cdot\frac{2}{2\pi e}\cdot I_{K\cdot
D_{l}\cdot T_{l}})$.

We define the indication function of a $2\cdot K\cdot D_{l}\cdot
T_{l}$ dimensional ball with radius $2R$ centered around zero by
$$I_{Ball(2R)}(\udl{x})=\left\{ \begin{array}{ll}
1, & \lv\udl{x}\rv\le 2R\\
0, & else
\end{array}\right..$$
In addition let us define the continues function of bounded support
$f_{rc}(\udl{x})=I_{{Ball(2R_{\eff})}}(\udl{x})\cdot Pr(\rv
\udl{\tilde{n}}_{\ex}\lv>\rv \udl{x}-\udl{\tilde{n}}_{\ex}\lv)$.
Based on \eqref{eq:DirectUpperBoundErrorProb2} we can state that for
each lattice induced at the receiver, $\Lambda_{\rc}$, the lattice
zero codeword error probability is upper bounded by
\begin{equation}\label{eq:LatticeUpperBoundErrorProb}
\sum_{\udl{x}\in\Lambda_{\rc},\udl{x}\neq
0}f_{\mathrm{rc}}(\udl{x})+Pr(\lVert\underline{\tilde{n}}_{\mathrm{ex}}\rVert\ge
R_{\eff}).
\end{equation}
where
$\frac{R_{\eff}^{2}}{2K_{l}T_{l}\sigma^{2}}=\mu_{rc}=\rho^{1-\frac{\sum_{i=1}^{K}r_{i}}{K\cdot
D_{l}}}\cdot |H_{eff}^{\lf(l\ri),K\dagger}\cdot
H_{eff}^{\lf(l\ri),K}|^{\frac{1}{K\cdot D_{l}}}$. For regular
lattice decoding we can equivalently consider
\begin{equation}\label{eq:TheEffectiveChannelModelWithTheTransmitter}
\udl{y}_{\ex}^{'}=B^{-1}\cdot\udl{y}_{\ex}=\udl{x}+\hat{\udl{n}}_{\ex}.
\end{equation}
where $\udl{\hat{n}}_{\ex}\sim
\CN\big(0,(H_{\eff}^{(l),K\dagger}H_{\eff}^{(l),K})^{-1}\big)$, i.e.,
the lattice at the receiver remains $\Lambda_{tr}$ and the affect of
the channel realization is passed on to the additive noise. In
addition let us denote an indication function over an ellipse
centered around zero by
$$I_{ellipse(B,2R)}(\udl{x})=\left\{ \begin{array}{ll}
1, & \rv B\cdot\udl{x}\lv\le 2R\\
0, & else
\end{array}\right.,$$
By defining the continues function
$g_{rc}(\udl{x})=I_{{ellipse(B,2R_{\eff})}}(\udl{x})\cdot Pr\big(\rv
B\udl{\hat{n}}_{\ex}\lv> \rv B(\udl{x}-\udl{\hat{n}}_{\ex})\lv\big)$
we get the following upper bound for the error probability
\begin{equation}\label{eq:LatticeTransmiterEnsembleUpperBoundErrorProbCrudeForm}
\sum_{\udl{x}\in\Lambda_{\tr},\udl{x}\neq
0}g_{rc}(\udl{x})+Pr(\lVert
B\cdot\underline{\hat{n}}_{\mathrm{ex}}\rVert\ge R_{\eff})
\end{equation}
that equals to the upper bound in
\eqref{eq:LatticeUpperBoundErrorProb}. In addition, since
$f_{rc}\lf(B\cdot\udl{x}\ri)=g_{rc}\lf(\udl{x}\ri)$, and based on
the fact that $H_{eff}^{\lf(l\ri),K}$ is a block diagonal matrix we
get
\con{\begin{equation}\label{eq:LaticeEnsembleEqualityBetweenExpressions}
|H_{eff}^{\lf(l\ri),\lf(S\ri)\dagger}H_{eff}^{\lf(l\ri),\lf(S\ri)}|^{-1}\cdot\int_{\udl{x}\in\mathbb{R}^{2\cdot
|S|\cdot D_{l}\cdot
T_{l}}}f_{rc}\lf(\udl{x}^{\lf(S\ri)}\ri)d\udl{x}^{\lf(S\ri)}=\int_{\udl{x}\in\mathbb{R}^{2\cdot
|S|\cdot D_{l}\cdot
T_{l}}}g_{rc}\lf(\udl{x}^{\lf(S\ri)}\ri)d\udl{x}^{\lf(S\ri)}\quad\forall
S\subseteq \lf\{1,\dots,K\ri\}
\end{equation}}{\bal{eq:LaticeEnsembleEqualityBetweenExpressions}
{&|H_{eff}^{\lf(l\ri),\lf(S\ri)\dagger}H_{eff}^{\lf(l\ri),\lf(S\ri)}|^{-1}\cdot\int_{\udl{x}\in\mathbb{R}^{2\cdot
|S|\cdot D_{l}\cdot
T_{l}}}f_{rc}\lf(\udl{x}^{\lf(S\ri)}\ri)d\udl{x}^{\lf(S\ri)}=&
\nn\\
&\int_{\udl{x}\in\mathbb{R}^{2\cdot
|S|\cdot D_{l}\cdot
T_{l}}}g_{rc}\lf(\udl{x}^{\lf(S\ri)}\ri)d\udl{x}^{\lf(S\ri)}\quad\forall
S\subseteq \lf\{1,\dots,K\ri\}&}}
where $\udl{x}^{\lf(S\ri)}$ equals zero in the entries corresponding
to $\lf\{1,\dots,K\ri\}\setminus S$ and the other entries are in
$\mathbb{R}^{2\cdot |S|\cdot D_{l}\cdot T_{l}}$.

In \cite[Theorem 2]{NamElGamalOptimalDMT2007} Nam and El Gamal
extended the Minkowski-Hlawka theorem to the multiple-access channel
by using Loeliger ensembles of lattices
\cite{LoeligerAveragingBounds} for each user. From this theorem we
get that for a certain Riemann integrable function of bounded
support $f \lf(\udl{x}\ri)$
\con{\begin{equation}\label{eq:LatticeTransmiterEnsembleMinkHlwForMAC}
E_{\Lambda_{tr}} \lf(\sum_{\udl{x}\in\Lambda_{tr},\udl{x}\neq 0}f
\lf(\udl{x}\ri)\ri)=\sum_{S\subseteq \lf\{1,\dots,K\ri\}}\prod_{s\in
S}\frac{1}{V_{f}^{\lf(s\ri)}}\int_{\udl{x}^{\lf(S\ri)}\in\mathbb{R}^{2\cdot
|S|\cdot D_{l}\cdot T_{l}}}f \lf(\udl{x}^{\lf(S\ri)}\ri)d
\udl{x}^{\lf(S\ri)}.
\end{equation}}{\bal{eq:LatticeTransmiterEnsembleMinkHlwForMAC}
{&E_{\Lambda_{tr}} \lf(\sum_{\udl{x}\in\Lambda_{tr},\udl{x}\neq 0}f
\lf(\udl{x}\ri)\ri)=\sum_{S\subseteq \lf\{1,\dots,K\ri\}}\prod_{s\in
S}\times&
\nn\\
&\frac{1}{V_{f}^{\lf(s\ri)}}\int_{\udl{x}^{\lf(S\ri)}\in\mathbb{R}^{2\cdot
|S|\cdot D_{l}\cdot T_{l}}}f \lf(\udl{x}^{\lf(S\ri)}\ri)d
\udl{x}^{\lf(S\ri)}.&}}
For each channel realization $B$, the function
$g_{rc}\lf(\udl{x}\ri)$ is bounded, and so by averaging over the
Loeliger ensembles for the multiple-access channel, we get based on
\eqref{eq:LatticeTransmiterEnsembleUpperBoundErrorProbCrudeForm},
\eqref{eq:LatticeTransmiterEnsembleMinkHlwForMAC} that the upper
bound on the error probability using regular lattice decoding is
\con{\begin{equation}\label{eq:LatticeTransmiterEnsembleMinkHlwForMACElipseErrorUpBound}
\sum_{S\subseteq \lf\{1,\dots,K\ri\}}\prod_{s\in
S}\frac{1}{V_{f}^{\lf(s\ri)}}\int_{\udl{x}^{\lf(S\ri)}\in\mathbb{R}^{2\cdot
|S|\cdot D_{l}\cdot T_{l}}}g_{rc} \lf(\udl{x}^{\lf(S\ri)}\ri)d
\udl{x}^{\lf(s\ri)}+Pr(\lVert
B\cdot\underline{\hat{n}}_{\mathrm{ex}}\rVert\ge R_{\eff}).
\end{equation}}{\bal{eq:LatticeTransmiterEnsembleMinkHlwForMACElipseErrorUpBound}
{&\sum_{S\subseteq \lf\{1,\dots,K\ri\}}\prod_{s\in
S}\frac{1}{V_{f}^{\lf(s\ri)}}\int_{\udl{x}^{\lf(S\ri)}\in\mathbb{R}^{2\cdot
|S|\cdot D_{l}\cdot T_{l}}}g_{rc} \lf(\udl{x}^{\lf(S\ri)}\ri)d
\udl{x}^{\lf(s\ri)}&
\nn\\
&+Pr(\lVert
B\cdot\underline{\hat{n}}_{\mathrm{ex}}\rVert\ge R_{\eff}).&}}
By assigning the relation of
\eqref{eq:LaticeEnsembleEqualityBetweenExpressions} in \eqref{eq:LatticeTransmiterEnsembleMinkHlwForMACElipseErrorUpBound} we get
\con{\begin{equation}\label{eq:LaticeEnsembleEqualityIntegralFrc}
\sum_{S\subseteq \lf\{1,\dots,K\ri\}}\rho^{T_{l}\sum_{s\in
S}r_{s}} \cdot |H_{eff}^{\lf(l\ri),\lf(S\ri)\dagger}H_{eff}^{\lf(l\ri),\lf(S\ri)}|^{-1}\int_{\udl{x}^{\lf(S\ri)}\in\mathbb{R}^{2\cdot
|S|\cdot D_{l}\cdot T_{l}}}f_{rc} \lf(\udl{x}^{\lf(S\ri)}\ri)d
\udl{x}^{\lf(s\ri)}+Pr(\lVert
\underline{\tilde{n}}_{\mathrm{ex}}\rVert\ge R_{\eff}).
\end{equation}}{\bal{eq:LaticeEnsembleEqualityIntegralFrc}
{&\sum_{S\subseteq \lf\{1,\dots,K\ri\}}\rho^{T_{l}\sum_{s\in
S}r_{s}}\cdot |H_{eff}^{\lf(l\ri),\lf(S\ri)\dagger}H_{eff}^{\lf(l\ri),\lf(S\ri)}|^{-1}\times &
\nn\\
&\int_{\udl{x}^{\lf(S\ri)}\in\mathbb{R}^{2\cdot
|S|\cdot D_{l}\cdot T_{l}}}f_{rc} \lf(\udl{x}^{\lf(S\ri)}\ri)d
\udl{x}^{\lf(s\ri)}&
\nn\\
&+Pr(\lVert
\underline{\tilde{n}}_{\mathrm{ex}}\rVert\ge R_{\eff}).&}}
Based on the bounds derived in \cite[Theorem
3]{YonaFederICOptimalDMT}, we can upper bound the integral of the
first term in \eqref{eq:LaticeEnsembleEqualityIntegralFrc} by
\con{\begin{equation*}
\sum_{S\subseteq \lf\{1,\dots,K\ri\}}\frac{4^{|S|\cdot D_{l}\cdot
T_{l}}}{2e^{|S|\cdot D_{l}\cdot T_{l}}}\rho^{-T_{l}\lf(|S|\cdot
D_{l}-\sum_{s\in
S}r_{s}\ri)}|H_{eff}^{\lf(l\ri),\lf(S\ri)\dagger}H_{eff}^{\lf(l\ri),\lf(S\ri)}|^{-1}.
\end{equation*}}{\baln{}{&\sum_{S\subseteq \lf\{1,\dots,K\ri\}}\frac{4^{|S|\cdot D_{l}\cdot
T_{l}}}{2e^{|S|\cdot D_{l}\cdot T_{l}}}\rho^{-T_{l}\lf(|S|\cdot
D_{l}-\sum_{s\in
S}r_{s}\ri)}\times&
\nn\\
&|H_{eff}^{\lf(l\ri),\lf(S\ri)\dagger}H_{eff}^{\lf(l\ri),\lf(S\ri)}|^{-1}.&}}
Since we consider radius of $R_{eff}$, for large values of $\rho$
the second term in \eqref{eq:LaticeEnsembleEqualityIntegralFrc} is
negligible compared to the first term \cite[Theorem
3]{YonaFederICOptimalDMT}. Hence, the remaining step is calculating
the average over all channel realizations. We divide the average
into two ranges $\mathcal{A}$ and $\ol{\mathcal{A}}$ as depicted in
Theorem \ref{Th:DirectLowerBoundDiversityOrder}. For each channel
realizations in $\ol{\mathcal{A}}$ we upper bound the error
probability by one. As shown in Theorem
\ref{Th:DirectLowerBoundDiversityOrder}, the probability of
receiving channel realizations in this range has exponent that is
lower bounded by the optimal DMT. For channel realizations in
$\mathcal{A}$ we get that $g_{rc}\lf(\udl{x}\ri)$ has bounded support, and so
we can use the Minkowski-Hlawka theorem to get the upper bound in
\eqref{eq:LaticeEnsembleEqualityIntegralFrc}. This bound coincides
with the upper bound in Theorem \ref{Th:DirectLowerBoundDiversityOrder}
which was shown to obtain the optimal DMT. this concludes the proof.

\section{Proof of Corollary \ref{Cor:SequencesLattocesAttainOptimalDMT}}\label{Append:SequencesLattocesAttainOptimalDMT}
We first consider the symmetric case $r_{1}=\dots =r_{K}=r_{\max}$.
Similarly to \cite[Corollary 3]{YonaFederICOptimalDMT} we can state
that if a sequence of $K$ lattices attains diversity order $d$ for
symmetric multiplexing gain $r_{max}=0$, it also attains diversity
order
\begin{equation}\label{Eq:AppendSequencesLattocesAttainOptimalDMT}
d \lf(1-\frac{r_{max}}{D_{\lfloor r_{max}\rfloor}T_{\lfloor
r_{max}\rfloor}}\ri)
\end{equation}
for any symmetric multiplexing gain $0<r_{max}\le D_{\lfloor
r_{max}\rfloor}T_{\lfloor r_{max}\rfloor}$. This is due to the fact
that changing $r_{max}$ merely has the effect of scaling the
effective lattice at the receiver. From Theorem
\ref{Th:LowerBoundDiversityOrderLattices} we get that there exists a
sequence of $K$ lattices (one for each user) that attains for
symmetric multiplexing gain $r_{max}=l$ the optimal DMT
$d^{\ast,\lf(FC\ri)}_{M,N}\lf(l\ri)$, where $l=0,\dots, M-1$. In
this case we also get from
\eqref{Eq:AppendSequencesLattocesAttainOptimalDMT} and Theorem
\ref{Th:LowerBoundDiversityOrderLattices} that this sequence also
attains the optimal DMT $d^{\ast,\lf(FC\ri)}_{M,N}\lf(r_{max}\ri)$,
when the symmetric multiplexing gain is in the range $l\le
r_{max}\le l+1$.

Now consider for the same sequence of lattices a multiplexing gains
tuple $\lf(r_{1},\dots, r_{K}\ri)$ with $r_{max}$ as its maximal
multiplexing gain. The performance can only improve compared to the
symmetric case since some of the multiplexing gains of the users are
smaller than $r_{max}$. Since the DMT can not be any larger than
$d^{\ast,\lf(FC\ri)}_{M,N}\lf(r_{max}\ri)$, which is already
obtained in the symmetric case, we get that
$d^{\ast,\lf(FC\ri)}_{M,N}\lf(r_{max}\ri)$ is obtained by any
multiplexing gains tuple with $r_{max}$ as its maximal value.

\end{appendices}

\bibliographystyle{IEEEtran}
\bibliography{IEEEabrv,YairRef}

\end{document}